\documentclass[12pt,epsf,amssymb]{article}

\usepackage{graphicx}
\usepackage{amsfonts}
\usepackage{amsmath}
\usepackage[usenames]{color}
\usepackage{amssymb}
\usepackage{amsthm}
\usepackage[mathscr]{eucal} 
\usepackage{url}

\usepackage[all]{xy}

\def\N{\mathbb{N}}
\def\Z{\mathbb{Z}}
\def\Q{\mathbb{Q}}
\def\R{\mathbb{R}}
\def\C{\mathbb{C}}
\def\P{\mathbb{P}}

\begin{document}

\baselineskip 0.6cm
\newcommand{\vev}[1]{ \left\langle {#1} \right\rangle }
\newcommand{\bra}[1]{ \langle {#1} | }
\newcommand{\ket}[1]{ | {#1} \rangle }
\newcommand{\Dsl}{\mbox{\ooalign{\hfil/\hfil\crcr$D$}}}
\newcommand{\nequiv}{\mbox{\ooalign{\hfil/\hfil\crcr$\equiv$}}}
\newcommand{\nsupset}{\mbox{\ooalign{\hfil/\hfil\crcr$\supset$}}}
\newcommand{\nni}{\mbox{\ooalign{\hfil/\hfil\crcr$\ni$}}}
\newcommand{\nin}{\mbox{\ooalign{\hfil/\hfil\crcr$\in$}}}
\newcommand{\Slash}[1]{{\ooalign{\hfil/\hfil\crcr$#1$}}}
\newcommand{\EV}{ {\rm eV} }
\newcommand{\KEV}{ {\rm keV} }
\newcommand{\MEV}{ {\rm MeV} }
\newcommand{\GEV}{ {\rm GeV} }
\newcommand{\TEV}{ {\rm TeV} }

\def\diag{\mathop{\rm diag}\nolimits}
\def\tr{\mathop{\rm tr}}

\def\Spin{\mathop{\rm Spin}}
\def\SO{\mathop{\rm SO}}
\def\SU{\mathop{\rm SU}}
\def\U{\mathop{\rm U}}
\def\Sp{\mathop{\rm Sp}}
\def\SL{\mathop{\rm SL}}

\def\change#1#2{{\color{blue}#1}{\color{red} [#2]}\color{black}\hbox{}}

\theoremstyle{definition}
\newtheorem{thm}{Theorem}[subsection]
\newtheorem{defn}[thm]{Definition}
\newtheorem{notn}[thm]{Notation}
\newtheorem{exmpl}[thm]{Example}
\newtheorem{props}[thm]{Proposition}
\newtheorem{lemma}[thm]{Lemma}
\newtheorem{rmk}[thm]{Remark}
\newtheorem{anythng}[thm]{}
\newtheorem{cor}[thm]{Corollary}

\begin{titlepage}
  
\begin{flushright}
IPMU19-0191
\end{flushright}
 
 \vskip 1cm
 \begin{center}
  
  {\large \bf Modular Parametrization as Polyakov Path Integral: \\
---Cases with CM Elliptic Curves as Target Spaces---}
 
 \vskip 1.2cm
  
 Satoshi Kondo$^{1,2}$ and Taizan Watari$^2$
  
 \vskip 0.4cm
  
  {\it $^1$ Middle East Technical University, Northern Cyprus Campus, 
    99738 Kalkanli, Guzelyurt, Mersin 10, Turkey
  \\[2mm] 
     $^2$Kavli Institute for the Physics and Mathematics of the Universe, 
    University of Tokyo, Kashiwa-no-ha 5-1-5, 277-8583, Japan
    }
 
 \vskip 1.5cm
    
 \abstract{
For an elliptic curve $E$ over an abelian extension $k/K$ with CM by $K$ of
Shimura type, the L-functions of its $[k:K]$ Galois
representations are Mellin transforms of Hecke theta functions;
a modular parametrization (surjective map) from a modular curve to $E$
pulls back the 1-forms on $E$ to give the Hecke theta functions.
This article refines the study of our earlier work and shows that
certain class of chiral correlation functions in Type II string theory
with $[E]_\C$ ($E$ as real analytic manifold) as a target space yield 
the same Hecke theta functions as objects on the modular curve.   
The K\"{a}hler parameter of
the target space $[E]_\C$ in string theory plays the role of the index
(partially ordered) set in defining the projective/direct limit of
modular curves.
 } 
 \end{center}
 \end{titlepage}

\tableofcontents

\newpage 
\section{Introduction}
\label{sec:Intro}

The $L$-function of an elliptic curve with complex multiplication 
defined over $\Q$, or over the imaginary quadratic field $K$ in question, 
is the Mellin transform of a Hecke theta function, or of the product
of them \cite{Deuring}. 
This statement has been generalized to multiple directions; one is 
to think of elliptic curves of Shimura type \cite{Shimura-AA}, 
and the other is to think of elliptic curves defined over $\Q$ 
not necessarily with complex multiplication \cite{Taylor-Wiles}.

It is not immediately clear why the inverse Mellin transform 
of the $L$-functions $f(\tau)$ come to have modular invariance.
Another question closely related to the one above is what 
the argument $\tau$ and the value $f(\tau)$ stand for; functions appearing 
in natural science describe how the output changes as an input changes, and 
the input/output often has substance such as time, energy, entropy, and 
scattering amplitude, not just an abstract number in $\Z$, $\R$, or $\C$. 
When one comes to have an idea about the substance in the argument $\tau$ 
and value $f(\tau)$ of those modular forms, one might also begin to have 
an idea where the modularity also comes from.   

In our previous paper \cite{prev.paper},\footnote{For an earlier attempt
in this direction, see \cite{Schimmrigk-05}.} 
we identified a class 
of observables in superstring theory with those Hecke theta functions, 
in the cases of elliptic curves of Shimura type. 
To be more specific, Ref. \cite{prev.paper} pointed out that 
chiral correlation functions of the form 
\begin{align}
  f_{1\Omega'}^{\rm II}(\tau_{ws};\beta) := {\rm Tr}_{V_\beta}^{\rm Rmnd}
   \left[ F e^{\pi i F} q_{ws}^{L_0 - \frac{c}{24}} (\partial X^\C) \right]
   \label{eq:the-CCF}
\end{align}
are equal to those Hecke theta functions $f(\tau)$ after appropriate linear 
combinations of those $f_{1\Omega'}^{\rm II}$ are formed and the argument
$\tau_{ws}$ is rescaled somewhat. The observation in \cite{prev.paper}
therefore hints that the modular nature of the inverse Mellin transform 
of the $L$-function may originate from (a part of) the ${\rm SL}(2;\Z)_{ws}$
transformation on the homology basis of the genus-one world sheet in 
string theory, the complex structure modulus $\tau_{ws}$ of the genus-one 
world sheet is somewhat the argument $\tau$ of the modular form in question, 
and the value of the modular forms are the chiral correlation functions 
in string theory. See the main text for notations.  

Missing the most in \cite{prev.paper} was why the specifically chosen 
chiral correlation functions (\ref{eq:the-CCF}) are relevant to the 
$L$-functions of elliptic curves. One can think of infinitely many 
observables in string/superstring theory even after a target space is 
chosen and fixed; we have seen in \cite{prev.paper} that the observables 
on the right hand side of (\ref{eq:the-CCF}) do reproduce the Hecke theta 
functions corresponding to the $L$-functions (after taking linear combinations 
and rescaling of the argument), but no explanation was given to why 
the specific choice of the observables (\ref{eq:the-CCF}) is natural 
or sensible for that purpose. It is one of the primary objectives 
of this article to explain why (\ref{eq:the-CCF}) is a natural choice.
Here is a digest version of the explanation.
In the theory of modular parametrization, a modular form $f(\tau)$
regarded as a 1-form $d\tau f(\tau)$ on a modular curve corresponds 
to the pull-back of a 1-form of an elliptic curve. The observables 
(\ref{eq:the-CCF}) also measure the same 1-form of the target-space 
elliptic curve (pulled-back to the genus-one world sheet of string theory) 
because of the operator $(\partial X^\C)$ in (\ref{eq:the-CCF});
see (\ref{eq:pulled-back-1forms}).
The modular invariance of the chiral correlation functions indicates 
that they are functions on a quotient of the upper complex half plane,  
${\cal H}/\Gamma$. The subgroup $\Gamma \subset {\rm SL}(2;\Z)_{ws}$
is determined by the monodromy representation associated with the 
rational model of $T^2$-target conformal field theory (CFT); 
see Prop. \ref{props:chi-corrl-II-lift}.  
In this way, both the 1-form $d\tau f(\tau)$ in the Langlands 
correspondence and the chiral correlation functions $f^{\rm II}_{1\Omega'}$
are regarded as objects on the modular curve ${\cal H}/\Gamma$, and 
are identified with each other; see a schematic diagram in 
(\ref{eq:schematic-triangle}). The idea is presented in 
section \ref{sec:lift}, and the triangle in (\ref{eq:schematic-triangle}) 
closes in Thm. \ref{thm:reproduce-LLdual-by-CCF}
after details are discussed.
 
This article also introduces improved/clarified understandings on 
a couple of things that have already been treated in \cite{prev.paper}. 
For example, 
the rescaling of the argument necessary in seeing $f^{\rm II}_{1\Omega'}(\tau_{ws})$ 
as a modular form of some $\Gamma_1(M)$ is now given by a simpler 
formula (\ref{eq:the-argument-rescaling}), which is in association 
with the isomorphism in Lemma \ref{lemma:isom-strokeOp-cuspformSpace}. 
Another improvement is in identifying the action of complex 
multiplication as a guiding principle\footnote{
In \cite{prev.paper}, we determined the linear combinations 
in the way Hecke theta functions are reproduced from $f^{\rm II}_{1\Omega'}$.
So, the result of Ref. \cite{prev.paper} should be read as follows: {\it there 
exist} appropriate linear combinations of $f^{\rm II}_{1\Omega'}$ 
{\it that become} the Hecke theta functions. In this article, however, 
we say that the linear combinations of $f^{\rm II}_{1\Omega'}$ that diagonalize 
the action of complex multiplication {\it are} Hecke theta functions if $h=1$, 
and {\it are} their building blocks if $h > 1$. } 
in determining choices of linear combinations of $f^{\rm II}_{1\Omega'}$, 
as explained in section \ref{sec:CCF-eigenform}. 

We also write down an observation that possible choices of the K\"{a}hler 
parameter on the target-space elliptic curve, which is necessary in 
formulating string theory but not in defining the $L$-function, can be 
regarded as a directed partially ordered set. By taking the direct limit 
of the vector space of chiral correlation functions with respect to this 
set of K\"{a}hler parameters, we arrive at the notion of the vector space 
of certain string-theory observables that does not depend on a specific choice 
of a K\"{a}hler parameter. This way of thinking fits very well with the 
subject that we deal with in this article. 
See sections \ref{ssec:KahlerVsShimuraC} and \ref{ssec:CG-D}.

In section \ref{ssec:CG-D}, we add one remark (\ref{statmnt:GTtheory-final}) 
on the relation between (a) the theory of modular parametrization and 
(b) Galois action on the monodromy representations of models of rational 
CFT; both (a) and (b) are placed within the Galois--Teichm\"{u}ller theory.  

Contents of section \ref{sec:CCF-eigenform} may also be regarded 
as how chiral correlation functions of various rational vertex operator 
algebras occupy the vector space of modular forms for various congruence 
subgroups of ${\rm SL}(2;\Z)$, and how Hecke operators act on those 
chiral correlation functions (of course Hecke operators act on modular forms). 
Section \ref{ssec:HTheta-TypeII-CCF} shows, in particular, that 
the Hecke operators bring the chiral correlation functions 
($f^{\rm II}_{1\Omega'}$'s) of one rational 
vertex operator algebra associated with one CM elliptic curve 
$([E_z]_\C, f_\rho)$ to linear combinations of the chiral correlation 
functions ($f^{\rm II}_{1\Omega'}$'s) of other rational vertex operator 
algebras, in general, but the action of the Hecke operators is closed 
within the set of those $f^{\rm II}_{1\Omega'}$'s associated with the elliptic 
curves with complex multiplication by the maximal order ${\cal O}_K$
forming a Galois orbit of size $h({\cal O}_K)$. So, the study in 
section \ref{sec:CCF-eigenform} is not far in spirit from the work 
\cite{Harvey} on characters of general models of rational CFT. 
 
{\bf Reading Guide}: 
Large fraction of the materials in 
sections \ref{ssec:string}--\ref{ssec:modC4CCF} are review on elementary 
things in vertex operator algebra, but a small number of 
ideas crucial in this article are scattered here and there; 
most of those important ideas are already referred to above, so 
one may jump directly to those places. We wish that the presentation 
in sections \ref{ssec:string}--\ref{ssec:modC4CCF}
are readable enough for number-theory experts. 
Section \ref{ssec:review-modular}, on the other hand, is a 
densely packed review on almost all we need about modular 
forms in this article; they are all textbook materials 
but we wish this review saves time for people in string-theory 
community. 

Materials in sections \ref{ssec:map2shimuraEC} 
and \ref{ssec:Grothendieck} are almost entirely a review
on things known in arithmetic geometry. 
The theory of modular parametrization has been explained in many places 
for elliptic curves defined over $\Q$, but there is no reference 
with a systematic exposition of the one for elliptic curves of 
Shimura type (except \cite{Wortmann} covering special cases referred 
to in \ref{statmnt:wortmann} and \ref{statmnt:wortmann2}). So, 
section \ref{ssec:map2shimuraEC} brings various known facts together 
from the literature, and add some statements (in 
sections \ref{sssec:surjective-morph-2BK} and \ref{sssec:Langlands}) 
not found in the literature, to provide a systematic one, and prepares 
for discussions in section \ref{ssec:KahlerVsShimuraC}. 
Section \ref{ssec:Grothendieck}, on the other hand, reviews some aspects 
of Galois--Teichm\"{u}ller theory. This material is included in this 
article only for the purpose of setting notations to be used in 
section \ref{ssec:CG-D}. 

An article \cite{prev.paper} by the authors has already provided
a concise review (in \S4.2) on Hecke characters of a number field and 
also arithmetic models of elliptic curves of Shimura type. We do not include 
those reviews in this article again, but they will be helpful for 
readers not familiar with those materials. Almost the same set of 
notations is being used in \cite{prev.paper} and in this article.   
 
This article has partially adopted a style of presentation with {\bf Theorem}, 
{\bf Remark}, etc., because that is useful in making contexts explicit, 
and also in referring to a specific result/logic. 
Because it often happens in physics that certain observations are more 
important than solid technical relations between well-defined objects, 
it is not necessarily meant in this article that Theorems are more important 
than Remarks, Statements, and Discussions.

\section{Modular Curves for a $T^2$-target RCFT Model}
\label{sec:lift}

\subsection{String Theory}
\label{ssec:string}

There are a couple of different versions (formulations) of string theory, 
such as bosonic string theory and Type II string theory; each of those versions
assigns\footnote{String theory does more than assigning a model of CFT 
to a set of data. String theory 
characterizes and/or constrains the kinds of data to which models of CFT 
can be assigned. Here, however, we will not pay attention to 
this aspect of string theory. It is well understood that at least 
a pair of data $([E]_\C, f_\rho)$ in the main text is well within 
the kinds of data to which a model of CFT is assigned, and we will 
focus our attention to the mathematical relation between i) the input 
geometry data $([E]_\C, f_\rho)$ for string theory, ii) the model 
of vertex operator algebra, and iii) arithmetic models of $[E]_\C$,
by using the geometry of modular curve as a hub among them.} 
a model of CFT to a set of data of geometry. Here, our primary interest 
is in a set of data of the form $([E]_\C, f_\rho)$, where $[E]_\C$ is 
a $\C$-isomorphism class of elliptic curves with complex multiplication, 
and $f_\rho$ a positive integer; bosonic string theory assigns a model 
of CFT to a set of data $([E]_\C,f_\rho)$. 

We begin by recalling basic facts about elliptic curves with complex 
multiplication.
First, let $z$ be a complex number in the upper half plane, $z \in {\cal H}$, 
and $\C/(\Z \oplus z\Z)$ be an analytic representation of an elliptic curve:
\begin{align}
  [E_z]_\C = \left\{ X^\C \in \C \right\} / \left( X^\C \sim X^\C + 1, \; 
  X^\C \sim X^\C + z \right) = \C/\mathfrak{b}_z; \qquad
     \mathfrak{b}_z := (\Z \oplus z\Z). 
\end{align}
The $\C$-isomorphism class of the elliptic curve specified by 
$z \in {\cal H}$ in this way is denoted by $[E_z]_\C$. Now, suppose that 
an elliptic curve $[E_z]_\C$ has complex multiplication. This extra condition 
is known to be equivalent to the presence of a set of integers $[a_z,b_z,c_z]$ 
prime to each other, $(a_z,b_z,c_z)=1$, such that 
\begin{align}
 a_zz^2+b_zz+c_z=0. 
\end{align}
The root $z \in {\cal H}$ of this equation is algebraic, and hence 
there is a quadratic extension field $K:=\Q(z)$ associated with this 
elliptic curve $[E_z]_\C$ with complex multiplication. 
The discriminant $D_K < 0$ of the degree-2 extension field $K$ and the 
discriminant of the quadratic equation are related by 
\begin{align}
  D_z := 4a_zc_z - b_z^2 = (-D_K)f_z^2, \qquad {}^\exists f_z \in \N_{>0}.
\end{align}
We have $K=\Q(\sqrt{D_K})$ because 
\begin{align}
   z = \frac{-b_z + f_z \sqrt{D_K}}{2a_z}. 
\end{align}

Algebraic integers in $K$ form a ring, which is denoted by ${\cal O}_K$. 
It is expressed, as an abelian group, as ${\cal O}_K = \Z + w_K \Z$ with 
$w_K = \sqrt{D_K}/2$ if $D_K \equiv 0$ mod 4, and with 
$w_K = (1+\sqrt{D_K})/2$ 
if $D_K \equiv -3$ mod 4. A subring of ${\cal O}_K$ in the form of 
$\Z + w_K f_z \Z$ with an integer $f_z \in \N_{>0}$ is denoted by 
${\cal O}_{f_z}$, and is said to be an {\it order} of $K$. 
When any element $\alpha \in {\cal O}_{f_z}$ is given, 
multiplication 
of $\alpha$ on the complex coordinate $X^\C$ maps the rank-2 lattice 
$\mathfrak{b}_z$ to $\mathfrak{b}_z$ (not necessarily one-to-one).
Conversely, the set of elements in $K$ that maps $\mathfrak{b}_z$ to 
$\mathfrak{b}_z$ is the ring ${\cal O}_{f_z}$. So, a choice of an elliptic 
curve $[E_z]_\C$ with complex multiplication specifies an imaginary quadratic 
field $K$, its order ${\cal O}_{f_z}$, and an (${\rm SL}(2;\Z)$ orbit of) 
algebraic number(s) $z \in {\cal H}$ (equivalently a fractional ideal\footnote{
$a_z \mathfrak{b}_z \subset {\cal O}_{f_z}$.} 
of ${\cal O}_{f_z}$ that is proper, $\mathfrak{b}_z$).  Conversely, when 
an imaginary quadratic field $K$ and its order ${\cal O}_{f_z}$ are specified, 
there is only a finite number of $\C$-isomorphism classes of elliptic curves
with complex multiplication by ${\cal O}_{f_z}$, 
\begin{align}
  {\cal E}ll({\cal O}_{f_z}) := 
       \left\{ [E_{z_a}]_\C \; | \; a = 1,\cdots, h({\cal O}_{f_z}) \right\}.
  \label{eq:def-Ell-Ofz}
\end{align}

A set of data $([E_z]_\C, f_\rho)$, to which either bosonic string theory 
or Type II string theory assigns a model of CFT, contains extra 
information $f_\rho$. The positive integer $f_\rho$ is used to 
set the complexified K\"{a}hler parameter $\rho = f_\rho a_z z$, which 
determines the K\"{a}hler metric on $[E_z]_\C$:
\begin{align}
  ds^2 = \frac{1}{2} \left(
    dX^\C \otimes d\overline{X}^\C + d\overline{X}^\C \otimes dX^\C\right)
  \; \frac{{\rm Im}(\rho)}{{\rm Im}(z)} (2\pi)^2\alpha'. 
\end{align}
A set of data $([E_z]_\C, f_\rho)$ therefore specifies a $\C$-isomorphism 
class of elliptic curves with a (complexified) K\"{a}hler metric on it, 
rather than $[E_z]_\C$ with just the complex structure specified. 

Let us now move on and describe the model of CFT that bosonic string theory 
or Type II string theory assigns to a set of data $([E_z]_\C, f_\rho)$.
As a preparation, we introduce three lattices that are determined by the data 
$([E_z]_\C, f_\rho)$. First, think of a free abelian group 
$H^1(E;\Z) \oplus H_1(E;\Z)$ of rank 4, where $E$ is an elliptic curve, 
and introduce an intersection form on this rank-4 abelian group by 
using the cohomology--homology natural pairing. The lattice obtained in this 
way is ${\rm II}_{2,2}$, the even unimodular lattice of signature (2, 2), 
for any complex structure of the elliptic curve $E=[E_z]_\C$.

The remaining two lattices are denoted by $\Lambda_-$ and $\Lambda_+$, which 
are both sublattices of ${\rm II}_{2,2}$ introduced above. 
Let $\{ \alpha, \beta\}$ be a basis of $H_1(E;\Z)$ and 
$\{ \hat{\alpha}, \hat{\beta} \}$ the basis of $H^1(E;\Z)$ dual to 
$\{ \alpha, \beta\}$. Then think of four linear maps, $\Omega'_-$, 
$\overline{\Omega}'_-$, $\Omega'_+$ and $\overline{\Omega}'_+$, from 
${\rm II}_{2,2}$ to $\C$.
\begin{align}
  \Omega'_-: (\hat{\alpha}, \alpha, \hat{\beta}, \beta) & \; \longmapsto 
  \frac{i}{\sqrt{2{\rm Im}(z){\rm Im}(\rho)}} (-z, \bar{\rho}, 1, z\bar{\rho}),
    \\ 
  \overline{\Omega}'_-: (\hat{\alpha}, \alpha, \hat{\beta}, \beta) & \; 
    \longmapsto 
  \frac{-i}{\sqrt{2{\rm Im}(z){\rm Im}(\rho)}} (-\bar{z}, \rho, 1, \bar{z}\rho),
  \\
  \Omega'_+: (\hat{\alpha}, \alpha, \hat{\beta}, \beta) & \; \longmapsto 
  \frac{i}{\sqrt{2{\rm Im}(z){\rm Im}(\rho)}} (-z, \rho, 1, z\rho), \\
  \overline{\Omega}'_+: (\hat{\alpha}, \alpha, \hat{\beta}, \beta) & \; 
    \longmapsto   \frac{-i}{\sqrt{2{\rm Im}(z){\rm Im}(\rho)}} 
    (-\bar{z}, \bar{\rho}, 1, \bar{z}\bar{\rho}),
\end{align}
so $\overline{\Omega}'_{\mp} = cc \cdot \Omega'_{\mp}$. The sublattices 
$\Lambda_-$ and $\Lambda_+$ of ${\rm II}_{2,2}$ are defined to be 
the kernel of $\Omega'_+: {\rm II}_{2,2} \rightarrow \C$ and 
$\Omega'_-: {\rm II}_{2,2} \rightarrow \C$, respectively. So, the choice 
of the sublattices $\Lambda_{\mp}$ in ${\rm II}_{2,2}$ depends on 
the data $z$ and $\rho = f_\rho a_z z$. To be more concrete, 
\begin{align}
 \Lambda_- & \; = {\rm Span}_\Z \left\{ 
   f_\rho a_z \hat{\alpha} + \alpha, \;
   f_\rho b_z \hat{\alpha} -f_\rho c_z \hat{\beta} - \beta \right\} 
   =: \left\{ e_{-2}, e_{-1} \right\} \\
 \Lambda_+ & \; = {\rm Span}_\Z \left\{ 
   f_\rho a_z \hat{\alpha} -\alpha -f_\rho b_z \hat{\beta},  \;
   -f_\rho c_z \hat{\beta} + \beta \right\} =:
   \left\{ e_{+2}, e_{+1} \right\}
\end{align}
as free rank-2 abelian groups, and the intersection form of $\Lambda_-$ and 
$\Lambda_+$ on the basis above is given by 
\begin{align}
   f_\rho \left[ \begin{array}{cc} 2a_z  & b_z \\ b_z & 2c_z \end{array} \right],
 \quad {\rm and} \quad 
 -  f_\rho \left[ \begin{array}{cc} 2a_z  & b_z \\ b_z & 2c_z \end{array} \right],
  \label{eq:mom-latt-int-form}
\end{align}
respectively. Each is an even lattice, and the signature is 
$(2,0)$ and $(0,2)$, respectively. The discriminant is 
${\rm discr}(\Lambda_{\mp}) = f_\rho^2D_z = f_\rho f_z^2 |D_K|$. 
(We may use, later on, a lattice $\Lambda_{[E_z]_\C} := \Lambda_-[1/f_\rho]$, 
which depends only on the information $[E_z]_\C$, not on $f_\rho$.)
The dual lattice of $\Lambda_{\mp}$ is denoted by $\Lambda_{\mp}^\vee$ and is 
isomorphic to ${\rm II}_{2,2}/\Lambda_{\pm}$ as an abelian group.
Now, we have finished introducing three lattices ${\rm II}_{2,2}$ and 
$\Lambda_{\mp}$; let us now use those three lattices to describe the 
models of CFT assigned to $([E_z]_\C, f_\rho)$.

There are a multiple different ways to specify a model of CFT.
One way to do this is to specify a couple of things starting with 
i) a pair of vertex operator algebra\footnote{The terminology 
``vertex operator algebra'' in math corresponds to ``chiral algebra'' in 
the language of string theorists; $V_-$ [resp. $V_+$] is the set of 
states in a model of CFT whose vertex operator is purely holomorphic [resp. 
purely anti-holomorphic], and $Y_-$ [resp. $V_+$] is the linear operation 
assigning the vertex operator to such a state.} $(V_\mp,Y_\mp)$ labeled 
by $\{ -, +\}$, where $(V_-, Y_-)$ contains the Virasoro algebra associated 
with the holomorphic coordinate reparametrization of a local neighborhood 
of a Riemann surface $\Sigma$, and $(V_+, Y_+)$ contains 
that with the anti-holomorphic 
coordinate reparametrization. This is followed by ii) a set $A_-$ of 
representations of $(V_-,Y_-)$ and another set $A_+$ of representations of 
$(V_+,Y_+)$, iii) a $\Z$-valued matrix $d_{\alpha\beta}$ ($\alpha \in A_-$, 
$\beta \in A_+$), and iv) a choice of a set of intertwining operators (chiral 
vertex operators, OPE coefficients) subject to some constraints. 
Bosonic string theory assigns one set of data i--iv) of a model of CFT 
to a set of data $([E]_\C, f_\rho)$, and Type II superstring theory assigns 
another set of data i--iv) of a model of CFT to the same set of data 
$([E]_\C, f_\rho)$.

\begin{props}
\label{props:bos-str-T2target}
Bosonic string theory assigns to a set of data $([E]_\C, f_\rho)$ 
a model of CFT in the following way. $(V_-,Y_-)$ is the vertex operator 
algebra associated with the even rank-2 positive definite lattice $\Lambda_-$, 
and $(V_+, Y_+)$ that with the even rank-2 positive definite lattice 
$\Lambda_+[-1]$. The set of irreducible representations are 
\begin{align}
 A_- = \Lambda_-^\vee/\Lambda_- \cong {\rm II}_{2,2}/(\Lambda_- \oplus \Lambda_+)
  \cong \Lambda_+^\vee/\Lambda_+ = A_+,
\end{align}
where the isomorphism from ${\rm II}_{2,2}/(\Lambda_- \oplus \Lambda_+)$
to $\Lambda_\mp^\vee/\Lambda_\mp$ is given by $\Omega'_{\pm}$. Let 
$\phi: A_- \rightarrow A_+$ be the isomorphism above 
between the two sets (two 
abelian groups in fact). Then $d_{\alpha\beta} = \delta_{\phi(\alpha)\beta}$.
\end{props}
   
\begin{rmk}
The model of CFT assigned to $([E]_\C,f_\rho)$ by the bosonic string theory 
is a model of rational CFT that is diagonal. Due to the isomorphism between 
the vertex operator algebra for the holomorphic part and anti-holomorphic 
part ($\Lambda_+[-1] \cong \Lambda_-$), we will sometimes drop the reference 
to the distinction between them, and use the notation $(V,Y)$. The sets of 
representations $A_-$ and $A_+$ are written as $iReps$ instead. 
$iReps \cong \Lambda^\vee/\Lambda$, where 
$\Lambda := \Lambda_- = \Lambda_+[-1]$.
\end{rmk}
 
Let us now leave a statement about how Type II string theory assigns a model 
of SCFT to a set of data $([E_z]_\C, f_\rho)$. A model of ${\cal N}=(2,2)$ 
SCFT is specified \cite{SVOA} by things such as i) a pair of ${\cal N}=2$ 
super vertex 
operator algebras $(V_\pm, Y_\pm)$ of NS-type, ii) the sets $A_-^{NS}$ and 
$A_-^{R}$ of NS-type representations and R-type representations\footnote{
The Ramond sector in string-theory language corresponds to  
parity-twisted modules of a super vertex operator algebra \cite{Rsectr-SVOA}. } 
of the 
${\cal N}=2$ super vertex operator algebra $(V_-,Y_-)$, and the sets $A_+^{NS}$
and $A_+^R$ for $(V_+,Y_+)$, iii) $\Z$-valued matrices $d_{\alpha\beta}^{NSNS}$, 
$d_{\alpha\beta}^{NSR}$, $d_{\alpha\beta}^{RNS}$, and $d_{\alpha\beta}^{RR}$, and iv) 
a choice of intertwining operators (OPE coefficients) satisfying certain 
sets of conditions.  

\begin{props}
The Type II superstring theory assigns to a set of data $([E_z]_\C, f_\rho)$ 
a model of ${\cal N}=(2,2)$ superconformal field theory (SCFT) in the 
following way.
Both of the ${\cal N}=2$ super vertex algebras $(V_-, Y_-)$ and $(V_+,Y_+)$ are 
that of ${\cal N}=2$ superconformal symmetry and the lattice $\Lambda$.
$A_+^{NS}=A_+^R = A_-^{NS}=A_-^R = \Lambda^\vee/\Lambda$ (due to the 
spectral flow in the case of $T^2$-target models). The common set 
$\Lambda^\vee/\Lambda$ is denoted by $iReps.$. 
$\phi: A_- \cong A_+$ is the trivial identification, and 
$d_{\alpha\beta}=\delta_{\phi(\alpha)\alpha}$ for all the four matrices 
$d_{\alpha\beta}^{**}$.
\end{props}

\subsection{The Chiral Correlation Functions of Interest: Bosonic String Theory}

\subsubsection{Objects of Interest in the Operator Formalism}

In our previous paper \cite{prev.paper}, we have seen that certain sets of 
operator matrix elements in the model of rational CFT corresponding to 
$([E_z]_\C, f_\rho)$ are closely related to the $L$-functions 
of arithmetic models\footnote{{\bf Definition}: 
We refer to an elliptic curve $E$ defined over a number field $k$ (modulo 
isomorphism over $k$) as an {\it arithmetic model} of a $\C$-isomorphism class 
of elliptic curves, $[E]_\C$, if $E \times_k \C=[E]_\C$.} 
of the elliptic curve $[E_z]_\C$.
They were 
\begin{align}
 f_0^{\rm bos}(\tau_{ws}; \beta) & \; 
     :=  {\rm Tr}_{V_\beta}\left[ q^{L_0-\frac{c}{24}} \right],  
    \label{eq:chi-corrl-bos-0} \\
 f_{1\Omega'_-}^{\rm bos}(\tau_{ws};\beta) & \;
     := {\rm Tr}_{V_\beta} \left[ q^{L_0-\frac{c}{24}} \alpha_0^\C \right], 
     \label{eq:chi-corrl-bos-1O} \\
 f_{1\overline{\Omega}'_-}^{\rm bos}(\tau_{ws}; \beta) & \;
     := {\rm Tr}_{V_\beta} \left[ q^{L_0-\frac{c}{24}} \overline{\alpha}_0^\C \right],
     \label{eq:chi-corrl-bos-1Obar}
\end{align}
which are functions of $\tau_{ws}$ in the upper half complex plane ${\cal H}$
defined for individual irreducible representations $\beta \in iReps.$ of 
the chiral algebra $(V_-,Y_-)$. $V_\beta$ is the representation space of 
$(V_-,Y_-)$, $L_0$ is the Cartan generator of the Virasoro algebra of 
the holomorphic coordinate reparametrization of Riemann surfaces, and $c = 2$ 
is the central charge of the Virasoro algebra. 
The linear operators $\alpha_0^\C$ and $\overline{\alpha}_0^\C$ on the vector 
space $V_\beta$ are described by using the decomposition 
\begin{align}
  V_\beta = \oplus_{w \in \Lambda_-^\vee, \; w + \Lambda_- = \beta} \; V_w; 
  \label{eq:KK-sum}
\end{align}
the operators $\alpha_0^\C$ and $\overline{\alpha}_0^\C$ take $V_w$'s as 
eigenspaces, 
and the eigenvalues are given by $\Omega'_-(w)$ and $\overline{\Omega}'_-(w)$, 
respectively, where $\Omega'_-$ and $\overline{\Omega}'_-$ have been 
restricted to $\Lambda_- \subset {\rm II}_{2,2}$ and then extended 
$\Q$-linearly from $\Lambda_-$ to $\Lambda_-^\vee$.

With the definitions given, it is straightforward (for experts of string 
theory and conformal field theory) to find an explicit formula for 
$f_0^{\rm bos}$ for the model of RCFT assigned to a set of data 
$([E_z]_\C, f_\rho)$:
\begin{align}
 f_0^{\rm bos} 
  = \frac{\vartheta_{\Lambda_-}(\tau_{ws};\beta)}
        {(\eta(\tau_{ws}))^2}.  
  \label{eq:f0-bos-formula}
\end{align}
Here, $\eta(\tau) := q^{\frac{1}{24}} \prod_{n=1}^\infty(1-q^n)$ for 
$q := e^{2\pi i \tau}$ is the Dedekind $\eta$-function of $\tau  \in {\cal H}$.
The theta function $\vartheta_L(\tau; x)$ for an even lattice $L$, 
an element $x \in L^\vee/L$, and $\tau \in {\cal H}$ is given by 
\begin{align}
  \vartheta_L(\tau; x) := \sum_{w \in L^\vee, \; w+L = x} e^{2\pi i \tau \frac{(w,w)_L}{2}}
   = \sum_{w \in x} q^{\frac{(w,w)_L}{2}}, 
  \label{eq:theta0-def}
\end{align}
where $(-,-)_L$ is the intersection form of $L$ extended to 
$L^\vee \subset L \otimes \Q$ linearly. 
In the expression (\ref{eq:f0-bos-formula}, \ref{eq:theta0-def}), 
the sum over $w$'s in a coset $\beta \in \Lambda_-^\vee/\Lambda_-$ 
is from the decomposition (\ref{eq:KK-sum}) of the vector space $V_\beta$. 

The expressions for $f_{1\Omega'_-}^{\rm bos}$ and $f_{1\overline{\Omega}'_-}^{\rm bos}$
are 
\begin{align}
\qquad 
 f_{1\Omega'_-}^{\rm bos} 
   = \frac{\vartheta_{\Lambda_-}^{1\Omega'_-}(\tau_{ws};\beta)}
        {(\eta(\tau_{ws}))^2},  \qquad 
 f_{1\overline{\Omega}'_-}^{\rm bos} 
   = \frac{\vartheta_{\Lambda_-}^{1\overline{\Omega}'_-}(\tau_{ws};\beta)}
        {(\eta(\tau_{ws}))^2};  
\end{align}
type-1 theta functions $\vartheta^{1\omega}_L(\tau;x)$ for an even lattice $L$, 
an element $x \in L^\vee/L$, and a linear map $\omega: L^\vee \rightarrow \C$
is given by 
\begin{align}
  \vartheta^{1\omega}_L(\tau;x) := \sum_{y \in L^\vee, \; y +L = x}
       \omega(y)q^{\frac{(y,y)_L}{2}}.
\end{align}
%

\subsubsection{Operator Formalism and Conformal Blocks}

The three functions $f_0^{\rm bos}(\tau_{ws};\beta)$, $f_{1\Omega'_-}^{\rm bos}$ 
and $f_{1\overline{\Omega}'_-}^{\rm bos}$ on $\tau_{ws} \in {\cal H}$ can be 
thought of as chiral correlation functions of one local operator on 
a Riemann surface (worldsheet) $\Sigma_{ws}$ of genus $g=1$. A genus $g=1$ 
Riemann surface is given an analytic representation $\Sigma_{ws} = \C/(\Z + \tau_{ws}\Z)$, 
or 
\begin{align}
  \Sigma_{ws} = \left\{ u \in \C \right\} / 
       (u \sim u+1, \; u \sim u + \tau_{ws}).
  \label{eq:anal-repr-SigmaWs}
\end{align}
For a point $p \in \Sigma_{ws}$, $u(p)$ stands for the complex value $u$ 
of the point $p$ modulo $\Z + \tau_{ws}\Z$. First, 
\begin{align}
  f_0^{\rm bos}(\tau_{ws};\beta) = \langle \varphi_\beta, \; {\bf 1} 
      \rangle_{\Sigma_{ws}},       \label{eq:bos-confBl-0}
\end{align}
where the notation $\langle \varphi_\beta, \; {\cal O}(u(p)) \rangle_{\Sigma_{ws}}$
represents a $(g,n)=(1,1)$ chiral correlation function computed by inserting 
a local operator ${\cal O}(u)$ at the point $p$, and  
having restricted the states propagating in the handle part of 
Figure~\ref{fig:g1n1} to be those in $V_\beta$. 
\begin{figure}[tbp]
\begin{center}
\begin{tabular}{cc}
  \includegraphics[width=0.4\linewidth]{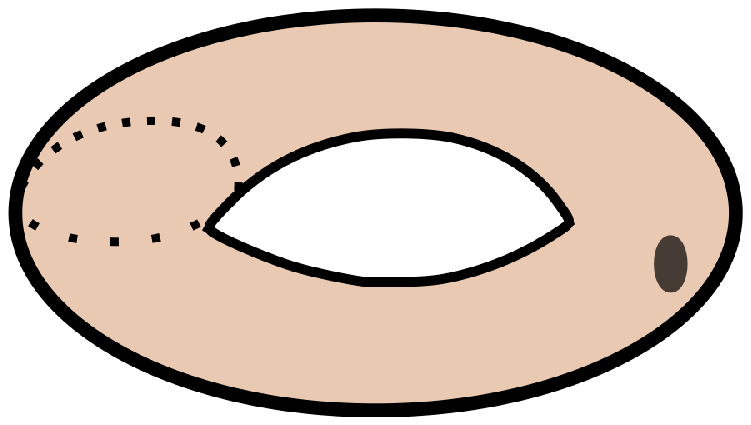} &
  \includegraphics[width=0.4\linewidth]{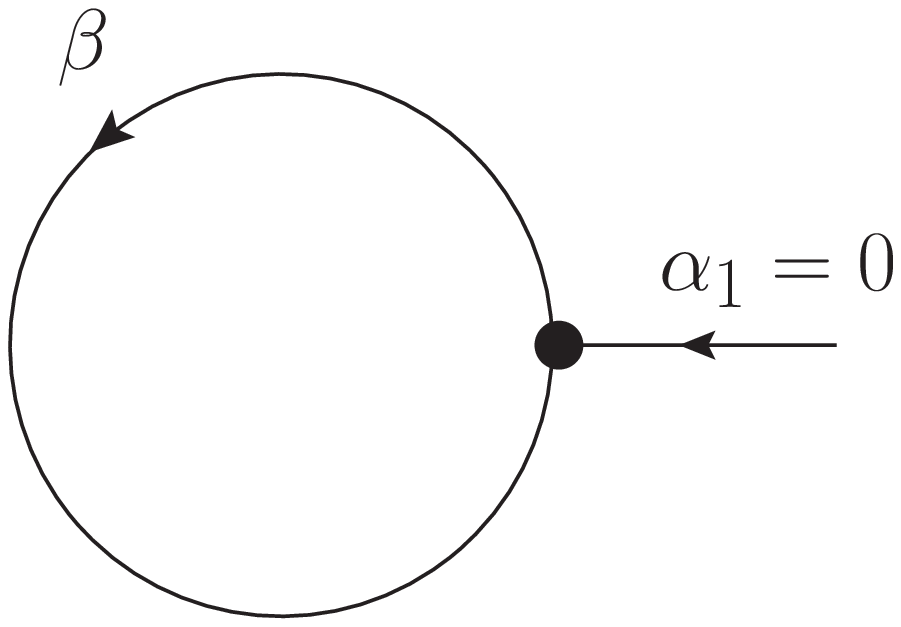} \\
 (a) & (b) 
\end{tabular}
 \caption{\label{fig:g1n1}A $g=1$ Riemann surface $\Sigma_{ws}$ with $n=1$ 
marked point is regarded as one pants with its two legs sewn together 
along the dotted line in (a). 
The degeneration limit corresponding to the pants decomposition (a) 
is described by the skeleton graph in (b). }
\end{center}
\end{figure}

For $f_{1\Omega'_-}^{\rm bos}$ 
and $f_{1\overline{\Omega}'_-}^{\rm bos}$, string theorists realize immediately 
that 
\begin{align}
  {\rm Tr}_{V_\beta}\left[ q^{L_0-\frac{c}{24}} \alpha^\C_0 \right]
  & \; = -\frac{1}{2\pi} \sqrt{\frac{2}{\alpha'}} 
  {\rm Tr}_{V_\beta} \left[ q^{L_0-\frac{c}{24}} (\partial_u X^\C)(u(p)) \right], 
    \label{eq:bos-alpha0=partX} \\
  {\rm Tr}_{V_\beta}\left[ q^{L_0-\frac{c}{24}} \overline{\alpha}^\C_0 \right]
  & \; = -\frac{1}{2\pi} \sqrt{\frac{2}{\alpha'}} 
  {\rm Tr}_{V_\beta} \left[ q^{L_0-\frac{c}{24}} (\partial_u \overline{X}^\C)(u(p))
      \right], 
    \label{eq:bos-alpha0=partX-cc}
\end{align}
so that the operators $\alpha_0^\C$ and $\overline{\alpha}_0^\C$ can be 
replaced by local operators $(\partial_u X^\C)$ and 
$(\partial_u \overline{X}^\C)$. Here, 
\begin{align}
  \sqrt{\frac{2}{\alpha'}} (\partial_u X^\C)(u) & \; = -(2\pi) \alpha_0^\C 
   -(2\pi) \sum_{0 \neq m \in \Z} \alpha_m^\C e^{2\pi i m u}, 
     \label{eq:partX-expansion} \\ 
  \sqrt{\frac{2}{\alpha'}} (\partial_u \overline{X}^\C)(u) & \; =
   -(2\pi) \overline{\alpha}_0^\C 
   -(2\pi) \sum_{0 \neq m \in \Z} \overline{\alpha}_m^\C e^{2\pi i m u}, 
    \label{eq:partXbar-expansion}
\end{align}
with the commutation relation $[\alpha^\C_m, \overline{\alpha}_n^\C] = 
2m \delta_{m+n}$, $[\alpha^\C_m, \alpha_n^\C] = 0$, and 
$[\overline{\alpha}_m^\C, \overline{\alpha}^\C_n] = 0$. Because 
the operators $\alpha_m^\C$ and $\overline{\alpha}^\C_m$ with $m \neq 0$ 
map an $L_0=h$ eigenstate into $L_0 = h-m$ eigenstates with the eigenvalue 
different from the one before, the $\alpha^\C_{m \neq 0}$ terms in 
$\partial_uX^\C$ do not contribute to the trace, which justifies the 
relation (\ref{eq:bos-alpha0=partX}, \ref{eq:bos-alpha0=partX-cc}).
So, 
\begin{align}
  f_{1\Omega'_-}^{\rm bos} & \; = - \frac{1}{2\pi} \sqrt{\frac{2}{\alpha'}} \; \;
      \langle \varphi_\beta, \; (\partial_u X^\C)(u(p)) \rangle_{\Sigma_{ws}}, 
    \label{eq:bos-confBl-1O} \\
  f_{1\overline{\Omega}'_-}^{\rm bos} & \; =
    - \frac{1}{2\pi} \sqrt{\frac{2}{\alpha'}} \; \; 
       \langle \varphi_\beta, \; (\partial_u \overline{X}^\C)(u(p))
        \rangle_{\Sigma_{ws}}. 
    \label{eq:bos-confBl-1Obar}
\end{align}
The expressions on the right-hand side of (\ref{eq:bos-confBl-0}, 
\ref{eq:bos-confBl-1O}, \ref{eq:bos-confBl-1Obar}) do not depend 
on which point $p \in \Sigma_{ws}$ we choose to insert those operators; 
we can see this both from the translational symmetry of $\Sigma_{ws}$, and 
also from the fact that the $\alpha_0^\C$ term in (\ref{eq:partX-expansion}) 
and the $\overline{\alpha}_0^\C$ term in  (\ref{eq:partXbar-expansion})
do not have $u(p)$-dependence. 

All of (\ref{eq:bos-confBl-0}, \ref{eq:bos-confBl-1O}, 
\ref{eq:bos-confBl-1Obar}) are in the form of a $(g,n)=(1,1)$ chiral 
correlation functions of one local operator, 
\begin{align}
   \langle \varphi_\beta, \; Y_-(v;u_*) \rangle_{\Sigma_{ws}}
  \label{eq:bos-confBl-gen}
\end{align}
for a state $v \in V_{\alpha}$ with $\alpha = 0 \in \Lambda_-^\vee/\Lambda_- 
= iReps$. For the vertex operators $Y_-(v; u)$ to be ${\bf 1}$, 
$-(2\pi)^{-1} \sqrt{(2/\alpha')} \; (\partial_uX^\C)$ and 
$-(2\pi)^{-1}\sqrt{(2/\alpha')} (\partial_u\overline{X}^\C)$, the states $v$
should be 
\begin{align}
   v = \ket{0}, \qquad  \frac{i}{2\pi} \alpha_{-1}^\C \ket{0}, \qquad 
   \frac{i}{2\pi} \overline{\alpha}_{-1}^\C \ket{0}.
  \label{eq:bos-states-VirPrim-011}
\end{align}
All those states are Virasoro primary, with the conformal weight $h=0$, $h=1$, 
and $h=1$, respectively. Let $h(v)$ be the conformal weight, the 
$L_0$-eigenvalue, of a Virasoro primary state $v$ in $V_{\alpha = 0}$. Then 
\begin{align}
    \langle \varphi_\beta, \; Y_-(v; u(p)) \rangle_{\Sigma_{ws}} \; (du)^{h(v)}
  \label{eq:bos-confBl-gen-wDiff}
\end{align}
with an extra factor of formal differential $du$ to the $h(v)$-th power 
is independent of the choice of a local coordinate at the point of 
$\Sigma_{ws}$ where the vertex operator $Y_-(v;u(p))$ is inserted.

One can see (\ref{eq:bos-confBl-gen}) as a definition of a $\C$-valued 
linear functional of the space of local sections of the bundle 
${\cal V}_{\alpha=0}$ of the vertex algebra $(V_-,Y_-)$ on a genus $g=1$ 
Riemann surface $\Sigma_{ws}$ \cite[\S6.5]{FBz}. We need to invoke a notion 
of the bundle ${\cal V}_{\alpha}$ over $\Sigma_{ws}$ for $\alpha \in iReps$
because the vertex operator $Y_-(v;u)$ for a state $v \in V_{\alpha = 0}$, 
in fact, depends also on a choice of local coordinate on $\Sigma_{ws}$;  
a set of $(v,u)$, a combination of 
a state $v \in V_{\alpha = 0}$ and a choice of local coordinate of 
$\Sigma_{ws}$, determines a local section $Y_-(v;u)$ of the bundle 
${\cal V}_{\alpha = 0}$ on $\Sigma_{ws}$; a set $(v, u')$ of the same state 
$v$ but with a different choice of a local coordinate $u'$ determines 
a local section $Y_-(v;u')$ of ${\cal V}_{\alpha = 0}$ that is different 
from the local section $Y_-(v;u)$. The expression (\ref{eq:bos-confBl-gen}) 
is a linear functional of the space of such local sections over $\Sigma_{ws}$; 
$[f_{\{0\}}^{\rm bos}]_{\{\beta\}}: \Gamma(U; {\cal V}_{\alpha=0}) \rightarrow \C$, 
where $U$ is a patch of local coordinate $u$ of $\Sigma_{ws}$. 
Such a linear functional is called a conformal block.

For different $\beta$'s in $iReps$ and for a $g=1$ Riemann surface 
$\Sigma_{ws}$ with different complex structure parameter $\tau_{ws}$'s, 
the conformal blocks $[f_{\{0\}}^{\rm bos}]_{\{\beta\}}$ introduced above are 
different. The parameter $\tau_{ws}$ should be regarded as a local coordinate 
of the moduli space ${\cal M}_{1,1}$ of $(g,n)=(1,1)$ pointed Riemann surfaces, 
Riemann surfaces of genus $g=1$ with $n=1$ point $u=u_*$ specified; 
the one point at $u=u_*$ can be brought
to the origin in the analytic representation $\Sigma_{ws} = \C/(\Z+\tau_{ws}\Z)$
(by translational symmetry or by coordinate redefinition $u' = u-u_*$), 
and hence just $\tau_{ws} \in {\cal H}$ is enough. 
The mapping class group of 
$(g,n)=(1,1)$ pointed Riemann surfaces is ${\rm SL}(2;\Z)$, which acts on 
$\tau_{ws}$ and $u_*$ through 
\begin{align}
  \tau_{ws} \longmapsto \gamma \cdot \tau_{ws} := \frac{p \tau_{ws} + q}{r \tau_{ws}+s}, \qquad 
  u_* \longmapsto u^\gamma_* := \frac{u_*}{r \tau_{ws} + s} 
\label{eq:SL2Z-act-on-tauNu}
\end{align}
for 
\begin{align}
\gamma = [p,q; r,s] := \left( \begin{array}{cc} p & q \\ r & s \end{array} \right) \in {\rm SL}(2;\Z), 
\end{align}
so the moduli space ${\cal M}_{1,1}$ is ${\cal H}/{\rm PSL}(2;\Z)$, but 
we keep in mind that there is a residual $\Z/(2) \subset {\rm SL}(2;\Z)$ 
symmetry acting on the differential $du$ even after setting $u_*=0$.
The way the conformal block $[f_{\{0\}}^{\rm bos}]_{\{\beta\}}$ changes 
over ${\cal M}_{1,1}$ is described by the Knizhnik--Zamoldchikov equation 
adapted to the model of rational CFT for $([E_z]_\C, f_\rho)$.

Generally in a model of rational CFT, we have a set of i) 
the vertex operator algebra $(V_-,Y_-)$, ii) the set of irreducible 
representations $A_-$ of $(V_-,Y_-)$ and iii) the set of OPE coefficients 
satisfying certain conditions (cf comments just before 
Prop. \ref{props:bos-str-T2target}). 
For any choice of a $(g,n)$ pointed Riemann 
surface $m_* \in {\cal M}_{g,n}$, which is a Riemann surface $\Sigma_{ws}$ of 
genus $g$ and $n$ points on it, we can specify a set of $n$ irreducible 
representations 
$\{ \alpha\} = \{ \alpha_1, \alpha_2, \cdots, \alpha_n\}$ from $A_-$; 
an idea is to insert at the $i$-th point $u_i \in \Sigma_{ws}$ the vertex 
operator for a state $v_i$ in the representation $\alpha_i$ of $(V_-,Y_-)$.
We can further introduce a pants decomposition of the $(g,n)$ pointed 
Riemann surface $m_*$, where $\Sigma_{ws}$ with $n$ points are decomposed 
into $2g-2+n$ pants sewn together along $3g-3+n$ pairs of pants-boundaries.
Once a pants decomposition is given, then one can think of assigning 
an irreducible representation $\beta_a \in A_-$ to the $a$-th sewing 
locus for each one of $a=1,\cdots, 3g-3+n$; an idea here is to have 
states in the representation $\beta_a$ propagate at the $a$-the sewing locus
in computing chiral correlation functions of the model of rational CFT.
So, in a given model of rational CFT, for a given $(g,n)$ pointed Riemann 
surface $m_* \in {\cal M}_{g,n}$, once $\{\alpha\} = 
\{ \alpha_1, \cdots, \alpha_n\} \in (A_-)^n$ is chosen, then 
\begin{align}
 [f_{\{\alpha\}}]_{\{\beta\}}: \{ v_1,\cdots, v_n\} \longmapsto 
      \langle \varphi_{\{\beta\}}, \; \prod_{i=1}^n Y_-(v_i;u_i) \rangle_{\Sigma_{ws}}
\end{align}
for each choice of $\{ \beta \} \in (A_-)^{3g-3+n}$ is a linear functional 
on the space of choice of states for $v_i \in V_{\alpha_i}$, where $V_{\alpha_i}$ 
is the vector space of the representation $\alpha_i$ of $(V_-,Y_-)$. 
To be a little more\footnote{Because of the fact that the conformal weights 
(the eigenvalues of $L_0$ in the Virasoro algebra) are not 
integer valued in all the irreducible representations $\alpha \in A_-$ 
except $\alpha = 0$, such notions as ``the bundle ${\cal V}_{\alpha}$ on 
a Riemann surface $\Sigma_{ws}$'' fail in capturing mathematical structures 
of rational CFT appropriately \cite[\S7.3]{FBz}.  
Discussions in the main text completely ignore this aspect.
This article only deals with conformal blocks for $\{\alpha \} = \{0\}$, 
where this issue is absent. So we consider that too much reference 
to this issue would be rather distracting. } precise, a choice of a state 
$v_i$ (along with a choice 
of local coordinate $u_i$ in a neighborhood of the $i$-th point) specifies 
a local section of the bundle $\widetilde{\cal V}_{\alpha_i}$ on 
${\cal M}_{g,1}$, the bundle obtained by treating as a family the bundles 
${\cal V}_{\alpha_i}$ on individual genus $g$ Riemann surfaces $\Sigma$ 
\cite[6.5.3]{FBz}; the bundles $\widetilde{\cal V}_{\alpha_i}$ 
for $i=1,\cdots, n$ are pulled back from ${\cal M}_{g,1}$ to ${\cal M}_{g,n}$ 
to form a bundle 
$\otimes_{i=1}^n \widetilde{\cal V}_{\alpha_i}$, and a choice of $n$ states 
$\{ v_1, \cdots, v_n\}$ specifies a local section of the bundle 
$\otimes_{i=1}^n \widetilde{\cal V}_{\alpha_i}$. The conformal block 
$[f_{\{\alpha\}}]_{\{\beta\}}$ for a given choice $\{\beta\} \in (A_-)^{3g-3+n}$
is a linear functional on the space of such local sections.

So long as we deal with the $(g,n)=(1,1)$ conformal blocks in $T^2$-target 
models of rational CFT (and their ${\cal N}=(2,2)$ supersymmetric versions), 
there is just one $(3g-3+n=1)$ sewing locus in the $(g,n)=(1,1)$ pointed 
Riemann surface, so the label $\{\beta\} \in (A_-)^{3g-3+n}$ is just 
a choice $\beta \in A_-$. Moreover, within this class of models, 
the conformal blocks can be non-trivial only for $\{ \alpha \} = 
\{ \alpha_1\}$ with $\alpha_1 = 0 \in \Lambda_-^\vee/\Lambda_- = A_-$ 
because of the momentum conservation (in the $T^2$ target space). 
Therefore, $[f_{\{0\}}]_{\{\beta\}}$ is the only form of conformal 
blocks that can be non-trivial. 

\subsubsection{Conformal Blocks and Monodromy Representations}

For a model of rational CFT, the Friedan--Shenker bundle ${\cal F}_{\{\alpha\}}$ 
on ${\cal M}_{g,n}$ for a choice of $\{ \alpha_{i=1,\cdots,n}\}$ is the 
bundle of linear functionals of $\otimes_{i=1}^n \widetilde{\cal V}_{\alpha_i}$ 
on ${\cal M}_{g,n}$ that are consistent with the action of $(V_-,Y_-)$. 
The $(g,n)$ conformal blocks of that model, $[f_{\{\alpha\}}]_{\{\beta\}}$ for 
$\{\beta\} \in (A_-)^{3g-3+n}$, are sections of the Friedan--Shenker bundle 
${\cal F}_{\{\alpha\}}$ defined locally in ${\cal M}_{g,n}$.
The bundle ${\cal F}_{\{\alpha\}}$ further has a structure locally in 
${\cal M}_{g,n}$ of 
\begin{align}
  {\cal F}_{\{\alpha\}} \cong
      \oplus_{\{\beta \} \in (A_-)^{3g-3+n}} [{\cal F}_{\{\alpha\}}]_{\{\beta\}},
\end{align}
and $[f_{\{\alpha\}}]_{\{\beta\}}$ for a given $\{\beta\} \in (A_-)^{3g-3+n}$ 
is a section of $[{\cal F}_{\{\alpha\}}]_{\{\beta\}}$. Neither 
the bundle ${\cal F}_{\{\alpha\}}$ has this structure of direct sum globally 
over ${\cal M}_{g,n}$, nor the local sections $[f_{\{\alpha\}}]_{\{\beta\}}$ can 
be extended globally over ${\cal M}_{g,n}$ without giving up single-valuedness.

\begin{props}
\label{props:bos-mondrmy-repr-g-n}
A conformal block $f_{\{ \alpha\}} = ([f_{\{\alpha\}}]_{\{\beta\}})$, a local 
section of the bundle ${\cal F}_{\{\alpha\}}$ on ${\cal M}_{g,n}$, undergoes 
monodromy transformation as one moves along a closed path on ${\cal M}_{g,n}$. 
Let $\gamma \in \pi_1({\cal M}_{g,n}; m_*)$, where $m_* \in {\cal M}_{g,n}$ is 
a base point. For any point $m \in {\cal M}_{g,n}$, let $\gamma_m$ be 
a path from $m_*$ to $m$. Then another local section 
$\gamma \circ f_{\{\alpha\}} = ([\gamma \circ f_{\{\alpha\}}]_{\{\beta\}})$ of 
the bundle ${\cal F}_{\{\alpha\}}$ on ${\cal M}_{g,n}$, is given by 
\begin{align}
   [\gamma \circ f_{\{\alpha\}}]_{\{\beta\}} (m) 
 := [f_{\{\alpha\}}]_{\{\beta\}} 
    ( \gamma_m \cdot \gamma^{-1} \cdot \gamma_m^{-1} \cdot m), 
  \qquad m \in {\cal M}_{g,n}, 
\end{align}
which is the analytic continuation of $[f_{\{\alpha\}}]_{\{\beta\}}$ along 
the path $\gamma_m \cdot \gamma \cdot \gamma_m^{-1}$.
Let 
\begin{align}
  [\gamma \circ f_{\{\alpha\}}]_{\{\beta\}} = \sum_{\{\beta'\} \in (iReps)^{3g-3+n}}
    [f_{\{\alpha\}}]_{\{\beta'\}} \;  
    [\rho^{(g,n)}_{\{\alpha\}}(\gamma)]_{\{\beta'\} \{\beta\}}. 
\end{align}
For $\gamma_1, \gamma_2 \in \Gamma_{g,n}$, 
$(\gamma_1\gamma_2)\circ f_{\{\alpha\}} = \gamma_1 \circ 
\left( \gamma_2 \circ f_{\{\alpha\}}\right)$. So, this 
$\rho^{(g,n)}_{\{\alpha\}}$ is a representation of 
$\Gamma_{g,n} := \pi_1^{\rm top}({\cal M}_{g,n})$ 
(only a projective representation for general $(g,n)$'s, even 
for $\{\alpha\}=\{0\}$), and 
is called the monodromy representation of the mapping class group 
$\Gamma_{g,n}$ of $n$-pointed genus-$g$ Riemann surfaces. The representation 
$\rho^{(g,n)}_{\{\alpha\}}$ is determined by a model of rational CFT, $(g,n)$, and 
a choice of $\{\alpha_1, \cdots, \alpha_n\} \in (A_-)^n$. 
See \cite{FS, TK, Kohno, MS, BK, Schneps-PS} for more information. $\bullet$
\end{props}

Example \ref{exmpl:bos-T2models-monodromy} will describe the monodromy 
representations $\rho$ for $T^2$-target models of rational CFT, for $(g,n)=(1,1)$, 
and the only non-trivial choice of $\{ \alpha\} = \{ \alpha_1\}$, which is 
$\alpha_1 = 0 \in \Lambda_-^\vee/\Lambda_- = A_-$. As a preparation for 
the Example \ref{exmpl:bos-T2models-monodromy}, we write down a few definitions. 

\begin{defn}
A pair $(M, q_M)$ of a finite abelian group $M$ and a map 
$q_M: M \rightarrow \Q/2\Z$ is called a {\it finite quadratic module}, if 
i) $q_M(-x)=q_M(x)$ for any $x \in M$, 
ii) $(-,-)_M: M \times M \rightarrow \Q/\Z$ defined by 
$(x,y)_M = [q_M(x+y)-q_M(x)-q_M(y)]/2$ is $\Z$-bilinear. 
Such a pair may be denoted by $\underline{M} = (M, q_M)$.

It follows from the conditions i) and ii) that 
$q_M(ax) = a^2q_M(x)$ for any $a \in \Z$ and $x \in M$. 
In particular, $q_M(0)=0$.  $\bullet$
\end{defn}

\begin{defn}
A finite quadratic module $\underline{M} = (M, q_M)$ is said to be 
{\it non-degenerate}, if, for $u \in M$, there exists $y \in M$ so that 
$(u,y)_M \neq 0 \in \Q/\Z$, unless $u=0$. 
\end{defn}

\begin{defn}
For an even positive definite lattice $L$, which consists of a free abelian group 
$L$ and a bilinear form $(-,-)_L$, its {\it discriminant form}
$DL := (G_L, q_L)$ is a pair of finite abelian group $G_L := L^\vee/L$ and 
a quadratic form $q_L: G_L \ni x+L \mapsto (x,x)_L + 2\Z \in \Q/2\Z$. 
The symmetric bilinear form $G_L \times G_L \ni (x,y) \mapsto 
[q_L(x+y)-q_L(x)-q_L(y)]/2 \in \Q/\Z$ is denoted by $(-,-)_L$ by recycling 
the notation. The discriminant form $DL = (G_L, q_L)$ is a non-degenerate 
finite quadratic form.
\end{defn}

\begin{defn}
Let $DL = (G_\Lambda, q_L)$ be the discriminant form associated with an even 
lattice $L$ of even rank $2r_0$, with the signature either $(2r_0,0)$ or 
$(0,2r_0)$. The {\it Weil representation} of ${\rm SL}(2;\Z)$
of $DL$, denoted by $\rho_{DL}^{\rm Weil}$, takes the vector space 
$\C[G_L] = {\rm Span}_\C\{ e_x \; | \; x \in G_L\}$ as the representation 
space, and is given by the following representation of the two generators 
$T := [1, 1; 0, 1] \in {\rm SL}(2;\Z)$ and $S := [0, -1; 1, 0] \in {\rm SL}(2;\Z)$: 
\begin{align}
  T \cdot e_x & \; = \sum_{y \in G_L} e_y \delta_{x,y} 
      \mathbb{E}\left[ \frac{q_L(x,x)}{2} \right] =: 
  \sum_{y} e_y \; (\rho_{DL}^{\rm Weil}(T))_{yx}, \\
  S \cdot e_x & \; = \frac{\mathbb{E}[-{\rm sgn}(L)/8]}{\sqrt{|G_L|}} 
         \sum_{y \in G_L} e_y \mathbb{E}\left[-(y,x)_L \right]
   =: \sum_{y \in G_L} e_y \; (\rho_{DL}^{\rm Weil}(S))_{yx}.
\end{align}
The group ${\rm SL}(2;\Z)$ is generated by the two elements $T$ and $S$ with the 
two following relations:
\begin{align}
   S^2 = (ST)^3, \qquad (S^2)^2 = 1 \in {\rm SL}(2;\Z).
  \label{eq:relatn-SL2Z}
\end{align}
It is known that the representation matrices $\rho_{DL}^{\rm Weil}(T)$ and 
$\rho_{DL}^{\rm Weil}(S)$ also satisfy the same relations, and hence the Weil 
representation is a linear representation of ${\rm SL}(2;\Z)$ 
rather than a projective representation of ${\rm SL}(2;\Z)$.  
\end{defn}

\begin{defn}
A 1-dimensional representation $(\vartheta_\eta)^2: {\rm SL}(2;\Z) 
\rightarrow S^1$
is given by 
\begin{align}
(\vartheta_\eta)^2: T \longmapsto \mathbb{E}[-1/12], \quad {\rm and} \quad
S \longmapsto \mathbb{E}[+3/12].
\end{align}
This is a linear representation of ${\rm SL}(2;\Z)$, rather than a 
projective representation; this can be verified easily by substituting 
the complex phases $\mathbb{E}[-1/12]$ and $\mathbb{E}[+3/12]$ into 
$T$ and $S$ in (\ref{eq:relatn-SL2Z}).
\end{defn}

Now, we are ready to write down the following well-known 

\begin{exmpl}
\label{exmpl:bos-T2models-monodromy}
Think of a model of rational CFT for $([E_z]_\C, f_\rho)$, where 
$[E_z]_\C$ is a $\C$-isomorphism class of elliptic curves with complex 
multiplication and $f_\rho \in N_{>0}$ controlling the complexified 
K\"{a}hler parameter $\rho = f_\rho a_z z$. For $(g,n)=(1,1)$, and 
$\alpha = 0 \in G_{\Lambda_-} = A_-$, the monodromy representation 
$\rho^{(1,1)}_{\{\alpha\}}$ of the mapping class group $\Gamma_{1,1} = 
{\rm SL}(2;\Z)$ determined by the non-trivial conformal block $f_{\{\alpha\}}$ 
(Prop. \ref{props:bos-mondrmy-repr-g-n}) is 
\begin{align}
  \rho^{(1,1)}_{\{0\}} = (\vartheta_{\eta})^{-2} \rho_{D\Lambda_-[-1]}^{\rm Weil} .
\end{align}
\end{exmpl}

\begin{proof} The moduli space ${\cal M}_{1,1}$ of $(g,n)=(1,1)$ pointed 
Riemann surfaces is isomorphic to ${\cal H}/{\rm PSL}(2;\Z)$, where 
$\tau_{ws}$ can be used as a coordinate of ${\cal H}$. Let us choose a base 
point $[\tau_*]$ of $\Gamma_{1,1} = \pi_1^{\rm top}({\cal M}_{1,1}; [\tau_*])$ 
at the ${\rm PSL}(2;\Z)$-orbit of $\tau_{*} = i\infty - i \epsilon \in {\cal H}$.
The generator $\gamma_T \in \Gamma_{1,1} = 
\pi_1^{\rm top}({\cal M}_{1,1}; [\tau_*])$ is a path in ${\cal H}$ that starts 
from $\tau_{ws}=\tau_*$ and ends at $\tau_{ws}=\tau_*+1$
represented by a line segment horizontal to the real axis; 
the path $\gamma_T^{-1}$ is from $\tau_{ws}=\tau_*$ to $\tau_{ws}=\tau_*-1$ then. 
The other generator $\gamma_S \in \Gamma_{1,1} = 
\pi_1^{\rm top}({\cal M}_{1,1}; [\tau_*])$ is a path starting from the same 
point $\tau_{ws}=\tau_*$ and ending at $\tau_{ws}= -1/\tau_*$ which avoids 
the point $\tau_{ws}=+i$ by staying within the second quadrant; 
$\gamma_S^{-1}$ is a path from $\tau_{ws}=\tau_*$ to $\tau_{ws}=-1/\tau_*$
that goes through the first quadrant.

Let us first work out the monodromy representation $\rho^{(1,1)}_{\{0\}}$
by using the image of the state $v = \ket{0} \in V_{\alpha = 0}$ under the 
conformal block $f_{\{0\}}$, 
\begin{align}
  [f_{\{0\}}]_{\{\beta\}} : \ket{0} \longmapsto f_0^{\rm bos}(\tau_{ws};\beta),  
\end{align}
rather than $f_{\{0\}}$ as a whole. The analytic continuation of 
$f_0^{\rm bos}(\tau_{ws};\beta)$ in (\ref{eq:f0-bos-formula}) by $T$ and by $S$ 
turns them into 
\begin{align}
  \left(\gamma_T \circ f_0^{\rm bos}\right)(\tau_{ws};\beta) & \; =
  f_0^{\rm bos}(\tau_{ws}-1;\beta)
    = \frac{ \mathbb{E}[-q_{\Lambda_-}(\beta,\beta)/2] \;
                  \vartheta_{\Lambda_-}(\tau_{ws};\beta) }
           { \mathbb{E}[-1/12] \; (\eta(\tau_{ws}))^2 },    \nonumber \\
   & \; = \sum_{\gamma \in G_{\Lambda_-}} \; 
     [(\vartheta_\eta)^{-2} \rho_{D\Lambda_-[-1]}^{\rm Weil}](T)_{\beta \gamma} \; 
         f_0^{\rm bos}(\tau_{ws};\gamma).    \\
  \left(\gamma_S \circ f_0^{\rm bos}\right)(\tau_{ws};\beta) & \; =
    f_0^{\rm bos}(-1/\tau_{ws};\beta) 
    = \frac{ \frac{i(-\tau_{ws})}{\sqrt{|G_{\Lambda_-}|}}
               \sum_{\gamma \in G_{\Lambda_-}} \mathbb{E}[(\gamma,\beta)] \;
                         \vartheta_{\Lambda_-}(\tau_{ws};\gamma)}
          { (\mathbb{E}[+1/8]\sqrt{-\tau_{ws}})^2 (\eta(\tau_{ws})^2 } ,
              \nonumber \\
   & \; = \sum_{\gamma \in G_{\Lambda_-}}
       [(\vartheta_\eta)^{-2} \rho_{D\Lambda_-[-1]}^{\rm Weil}](S)_{\beta \gamma} \;
          f_0^{\rm bos}(\tau_{ws};\gamma).
\end{align}
Here, we used Lemma \ref{lemma:eta-transf} and Lemma \ref{lemma:theta-transf}
for the modular transformation law of the $\eta$-function and 
$\vartheta$-functions, respectively. 
This proves that the monodromy 
representation $\rho^{(1,1)}_{\{0\}}$ of a $T^2$-target model of rational CFT 
is $(\vartheta_\eta)^{-2}\rho_{D\Lambda_-[-1]}^{\rm Weil}$. 

The monodromy representation $\rho^{(1,1)}_{\{0\}}$ can also be computed by 
using a state $v \in V_{\alpha = 0}$ other than $v = \ket{0}$; let us see 
that the results remain the same for the two other choices of $v$ 
in (\ref{eq:bos-states-VirPrim-011}). Naive analytic continuation of 
$f_{1\Omega'_-}^{\rm bos}(\tau_{ws};\beta)$ and $f_{1\overline{\Omega}'_-}^{\rm bos}$ 
in the argument $\tau_{ws}$ from $\tau_{ws}=\tau_*$ to 
$\tau_{ws} = \gamma^{-1} \cdot \tau_{*}$ in ${\cal H}$ would then turn 
$\langle \varphi_\beta, \; \partial_uX^\C \rangle_{\Sigma_{ws}(\tau_*)}$ to 
$\langle \varphi_\beta, \; \partial_{u^{\gamma^{-1}}}X^\C \rangle_{\Sigma_{ws}(\gamma^{-1} \cdot \tau_*)}$, 
rather than to 
$\langle \varphi_\beta, \; \partial_uX^\C \rangle_{\Sigma_{ws}(\gamma^{-1} \cdot \tau_*)}$; see 
(\ref{eq:SL2Z-act-on-tauNu}) for the definition of the choice of local coordinate 
$u^{\gamma^{-1}}$. So, the inserted vertex operator $\partial_uX^\C$ in $f_{1\Omega'_-}^{\rm bos}$ 
before the analytic continuation is $(\partial u^{\gamma^{-1}}/\partial u)=1/(-c\tau_*+a)$ 
times the inserted vertex operator in $f_{1\Omega'_-}^{\rm bos}$ after the continuation.\footnote{An alternative (but equivalent) to this argument is to study analytic 
continuation of $f_{1\Omega'_-}^{\rm bos}(\tau_{ws};\beta) du$; for the state 
$v \propto \alpha_{-1}^\C \ket{0}$, which is a Virasoro primary state, the 
combination $(\partial_u X^C)du$ is independent of the choice of a local coordinate 
on the Riemann surface $\Sigma_{ws}$. The continuation of $du$ to $du^\gamma$ then 
yields a factor $du^\gamma/du = 1/(c\tau_{ws}+d)$. }
Therefore, in order to work out the monodromy representation of the conformal 
block $f_{\{0\}}: {\cal V}_{\alpha = 0} \rightarrow \C$, we should look at the following:
\begin{align}
  \left(\gamma_T \circ f_{1\Omega'_-}^{\rm bos}\right)(\tau_{ws};\beta) & \; =
    1 \cdot f_{1\Omega'_-}^{\rm bos}(\tau_{ws}-1; \beta)
     = \frac{\mathbb{E}[-q_{\Lambda_-}(\beta, \beta)] \;
                  \vartheta_{\Lambda_-}^{1\Omega'_-}(\tau_{ws};\beta)}
            { \mathbb{E}[-1/12] \; (\eta(\tau_{ws}))^2 },   \nonumber \\
  & \; = \sum_{\gamma \in G_{\Lambda_-}} 
            [ (\vartheta_\eta)^{-2}\rho_{D\Lambda_-[-1]}^{\rm Weil}](T)_{\beta \gamma} \; 
            f_{1\Omega'_-}^{\rm bos}(\tau_{ws};\gamma). \\
  \left( \gamma_S \circ f_{1\Omega'_-}^{\rm bos}\right)(\tau_{ws};\beta) & \; =
    \frac{1}{(-\tau_{ws})} f_{1\Omega'}^{\rm bos}(-1/\tau_{ws};\beta), \nonumber \\
   & \;  = \frac{1}{(-\tau_{ws})}
     \frac{ \frac{i (-\tau_{ws})^2 }{\sqrt{|G_{\Lambda_-}|}}
             \sum_{\gamma \in G_{\Lambda_-}} \mathbb{E}[(\beta, \gamma)] \; 
                             \vartheta_{\Lambda_-}^{1\Omega'_-}(\tau_{ws};\gamma) }
         { (\mathbb{E}[1/8] \sqrt{-\tau_{ws}})^2 \; 
            (\eta(\tau_{ws}))^2 }, \nonumber \\
  & \; = \sum_{\gamma \in G_{\Lambda_-}} 
           [ (\vartheta_\eta)^{-2}\rho_{D\Lambda_-[-1]}^{\rm Weil}](S)_{\beta \gamma} \; 
            f_{1\Omega'_-}^{\rm bos}(\tau_{ws};\gamma).
\end{align}
Here, we used Lemma \ref{lemma:theta-transf} for the modular transformation 
of the type-1 theta functions. The extra factor $1/(-c\tau_{ws}+a)^{w=1}$ of 
the $(\partial u^{\gamma^{-1}} / \partial u)^{w=1}$ for the state level-$L_0=w=1$ 
Virasoro 
primary state $v$ cancels against the extra $(-c\tau_{ws}+a)^w$ in the modular 
transformation law of the weight $w=1$ theta function, and eventually 
the same monodromy representation $\rho^{(1,1)}_{\{0\}} = (\vartheta_\eta)^{-2} 
\rho_{D\Lambda_-[-1]}^{\rm Weil}$ of the conformal block $f_{\{0\}}$ is reproduced 
when we use a state $v \propto \alpha_{-1}^\C\ket{0}$. The same result is 
obtained for a state $v \propto \overline{\alpha}_{-1}^\C\ket{0}$. 

Note that the Weil representation on $\C[G_L]$ splits on the subspace spanned 
by $\{ (e_x + e_{-x}) \; | \; x \in G_L\}$ and another subspace 
by $\{ (e_x - e_{-x}) \; | \; x \in G_L\}$. 
Because $f_0^{\rm bos}(\tau_{ws};-\beta) = f_0^{\rm bos}(\tau_{ws};\beta)$, and 
$f_{1\Omega'_-}^{\rm bos}(\tau_{ws};-\beta) = - f_{1\Omega'_-}^{\rm bos}(\tau_{ws};\beta)$, 
both of the sub-representations are covered by the study above. 
\end{proof}

\begin{lemma}
\label{lemma:eta-transf} 
The Dedekind $\eta(\tau)$ function transforms under the action of 
$\gamma=[a,b; c,d] \in {\rm SL}(2;\Z)$ on the argument $\tau \in {\cal H}$ 
as follows:
\begin{align}
  (\eta(\gamma \cdot \tau))^2 = \vartheta_\eta^2(\gamma^{-1}) 
        \; (c\tau+d) \; (\eta(\tau))^2,
  \label{eq:eta-transf}
\end{align}
where $\vartheta_\eta^2(\gamma^{-1})=\mathbb{E}[2b/24]$ when $c=0$ and $d=1$, 
while\footnote{
The flooring function rounds down $x \in \R$ to the largest integer 
$\lfloor x \rfloor$ not greater than $x$. 
$\lfloor -2.4 \rfloor = -3$ and $\lfloor 2\rfloor =2$.} 
\begin{align}
 \vartheta^2(\gamma^{-1}) = \mathbb{E}\left[ \frac{2(a+d)}{24c} - s(d,c)
     - \frac{2}{8}\right],  \quad 
 s(d,c) := \sum_{r=1}^{c-1}\frac{r}{c}\left( \frac{dr}{c}- \left\lfloor \frac{dr}{c}\right\rfloor-\frac{1}{2} \right)
  \label{eq:vartheta-eta-repr-formula}
\end{align}
if $c>0$. 
For $\gamma$'s with $c = 0$ and $d=-1$ or $c < 0$, let $\gamma' := (-{\bf 1}_{2\times 2}) \cdot \gamma$, and $\vartheta_{\eta}^2(\gamma^{-1}) := 
- \vartheta_\eta^2(\gamma^{'-1})$.

In particular, 
\begin{align}
  {\rm for~}\gamma = T^{-1}, 
  & \quad 
 (\eta(\tau-1))^2 = \mathbb{E}[-2/24] (\eta(\tau))^2, \qquad 
   \vartheta^2_\eta(T) = \mathbb{E}[-2/24], \\
  {\rm for~}\gamma = S^{-1}, 
  & \quad 
  (\eta(-1/\tau))^2 = \mathbb{E}[+2/8] (\sqrt{-\tau} \eta(\tau))^2,  \qquad 
    \vartheta^2_\eta(S) = \mathbb{E}[+2/8].
\end{align}
For $\gamma = (- {\bf 1}_{2 \times 2})$, $\vartheta_\eta^2(\gamma) = -1$ and 
$(c\tau+d)=-1$, so (\ref{eq:eta-transf}) is trivial.

The multiplier $\vartheta_\eta^2(S) = \mathbb{E}[+2/8]$ for 
$S: \tau \rightarrow -1/\tau$ is the inverse-cubic power of $\vartheta^2_\eta(T) 
= \mathbb{E}[-2/24]$ for $T: \tau \rightarrow \tau +1$, which is not 
surprising because the $\tau \rightarrow -1/\tau$ transformation is obtained 
by a consecutive three Dehn twists of a $g=1$ curve $\Sigma = \Sigma_{ws}$.
$\bullet$ 
\end{lemma}

\begin{lemma}
\label{lemma:theta-transf}
Let $L$ be an even positive definite lattice of rank $2r_0$. The theta 
functions $\vartheta_L(\tau;x)$ for $x \in G_L$ transform under the 
action of $\gamma = [a,b; c,d] \in {\rm SL}(2;\Z)$ on the argument 
$\tau \in {\cal H}$ as follows:
\begin{align}
 {\rm for~}\gamma = T^{-1},  & \quad 
      \vartheta_L(\tau-1;x) = \mathbb{E}\left[ -\frac{q_L(x,x)}{2} \right] \;
   \vartheta_L(\tau; x), \\
 {\rm for~}\gamma = S^{-1}, & \quad 
  \frac{\vartheta_L(-1/\tau;x)}{(-\tau)^{r_0}} = \frac{\mathbb{E}[2r_0/8]}{\sqrt{|G_L|}}
     \sum_{y \in G_L} \mathbb{E}\left[ (x,y)_L \right] \; 
       \vartheta_L(\tau;y).
\end{align}

Let us fix a linear functional $\omega: L^\vee \rightarrow \C$, and think 
of how the type-1 theta functions $\vartheta_L^{1\omega}(\tau; x)$ for 
$x \in G_L$ transform under the action of $\gamma \in {\rm SL}(2;\Z)$.
Now, 
\begin{align}
  {\rm for~} \gamma = T^{-1}, & \quad 
  \vartheta_L^{1\omega}(\tau-1;x) =
        \mathbb{E}\left[ - \frac{q_L(x,x)}{2} \right] \; 
        \vartheta_L^{1\omega}(\tau; x), \\
  {\rm for~} \gamma = S^{-1}, & \quad 
    \frac{\vartheta_L^{1\omega}(-1/\tau;x)}{(-\tau)^{r_0+1}} = 
        \frac{\mathbb{E}[2r_0/8]}{\sqrt{|G_L|}} \; 
        \sum_{y \in G_L} \; \mathbb{E}\left[ (x,y)_L \right] \;
            \vartheta_L^{1\omega}(\tau;y).
\end{align}
$\bullet$
\end{lemma}

\begin{rmk}
We have remarked that both of the representations $\vartheta_\eta^2$ and 
$\rho_{D\Lambda_-[-1]}^{\rm Weil}$ are not projective representations of 
${\rm SL}(2;\Z)$ but its linear representations; this means that 
the monodromy representation $\rho^{(1,1)}_{\{0\}}$ for the $T^2$-target 
models of rational CFT are also linear (rather than projective) 
representations of the mapping class group 
${\rm SL}(2;\Z) = \Gamma_{1,1}$.

In general, monodromy representations $\rho^{(1,1)}_{\{\alpha\}}$ of 
a $(g,n)=(1,1)$ conformal block $f_{\{\alpha\}}$ of a general model 
of rational CFT are projective representations of $\Gamma_{1,1}$, 
not necessarily linear representations. The special fact above for 
the $T^2$-target models of rational CFT is consistent with the observation 
in \cite{MS}; Ref. \cite{MS} states that the projective representation 
of $\Gamma_{1,1}$ can be regarded a linear representation, when $\Gamma_{1,1}$
is replaced by its central extension of $\Gamma_{1,1}$ by a central element 
$R$, which multiplies\footnote{$\Delta(\alpha) +\Z \in \Q/\Z$ is the 
fractional part of the $L_0$-eigenvalues (conformal weights) of all the 
states in $V_\alpha$ to be inserted at the $n=1$ point of a genus $g$ Riemann 
surface $\Sigma_{ws}$.}
the complex phase $\mathbb{E}[\Delta(\alpha)] \in \C$; the second 
relation in (\ref{eq:relatn-SL2Z}) among the generators of $\Gamma_{1,1}$
is modified into $(S^2)^2 = R$; the first relation $S^2 = (ST)^3$ remains 
the same. For the obvious representation $\alpha = 0$ in $A_-$, 
$\Delta(\alpha = 0) = 0 \in \Q/\Z$, 
so the element $R$ is just an identity operator for the $(g,n)=(1,1)$ 
conformal block $f_{\{\alpha\}}$ for any model of rational CFT. We have 
remarked earlier that $f_{\{0\}}$ is the only non-zero $(1,1)$-conformal 
blocks in the case of the $T^2$-target models of rational CFT. $\bullet$
\end{rmk}

\subsection{The Chiral Correlation Functions of Interest: Type II String Theory}
\label{ssec:TypeII-confBlock}

In our previous paper \cite{prev.paper}, 
we focused on the following set of chiral 
correlation functions of the rational model of ${\cal N}=(2,2)$ SCFT 
for a set of data $([E_z]_\C, f_\rho)$:
\begin{align}
  f_0^{\rm II}(\tau_{ws};\beta) & \; := {\rm Tr}_{V_\beta^R} \left[ 
     e^{\pi i F_-} q^{L_0-\frac{c}{24}} F_- \right] 
        = \vartheta_{\Lambda_-}(\tau_{ws};\beta), \label{eq:chi-corrl-fcn-4-L0} \\
  f_{1\Omega'_-}^{\rm II}(\tau_{ws}; \beta) & \; :=
    - \frac{-i}{2\pi}\sqrt{\frac{2}{\alpha'}}{\rm Tr}_{V_\beta^R} \left[ 
     e^{\pi i F_-} q^{L_0-\frac{c}{24}} F_-  (\partial_u X^\C)(u=u_1) \right] 
        = \vartheta_{\Lambda_-}^{1\Omega'_-}(\tau_{ws}; \beta), 
       \label{eq:chi-corrl-fcn-4-L1}\\
  f_{1\overline{\Omega}'_-}^{\rm II}(\tau_{ws}; \beta) & \; := 
       - \frac{-i}{2\pi}\sqrt{\frac{2}{\alpha'}} {\rm Tr}_{V_\beta^R} \left[
      e^{\pi i F_-} q^{L_0-\frac{c}{24}} F_- (\partial_u \overline{X}^\C)(u=u_1)
      \right] = \vartheta_{\Lambda_-}^{1\overline{\Omega}'_-}(\tau_{ws};\beta),
  \label{eq:chi-corrl-fcn-4-L1Bar}
\end{align}
where $\tau_{ws} \in {\cal H}$ is the complex structure parameter 
of the worldsheet (Riemann surface) $\Sigma_{ws}$ of genus 1, and 
$\beta \in \Lambda_-^\vee/\Lambda_- = A_-$ is one of the irreducible 
representations of the left-mover (holomorphic) chiral algebra $(V_-,Y_-)$. 
The trace is taken over $V_\beta^R$, the vector space of all the left mover 
sates in 
the Ramond sector. The linear operator $F_-$ is the fermion number operator,
which is also the zero-modes of the U(1) current $J_-$ in the left-mover 
(holomorphic) ${\cal N}=2$ superconformal algebra on $\Sigma_{ws}$.
We have seen that the Mellin transform of $f_{1\Omega'_-}^{\rm II}$ and 
$f_{1\overline{\Omega}'_-}^{\rm II}$ [resp. $f_0^{\rm II}$] with respect to $\tau_{ws}$ 
can be used to write down the $L$-functions $L(H^1_{et}(E_z))$ 
[resp. Dedekind zeta function $\zeta_K(s)$ of the CM field $K$] of 
(some class of) arithmetic models $E_z$ of the $\C$-isomorphism 
class $[E_z]_\C$ we have started with. 

Now, beginning in this section \ref{ssec:TypeII-confBlock}, we 
wish to elaborate more on the observation above by exploiting 
the geometry of modular curves. For that purpose, we start off
by  re-capturing the chiral correlation functions above as 
objects defined on the moduli space of pointed Riemann surfaces 
of genus $g=1$.

Note, first, that $f_{1\Omega'_-}^{\rm II}$ and $f_{1\overline{\Omega}'_-}^{\rm II}$ 
could have been defined by using the operators $\alpha^\C_{-1}$ and 
$\overline{\alpha}^\C_{-1}$ 
in (\ref{eq:partX-expansion}, \ref{eq:partXbar-expansion}). In a basis 
of $V_\beta^R$ where the linear operators $L_0$ and $F_-$ are diagonalized, 
$\alpha_0^\C$ and $\overline{\alpha}^\C_0$ are diagonal;  
all other $\alpha^\C_m$'s and $\overline{\alpha}^\C_m$'s have zero diagonal 
entries, however, and hence they do not contribute to the trace. 
Exploiting this observation (as in (\ref{eq:bos-alpha0=partX}, 
\ref{eq:bos-alpha0=partX-cc})), we have chosen to use an expression 
with a local operator inserted, $\partial_uX^\C$ and 
$\partial_u \overline{X}^C$, rather than $\alpha^\C_0$ and 
$\overline{\alpha}^\C_0$.

Next, remember that 
\begin{align}
 F_- = \int_0^1 d\sigma_1 \; J_-(u), 
\end{align}
where $u$ is the analytic coordinate of the genus $g=1$ worldsheet 
Riemann surface $\Sigma_{ws}$, as in (\ref{eq:anal-repr-SigmaWs}), 
and let $(\sigma_1,\sigma_2) \in [0,1] \times [0,1]$ be a set of real 
coordinates on $\Sigma_{ws}$ so that $u = \sigma_1 + \tau_{ws} \sigma_2$ 
in a unit cell of the torus $\Sigma_{ws} = \C/(\Z+\tau_{ws}\Z)$.
Due to the conservation of the fermion number $F_-$, we can further replace 
this expression by 
\begin{align}
  F_- = \int_{\Sigma_{ws}} d\sigma_1 d\sigma_2 \; J_-(u) =
   \int_{\Sigma_{ws}} \frac{dud\bar{u}}{2{\rm Im}(\tau_{ws})} \; J_-(u).
\end{align}
So,\footnote{The $J_-$--$(\partial_uX^\C)$ chiral correlation function
and also the $J_-$--$(\partial_u\overline{X}^\C)$ chiral correlation function 
do not have singularity as a function on the coordinate $u_2-u_1$ on 
$\Sigma_{ws}$.} the chiral correlation functions $f_0^{\rm II}$, 
$f_{1\Omega'_-}^{\rm II}$, and $f_{1\overline{\Omega}'_-}^{\rm II}$ have alternative 
expressions:
\begin{align}
  f_0^{\rm II}(\tau_{ws};\beta) & \; =
     \int_{\Sigma_{ws}} \frac{du_2d\bar{u}_2}{2{\rm Im}(\tau_{ws})} \; 
      \langle \varphi_\beta, \; J_-(u_2) \; Y_-(v_0; u_1) \;
    \rangle_{\Sigma_{ws}},  \label{eq:f0-as-12confBl} \\
  f_{1\Omega'_-}^{\rm II}(\tau_{ws};\beta) & \; =
     \int_{\Sigma_{ws}} \frac{du_2d\bar{u}_2}{2{\rm Im}(\tau_{ws})} \; 
      \langle \varphi_\beta, \; J_-(u_2) \; Y_-(v_{1\Omega'_-}; u_1) \;
      \rangle_{\Sigma_{ws}},   \label{eq:f1Omega-as-12confBl} \\
 f_{1\overline{\Omega}'_-}^{\rm II} (\tau_{ws};\beta) & \; =
     \int_{\Sigma_{ws}} \frac{du_2d\bar{u}_2}{2{\rm Im}(\tau_{ws})} \; 
      \langle \varphi_\beta, \; J_-(u_2) \; Y_-(v_{1\overline{\Omega}'_-}; u_1)
        \; \rangle_{\Sigma_{ws}},   \label{eq:f1OmegaBar-as-12confBl}
\end{align}
where the three states $v_0$, $v_{1\Omega_-}$, and $v_{1\overline{\Omega}'_-}$ in 
$V_{\alpha_1 = 0}^{NS}$ are those in (\ref{eq:bos-states-VirPrim-011}). Here, 
we take $p_1$ and $p_2$ to be arbitrary points in 
$\Sigma_{ws} = \C/(\Z+\tau_{ws}\Z)$, and $u_1$ and $u_2$ 
to be the values of their 
complex analytic $u$-coordinate (modulo $\Z+\tau_{ws}\Z$) in the analytic 
representation (\ref{eq:anal-repr-SigmaWs}). 
The notation 
\begin{align}
  \langle \varphi_\beta, \; {\cal O}_2(p_2) \; {\cal O}_1(p_1) \; 
  \rangle_{\Sigma_{ws}}
\end{align}
for two vertex operators ${\cal O}_i$ inserted at two points 
$p_i \in \Sigma_{ws}$ ($i=1,2$) means the chiral correlation functions 
of the two vertex operators with contributions from all the states in the 
representation $\beta \in \Lambda_-^\vee/\Lambda_-$ (but nothing else) 
in the two pants-leg-sewing loci designated in Figure \ref{fig:g1n2}.
\begin{figure}[tbp]
\begin{center}
\begin{tabular}{cc}
  \includegraphics[width=0.4\linewidth]{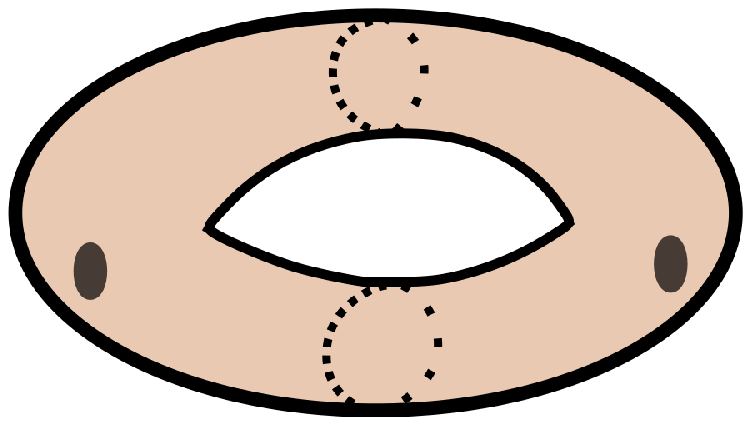} &
  \includegraphics[width=0.4\linewidth]{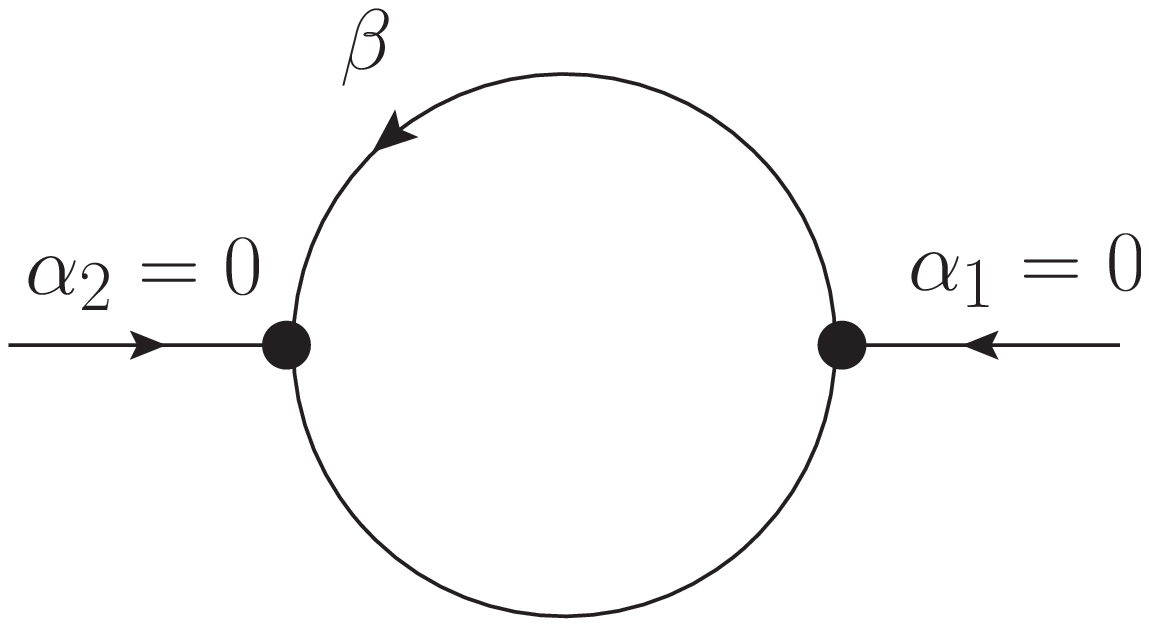} \\
 (a) & (b) 
\end{tabular}
 \caption{\label{fig:g1n2}A $g=1$ Riemann surface with $n=2$ marked points 
is regarded as two pants sewn together at their legs (along the two 
dotted lines) in (a). The degeneration limit corresponding to this pants 
decomposition is described by the skeleton graph (b).}
\end{center}
\end{figure}

The observations up to this point show that the integrands of 
(\ref{eq:f0-as-12confBl}, \ref{eq:f1Omega-as-12confBl}, 
\ref{eq:f1OmegaBar-as-12confBl}) are all $(g,n)=(1,2)$ chiral correlation 
functions of local operators on $m_* = [(\Sigma_{ws}(\tau_{ws}); \{ p_1, p_2\})] 
\in {\cal M}_{1,2}$, a $(g,n)=(1,2)$ pointed Riemann surface.
Those integrands (\ref{eq:f0-as-12confBl}, \ref{eq:f1Omega-as-12confBl}, 
\ref{eq:f1OmegaBar-as-12confBl}) are the $(g,n)=(1,2)$ conformal blocks 
$[f_{\{0NS,0NS\}}]_{\{\beta R,\beta R\}}$ on $m_* \in {\cal M}_{1,2}$ for the choice 
$\{ \alpha_1, \alpha_2\} = \{ 0NS,0NS \}$ of the irreducible representations 
for the vertex operator insertion at $n=2$ points, evaluated for 
the choice of $Y_-(v_{0,1\Omega'_-,1\overline{\Omega}'_-}, u)$ for the vertex 
operator at $p_1$ and $J_-(u)$ for the one at $p_2$.

The vertex operators to be inserted, $Y_-(v_{1\Omega'_-, 1\overline{\Omega}'_-}; u)$ 
and $J_-(u)$, depend on the choice of a local coordinate of $\Sigma_{ws}$ 
at the point of insertion.\footnote{Although the complex $u$-coordinate 
in (\ref{eq:anal-repr-SigmaWs}) looks like ``the most canonical'' choice 
of such a local coordinate for the case of $g=1$ Riemann surface $\Sigma_{ws}$, 
in fact, there is ${\rm SL}(2;\Z)$ transformation acting on such a choice 
of the holomorphic coordinate of the analytic 
representation (\ref{eq:anal-repr-SigmaWs}).}  Only a pair of 
a state $v \in V_{\alpha = 0}^{NS}$ and a choice of a local coordinate of 
$\Sigma_{ws}$ determines a local section of the bundle 
${\cal V}_{\alpha = 0}^{NS}$ on $\Sigma_{ws}$ of the vertex algebra of 
a rational model of ${\cal N}=(2,2)$ SCFT, and 
a local section of the bundle $\widetilde{\cal V}_{\alpha = 0}^{NS}$ on 
${\cal M}_{g,1}$. 
However, in the case the vertex operator $Y_-(v;u)$ is Virasoro primary, 
which is the case for all of $Y_-(v_{0,1\Omega'_-,1\overline{\Omega}'_-};u)$ and 
$J_-(u)$, the combination $Y_-(v;u) (du)^{h(v)}$ with $h(v)$ the $L_0$ 
eigenvalue of the state $v$ determines a local section of the bundle 
${\cal V}_{\alpha = 0}^{NS} \otimes (T^*\Sigma_{ws})^{\otimes h(v)}$ on $\Sigma_{ws}$, 
and a local section of the bundle $\widetilde{\cal V}_{\alpha = 0}^{NS} \otimes 
(T^*\Sigma_{ws})^{\otimes h(v)}$ over ${\cal M}_{g,1}$ that depends only on 
the choice of a Virasoro primary state $v$, but not on the choice of 
a local coordinate on $\Sigma_{ws}$ \cite{TUY, FBz}. In this perspective, 
the $(g,n)=(1,2)$ conformal blocks $[f_{\{0NS,0NS\}}]_{\{\beta R,\beta R\}}$ send 
linearly the local sections of the bundle 
$\otimes_{i=1}^n \pi_i^* \left( \widetilde{\cal V}_{\alpha_i=0}^{NS} \otimes 
(T^*\Sigma)^{h(v_i)} \right)$ (when all the $v_i$'s are Virasoro primary) to 
sections of $\otimes_{i=1}^n \pi_i^*((T^*\Sigma_{ws})^{h(v_i)})$ over 
${\cal M}_{g,n}$. The monodromy representation $\rho^{(1,2)}_{\{0NS,0NS\}}$ of 
the group $\Gamma_{1,2}$ (the mapping class group of $(g,n)=(1,2)$ pointed 
Riemann surfaces) for the conformal blocks $[f_{\{0NS,0NS\}}]_{\{\beta R,\beta R\}}$, 
or for the chiral correlation 
functions in the integrands of $f_0^{\rm II}$, $f_{1\Omega'_-}^{\rm II}$ and 
$f_{1\overline{\Omega}'_-}^{\rm II}$, can be studied without paying attention 
to the choice of a local coordinate on $\Sigma_{ws}$. The representation 
matrices of the elements of $\Gamma_{1,2}$ are locally constant.

We are interested in the integrals, $f_0^{\rm II}$, $f_{1\Omega'_-}^{\rm II}$, and 
$f_{1\overline{\Omega}'_-}^{\rm II}$, rather than their integrands. The local chiral 
correlation functions in (\ref{eq:f0-as-12confBl}--\ref{eq:f1OmegaBar-as-12confBl})---the integrands---are completely independent of the choice of two points 
$p_1, p_2 \in \Sigma_{ws}$
in the genus $g=1$ Riemann surface $\Sigma_{ws}$, so there is no non-trivial 
monodromy associated with the choice of the point $p_2$ relatively to $p_1$, 
or there is no ambiguity in the definition of the integrand. So, the 
integration $\int_{\Sigma_{ws}} \; du_2 d\bar{u}_2$ is simple and straightforward. 
Suppose that 
\begin{align}
 \gamma \circ 
  \langle \varphi_{\beta}, & \; J_-(u_2) \; Y_-(v_{1\Omega'_-}; u_1) \rangle \; 
      du_2 du_1 \nonumber \\
 &  = \sum_{\beta' \in iReps.^-} \left(\rho^{(1,2)}_{\{0,0\}}(\gamma)\right)_{\beta,\beta'} 
      \langle \varphi_{\beta'}, \; J_-(u_2) \; Y_-(v_{1\Omega'_-}; u_1) \rangle \; 
      du_2 du_1
\end{align}
where $\gamma$ is a closed path in ${\cal M}_{1,2}$. 
Because the integration measure $du_2d\bar{u}_2 / {\rm Im}(\tau_{ws})$ 
is invariant under the ${\rm SL}(2;\Z)$ transformation, 
we have the following transformation law under the analytic continuation along 
$[\gamma] \in \Gamma_{1,1} = {\rm SL}(2;\Z)$:
\begin{align}
 \left( [\gamma] \circ f_{1\Omega'_-}^{\rm II} du_1 \right)(\tau_{ws}; \beta) & \; = 
   \sum_{\beta' \in iReps} \left( \rho^{(1,2)}_{\{0,0\}}(\gamma) \right)_{\beta,\beta'}
    (c\tau_{ws}+d) f_{1\Omega'_-}^{\rm II}du_1, \\
 \left( [\gamma] \circ f_{1\Omega'_-}^{\rm II} \right)(\tau_{ws}; \beta) & \; = 
   \sum_{\beta' \in iReps} \left( \rho^{(1,2)}_{\{0,0\}}(\gamma) \right)_{\beta,\beta'}
    (c\tau_{ws}+d)^2 f_{1\Omega'_-}^{\rm II}, 
\end{align}
if $[\gamma] = [a, b; c,d]$ in ${\rm SL}(2;\Z)$. We find in this way 
that $f_{1\Omega'_-}^{\rm II}(\tau_{ws};\beta)$ labeled by $\beta \in iReps$ 
are weight-2 modular forms taking their values in the vector space 
$\C[iReps]$. When they are multiplied by a differential $d\tau_{ws}$ on 
${\cal H}$, each---$f_{1\Omega'_-}^{\rm II}(\tau_{ws};\beta) d\tau_{ws}$---can be 
regarded as a vector-valued holomorphic 1-form without a pole on the 
compactification of the curve ${\cal H}/{\rm SL}(2;\Z)$, subject to the 
monodromy determined by $\rho^{(1,2)}_{\{0,0\}}$; due to the $p_1$--$p_2$ 
independence of the $(g,n)=(1,2)$ local correlation function with 
the two vertex operators $\partial_u X^\C$ and $J_-$ inserted, 
the monodromy matrix $\rho^{(1,2)}_{\{0NS,0NS\}}(\gamma)$ of $[\gamma] \in 
\Gamma_{1,1}$ does not depend on the choice of a 
representative $\gamma \in \Gamma_{1,2}$.
Repeating the same logic, we obtain the same transformation law for 
$f_{1\overline{\Omega}'_-}^{\rm II}$ (cf. \cite{prev.paper}).

\begin{defn}
There is no convenient jargon for $(g,n)$ chiral correlation functions 
on ${\cal M}_{g,n}$ that are integrated along the fiber of 
${\cal M}_{g,n} \rightarrow {\cal M}_{g,n-1}$. As we have to refer to 
$f_{1\Omega'_-}^{\rm II}$ and $f_{1\overline{\Omega}'_-}^{\rm II}$ many times in this 
article, we refer to them as ``$(g,n)=(1,2)$ chiral correlation functions'',
or ``chiral correlation functions'' with the quotation marks.
\end{defn}

\begin{props}[\S4.1 of \cite{prev.paper}]
Think of a rational model of ${\cal N}=(2,2)$ SCFT associated with 
a set of data $([E_z]_\C, f_\rho)$. We have seen that the ``$(g,n)=(1,2)$ 
chiral correlation functions'' $\{ f_{1\Omega'_-}^{\rm II}(\tau_{ws};\beta) \; 
| \; \beta \in iReps \}$ are regarded as a local section on 
${\cal M}_{1,1}$ of the bundle $\C^{\oplus |iReps|} \otimes T^*{\cal M}_{1,1}$, 
which fails to be a global section on ${\cal M}_{1,1}$, due to the 
monodromy that descends from the monodromy representation of 
$\Gamma_{1,2}$. The monodromy representation 
of $\Gamma_{1,1} \cong {\rm SL}(2;\Z)$ for the sections
$\{ d\tau_{ws} f_{1\Omega'_-}^{\rm II}(\tau_{ws};\beta) \; | \; \beta \in iReps\}$
is $\rho_{D\Lambda_-}^{\rm Weil}$.
\end{props}

\begin{proof} 
The last sentence is from a property of theta functions; 
see Lemma \ref{lemma:theta-transf}, or many other references.
\end{proof}

\begin{rmk}
\label{rmk:triangle-shcematic}
In the theory of modular parametrization, an elliptic curve $E$ defined 
over $\Q$ has its $L$-function $L(H^1_{et}(E \times_\Q \overline{\Q}),s)$, 
and it is known that 
there must be a weight-2 modular form $f(\tau)$ of a congruence subgroup 
$\Gamma_0(N)$ of ${\rm SL}(2;\Z)$ for some $N \in \N_{>0}$ so that 
i) the Mellin transform of $f(\tau)$ 
agrees with $L(H^1_{et}(E \times_\Q \overline{\Q}),s)$ 
(i.e., the $L$-function of the automorphic representation on $f$ agrees with 
$L(H^1_{et}(E \times_\Q \overline{\Q} ),s)$), and ii) the correspondence 
between $f$ and $E/\Q$ is given by a morphism $(\nu \circ \mu): X_0(N) 
\rightarrow E$ defined over $\Q$, in that $(\nu \circ \mu)^*(\omega_E) = 
d\tau f(\tau)$ for a holomorphic (1,0)-form $\omega_E$ on $E$ 
(cf. \cite[p. 845]{Taylor-Wiles}).

In the previous paper \cite{prev.paper}, we have seen for a class of elliptic 
curves $E$ (explained shortly) that 
the ``$(g,n)=(1,2)$ chiral correlation functions'' 
$f_{1\Omega'_-}^{\rm II}(\tau_{ws};\beta)$ in the rational model of 
${\cal N}=(2,2)$ SCFT for $([E]_\C, f_\rho)$ for some appropriately chosen 
$f_\rho \in \N$ can play a role a little more general than that of $f(\tau)$ 
above for $L(H^1_{et}(E \times \overline{\Q}),s)$. The class of arithmetic 
models of elliptic curves under consideration (in \cite{prev.paper} as well 
as in this article) is the {\it elliptic curves of Shimura type}; 
section \ref{sssec:shimuraEC} of this article explains what it is; 
the field of definition of $E$ is not necessarily $\Q$ (more general), 
but $E$ has complex multiplication (more special). 

It is one of central questions in this article how general the relation 
between the CFT correlation functions and the arithmetic modular forms is. 
Here is an observation that makes us feel that the relation is not just 
outright coincidence. Remember that $f_{1\Omega'_-}^{\rm II}$'s are obtained (by 
integrating partially) from the $n=2$ point correlation functions 
\begin{align}
  \langle \varphi_\beta, \; (J_-du)(p_2) \; (du\partial_uX^\C)(p_1) \rangle,
\end{align}
and one of the two operators to be inserted is $du (\partial_u X^\C)$.  
This $dX^\C$ is the holomorphic (1,0)-form of the target space $[E]_\C$, 
so it is $\omega_E$. The model of superconformal field theory in 
consideration is formulated (roughly speaking) by using 
path-integration over the space $Map(\Sigma_{ws}, ([E]_\C,f_\rho))$, and 
the insertion of an operator $(du \partial_uX^\C)$ in a correlation function 
is to include a factor 
\begin{align}
  P^{(1,0)} \left( \phi^*(dX^\C) \right) \propto 
  P^{(1,0)} \left( \phi^*(\omega_E) \right) 
   \label{eq:pulled-back-1forms}
\end{align}
in the integrand of the path-integration over 
$Map(\Sigma_{ws}, ([E]_\C, f_\rho))$; here, $P^{(1,0)}$ is the 
projection $H^1(\Sigma_{ws};\C) \rightarrow H^{1,0}(\Sigma_{ws};\C)$,   
and $\phi \in Map(\Sigma_{ws}; ([E]_\C, f_\rho))$. So, the ``$(g,n)=(1,2)$ 
chiral correlation functions'' $f_{1\Omega'_-}^{\rm II} du$, associated with 
$\omega_E \in H^{1,0}([E]_\C;\C)$ in its definition, has a good reason to 
play a role similar to that of the modular form $(\nu \circ \mu)^*(\omega_E)
 \propto f(\tau)d\tau$ in the 
theory of modular parametrization. In the rest of this article, we will 
elaborate more on this observation, by closing the remaining edges of the 
following triangle (\ref{eq:schematic-triangle}). 
\begin{align}
  \vcenter{ \xymatrix{
 f \in S_2(\Gamma) &   X_0(N) \; {\rm or} \; X_\Gamma \ar[rdd]^{\nu \circ \mu} & (\nu \circ \mu)^*(\omega_E)=d\tau f(\tau)
   \\ f^{\rm II}_{1\Omega'_-}(\tau_{ws};\beta) \propto \int_{\Sigma_{ws}} d^2u \langle J_-(u) \phi^*(\omega_E) \rangle & & \\
Map(\Sigma_{ws},([E]_\C,f_\rho)) \ni \phi:\Sigma_{ws} \ar[rr] \ar@{.}[ruu] &   & E/\Q \; {\rm or} \; E/k 
  }}
  \label{eq:schematic-triangle}
\end{align}
Sections \ref{ssec:modC4CCF}--\ref{ssec:diff-frho} are for the upper left 
edge, while section \ref{ssec:map2shimuraEC} is for the upper right edge. 
$\bullet$
\end{rmk}

\subsection{Modular Curves for the $g=1$ Chiral Correlation Functions}
\label{ssec:modC4CCF}

\subsubsection{Bosonic String Theory: Warming Up}

The $(g,n)=(1,1)$ chiral correlation functions $f_0^{\rm bos}(\tau_{ws};\beta)$, 
$f_{1\Omega'_-}^{\rm bos}(\tau_{ws};\beta)$, and 
$f_{1\overline{\Omega}'_-}^{\rm bos}(\tau_{ws};\beta)$ for $\beta \in iReps$ in 
the given model of rational CFT for a set of data $([E_z]_\C, f_\rho)$ can 
be regarded as functions on the entire upper complex half plane, 
$\tau_{ws} \in {\cal H}$; for each one of $\beta \in iReps$, 
they are single valued functions.

An alternative way to look at them is to think of 
$[f_0^{\rm bos}(\beta)]_{\beta \in iReps}$ as a section of rank-$|iReps|$
trivial bundle $\C^{|iReps|}$ over ${\cal M}_{1,1} = {\cal H}/{\rm PSL}(2;\Z)$
that is defined locally; $[f_{1\Omega'_-}^{\rm bos}]_{\beta \in iReps}$ as a whole 
can be regarded as a section of the rank=$|iReps|$ 
bundle $\C^{|iReps|} \otimes (T^*{\cal M}_{1,1})^{\otimes 1/2}$ over 
${\cal M}_{1,1}$ defined locally. The sections 
$[f_0^{\rm bos}]_{\beta \in iReps}$ and $[f_{1\Omega'_-}^{\rm bos}]_{\beta \in iReps}$
fail to be extended to a well-defined section defined globally on 
${\cal M}_{1,1}$ because of the non-trivial monodromy representation 
$\rho^{(1,1)}_{\{0\}}$ of the mapping class group 
$\Gamma_{1,1} = {\rm SL}(2;\Z)$ (Example \ref{exmpl:bos-T2models-monodromy}).

There is a third way that comes in between the two different (but compatible)
ways described above in thinking of $f_0^{\rm bos}$, $f_{1\Omega'_-}^{\rm bos}$ and 
$f_{1\overline{\Omega}'_-}^{\rm bos}$. Let $\Gamma^{(1,1)}_{\{0\}} \subset 
{\rm SL}(2;\Z)$ be the kernel of the monodromy representation 
$\rho^{(1,1)}_{\{0\}}$. For any subgroup $\Gamma$ that is contained in 
$\Gamma^{(1,1)}_{\{0\}}$, we can think of a curve $X_\Gamma$, which is 
a compactification of the open space ${\cal H}/\Gamma$, and a projection 
morphism $p_\Gamma: X_\Gamma \rightarrow {\cal H}/{\rm SL}(2;\Z) \cong 
{\cal M}_{1,1}$. The pull-back of the Friedan--Shenker bundle ${\cal F}_{\{0\}}$, 
$p_\Gamma^*({\cal F}_{\{0\}})$ on $X_\Gamma$, is a trivial bundle 
$\C^{|iReps|}$ then. The locally defined section 
$[f_0^{\rm bos}]_{\beta \in iReps}$ of $\C^{|iReps|}$ over ${\cal M}_{1,1}$ is 
also pulled back to $p_\Gamma^*(\C^{|iReps|})$ over $X_\Gamma$ where 
the section $p_\Gamma^*([f_0^{\rm bos}]_{\beta \in iReps})$ is globally 
well-defined (the monodromy is trivial). Similarly, the section 
$p_\Gamma^*([f_{1\Omega'_-}^{\rm bos}]_{\beta \in iReps})$ is a section 
of the bundle $(T^*X_\Gamma)^{\otimes 1/2}$ that is well-defined globally over 
the curve $X_\Gamma$ (the monodromy is trivial). 

This intermediate way to look at the conformal blocks and chiral correlation 
functions will sound very natural to arithmetic geometers, but may sound 
a little strange to string theorists. If we do not want to have 
non-trivial monodromy around, then why don't we lift all the chiral 
correlation functions all the way to the upper complex half plane, which 
is available for any model, and for any choice of the representation 
$\{ \alpha\}$ for the 
vertex operator insertion. Why do we bother to think of a partial lift 
to $X_\Gamma$ that depends on a model of rational CFT? An answer to this 
question is given in Remark \ref{rmk:motivatn4partialLift} after 
we repeat the same discussion for the ``$(g,n)=(1,2)$ chiral correlation 
functions'' $f_0^{\rm II}$, $f_{1\Omega'_-}^{\rm II}$, and 
$f_{1\overline{\Omega}'_-}^{\rm II}$ of Type II string theory. 

Let us stick to the model of rational CFT assigned by bosonic string theory 
for a set of data $([E_z]_\C, f_\rho)$ for a while. In this case, 
the monodromy representation $\rho^{(1,1)}_{\{0\}}$ is known explicitly; 
see Example \ref{exmpl:bos-T2models-monodromy}. So, we can extract more 
information on the kernel $\Gamma^{(1,1)}_{\{0\}} \subset \Gamma_{1,1} = 
{\rm SL}(2;\Z)$. 

\begin{lemma}
\label{lemma:kern-theta-eta}
The kernel of the 1-dimensional representation $(\vartheta_\eta)^2: 
{\rm SL}(2;\Z) \rightarrow S^1$ contains the principal congruence subgroup 
$\Gamma(12)$.
\end{lemma}

\begin{proof} One can use \cite{SAGE} to generate a set of generators of 
$\Gamma(12)$, 
which consists of $A = [1,12; 0,1] \in {\rm SL}(2;\Z)$, 66 elements 
of the form $A = [a,b; c,d] \in {\rm SL}(2;\Z)$ with $c>0$, and 
30 other elements of the form $A = [a,b;c,d] \in {\rm SL}(2;\Z)$ 
with $c<0$. To prove that $\Gamma(12)$ is contained within the kernel, it is 
enough to verify that $\vartheta_\eta^2(A) = 1$ for all those generators 
($A$'s) of $\Gamma(12)$. 

As for the first generator $A = [1,12; 0,1] = T^{12}$ of $\Gamma(12)$, 
$(\vartheta_\eta)^2(T^{12}) = (\mathbb{E}[2/24])^{12}=1$. 

For $A$'s with $c>0$, we have evaluated 
(\ref{eq:vartheta-eta-repr-formula}) for all the 66 elements, and confirmed 
that they all satisfy $\vartheta^2_\eta(A)=1$. 

For $A$'s with $c<0$, we have verified that $\vartheta^2_\eta(-A)=
-1$ by evaluating (\ref{eq:vartheta-eta-repr-formula}) 
explicitly for all the 30 generators of this type; this means once again 
that $\vartheta^2_\eta(A)=1$ for those 30 generators of $\Gamma(12)$. 
\end{proof}

Let us now move on to the kernel of the Weil representation 
$\rho_{D\Lambda_-}^{\rm Weil}$; we quote the relevant result in 
Lemma \ref{lemma:Weil-repr-kernel}, but before doing so, we need 
to review more technical results. 

\begin{defn}
\label{defn:level-DGamma}
Let $DL :=(G_L,q_L)$ be the discriminant form associated with an even 
lattice $L$. The {\it level of the discriminant form $DL$} is defined by 
\begin{align}
N_{DL} := {\rm GCD}[ N \in \Z \; | \; N(x,x) \in 2\Z \; 
{\rm for~}{}^\forall x \in G_L].
\end{align}
We also introduce a closely related integer, 
\begin{align}
N_{DL}' := {\rm GCD}[ N \in \Z \; | \; 
Nx = 0 \in G_L \; {\rm for~}{}^\forall x \in G_L]. 
\end{align}
$N_{DL}$ is either equal to $N'_{DL}$ or $2N'_{DL}$.  
\end{defn}

\begin{lemma}
\label{lemma:NDLambda}
Think of the discriminant form $D\Lambda_- = (G_{\Lambda_-}, q_{\Lambda_-})$ 
of the even rank-2 positive definite lattice $\Lambda_-$ introduced 
in section \ref{ssec:string}. Then $N_{D\Lambda_-}=D_z f_\rho$. Furthermore, 
\begin{align}
  N'_{D\Lambda_-} = \left\{ \begin{array}{ll} 
    N_{D\Lambda_-},  & b_z \equiv 1 (2), \\
    N_{D\Lambda_-}/2, & b_z \equiv 0 (2).
    \end{array} \right.
\end{align}
\end{lemma}

\begin{proof}
For any $x, y \in \Z$, $N_{D\Lambda_-}$ has to satisfy 
\begin{align}
  \left( x , y \right) \frac{N_{D\Lambda_-}}{f_\rho D_z}
     \left( \begin{array}{cc} 2c_z & -b_z \\ -b_z & 2a_z \end{array} \right)
  \left( \begin{array}{c} x \\ y \end{array} \right) \in 2\Z.
\end{align}
It is easy to see that $N_{D\Lambda_-} | f_\rho D_z $, because the condition 
above is satisfied when $N_{D\Lambda_-}$ is replaced by $f_\rho D_z$.
Conversely, we also see that $2f_\rho D_z | 2c_z N_{D\Lambda_-}$ is necessary 
for $(x,y)=(1,0)$. Similarly, $2f_\rho D_z | 2a_z N_{D\Lambda_-}$ and 
$2f_\rho D_z | 2b_z N_{D\Lambda_-}$ are necessary for $(x,y)=(0,1)$ and 
 for $(x,y)=(1,1)$, respectively. 
Because $(a,b,c)=1$, it is also necessary that $f_\rho D_z | N_{D\Lambda_-}$. 
So, we see that $N_{D\Lambda_-} = f_\rho D_z$. 

As for $N'_{D\Lambda_-}$, note that the discriminant group of $\Lambda_-$ is always 
of the form $G_{\Lambda_-} \cong \Z/(f_\rho \delta_1) \oplus \Z/(f_\rho \delta_2)$, 
with $\delta_1 | \delta_2$ and $\delta_1\delta_2 = D_z = |D_K|f_z^2$. 
If $\delta_1 > 1$, then that means that the matrix $[2a_z,b_z; b_z,2c_z]$---the intersection 
form (\ref{eq:mom-latt-int-form}) without the factor $f_\rho$---is divisible 
by $\delta_1$. So, $\delta_1 > 1$ if and only if $b_z$ is even (because 
$(a_z,b_z,c_z)$ are mutually prime). In that case, $\delta_1 = 2$ and  
$N'_{D\Lambda_-}=f_\rho \delta_2$, which is $f_\rho D_z/\delta_1 = N_{D\Lambda_-}/2$. 
If $b_z$ is odd, instead, then $\delta_1 = 1$, and 
$N'_{D\Lambda_-} = f_\rho \delta_2 = f_\rho D_z = N_{D\Lambda_-}$. 
\end{proof}

\begin{cor}
\label{cor:discr-grp-isom}
The discriminant group $G_\Lambda$ (as an abelian group) is determined by 
$D_K$, $f_z$, and $f_\rho$. In particular, even when there are more than one 
$\C$-isomorphism classes of elliptic curves with complex multiplication by 
${\cal O}_{f_z}$, $h({\cal O}_{f_z}) > 1$, the abelian groups $G_{\Lambda}$ 
of all those $h({\cal O}_{f_z})$ CM elliptic curves in (\ref{eq:def-Ell-Ofz}) 
are all isomorphic to one another. 
\end{cor}

\begin{proof} Discussion in the proof of the preceding Lemma almost already 
proved this. To add a little more, we see that $b_z$ is odd if both $f_z$ and 
$D_K$ are odd; $b_z$ is even otherwise, which means that either all the 
$h({\cal O}_{f_z})$ CM elliptic curves have even $b_z$, or all of them have 
odd $b_z$. 
\end{proof}

\begin{lemma}
Let $\underline{M} = (M, q_M)$ be a non-degenerate finite quadratic module. 
Then there exists a decomposition 
\begin{align}
  \underline{M} \cong \oplus_{p: {\rm prime}} \; \underline{M}(p); 
\end{align}
the finite abelian group $M$ is decomposed into the direct sum of 
the groups $M(p)$ labeled by prime numbers $p$, where $M(p)$ is the set of 
elements of $M$ whose order is a power of $p$; the bilinear form $(-,-)_M$ 
derived from $q_M$ splits (becomes block diagonal) into the pieces labeled 
by the prime numbers. Moreover, for an odd prime number $p$, 
\begin{align}
  \underline{M}(p) \cong \oplus_{i=1}^{n_p} \left( \underline{A}_{p^{k_i}}^{t_i} \right)
\end{align}
for an appropriately chosen set of $\{ (k_i, t_i)_{i=1,\cdots, n_p} \}$ where 
$k_i \in \N_{>0 }$ and $t_i \in [\Z/(p^{k_i})]^\times$. 
$\underline{A}_{p^k}^t$ is a finite quadratic module on the abelian group $\Z/(p^k)$ 
where the quadratic form assigns the generator $1$ of the cyclic group to 
$2t/p^k + 2\Z \in \Q/2\Z$.

The finite quadratic module $\underline{M}(2)$ also has a similar result 
on its structure, although that is a little more complicated. 
See section 2.3 of \cite{Stroemberg} for more details. $\bullet$
\end{lemma}

\begin{defn}
Let $\underline{M} = (M,q_M)$ be a non-degenerate finite quadratic module. 
Then it is known that 
\begin{align}
  \frac{1}{\sqrt{|M|}} \sum_{y \in M} \mathbb{E}\left[ \frac{q_M(y)}{2} \right]
  =: \mathbb{E}\left[ \frac{{\rm sgn}(\underline{M})}{8} \right]
\end{align}
is an 8-th root of unity. An integer ${\rm sgn}(\underline{M})$ mod 8 is 
defined by the relation above, and is called the {\it signature of the finite 
quadratic module}.  
If a finite quadratic module is that of an even positive 
definite lattice, its signature is the same as the rank of the lattice modulo 8. 
$\bullet$
\end{defn}

\begin{lemma} [cf. Lemma 5.5 of \cite{Stroemberg}; also \cite{Weil-repr}] 
\label{lemma:Weil-repr-kernel}
Let $L$ be a rank-2 even lattice that is either positive definite or 
negative definite, and $DL$ its discriminant form. 
Then the kernel of the Weil representation $\rho^{\rm Weil}_{DL}$ contains 
the group $\Gamma(N_{DL})$. 
Furthermore, for an element\footnote{
The group $\Gamma_0^0(N) \subset {\rm SL}(2;\Z)$ is defined 
in (\ref{eq:def-Gamma00}).
} $A \in \Gamma_0^0(N_{DL}) \subset {\rm SL}(2;\Z)$, 
whose $2 \times 2$ matrix representation is $[a, b; c, d]$,  
\begin{align}
  A \cdot e_x = e_{d \cdot x} \; \epsilon_{DL}^{-1}(d), 
\end{align}
where $x$ is any element of the set $G_L = L^\vee/L$, and $e_x$ the corresponding 
generator of the vector space $\C[G_L]$ of the Weil representation. 
$\epsilon_{DL}^{-1}$ is a character of the multiplicative group 
$[\Z/(N_{DL})]^\times$ given\footnote{$(-/-)$ stands for the Jacobi symbol. 
Ref. \cite[\S1.2]{Stroemberg} uses a convention (*a) that $(x/2)=-1$ for 
$x \equiv 3$ mod 8, 
and $(x/2)=+1$ for $x \equiv 7$ mod 8; another convention (*b) is to assign 
$(x/2)=0$ for $x \equiv 3,7$ mod 8 \cite[\S3]{Miyake}. We follow the 
convention (*a) in this article. Then $(2/y)(y/2)=1$ for any odd integer $y$, 
and $(x/y)$ is fully multiplicative in both $x$ and $y$. 
} 
explicitly by 
\begin{align}
  \epsilon_{DL}^{-1} : d \longmapsto \left( \frac{d}{|G_L|} \right)
      \mathbb{E}\left[\frac{(d-1){\rm sgn}(DL(2))}{8} \right].  
  \qquad \qquad \qquad 
   \label{eq:def-epsilon-DL}
\end{align}
\end{lemma}

The character $\epsilon_{DL}$ for a rank-2 positive definite lattice $L$
has a simpler expression (Lemma \ref{lemma:epsiln=DK}). 
There is a good reason to believe that this Lemma has been known for 
a long time, but we could not find out a reference, so we write it 
down here along with a proof.  

\begin{lemma}
\label{lemma:epsiln=DK}
Let $L$, $DL$, and $\epsilon_{DL}^{-1}$ be the same as in the previous Lemma.
Then it is 
\begin{align}
  \epsilon_{DL}^{-1}(d) = \left( \frac{D_K}{d} \right)
\end{align}
for ${}^\forall d \in [Z/(N_{DL})]^\times$. It is $\{ \pm 1 \}$-valued. 
There are $h({\cal O}_{f_z})$ distinct rank-2 lattices $L$ sharing 
the same set of $D_K$, $f_z$, and $f_\rho$, but the character 
$\epsilon_{DL}^{-1}$ on $[\Z/(N_{DL})]^\times$ is the same for all of 
those lattices, even when $h({\cal O}_{f_z}) > 1$. 
\end{lemma} 
\begin{proof}
We use \cite[\S4 (p.1503)]{Weil-repr}
\begin{align}
  \mathbb{E}\left[ \frac{{\rm sgn}(DL(2))}{4} \right]
 = \mathbb{E}\left[ \frac{{\rm sgn}(L)}{4} \right] \; 
   \left( \frac{-1}{|G_L|} \right)
 = (-1) \left( \frac{-1}{|G_L|} \right),
\end{align}
where ${\rm sgn}(L)=2$ has been used in the second equality. 
When $D_K \equiv -3$ mod 4, it is always $- (-1/|D_K|) = +1$.
When $D_K =: -4 D'_K$ for an odd square-free positive $D'_K \equiv 1$ mod 4, 
it is always $- (-1/|D_K|) = - (-1/D'_K) = -1$. Finally, 
when $D_K =: -8 D''_K$ for an odd square-free positive $D''_K$, 
it is $- (-1/|D_K|) = - (-1/2)(-1/D''_K) = - (-1)^{(D''_K-1)/2}$.

Let us evaluate $\epsilon_{DL}^{-1}(d)$ for $d \in [\Z/(N_{DL})]^\times$
for the three separate cases of $D_K$. 
The simplest case is for $D_K = -4 D'_K$ with an odd positive square free 
$D'_K \equiv 1$ mod 4. In this case, $d$ is always odd. 
\begin{align}
  \epsilon_{DL}^{-1}(d) = \left( \frac{d}{|D_K|} \right) \times (-1)^{\frac{d-1}{2}}
   = \left( \frac{d}{D'_K} \right) \times (-1)^{\frac{d-1}{2}}
   = \left( \frac{D'_K}{d} \right) \times (-1)^{\frac{d-1}{2}}
   = \left( \frac{D_K}{d} \right).
\end{align}

In the case of $D_K = -8 D''_K$ with an odd positive square-free $D''_K$, 
\begin{align}
  \epsilon_{DL}^{-1}(d) = \left(\frac{d}{2D''_K}\right) \times 
     (-1)^{\frac{d-1}{2} + \frac{D''_K-1}{2}\frac{d-1}{2}}
& \;  = \left(\frac{2}{d}\right) \times
    \left(\frac{D''_K}{d} \right) (-1)^{\frac{D''_K-1}{2}\frac{d-1}{2}} \times
    (-1)^{\frac{d-1}{2}+\frac{D''_K-1}{2}\frac{d-1}{2}} \nonumber \\
& \; = \left( \frac{-2D''_K}{d} \right) = \left( \frac{D_K}{d} \right).
\end{align}

The last case $D_K \equiv -3$ mod 4 needs to be treated separately for 
an odd $d$ and an even $d$. For an odd $d$, 
\begin{align}
 \epsilon_{DL}^{-1}(d) = \left( \frac{d}{|D_K|}\right) \times 1^{\frac{d-1}{2}}
    = \left( \frac{|D_K|}{d} \right) (-1)^{\frac{d-1}{2}\frac{|D_K|-1}{2}}
    = \left( \frac{|D_K|}{d} \right) (-1)^{\frac{d-1}{2}}
    = \left( \frac{D_K}{d} \right).
\end{align}
We need to evaluate $\epsilon_{DL}^{-1}(d)$ for an even $d$, only when 
$N_{DL}$ is odd, which means that the abelian group $DL(2)$
is trivial, and ${\rm sgn}(DL(2))\equiv 0$ mod 8. So, 
$\mathbb{E}[(d-1) {\rm sgn}(DL(2)/8] = 1^{d-1}=1$. 
Now, if $d = 2^{\rm even} \cdot d''$ with an odd $d''$, 
\begin{align}
 \epsilon_{DL}^{-1}(d) = \left( \frac{d}{|D_K|} \right)
   = \left( \frac{d''}{|D_K|}\right)
   = (-1)^{\frac{d''-1}{2}\frac{|D_K|-1}{2}} \left( \frac{|D_K|}{d''} \right)
   = (-1)^{\frac{d''-1}{2}} \left(\frac{|D_K|}{d''}\right)
   = \left( \frac{D_K}{d} \right).
\end{align}
If $d = 2^{\rm odd} \cdot d'$ with an odd $d'$ instead, 
\begin{align}
 \epsilon_{DL}^{-1}(d) = \left( \frac{2^{\rm odd}d'}{|D_K|} \right)
  = \left(\frac{|D_K|}{2^{\rm odd}} \right)\left( \frac{D_K}{d'} \right)
  = \left(\frac{D_K}{2^{\rm odd}} \right) \left(\frac{D_K}{d'}\right)
  = \left( \frac{D_K}{d} \right).
\end{align}

For all the cases, we have now verified that $\epsilon_{DL}^{-1}$ on 
$[\Z/(N_{DL})]^\times$ is the same as $(D_K/-)$.
\end{proof}

\begin{props}
\label{props:mndrmy-kernel-bos-str}
In the model of rational CFT assigned by bosonic string theory 
to the data $([E_z]_\C,f_\rho)$, the monodromy representation 
$\rho^{(1,1)}_{\{0\}}$ of $\Gamma_{1,1} = {\rm SL}(2;\Z)$ 
associated with the $(g,n)=(1,1)$ conformal block $f_{\{0\}}$ has a kernel 
that contains a group $\Gamma(N_{\rm bos})$,  where 
\begin{align}
  N_{\rm bos} = {\rm LCM}(N_{D\Lambda_-}, 12) = {\rm LCM}(f_\rho D_z, 12).
\end{align}
In particular, the kernel is a congruence subgroup of ${\rm SL}(2;\Z)$.

When the Friedan--Shenker bundle ${\cal F}_{\{0\}}$ 
of the conformal blocks  
on ${\cal M}_{1,1}={\cal H}/{\rm PSL}(2;\Z)$ is lifted to the modular curve 
$X_\Gamma$ with the congruence subgroup $\Gamma$ contained in the kernel of 
the monodromy representation, $X(N_{\rm bos})$ for example, 
$p_\Gamma^*({\cal F}_{\{0\}})$ is a trivial bundle on $X_\Gamma$. 
The chiral correlation functions $f_0^{\rm bos}(\beta)$ become a single 
$\C$-valued function defined globally on $X_\Gamma$ for individual  
$\beta \in iReps$; $f_{1\Omega'_-}^{\rm bos}(\beta)$ and 
$f_{1\overline{\Omega}'_-}^{\rm bos}(\beta)$ are globally well-defined 
sections of $(T^*X_\Gamma)^{\otimes 1/2}$ for individual $\beta \in iReps$.
\end{props}

\begin{proof} We just have to use the Lemmas \ref{lemma:kern-theta-eta},
\ref{lemma:Weil-repr-kernel}, and \ref{lemma:NDLambda}, in the 
Example \ref{exmpl:bos-T2models-monodromy}.  
\end{proof}

\subsubsection{Type II String Theory}

\begin{props}
\label{props:chi-corrl-II-lift}
In the model of rational ${\cal N}=(2,2)$ SCFT assigned by Type II string theory to the data $([E_z]_\C, f_\rho)$, the ``$(g,n)=(1,2)$ chiral correlation function(s)'' 
\begin{align}
 f_0^{\rm II} & \; = ([f_0^{\rm II}(\beta)]_{\beta \in iReps}), \\
 [{\rm resp.} \quad 
 f_{1\Omega'_-}^{\rm II} & \; = ([f_{1\Omega'_-}^{\rm II}(\beta)]_{\beta \in iReps} ), 
    \qquad 
 f_{1\overline{\Omega}'_-}^{\rm II} = 
([f_{1\overline{\Omega}'_-}^{\rm II}(\beta)]_{\beta \in iReps} )  \quad ]
\end{align}
is a section of $(T^*{\cal M}_{1,1})^{\otimes 1/2}$ 
[resp. are sections of $T^*{\cal M}_{1,1}$] defined locally. 

They fail to be a global section on ${\cal M}_{1,1}$ due to the monodromy 
representation $\rho_{D\Lambda_-}^{\rm Weil}$ of $\pi_1^{\rm top}({\cal M}_{1,1}) 
= {\rm SL}(2;\Z)$. The kernel of this representation contains 
$\Gamma(N_{D\Lambda_-}) = \Gamma(f_\rho D_z)$, and hence is a congruence subgroup
of ${\rm SL}(2;\Z)$.

When the ``$(g,n)=(1,2)$ chiral correlation functions'' are lifted to 
the modular curves $X_\Gamma$ whose group $\Gamma$ is contained in the kernel, 
such as $X(N_{D\Lambda_-}) = X(f_\rho D_z)$, they become sections of 
$(T^*X_\Gamma)^{\otimes 1/2}$ and $T^*X_\Gamma$, respectively, defined globally 
over the entire modular curve $X_\Gamma$, without a non-trivial monodromy 
among those with different $\beta$'s in $iReps$. 
\begin{align}
 d\tau_{ws} f_{1\Omega'_-}^{\rm II}(\tau_{ws};\beta) =   
 d\tau_{ws} \;\vartheta_{\Lambda_-}^{1\Omega'_-}(\tau_{ws}; \beta) \in 
    \Gamma(X_\Gamma; T^*X_\Gamma ), 
   \qquad   \beta \in iReps, 
  \label{eq:pulledback-theta1} \\
 d\tau_{ws} f_{1\overline{\Omega}'_-}^{\rm II}(\tau_{ws};\beta) =   
 d\tau_{ws} \;\vartheta_{\Lambda_-}^{1\overline{\Omega}'_-}(\tau_{ws}; \beta) \in 
    \Gamma(X_\Gamma; T^*X_\Gamma ), 
   \qquad   \beta \in iReps. 
  \label{eq:pulledback-theta1Bar} 
\end{align}
$\bullet$
\end{props}

\begin{rmk}
We have emphasized already in Remark \ref{rmk:triangle-shcematic} 
that the chiral correlation functions $f_{1\Omega'_-}^{\rm II}$ take expectation 
values of an operator in quantum field theory whose classical interpretation 
is the pull-back of the holomorphic (1,0)-form of the target space 
elliptic curve $[E_z]_\C$. The Proposition above then implies that the lift 
of $f_{1\Omega'_-}^{\rm II}$ from ${\cal M}_{1,1}$ to an appropriate modular 
curve $X_\Gamma$, a finite covering of ${\cal M}_{1,1}$, allows us to 
see $d\tau_{ws} f_{1\Omega'_-}^{\rm II}(\tau_{ws}; \beta)$ as holomorphic 1-forms 
of the modular curve $X_\Gamma$ globally well-defined. Alternatively, 
we can see them as weight-2 cuspforms of the congruence subgroup $\Gamma$.

In \cite{prev.paper}, we have seen that the chiral correlation functions 
$f_{1\Omega'_-}^{\rm II}$ can be used to construct the $L$-functions, 
$L(H^1_{et}(E),s)$, of the CM elliptic curve $E$ defined over a certain 
class of number fields. We did so in \cite{prev.paper} by comparing the Mellin 
transform of $f_{1\Omega'_-}^{\rm II}$'s and the $L$-functions directly; this 
observation is now quite reasonable, and in line with the theory of 
modular parametrization, because the chiral correlation functions 
$d\tau_{ws} f_{1\Omega'_-}^{\rm II}$ can be seen as 1-forms on modular curves, 
once we unwind the monodromy representation of $\pi_1^{\rm top}({\cal M}_{1,1})$
by taking a finite covering $X_\Gamma$ over ${\cal M}_{1,1}$.  $\bullet$
\end{rmk}

\begin{rmk}
As reviewed in Lemmas \ref{lemma:isom-GammaN-GHNM} and 
\ref{lemma:isom-strokeOp-cuspformSpace}, a cusp form $f(\tau)$ 
of weight 2 for $\Gamma(N)$ for some $N$ turns into a cusp form 
of weight 2 for $\Gamma_1(N^2)$, when the argument $\tau$ is replaced 
by $N\tau$. For more background materials on modular forms, see 
section \ref{ssec:review-modular}, or introductory textbooks/lecture notes 
(e.g. \cite{DS, Ribet-Stein}). 
Therefore, the ``$(g,n)=(1,2)$ chiral correlation functions'' 
$f_{1\Omega'}^{\rm II}(\tau_{ws}; \beta)$ for $\beta \in iReps.^-$ in 
$\Gamma(N_{D\Lambda_-})$ can also be regarded, via the linear map 
$|[\diag(N_{D\Lambda_-},1)]_2$, as an element of $S_2(\Gamma_1(N_{D\Lambda_-}^2))$.

When the ``$(g,n)=(1,2)$ chiral correlation functions'' 
$f_{1\Omega'_-}^{\rm II}(\beta)$ are written as power series of 
\begin{align}
  e^{2\pi i (\tau_{ws}/N_{D\Lambda_-})} = q_{ws}^{1/N_{D\Lambda_-}} = q = e^{2\pi i \tau} 
  \qquad   ( \tau_{ws} = N_{D\Lambda_-} \tau ), 
  \label{eq:the-argument-rescaling}
\end{align}
we only need terms with integer-power of the variable\footnote{
\label{fn:commnt-tau-rescale}
Our previous paper \cite{prev.paper} introduced the following 
change-of-variable operation, 
\begin{align}
  \tau_{ws} = \frac{f_\rho D_z}{C^2a_z} \tau, 
\end{align}
where $C$ is an integer used in the presentation of \cite{prev.paper}, so 
$C^2a_z$ is also an integer. Because $N_{D\Lambda_-}=f_\rho D_z$, 
the variable $\tau$ in that paper and $\tau$ here are related by 
$\tau_{\rm there} = C^2 a_z \tau_{\rm here}$.
The normalization of the variable $\tau_{\rm there}$ was determined 
in \cite{prev.paper}, 
in effect, so that the power series expansion of $f_{1\Omega'_-}^{\rm II}(\beta)$'s
only have terms with integer-power of $e^{2\pi i \tau_{\rm there}}$ for 
a class of $\beta$'s within $iReps$, but not for all of $\beta$'s 
in $iReps$. }
$q$ for any $\beta \in iReps$. We can see 
this directly from the expression of $f_{1\Omega'_-}^{\rm II}$: contributions to 
$\vartheta_{\Lambda_-}^{1\Omega'_-}(\tau_{ws};\beta)$ are always of the form of a 
coefficient 
times $\mathbb{E}[ \tau_{ws} \; q_{\Lambda_-}(w)/2 ] = 
\mathbb{E}[ \tau \; N_{D\Lambda_-}q_{\Lambda_-}(w)/2 ]$, with $w \in \Lambda_-^\vee$, 
so one just have to remember that $N_{DL} q_L(w)/2 \in \Z$ for any lattice $L$
 (by definition). Alternatively, we can also tell from the fact 
that $f_{1\Omega'_-}^{\rm II}|[\diag(N_{D\Lambda_-},1)]_2 \in S_2(\Gamma_1(N_{D\Lambda_-}^2))$
has just terms that are integer-power of the variable $q=e^{2\pi i \tau}$;
the congruence subgroup $\Gamma_1(N_{D\Lambda_-}^2)$ contains the element 
$T = [1, 1; 0, 1]$.
\end{rmk}

\begin{rmk}
\label{rmk:motivatn4partialLift}
Here is a comment for string theorists. It may sound a little strange to think 
of lifting the conformal blocks and ``chiral correlation functions'' on 
${\cal M}_{1,1} = {\cal H}/{\rm SL}(2;\Z)$ not all the way to ${\cal H}$, 
but only half way to modular curves $X_\Gamma \supset {\cal H}/\Gamma$ 
where $\Gamma$' are 
congruence subgroups that are contained in the kernel of the monodromy 
representation of ${\rm SL}(2;\Z)$. Lifts to ${\cal H}$ always unwind the 
monodromy for any model of rational CFT, so one may wonder if there is any 
point in thinking of partial lifts to modular curves $X_\Gamma$ that vary 
from one model of rational CFT to another.

There is an important difference between ${\cal H}$ and ${\cal H}/\Gamma$ 
in fact, although they are both a covering space for ${\cal M}_{1,1}$. 
The projection  
${\cal H}/\Gamma \rightarrow {\cal H}/{\rm SL}(2;\Z) = {\cal M}_{1,1}$ is 
a morphism in {\it algebraic} geometry, while ${\cal H} \rightarrow 
{\cal H}/{\rm SL}(2;\Z) = {\cal M}_{1,1}$ is not; the $j$-function, which 
describes the latter projection, is a transcendental (non-polynomial, 
non-rational) function. So, by dealing with lifts of the chiral correlation 
functions and conformal blocks 
to various modular curves but not to ${\cal H}$, we can stay within 
the realm of algebraic geometry. By allowing to think of all kinds of 
possible modular curves (instead of just one covering space), 
in effect we still have a notion of ``universal cover'' within algebraic 
geometry (implemented in the form of the inverse system of modular curves
(cf sections \ref{sssec:surjective-morph-2BK} and \ref{ssec:KahlerVsShimuraC}, 
and the observation (ii) in \ref{statmnt:GTtheory-final})). 

Modular curves $X_\Gamma$ can be regarded as algebraic varieties (as stacks
for some $\Gamma$'s), while the upper complex half plane ${\cal H}$
or the unit disc in $\C$ is not. Moreover, the field of moduli $k_{\rm mod}$
of a modular curve is a number field within $\overline{\Q}$ \cite[\S7.6]{DS}, 
so one can 
think of the Galois group ${\rm Gal}(\overline{\Q}/k_{\rm mod})$
acting on the $\overline{\Q}$-coefficient function field and the 
cohomology groups of $X_\Gamma$. 
The Galois group action plays a central role in the discussion 
in sections \ref{sec:parametrizatn} and \ref{sec:two-Gal-actn}.

Another benefit of paying attention to intermediate things such as 
${\cal H}/\Gamma$ and $\Gamma \subset {\rm SL}(2;\Z)$ than just to 
the truly universal ${\cal H}$ and $\cap \Gamma$ is in the fact that 
the vector space of all the modular forms (for some congruence subgroup 
${}^\exists \Gamma \subset {\rm SL}(2;\Z)$) is horribly complicated. The 
vector space of modular forms for a given congruence subgroup $\Gamma$ 
forms a finite dimensional vector space, and moreover, one can derive 
a lot more knowledge on the substructure in this finite dimensional 
vector space (see section \ref{ssec:review-modular} for a review) 
by exploiting the action of the Hecke algebra. $\bullet$
\end{rmk}
%

\subsubsection{``$(g,n)=(1,2)$ Chiral Correlation Functions'' as Hecke 
Eigenforms: Easiest Examples}
\label{sssec:easiest-ex}

Let us fix one set of data $([E_z]_\C, f_\rho)$, and the rational model 
of ${\cal N}=(2,2)$ SCFT assigned by Type II string theory. For any 
$\beta \in iReps$, the ``$(g,n)=(1,2)$ chiral correlation function''
is an element of $S_2(\Gamma(N)) \cong S_2(\Gamma(H_N, N^2,1))$ 
with $N=N_{D\Lambda_-}$; see (\ref{eq:def-Hn-subgroup}) and 
Lemma \ref{lemma:isom-strokeOp-cuspformSpace} for the definition of 
$H_N$ and this isomorphism. 

\begin{anythng} [cf \cite{MooreArth, GV}]
\label{statmnt:iReps-as-torsion-pts}
Recall that $\C/\Omega'_-(\Lambda_-) \cong \C/\mathfrak{b}_z = [E_z]_\C$ 
(via homothety). This is because 
\begin{align}
   \Omega'_-(e_{-2}) = \sqrt{2a_zf_\rho}, \qquad 
   \Omega'_-(e_{-1}) = \sqrt{2a_zf_\rho} (-z), \qquad 
  \Omega'_-(\Lambda_-) = \sqrt{2a_zf_\rho} \mathfrak{b}_z.
\end{align}
When the linear map $\Omega'_-: \Lambda_- \rightarrow \C$ is extended linearly 
to $\Lambda_-^\vee \subset \Lambda_- \otimes \Q$, 
\begin{align}
  \Omega'_-(\hat{e}_{-2}) = -i \sqrt{\frac{2a_z}{f_\rho D_z}} z, \qquad 
  \Omega'_-(\hat{e}_{-1}) = -i \sqrt{\frac{2a_z}{f_\rho D_z}},
\end{align}
where $\{ \hat{e}_{-2}, \hat{e}_{-1}\}$ is the set of generators of the lattice 
$\Lambda_-^\vee$, 
\begin{align}
  \hat{e}_{-2} = \frac{1}{f_\rho D_z} (2c_z e_{-2}-b_ze_{-1}), \qquad 
  \hat{e}_{-1} = \frac{1}{f_\rho D_z} (-b_z e_{-2}+2a_z e_{-1}).
\end{align}
We have 
\begin{align}
  \Omega'_-(\Lambda_-^\vee) = \Omega'_-({\rm II}_{2,2}) = 
      -i \sqrt{\frac{2a_z}{f_\rho D_z}} \; \mathfrak{b}_z .
  \label{eq:image-Omega'-of-LambdaVee}
\end{align}
The set of irreducible representations $iReps \cong G_{\Lambda_-}$ of 
the chiral algebra of the model for $([E_z]_\C, f_\rho)$ can be identified 
with a subset of $[E_z]_\C$ through\footnote{
In the language of open string, the set $iReps$ is identified with 
the set of Cardy states of D0-brane type, and is identified naturally 
with a set of finite torsion points of $[E_z]_\C$ for this physical reason
\cite{GV}.
We will not discuss open string interpretations of $f_{1\Omega'_-}^{\rm II}$ 
in this article, however.} 
\begin{align}
  \Omega'(\Lambda_-^\vee/\Lambda_-) = 
   \Omega'_-(\Lambda_-^\vee)/\Omega'_-(\Lambda_-) \subset \C/\Omega'_-(\Lambda_-)
   \cong [E_z]_\C.
   \label{eq:isom-abstract-C/kL-originalE}
\end{align}
$\Omega'_-(iReps)$ forms a finite set of torsion points of $[E_z]_\C$.
\end{anythng}

Here, this identification comes with an ambiguity of the group of 
automorphisms of $[E_z]_\C$, ${\rm Aut}([E_z]_\C)$. We may fix one 
identification once and for all, and may sometimes refer to an irreducible 
representation $\beta \in iReps$ as a torsion point of $[E_z]_\C$.

\begin{anythng}
\label{statmnt:Fzfrho-is-in-S2GammaN}
When $iReps = G_{\Lambda_-}$ contains a 2-torsion element 
($2\beta = 0$ in the abelian group $G_{\Lambda_-} = \Lambda_-^\vee/\Lambda_-$), 
$f_{1\Omega'_-}^{\rm II}(\tau_{ws};\beta) = 
\vartheta_{\Lambda_-}^{1\Omega'_-}(\tau_{ws};\beta)$ is identically zero. 
This is because of the linearity of $\vartheta^{1\omega}_L(\tau;x)$ with 
respect to $x \in G_L$; $\vartheta_{\Lambda_-}^{1\Omega'_-}(\tau_{ws}; -\beta) 
= - \vartheta_{\Lambda_-}^{1\Omega'_-}(\tau_{ws};\beta)$.
Similarly, when $[E_z]_\C$ is that of $j([E_z]_\C) = 0$, which means that  
$[z] = [e^{2\pi i /3}]$ and the automorphism group is 
${\rm Aut}([E_z]_\C) \cong \Z/(6)$, 
$f_{1\Omega'_-}^{\rm II}(\tau_{ws}; \beta)$ vanishes for 
any $\beta \in \Lambda_-^\vee/\Lambda_-$ that remains fixed (modulo $\Lambda_-$)
under the multiplication of $e^{2\pi i/3}$. Let $G_{\Lambda_-}^*$ be all the 
elements of $G_{\Lambda_-} \cong iReps$ that are not a 2-torsion point 
or a fixed point of the $e^{2\pi i /3}$ multiplication. 

For the model for the data $([E_z]_\C, f_\rho)$, think of the vector space
\begin{align}
  F_{\{0,0\}}^{(\Omega,\int J)}([z],f_\rho)=
  F([z], f_\rho) & \; := {\rm Span}_\C \left\{ 
            f_{1\Omega'_-}^{\rm II}(\tau_{ws}; \beta) \; | \; 
            \beta \in G_{\Lambda_-}^* \right\}  \nonumber \\
   & \; = {\rm Span}_\C \left\{ \vartheta_{\Lambda_-}^{1\Omega'_-}(\tau_{ws}; \beta) 
       \; | \; \beta \in G_{\Lambda_-}^* \right\} 
     \subset S_2(\Gamma(N_{D\Lambda_-})).
\end{align}
We may also be interested in the vector space 
\begin{align}
  [F([z], f_\rho)]_K & \; = {\rm Span}_K \left\{ 
       f_{1\Omega'_-}^{\rm II}(\tau_{ws}; \beta) /\sqrt{2a_zf_\rho} \; | \; 
       \beta \in G_{\Lambda_-}^* \right\}.
\end{align}

In the homomorphism 
\begin{align}
 \C_{G_{\Lambda_-}} \ni e_\beta \longmapsto f_{1\Omega'_-}^{\rm II}(\tau_{ws};\beta)
 \in F([z], f_\rho), 
\label{eq:hom-CiReps-Fmodel-prelim}
\end{align}
where $\C_{G_{\Lambda_-}} := {\rm Span}_\C\{ e_\beta \; | \; \beta \in 
G_{\Lambda_-}\}$ is a vector space of dimension $|iReps| = |G_{\Lambda_-}|$, 
$e_\beta$'s with $\beta \in G_{\Lambda_-} \backslash G_{\Lambda_-}^*$ are 
in the kernel. Note also that, if there is an automorphism 
$\sigma \in {\rm Aut}([E_z]_\C) \cong {\cal O}_{f_z}^\times$,
\begin{align}
f_{1\Omega'_-}^{\rm II}(\tau_{ws}; \sigma \cdot \beta) = 
    \rho^1(\sigma) \; f_{1\Omega'_-}^{\rm II}(\tau_{ws}; \beta), 
  \label{eq:fII-AutE-relatn}
\end{align}
where $\rho^1$ is the embedding $K \rightarrow \C$ that maps 
$\sqrt{D_K}$ into the upper complex half plane; $\sigma$ is regarded 
as an element of ${\cal O}_{f_z} \subset K$ here. Therefore, the homomorphism
(\ref{eq:hom-CiReps-Fmodel-prelim}) factors through 
\begin{align}
 \C_{G_{\Lambda_-}} \rightarrow  
\C_{G_{\Lambda_-}^*} / {\rm Aut}([E_z]_\C) \longrightarrow F([z], f_\rho).
  \label{eq:CGLambda2F}
\end{align}
The same homomorphism can also be formulated with the coefficients in $K$, 
rather than in $\C$. 

The vector-space homomorphism 
$\C_{G_{\Lambda_-}^*}/{\rm Aut}([E_z]_\C) \longrightarrow F([z],f_\rho)$
is still not necessarily injective; there may be further linear relations 
among the ``chiral correlation functions'' than the relations that follow 
immediately from the automorphism ${\rm Aut}([E_z])$ of the target-space 
elliptic curve. We will see in section \ref{ssec:HTheta-TypeII-CCF}, 
in the cases of $f_z=1$, when and how this last homomorphism fails to 
be injective. $\bullet$
\end{anythng}

For concreteness, let us take a look at a few examples of how 
the vector space $F([z], f_\rho)$ of a model for $([E_z]_\C, f_\rho)$ fits 
into the vector space $S_2(\Gamma(N_{D\Lambda_-}))$. 

\begin{exmpl}
\label{exmpl:Dk4-fz1-N64}
Take an example of a CM elliptic curve $[E_z]_\C$ with $z = i$, the elliptic 
curve with the $j$-invariant $j(z)=1728$. Then 
\begin{align}
  G_{\Lambda_{[E_i]_\C}} \cong \Z/(2) \oplus \Z/(2), \qquad 
  G_{\Lambda} \cong \Z/(2f_\rho) \oplus \Z/(2f_\rho), \qquad 
  N_{D\Lambda_-} = D_z f_\rho = 4f_\rho.
\end{align}

In the case of the data $([z], f_\rho) = ([i]_\C, 1)$, the $f_\rho = 1$ case, 
$G_{\Lambda_-} \cong \Z/(2) \oplus \Z/(2)$ consists only of 2-torsion elements, 
and the ``$(g,n)=(1,2)$ chiral correlation function'' 
$f_{1\Omega'_-}^{\rm II}(\tau_{ws}; \beta)$ vanishes for any one of 
$\beta \in G_{\Lambda_-}$.  This observation fits very well with the fact 
that the kernel of $\rho^{\rm Weil}_{D\Lambda}$ for $f_\rho=1$ contains 
$\Gamma(4f_\rho) = \Gamma(4)$, and the vector space $S_2(\Gamma(4))$ 
is known to be empty.\footnote{There is no weight-2 cuspform for 
$\Gamma_0(16)$ with nebentypus $\chi_{16}: [\Z/(16)]^\times \rightarrow S^1$ 
whose conductor $\mathfrak{c}_f$ divides $(4)_\Z$. 
Although there are two independent weight-2 cuspforms for $\Gamma_1(16)$, 
they are for different nebentypuses: 
$\dim_\C [S_2(\Gamma_1(16))]
= \dim_\C[S_2(\Gamma_0(16), \chi_{16}(0,1))] 
+ \dim_\C [S_2(\Gamma_0(16),\chi_{16}(0,3))] = 1+1$ 
(\cite{SAGE} can be used to get this computation done), where 
$\chi_{16}(0,1): [15] \mapsto 1$ and $[5] \mapsto i$, and 
$\chi_{16}(0,3): [15] \mapsto 1$ and $[5] \mapsto (i)^3$. The conductor 
of the two characters $\chi_{16}(0,1)$ and $\chi_{16}(0,3)$ do not 
divide $(4)_\Z$. (Background materials, such as the nebentypus decomposition 
of the vector space $S_2(\Gamma_1(N))$, are reviewed in 
section \ref{ssec:review-modular}. Alternatively, see any basic textbooks on 
modular forms.)}

In the case of the data $([z],f_\rho)=([i]_\C, 2)$, the $f_\rho = 2$ case, 
the vector space $F([i],2)$ can be as large as 3-dimensional, because 
$G_{\Lambda_-}^*$ consists of $4^2 - 2^2 =12$ elements, and 
${\rm Aut}([E_i]_\C) \cong \Z/(4)$, so $|G_{\Lambda_-}^*|/|{\rm Aut}([E_i]_\C)|=3$.
To be more concrete, 
\begin{align}
  i \sqrt{|G_{\Lambda_-}|} \frac{f_{1\Omega'_-}^{\rm II}(\tau_{ws}; 1)}{\sqrt{2f_\rho}}
   & \; = \sqrt{a_z} \vartheta_{\Lambda_-}^{1\Omega_-}(\tau_{ws}; 1) 
        = (q -3 q^9 + 2 q^{17} + 5q^{25} + \cdots ), 
    \label{eq:CCF-II-Dk4-fr2-a} \\
  i \sqrt{|G_{\Lambda_-}|} \frac{f_{1\Omega'_-}^{\rm II}(\tau_{ws}; 3+2i)}{\sqrt{2f_\rho}}  & \; = \sqrt{a_z} \vartheta_{\Lambda_-}^{1\Omega_-}(\tau_{ws}; 3+2i) 
        = - 2 \left( q^5-q^{13} + \cdots \right), \\
  i \sqrt{|G_{\Lambda_-}|} \frac{f_{1\Omega'_-}^{\rm II}(\tau_{ws}; 1+i)}{\sqrt{2f_\rho}}
  & \; = \sqrt{a_z} \vartheta_{\Lambda_-}^{1\Omega_-}(\tau_{ws};1+i) 
        = (1+i) \left( q^2 -2 q^{10}-3q^{18} + \cdots \right),
     \label{eq:CCF-II-Dk4-fr2-c}
\end{align}
where we have introduced\footnote{In \cite{prev.paper}, the linear map 
$\Omega_-: \Lambda_-^\vee \rightarrow \C$ was introduced as 
$\Omega_-^{\rm there} = C^{-1} \sqrt{2a_zf_\rho} {\rm Im}(z) i \; \Omega'_-$, 
where $C$ is an integer used in the presentation in \cite{prev.paper};
this is a little different from the linear map $\Omega_-$ here 
(by normalization). 
Eq. (\ref{eq:def-Omega}) is equivalent to 
$\Omega_-^{\rm here} = \sqrt{2a_zf_\rho} {\rm Im}(z)i \; \Omega'_-$.
}
a linear map 
$\Omega_-: \Lambda_-^\vee \rightarrow \C$ simply rescaled from $\Omega'_-$ by 
\begin{align}
  \Omega_- = i \frac{\sqrt{|G_{\Lambda_-}|}}{\sqrt{2a_zf_\rho}} \Omega'_-, 
  \label{eq:def-Omega}
\end{align}
and used $q = e^{2\pi i \tau}$ in (\ref{eq:the-argument-rescaling}) 
as the variable on the right hand sides instead of $\tau_{ws}$.
The irreducible representations, $\beta$'s in $G_{\Lambda_-}$, are also 
referred to in (\ref{eq:CCF-II-Dk4-fr2-a}--\ref{eq:CCF-II-Dk4-fr2-c}) 
by the complex value (mod $\Omega_-(\Lambda_-)$) of their 
linear map $\Omega_-$ in the second argument of $f_{1\Omega'_-}^{\rm II}$ and 
$\vartheta_{\Lambda_-}^{1\Omega_-}$.
Obviously they are linearly independent over $\C$, and hence $F([i],2)$ 
is of 3-dimensions indeed. They must be a part of the space of weight-2 cusp 
forms for  $\Gamma(8)$. 

Independently from the discussion based on a rational model of 
${\cal N}=(2,2)$ SCFT, it is known (we used \cite{SAGE})\footnote{
SAGE allows us to deal with the vector space of cuspforms for 
the group $\Gamma(H_N, N^2,1) \subset \Gamma_0(N^2)$, and those for 
$\Gamma_0(N^2)$ with the nebentypus specified. To list up all the possible 
nebentypuses of the group $\Gamma_0(N^2)$ with $N=N_{D\Lambda}=8$, we see 
that $[\Z/(64)]^\times \cong \Z/(2) \times \Z/(16)$ where the first factor 
is generated by $63 \equiv (-1)$ mod 64, and the second factor by 
$5$ mod 64; the Dirichlet characters of 
$[\Z/(64)]^\times$---$\chi_{64}(a,b)$---are labeled by 
$(a,b) \in \Z/(2) \times \Z/(16)$, where $\chi_{64}(a,b): (-1) \mapsto (-1)^a$ 
and $\chi_{64}(a,b): [5] \mapsto \zeta_{16}^b$. A character $\chi_{64}(a,b)$
is trivial on the subgroup $H_N \subset [\Z/(N^2)]^\times$ with $N=N_{D\Lambda_-}=8$
if and only if $b \in \{ 0,8\} \subset \Z/(16)$. The two characters 
$\chi_{64}(0,0)$ and $\chi_{64}(0,8)$ of the group $[\Z/(N^2)]^\times \cong 
\Gamma_0(N^2)/\Gamma_1(N^2)$ with $N=N_{D\Lambda_-}=8$ correspond to 
the two characters $\chi_8(0,0)$ and $\chi_8(0,2)$, respectively, of the 
group $\Gamma_0^0(N)/\Gamma(N)$ with $N=N_{D\Lambda_-}=8$.  Whenever we need to 
carry out this kind of calculation using SAGE, we repeated this logic and 
procedure. 
}  
that\footnote{
The space of cusp forms with $\chi_8(1,b)$ is empty for any 
$b \in \Z/(2)$, because a non-zero weight $k=2$ cuspform should have 
nebentypus that is even for $-1 \equiv 7$ mod $N_{D\Lambda_-}=8$.
}
\begin{align}
 \dim_\C S_2(\Gamma_0^0(8), \chi_8(0,0)) = 3, \qquad 
 \dim_\C S_2(\Gamma_0^0(8),\chi_8(0,1)) = 2,
\end{align}
where $\chi_8(a,b)$'s with $(a,b) \in \Z/(2) \times \Z/(2)$
are Dirichlet characters $[\Z/(N_{D\Lambda_-})]^\times \rightarrow S^1$
modulo $N_{D\Lambda_-}=8$, such that $\chi_8(a,b): [7] \mapsto (-1)^a$ and 
$\chi_8(a,b): [5] \mapsto (-1)^b$.

As a basis of the vector space 
\begin{align}
  S_2(\Gamma_0^0(8),\chi_8(0,0)) \; |[\diag(8,1)]_2  \cong 
    S_2(\Gamma_0(64),\chi_8(0,0)),
\end{align}
we can choose the following three elements \cite{SAGE}, 
\begin{align}
  & (q -3q^9 + 2q^{17} - q^{25} +10q^{41} -7q^{49} -12q^{65}) \nonumber \\
  &  \qquad \pm 2(q^5 -3q^{13}+5q^{29}+q^{37}-3q^{45}-7q^{53}+5q^{61}), 
         \label{eq:nf-64-00} \\
  & (q^2 -2q^{10}-3q^{18}+6q^{26} + \cdots); 
    \label{eq:of-64-00}
\end{align}
this basis is chosen so that it suits the decomposition 
(\ref{eq:substr-Gamma0Chi-levelNr}, \ref{eq:substr-Gamma00Chi-levelNr}).
The cuspform (\ref{eq:nf-64-00}) with the $+$ sign is the newform 
of level $M_f=N_{D\Lambda_-}^2 =64$ generating a one-dimensional space of 
$[S_2(\Gamma_0(64), \chi_8(0,0))]^{\rm new}_{64}$, and the 
cuspform (\ref{eq:nf-64-00}) with the $-$ sign the newform of level $M_f=32$
generating a one-dimensional space $[S_2(\Gamma_0(32),\chi_8(0,0))]^{\rm new}_{32}$.
The last one (\ref{eq:of-64-00}) is the image of the previous level $M_f=32$ 
newform by $|[\diag(r,1)]_2$ with $r=2|(64/M_f)$. For more information, 
see a brief review in section \ref{ssec:review-modular} or literatures. 

By comparing the Fourier coefficients of the first several orders 
of (\ref{eq:CCF-II-Dk4-fr2-a}--\ref{eq:CCF-II-Dk4-fr2-c}) and 
(\ref{eq:nf-64-00}, \ref{eq:of-64-00}), we see that 
\begin{align}
  F([i],2) \cong S_2(\Gamma_0^0(8),\chi_8(0,0)).
\end{align}
Moreover, 
\begin{align}
    \frac{f_\rho f_z\sqrt{D_K}}{\sqrt{2f_\rho}}
      \left(  f_{1\Omega'_-}^{\rm II}(\tau_{ws};1)
            - f_{1\Omega'_-}^{\rm II}(\tau_{ws};3+2i) \right) \;
 |[\diag(N_{D\Lambda_-},1)]_2 \;  
   \in {\rm NewForm}(64,\chi_8(0,0)), \nonumber \\
    \frac{f_\rho f_z\sqrt{D_K}}{\sqrt{2f_\rho}}
      \left(  f_{1\Omega'_-}^{\rm II}(\tau_{ws};1)
            + f_{1\Omega'_-}^{\rm II}(\tau_{ws};3+2i) \right)  \;
  |[\diag(N_{D\Lambda_-},1)]_2 \;  
   \in {\rm NewForm}(32,\chi_8(0,0)), \label{eq:nf-64-00-byCCF} 
\end{align}
and 
\begin{align}
  f_{1\Omega'_-}^{\rm II}(\tau_{ws};1+i) = \left(   f_{1\Omega'_-}^{\rm II}(\tau_{ws};1)
            + f_{1\Omega'_-}^{\rm II}(\tau_{ws};3+2i) \right) |[\diag(2,1)]_2.
  \qquad  \bullet
   \label{eq:of-64-00-byCCF}
\end{align}
\end{exmpl}

\begin{exmpl}
\label{exmpl:Dk8-fz1-N64}
Take an example of a CM elliptic curve $[E_z]_\C$ with $z = \sqrt{2}i$, the 
elliptic curve with the $j$-invariant $j(z) = 8000$. In the rational model 
of ${\cal N}=(2,2)$ SCFT for the data $([E_z]_\C, f_\rho)$ with 
$(z,f_\rho)=([\sqrt{2}i], 1)$, 
\begin{align}
 G_{\Lambda_-} \cong \Z/(2) \oplus \Z/(4), \quad N_{D\Lambda_-} = D_z f_\rho = 
  8 \times 1 = 8.
\end{align}

The set $G_{\Lambda_-}^*/{\rm Aut}([E_z]_\C)$ consists of two elements, 
represented by two $\beta$'s in $iRpes$ that are mapped by 
$\Omega_-$ (introduced in (\ref{eq:def-Omega})) to $1+\Omega_-(\Lambda_-)$ 
(or $3+\Omega_-(\Lambda_-)$) 
and $(1+\sqrt{2}i)+\Omega_-(\Lambda_-)$ (or $(3+\sqrt{2}i)+\Omega_-(\Lambda_-)$).
The vector space $F([\sqrt{2}i], 1)$ is therefore generated by the 
two elements 
\begin{align}
 i \sqrt{|G_{\Lambda_-}|}\frac{f_{1\Omega'_-}^{\rm II}(\tau_{ws};1)}{\sqrt{2f_\rho}}
   & \; = (q - q^9 -6q^{17} +5q^{25} + \cdots ), 
     \label{eq:CCF-II-Dk8-fr1-a} \\
 i \sqrt{|G_{\Lambda_-}|}\frac{f_{1\Omega'_-}^{\rm II}(\tau_{ws};1+\sqrt{2}i)}{\sqrt{2f_\rho}}
   & \; = 2 (q^3 -3q^{11} + q^{19} +2 q^{27} + \cdots ). 
     \label{eq:CCF-II-Dk8-fr1-b}
\end{align}
They are linearly independent from each other, and hence the vector space 
$F([\sqrt{2}i],1)$ of the model for $([E_{\sqrt{2}i}]_\C, 1)$ is of 2-dimensions.

This vector space $F([\sqrt{2}i],1)$ should be a subspace of 
$S_2(\Gamma(N_{D\Lambda_-}))$ with $N_{D\Lambda_-}=8$ in this model. 
Using \cite{SAGE}, we see that we can choose a set of generators for 
the remaining component
\begin{align}
  S_2(\Gamma_0^0(8),\chi_8(0,1)) \;  |[\diag(8,1)]_2  \cong 
      S_2(\Gamma_0(64), \chi_8(0,1))
\end{align}
to be 
\begin{align}
& q - q^9 - 6 q^{17} + 5 q^{25} + 12q^{33} -6q^{41} -7q^{49} -4q^{57}
       \nonumber  \\
& \qquad \pm 2i ( q^3 - 3 q^{11} + q^{19} + 2 q^{27} + 5q^{43}-6q^{51}-3q^{59}) 
      + {\cal O}(q^{60}).
  \label{eq:nf-64-01}
\end{align}
Once again, this basis is chosen so that this is consistent with the 
substructure (\ref{eq:substr-Gamma0Chi-levelNr}, 
\ref{eq:substr-Gamma00Chi-levelNr}) of this vector space.
This time, the two elements are both in ${\rm NewForms}(64, \chi_{8}(0,1))$, 
the level $M_f=N_{D\Lambda_-}^2=64$ newforms of nebentypus $\chi_8(0,1)$.

Therefore, by comparing the first several Fourier coefficients 
of (\ref{eq:CCF-II-Dk8-fr1-a}, \ref{eq:CCF-II-Dk8-fr1-b}) 
and (\ref{eq:nf-64-01}), we find that 
\begin{align}
  F([\sqrt{2}i],1) \cong S_2(\Gamma_0^0(8), \chi_8(0,1)),
\end{align}
and 
\begin{align}
  \left\{ \frac{f_\rho f_z \sqrt{D_K}} {\sqrt{2f_\rho}}
    \left( f_{1\Omega'_-}^{\rm II}(\tau_{ws};1)
         \pm i f_{1\Omega'_-}^{\rm II}(\tau_{ws};1+\sqrt{2}i) \right) \right\}
   = {\rm NewForms}(64,\chi_8(0,1)). \qquad \bullet
    \label{eq:nf-64-01-byCCF}
\end{align}
\end{exmpl}

\begin{rmk}
We find, therefore, that the vector space $S_2(\Gamma(8))$ is completely 
recovered by the vector space of 
``$(g,n)=(1,2)$ chiral correlation functions'' $F([z],f_\rho)$ 
of the two rational models of ${\cal N}=(2,2)$ SCFT with $N_{D\Lambda_-}=8$; 
one is for $([z],f_\rho)=([i],2)$ and the other for $([\sqrt{2}i],1)$. 
Models with $([z]=[2i],f_\rho)$ never have $N_{D\Lambda_-}=8$
because $D_z = 16$.
\end{rmk}

\subsection{Appendix: Minimum on Modular Forms and Modular Curves}
  \label{ssec:review-modular}

\subsubsection{Congruence Subgroups}

There are a few series of subgroups of ${\rm SL}(2;\Z)$ that have 
a dedicated notation. 
\begin{defn}
Let $N$ be a positive integer. 
\begin{align}
\Gamma(N) & \; := \left\{ \left. 
     \left[\begin{array}{cc} a & b \\ c & d \end{array}\right]
        \in {\rm SL}(2;\Z) \; \right| \; 
            b, c \equiv 0 (N), \; a,d \equiv 1 (N) \right\}, \\
\Gamma_1(N) & \; := \left\{ \left. 
     \left[\begin{array}{cc} a & b \\ c & d \end{array}\right]
        \in {\rm SL}(2;\Z) \; \right| \; 
            c \equiv 0 (N), \; a,d \equiv 1 (N) \right\}, \\
\Gamma_0^0(N) & \; := \left\{ \left. 
     \left[\begin{array}{cc} a & b \\ c & d \end{array}\right]
        \in {\rm SL}(2;\Z) \; \right| \; 
            b, c \equiv 0 (N), \; a,d \in [\Z/(N)]^\times \right\}, 
     \label{eq:def-Gamma00} \\
\Gamma_0(N) & \; := \left\{ \left. 
     \left[\begin{array}{cc} a & b \\ c & d \end{array}\right]
        \in {\rm SL}(2;\Z) \; \right| \; 
            c \equiv 0 (N), \; a,d \in [\Z/(N)]^\times \right\}, 
\end{align}
so $\Gamma(N) \subset \Gamma_1(N) \subset \Gamma_0(N)$ and 
$\Gamma(N) \subset \Gamma_0^0(N) \subset \Gamma_0(N)$. $\Gamma(N)$ is a normal 
subgroup of ${\rm SL}(2;\Z)$, and hence also that of any subgroup of 
${\rm SL}(2;\Z)$ that contains $\Gamma(N)$.

All those subgroups can be regarded as special cases of a more general 
class of subgroups (\cite[\S3.3]{Shimura-AA} and \cite[\S9.1]{Ribet-Stein})
\begin{align}
 \Gamma(H,N,M) := \left\{ \left. 
      \left[\begin{array}{cc} a & b \\ c & d \end{array} \right]
         \in {\rm SL}(2;\Z) \; \right| \; 
         c \equiv 0 (N), \; b \equiv 0 (M), \; 
         a, d \in H \right\}
\end{align}
where $N$ and $M$ are positive integers, and $H$ is a subgroup of the 
multiplicative group $[\Z/(NM)]^\times$. $\Gamma(N)$ corresponds to 
the case of $M=N$ and $H=H_N$, where 
\begin{align}
H_N := \{1, N+1, 2N+1, \cdots, N^2-N+1 \} \subset [\Z/(N^2)]^\times.
   \label{eq:def-Hn-subgroup}
\end{align}
The subgroup $\Gamma_1(N)$ is reproduced by setting $M=1$ and $H = \{1\}$,
$\Gamma_0^0(N)$ by setting $M=N$ and $H = H_N$, and $\Gamma_0(N)$ by 
setting $M=1$ and $H=[\Z/(N)]^\times$.  $\bullet$
\end{defn}

\begin{defn}
A subgroup $\Gamma$ of ${\rm SL}(2;\Z)$ is said to be a 
{\it congruent subgroup}, if there exists a positive integer $N$ so that 
$\Gamma(N) \subset \Gamma$.
\end{defn}

\begin{lemma}
\label{lemma:isom-GammaN-GHNM}
There is an isomorphism from $\Gamma(H,N,M)$ to $\Gamma(H,NM,1)$ 
given by 
\begin{align}
  \left[ \begin{array}{cc} a & b \\ c & d \end{array} \right]
  \longmapsto 
  \left[ \begin{array}{cc} M^{-1} & \\ & 1 \end{array} \right]
  \left[ \begin{array}{cc} a & b \\ c & d \end{array} \right]
  \left[ \begin{array}{cc} M & \\ & 1 \end{array} \right] = 
  \left[ \begin{array}{cc} a & b/M \\ cM & d \end{array} \right].
\end{align}
In particular, this isomorphism identifies $\Gamma(N)$ with 
$\Gamma(H_N, N^2,1)$. 
\end{lemma}

\subsubsection{The Vector Space of Modular Forms}

\begin{notn}
Let $A$ be an element of $\left\{ [a,b;c,d] \in {\rm GL}(2;\Z) \; 
| \; ad-bc>0 \right\}$, and $k \in \Z$. Then $|[A]_k$ is an operator 
acting on the space of holomorphic functions on ${\cal H}$ as follows. 
For a holomorphic function $f(\tau)$ on $\tau \in {\cal H}$, 
$f|[A]_k$ is another holomorphic function on ${\cal H}$ given by 
\begin{align}
  (f|[A]_k)(\tau) := \frac{[{\rm det}(A)]^{k/2}}{(c\tau+d)^k} f(A \cdot \tau).
\end{align}
The operator $|[A]_k$ is linear on the vector space of holomorphic 
functions on ${\cal H}$.   
\end{notn}

\begin{defn}
Let $\Gamma$ be a congruence subgroup of ${\rm SL}(2;\Z)$.
When a holomorphic function $f(\tau)$ on $\tau \in {\cal H}$ satisfies 
$f|[\gamma]_k=f$ for $\gamma \in \Gamma$, it is said to be a {\it modular 
form of weight-$k$ for $\Gamma$}. When such a modular form $f$ vanishes 
at all the cusp points of $\Gamma$, it is called a {\it cuspform}. 
They form vector spaces over $\C$, and are denoted by $M_k(\Gamma)$ and 
$S_k(\Gamma)$, respectively.
\end{defn}

\begin{lemma}
\label{lemma:isom-strokeOp-cuspformSpace}
The linear operator $|[\diag(N,1)]_k$ induces an isomorphism 
\begin{align}
  S_k(\Gamma(N)) \cong S_k(\Gamma(H_N,N^2,1)).
\end{align}
This isomorphism is a special case of 
$|[\diag(M,1)]_k : S_k(\Gamma(H,N,M)) \rightarrow S_k(\Gamma(H,NM,1))$, or 
even a special case of 
$|[\diag(r,1)]_k : S_k(\Gamma(H,N,M)) \rightarrow S_k(\Gamma(H,Nr, M/r))$ 
for any positive divisor $r$ of $M$.
\end{lemma}

\begin{lemma}
Let $\Gamma$ and $\Gamma'$ be both congruent subgroups of ${\rm SL}(2;\Z)$. 
Suppose that $\Gamma$ is a normal subgroup of $\Gamma'$. 
The group $\Gamma$ acts trivially (=as identity) on 
$S_k(\Gamma)$ via the $|[-]_k$ operator; 
the fact that $\Gamma$ is a normal subgroup of $\Gamma'$ can be used 
to see that $f|[\gamma]_k \in S_k(\Gamma)$ for $f \in S_k(\Gamma)$ and 
$\gamma \in \Gamma'$, so the group $\Gamma'$ also acts on the vector 
space $S_k(\Gamma)$ via $|[-]_k$, but non-trivially. The vector space 
$S_k(\Gamma)$ can be decomposed under the action of $\Gamma'$ into 
the form of 
\begin{align}
  S_k(\Gamma) \cong \oplus_{\rho \in iReps(\Gamma'/\Gamma)} S_k(\Gamma',\rho), 
\end{align}
where $iReps(\Gamma'/\Gamma)$ is the set of irreducible representations 
of the finite group $\Gamma'/\Gamma$; cusp forms in $S_k(\Gamma',\rho)$
transform under $\gamma \in \Gamma'$ as 
\begin{align}
   S_k(\Gamma', \rho) \ni f|[\gamma]_k \longmapsto  \rho(\gamma) \cdot f 
   \in S_k(\Gamma',\rho).
\end{align}
When a cuspform $f$ belongs to the subspace $S_k(\Gamma',\rho)$, we call 
the choice of the representation $\rho \in iReps(\Gamma'/\Gamma)$ the 
{\it nebentypus} of $f$.

Here are two examples of this decomposition. 
The first example is for $\Gamma = \Gamma_1(N)$ and $\Gamma' = \Gamma_0(N)$.
\begin{align}
  S_k(\Gamma_1(N)) \cong \oplus_{\chi_N \in {\rm Char}([\Z/(N)]^\times)} 
      S_k(\Gamma_0(N),\chi_N),
  \label{eq:SGamma1--SGamma0-decomp}
\end{align}
so $\chi_N$ is a Dirichlet character modulo $N$. 

The other example is for $\Gamma = \Gamma(N) \cong \Gamma(H_N,N^2,1)$ and 
$\Gamma_0^0(N^2)$. First, 
\begin{align}
 S_k(\Gamma(N)) \cong \oplus_{\chi_N \in {\rm Char}([\Z/(N)]^\times)} \; 
    S_k(\Gamma_0^0(N), \chi_N).
 \label{eq:SGamma--SGamma00-decomp}
\end{align}
Due to the isomorphism $\Gamma(N) \cong \Gamma(H_N,N^2,1)$ and 
$\Gamma_0^0(N) \cong \Gamma_0([\Z/(N^2)]^\times, N^2,1) = \Gamma_0(N^2)$, 
we can also recycle the decomposition (\ref{eq:SGamma1--SGamma0-decomp}) 
and use it after replacing $N$ with $N^2$ to make a statement 
equivalent to (\ref{eq:SGamma--SGamma00-decomp}). 
Because the group $H_N$ fits into the chain of subgroups 
of the multiplicative group $\{1\} \subset H_N \subset [\Z/(N^2)]^\times$, 
a cuspform for $\Gamma_1(N^2) = \Gamma(\{1\}, N^2,1)$ is also a cuspform 
for $\Gamma(H_N, N^2,1)$ if the nebentypus $\chi: [\Z/(N^2)]^\times 
\rightarrow S^1$ of the cuspform vanishes on $H_N \subset [\Z/(N^2)]^\times$.
This idea is schematically written down as 
\begin{align}
  S_k(\Gamma(H_N,N^2,1)) \cong 
    \oplus_{\chi_N \in iReps([\Z/(N)]^\times) } \; S_k(\Gamma_0(N^2), \chi_N),
  \label{eq:SGamma--SGamma00-decomp-4SAGE}
\end{align}
which is equivalent to (\ref{eq:SGamma--SGamma00-decomp}). SAGE allows us 
to work out this decomposition (\ref{eq:SGamma--SGamma00-decomp-4SAGE}) 
explicitly, so we can just translate the result into the language of 
(\ref{eq:SGamma--SGamma00-decomp}) appropriately. $\bullet$
\end{lemma}

\begin{anythng}
We are not going to provide a review on such notions as 
(the subspace of) oldforms, (the subspace of) newforms, and Hecke operators
in this article. 
Single page is not enough to review those materials for string theorists, 
while they are quite standard for experts of number theory. We just refer 
the readers to such textbooks and lecture notes as \cite{DS} 
and \cite{Ribet-Stein}.
In the following, we quote from those references a result that we refer 
to in the text, without explaining the jargons. 
\end{anythng}

We begin by stating structure within the vector space $S_2(\Gamma_1(N))$ 
for positive integer $N$.

\begin{thm}(e.g., Thm. 5.8.2 of \cite{DS})
\label{thm:newform-SkGamma0Chi}
For any positive integer $N$, there is a vector subspace of 
$S_k(\Gamma_1(N))$ called the {\it space of newforms}, whose definition 
is found in standard textbooks and lecture notes on modular forms. 
This subspace is denoted by $[S_k(\Gamma_1(N))]^{\rm new}_{N}$, and 
the modular forms in this subspace is said to be of {\it level} $N$.
For any nebentypus $\chi_N$, $[S_k(\Gamma_0(N), \chi_N)]^{\rm new}_{N}$
is $S_k(\Gamma_0(N),\chi_N) \cap [S_k(\Gamma_1(N))]^{\rm new}_N$, 
so 
\begin{align}
  [S_k(\Gamma_1(N))]^{\rm new}_N \cong 
      \oplus_{\chi_N} [S_k(\Gamma_0(N),\chi_N)]^{\rm new}_N.
\end{align}

The Hecke algebra, an algebra over $\Z$ generated by Hecke operators, 
acts on the individual components $[S_k(\Gamma_0(N),\chi_N)]^{\rm new}_N$ 
within $S_k(\Gamma_1(N))$, and all the operators in the algebra 
are diagonalized simultaneously. There is no two simultaneous 
eigenvectors of the Hecke algebra in $[S_k(\Gamma_0(N),\chi_N)]^{\rm new}_N$ 
on which the 1-dimensional representation of the Hecke algebra are identical, 
except when the two eigenvectors are identical except the normalization, 
so there is no ambiguity in eigenspace decomposition 
$[S_k(\Gamma_0(N),\chi_N)]^{\rm new}_N$; this property is called 
the {\it multiplicity one theorem}. 
It is conventional to choose each one of those eigenvectors so that its 
power series expansion in $q = e^{2\pi i \tau}$ begins with the term $q$ 
(with the coefficient 1). Such eigenvectors in 
$[S_k(\Gamma_0(N),\chi_N)]^{\rm new}_N$ are called weight-$k$, level-$N$ 
{\it newforms} for $\Gamma_0(N)$ with the nebentypus $\chi_N$. 
Let ${\rm NewForms}(N, \chi_N)$ be the set of those eigenvectors (newforms). 
$\bullet$
\end{thm}

\begin{thm}(e.g, Thm. 5.8.3 of \cite{DS})
\label{thm:str-SkGamma0Chi}
Let $N$ be a positive integer. The vector space $S_k(\Gamma_0(N),\chi_N)$
contains a subspace $[S_k(\Gamma_0(N),\chi_N)]^{\rm new}_N$, as stated 
in Thm. \ref{thm:newform-SkGamma0Chi}; the whole space 
$S_k(\Gamma_0(N),\chi_N)$ has a structure that is made up of subspaces 
associated with the vector spaces $[S_k(\Gamma_0(M_f),\chi_N)]^{\rm new}_{M_f}$ 
labeled by all possible positive integers $M_f$ that divide $N$. 
To be more concrete, 
\begin{align}
  S_k(\Gamma_0(N),\chi_N) \cong \oplus_{M_f | N}
       \left[ S_k(\Gamma_0(N), \chi_N) \right]_{M_f}, 
\end{align}
and the individual components associated with $M_f | N$ are 
\begin{align}
 \left[ S_2(\Gamma_0(N), \chi_{N}) \right]_{M_f} 
  & \; \cong \oplus_{r | (N/M_f)} |[\diag(r,1)]_k \; 
     \left[ S_k(\Gamma_0(M_f),\chi_N) \right]^{\rm new}_{M_f}, \nonumber \\
  & \; = \oplus_{f \in {\rm NewForms}(M_f,\chi_{N} )} \oplus_{r | (N/M_f)} 
       \C f |[\diag(r,1)]_k ,   \label{eq:substr-Gamma0Chi-levelNr}  \\
  & \; =: \oplus_{f \in {\rm NewForms}(M_f,\chi_{N} )} \; [f]_{M_f}^{N}. 
\end{align}
A newform $f \in {\rm NewForms}(M_f, \chi_N)$ mapped into 
$S_k(\Gamma_0(N),\chi_N)$ via the $|[\diag(1,1)]_k$ operator is called 
a weight-$k$, level-$M_f$ {\it newform} for $\Gamma_0(N)$ of 
nebentypus $\chi_N$.

The set ${\rm NewForms}(M_f, \chi_N)$ and the vector space 
$[S_k(\Gamma_0(M_f),\chi_N)]^{\rm new}_{M_f}$ is empty, if 
the Dirichlet character $\chi_N: [\Z/(N)]^\times \rightarrow S^1$
cannot be induced from a Dirichlet character of $[\Z/(M_f)]^\times$.

The dimension of the vector space $[f]_{M_f}^N$ is the number of 
divisors of the integer $N/M_f$, including $1$ and $(N/M_f)$.
\begin{align}
   \dim_\C \left( [S_k(\Gamma_0(N),\chi_N)]_{M_f} \right) & \; = 
      | {\rm NewForms}(M_f, \chi_N) | \times | {\rm divisors~of~}(N/M_f) |, \\
  \dim_\C \left( S_k(\Gamma_0(N),\chi_N) \right) & \; = \sum_{M_f |N}
     | {\rm NewForms}(M_f, \chi_N) | \times | {\rm divisors~of~}(N/M_f) |.
  \quad 
   \label{eq:dim-str-SkGamma0Chi}
\end{align}
$\bullet$
\end{thm}

Now, let us write down the statement on the structure of the vector space 
$S_k(\Gamma(H,N,M))$.

\begin{thm}\cite[\S9.1]{Ribet-Stein}
Let $N$ and $M$ be positive integers, and $H$ a subgroup of the multiplicative 
group $[\Z/(NM)]^\times$. In the decomposition and the isomorphisms
\begin{align}
 S_k(\Gamma(H,N,M)) \cong 
     \oplus_{\chi_{NM}} S_k(\Gamma([\Z/(NM)]^\times, N, M), \chi_{NM}), \nonumber \\
 \cong S_k(\Gamma(H,NM,1)) \cong 
  \oplus_{\chi_{NM}} S_k(\Gamma([\Z/(NM)]^\times, NM,1), \chi_{NM}),  
  \label{eq:temp-2-5-11}
\end{align}
we already have the notion of the subspace of newforms (introduced 
in Thm. \ref{thm:newform-SkGamma0Chi})
\begin{align}
[S_k(\Gamma([\Z/(NM)]^\times,NM,1),\chi_{NM})]^{\rm new}_{NM} = 
[S_k(\Gamma_0(NM), \chi_{NM})]^{\rm new}_{NM} \subset S_k(\Gamma_0(NM), \chi_{NM}),
   \nonumber 
\end{align}
so we can also introduce the notion of 
{\it the subspace of newforms} $[S_k(\Gamma(H, N, M))]^{\rm new}_{NM}$ and 
$[S_k(\Gamma([\Z/(NM)]^\times, N, M), \chi_{NM})]^{\rm new}_{NM}$
through the identification (\ref{eq:temp-2-5-11}) . 
So long as the character $\chi_{NM}$ vanishes 
on the subgroup $H \subset [\Z/(NM)]^\times$, there is one to one correspondence 
\begin{align}
  |[\diag(M,1)]_k: 
   {\rm NewForms}(\Gamma([\Z/(NM)]^\times, N, M),\chi_{NM}) \cong 
   {\rm NewForms}(\Gamma_0(NM), \chi_{NM}). 
\end{align}

The vector spaces $S_k(\Gamma([\Z/(NM)]^\times, N, M),\chi_{NM})$ with 
$\chi_{NM}$'s that vanish on $H$ have a substructure precisely the same 
as that of $S_k(\Gamma_0(NM),\chi_{NM})$ stated in 
Thm. \ref{thm:str-SkGamma0Chi}; we can just translate the substructure 
by using the isomorphism $|[\diag(M,1)]_k$. 

Of particular interest in this article is the case $M=N$ and $H=H_N$. 
So, let us write down the result of the translation in the notation 
for this particular $M=N$ and $H=H_N$ case. 
\begin{align}
  S_k(\Gamma_0^0(N),\chi_N) \cong \oplus_{M_f|N^2} 
     \left[ S_k(\Gamma_0^0(N),\chi_N) \right]_{M_f}, 
\end{align}
and the individual components associated with $M_f|N^2$ are 
\begin{align}
  [S_k(\Gamma_0^0(N),\chi_N)]_{M_f} & \; \cong \oplus_{r|(N^2/M_f)} 
     |[\diag(N^{-1},1)]_k \cdot |[\diag(r,1)]_k 
          \; [S_k(\Gamma_0(M_f),\chi_N)]^{\rm new}_{M_f},  \nonumber \\
  & \; = \oplus_{f \in {\rm NewForms}(M_f, \chi_N)} \oplus_{r|(N^2/M_f)} 
      \C f|[\diag(r/N,1)]_k,     \label{eq:substr-Gamma00Chi-levelNr} \\
  & \; =: \oplus_{f \in {\rm NewForms}(M_f, \chi_N)} \; [f]_{M_f}^{N^2,N}.
\end{align}
$\bullet$
\end{thm}

\newpage

\section{Chiral Correlation Functions and Hecke Eigenforms}
\label{sec:CCF-eigenform}

We have seen a few examples where the vector space 
$F([z],f_\rho)$ of the ``$(g,n)=(1,2)$ chiral correlation functions'' 
$\{ f_{1\Omega'_-}^{\rm II}(\beta) \; | \; \beta \in iReps\}$ fits into the 
vector space of weight-2 cuspforms. In this section, we study systematically 
how the vector space $F([z],f_\rho)$ of the model for a set of data 
$([z],f_\rho)$ fits into the vector space of cuspforms on various modular 
curves for different choices of $([z], f_\rho)$.  

\subsection{CM Group in Models of Rational CFT: Definition}
\label{ssec:CM-grp-def}

In Examples \ref{exmpl:Dk4-fz1-N64} and \ref{exmpl:Dk8-fz1-N64}, we have 
considered the linear combinations, (\ref{eq:nf-64-00-byCCF}) and 
(\ref{eq:nf-64-01-byCCF}), respectively, of the ``$(g,n)=(1,2)$ chiral 
correlation functions'' $f_{1\Omega'_-}^{\rm II}$ 
%
%
in the rational models of ${\cal N}=(2,2)$ SCFT associated with CM elliptic 
curves.
The complex phases that determine those linear combinations are chosen
so that the resulting linear combinations of the theta functions become 
the Hecke theta functions, so that they become newforms of some level $M_f$
(\cite[Lemma 3]{Shimura-ellCM} quoted as 
Lemma \ref{lemma:Shimura-ellCM-Lemma3} in this article). 

There is nothing {\it wrong} mathematically in choosing the linear 
combinations of $f_{1\Omega'_-}^{\rm II}$'s by hand in that way, even when 
we do not find much {\it motivation} for those particular linear combinations 
from the perspectives of string theory / vertex operator algebra. 
That was an attitude of the authors in \cite{prev.paper}; 
$f_{1\Omega'_-}^{\rm II}$'s span a vector space $F([z];f_\rho)$ over $\C$, and 
we 
may call $\{ f_{1\Omega'_-}^{\rm II}(\tau_{ws};\beta) \; | \; \beta \in iReps\}$ 
the VOA basis. But the individual 
$f_{1\Omega'_-}^{\rm II}(\tau_{ws};\beta)$'s for $\beta \in iReps$ are not 
Hecke eigenforms. 

In this section, we go back to the question of how to choose linear 
combinations of the ``$(g,n)=(1,2)$ chiral correlation functions'' 
$f^{\rm II}_{1\Omega'_-}$, and present a little more idea, spoken in the 
language of rational vertex operator algebra, that motivates the choice 
of the linear combination coefficients, to a large extent. Along the way, 
we arrive at the perspective of identifying the set $iReps$ as a finite 
quotient of ``the Tate modules'' $\otimes_\ell {\rm Tate}_\ell(E)$, and also 
provide conceptual clarification at a few points where we were confused 
in \cite{prev.paper}. 

A guiding principle from arithmetic geometry is the following. 
First, for an abelian variety $B$ defined on a number field $K$, 
the representation $\rho_\ell^\vee(B/K)$ of ${\rm Gal}(\overline{K}/K)$ 
on ${\rm Tate}_\ell(B/K) \cong \Z_\ell^{\oplus 2}$ splits into 
$\oplus_{\lambda|(\ell)} \rho_\lambda^\vee(B/K)$ on $\Z_\lambda^{\oplus 2}$ 
where $\lambda$'s are all the finite primes of ${\cal O}_K$ dividing $(\ell)$.
In section \ref{sec:parametrizatn}, we will be concerned about 
such representations with $(B,K)=(E,k)$ of an elliptic curve of Shimura type 
$E$ defined over a number field $k$, and also about 
its Weil restriction $(B,K)$; see section \ref{sssec:shimuraEC}. 
Second, any complex multiplication operation of a CM abelian variety 
$B$ defined over $K$ can be regarded as an element of ${\rm End}_K(B/K)$; 
that is, we can choose the field of definition of the complex multiplication 
operation as $K$ \cite[Thm. 5]{Shimura-zetaA}. 
Thirdly, the action of any element of ${\rm End}_K(B/K)$ (and hence 
any complex multiplication operation) on ${\rm Tate}_\ell(B/K)$ commutes 
with the action of the Galois group on ${\rm Tate}_\ell(B/K)$. So, the complex 
multiplication operation and the representation of the Galois group 
${\rm Gal}(\overline{K}/K)$ can 
be diagonalized on ${\rm Tate}_\ell(B/K) \otimes \overline{K} = 
\oplus_{\lambda|(\ell)} \overline{K}_\lambda^{\oplus 2}$. So, the Hecke theta 
functions, which are supposed to be in one-to-one with Galois representations
in the Langlands correspondence 
(see sections \ref{sssec:surjective-morph-2BK} and \ref{sssec:Langlands}), 
should be identified with something on which the action of complex 
multiplication operations is diagonalized already. Given those observations 
in arithmetic geometry, we are therefore motivated to capturing action of 
complex multiplication operations in the language of rational CFT.

\subsubsection{The CM Group}
\label{sssec:CM-group}

\begin{defn}Think of the model of rational CFT assigned by bosonic string 
theory for a set of data $([E_z]_\C, f_\rho)$; here, $[E_z]_\C$ is an elliptic 
curve with complex multiplication by an order ${\cal O}_{f_z}$ of an imaginary 
quadratic field $K$, as a reminder. 

Due to the isomorphism (\ref{eq:isom-abstract-C/kL-originalE}), the 
elliptic curve $\C/\Omega'_-(\Lambda_-)$ of the model of rational CFT also 
has complex multiplication by ${\cal O}_{f_z}$, the same order as that of 
the target-space elliptic curve $[E_z]_\C$. So, any complex multiplication 
$\alpha \in {\cal O}_{f_z}$ maps $\Omega'_-(\Lambda_-)$ to itself (though not 
necessarily surjective). Similarly, one can see 
from (\ref{eq:image-Omega'-of-LambdaVee}) that $\C/\Omega'_-(\Lambda_-^\vee)$
is isomorphic to $[E_z]_\C$ (complex analytically), and hence any complex 
multiplication $\alpha \in {\cal O}_{f_z}$ maps $\Omega'_-(\Lambda_-^\vee)$ 
to itself. This means that, when $iReps$ is regarded 
(see Statement \ref{statmnt:iReps-as-torsion-pts}) as a set of torsion 
points of the elliptic curve $\C/\Omega'_-(\Lambda_-)$, any complex 
multiplication $\alpha \in {\cal O}_{f_z}$ induces a map from $iReps$ to 
$iReps$.

Now, we define a subset of ${\cal O}_{f_z}$, 
\begin{align}
  \underline{\mathfrak{c}}([E_z]_\C, f_\rho) & \; 
    := \left\{ \alpha \in 
{\cal O}_{f_z} \; | \; (1+\alpha) = {\rm id}: iReps \rightarrow iReps \right\}
   \nonumber  \\
  & \; = \left\{ \alpha \in {\cal O}_{f_z} \; | \; \alpha \Omega'_-(\Lambda_-^\vee)
    \subset \Omega'_-(\Lambda_-)  \right\},
\end{align}
which is an integral ideal of ${\cal O}_{f_z}$. Two complex multiplications 
$\alpha_1, \alpha_2\in {\cal O}_{f_z}$ induce an identical map from 
$iReps \subset \C/\Omega'_-(\Lambda_-)$ to itself $iReps \subset 
\C/\Omega'_-(\Lambda_-)$, if and only if $\alpha_1 - \alpha_2$ is in the ideal 
$\underline{c}$. So, a map from $iReps$ to itself is defined for individual 
elements of ${\cal O}_{f_z}/\underline{\mathfrak{c}}$. Such a map 
$iRpes \rightarrow iReps$ is not necessarily injective or surjective 
for arbitrary element of $[{\cal O}_{f_z}/\underline{\mathfrak{c}}]$. 
For an element in $[{\cal O}_{f_z}/\underline{\mathfrak{c}}]^\times$, however,
there is an inverse in the multiplication law of 
${\cal O}_{f_z}/\underline{\mathfrak{c}}$, and forms a multiplicative group.
We call this finite abelian group {\it the CM group of the model of rational CFT} for $([E_z], f_\rho)$.  $\bullet$
\end{defn}

\begin{props}
\label{props:cUnderline-formula}
The CM group is given by 
\begin{align}
[{\cal O}_{f_z}/(f_\rho \sqrt{D_z}i)_{{\cal O}_{f_z}}]^\times
  = [{\cal O}_{f_z}/(f_\rho f_z \sqrt{D_K})_{{\cal O}_{f_z}}]^\times
\end{align}
for the model of rational CFT associated with the data $([E_z]_\C, f_\rho)$.

Because $D_z = f_z^2 (-D_K)$ depends only on the imaginary quadratic field 
$K$ and $f_z$, which specifies the order ${\cal O}_{f_z} \subset {\cal O}_K$, 
the models of rational CFT for 
${}^\forall [E_z]_\C \in {\cal E}ll({\cal O}_{f_z})$ and a common $f_\rho$ 
have the same CM group. So, the ideal 
$\underline{\mathfrak{c}}([E_z]_\C,f_\rho)$ may be denoted by 
$\underline{\mathfrak{c}}({\cal O}_{f_z},f_\rho)$ or 
$\underline{\mathfrak{c}}(K,f_z,f_\rho)$ instead.
\end{props}

Before giving a proof of this Proposition, 

\begin{lemma}
Let $[E_z]_\C$ be an elliptic curve with complex multiplication by 
an order ${\cal O}_{f_z}$ of an imaginary quadratic field. We use the 
notation introduced in section \ref{ssec:string}. Then the lattice 
$\mathfrak{b}_z := (\Z + z \Z)$ for the analytic representation 
$[E_z]_\C \cong \C/\mathfrak{b}_z$ is a proper\footnote{
{\bf Definition}: 
A fractional ideal $\mathfrak{b}_z$ of an order ${\cal O}_{f_z}$ of an 
imaginary quadratic field $K$ is said to be {\it proper}, if its algebra 
of complex multiplication is no less, no larger than ${\cal O}_{f_z}$.} 
 ideal of ${\cal O}_{f_z}$. Let us then fix an isomorphism 
$\C/\Omega'_-(\Lambda_-) \cong \C/\mathfrak{b}_z$ 
(cf Lemma \ref{statmnt:iReps-as-torsion-pts}). Then the set 
$iReps$ corresponds to 
\begin{align}
  iReps \cong  (f_\rho\sqrt{D_z}i)^{-1} \mathfrak{b}_z/ \mathfrak{b}_z \subset 
   \C/\mathfrak{b}_z \cong \C/\Omega'_-(\Lambda_-).
\end{align}
\end{lemma}
\begin{proof}: In the Statement \ref{statmnt:iReps-as-torsion-pts}, we have 
almost already argued that $\Omega'_-(\Lambda_-^\vee)$ is mapped into 
$(1/(if_\rho \sqrt{D_z})) \times \mathfrak{b}_z$, when $\Omega'_-(\Lambda_-)$ 
is mapped to $\mathfrak{b}_z$. The claim of this Lemma is then obvious. 
\end{proof}

\begin{proof} of Proposition \ref{props:cUnderline-formula}: 
it is easy to understand that $(f_\rho \sqrt{D_z}i)_{{\cal O}_{f_z}}$ is 
contained in the ideal $\underline{\mathfrak{c}}$ of ${\cal O}_{f_z}$. 
At least by resorting to the honest-to-definition study of the subset 
$\underline{\mathfrak{c}} \subset {\cal O}_{f_z}$, we have verified that 
$\underline{\mathfrak{c}}$ is indeed equal to 
$(f_\rho \sqrt{D_z}i)_{{\cal O}_{f_z}}$. 
\end{proof}

\begin{rmk}
\label{rmk:OK-CMgrp-ideal}
In the case $f_z=1$, so ${\cal O}_{f_z}={\cal O}_K$, there is a 
relation\footnote{The different $\mathfrak{d}_{K/\Q}$ is known to be 
equal to the principal integral ideal $(\sqrt{D_K})_{{\cal O}_K}$ for 
an imaginary quadratic field $K$.} 
$\underline{\mathfrak{c}}=(f_\rho \sqrt{D_K})_{{\cal O}_K} = 
f_\rho \mathfrak{d}_{K/\Q}$. It then follows that 
\begin{align}
  [{\cal O}_K: \underline{\mathfrak{c}}]
 = {\rm Nm}_{K/\Q} (\underline{\mathfrak{c}}) 
 = f_\rho^2 |D_K|.
\end{align}
One will also notice that 
\begin{align}
  N_{D\Lambda_-}^2 = |D_K|^2 f_\rho^2 
    = |D_K| \cdot {\rm Nm}_{K/\Q}(\underline{\mathfrak{c}}).
\end{align}
\end{rmk}

\subsubsection{Decomposition of $F([z],f_\rho)$ under the Action 
of the CM Group}

\begin{anythng}
\label{statmnt:fusn-alg-actn}
Think of a model of diagonal rational CFT, and let $A_-$ be the set 
of irreducible representations of the chiral algebra $(V_-,Y_-)$. Let 
$\Z[A_-]$ be the fusion algebra; the generators of $\Z[A_-]$ are 
denoted by $\{ e_\beta \; | \; \beta \in A_-\}$.

The symmetry group of the fusion algebra, ${\rm Aut}(\Z[A_-])$, 
acts on the vector space of conformal blocks. 
Let $[f_{\{\alpha\}}] = \{ [f_{\{\alpha\}}]_{\{\beta\}} \; 
| \; \{ \beta \} \in (A_-)^{3g-3+n} \}$ be a $(g,n)$ conformal block 
for a choice of irreducible representations $\{\alpha \} = 
\{ \alpha_1, \cdots, \alpha_n \} \in (A_-)^n$ of the $n$ external lines 
(vertex operator insertions); the $|A_-|^{(3g-3+n)}$ components are for 
different choices of irreducible representations $\beta$ at $(3g-3+n)$ 
sewing loci of a pants decomposition of a $(g,n)$ pointed Riemann surface. 
The group ${\rm Aut}(\Z[A_-])$ acts on the set $(A_-)^n$ naturally, 
where $\{ \alpha \} = \{0,0,\cdots \} \in (A_-)^n$ forms a trivial orbit. 
We consider the action of ${\rm Aut}(\Z[A_-])$ on the vector space 
of $(g,n)$ conformal blocks with $\{ \alpha \} = \{ 0 \}$ in the following, 
given by 
\begin{align}
   g: \;   [f_{\{0\}}]_{\{ \beta_1, \cdots, \beta_{3g-3+n}\}} \mapsto 
     [f_{\{0\}}]_{\{ g^{-1} \cdot \beta_1, \cdots, g^{-1} \cdot \beta_{3g-3+n}\}} ,
  \qquad \quad g \in {\rm Aut}(\Z[A_-]).
  \label{eq:autZG-acton-confblcks}
\end{align}

Suppose that $G_0$ is a symmetry group of the vertex operator algebra 
$(V_-,Y_-)$; a symmetry operation $g_0 \in G_0$ induces a map 
$A_- \rightarrow A_-$, $\beta \mapsto g_0 \cdot \beta$, in a way 
the vertex operators for states $v \in V_{\beta = 0}$ satisfy 
\begin{align}
  Y( g_0 \cdot v; u(p)) = [g_0] \cdot Y(v; u(p)) \cdot [g_0]^{-1} 
\end{align}
for $[g_0] \in \oplus_{\beta \in A_-} {\rm Hom}_{(V_-,Y_-)}(V_{\beta},V_{g_0 \cdot \beta})$.
So, we have a homomorphism $G_0 \rightarrow {\rm Aut}(\Z[A_-])$.
When we choose a set of $n$ representations $\{ \rho_{1,\cdots,n}\}$ 
of $G_0$ where each $\rho_i$ is one of irreducible components of 
$G_0$ on $V_-$, the chiral correlation functions 
\begin{align}
   f_{\otimes_i v_i} :=  \langle f_{\{ 0\}}, \;  \otimes_{i=1}^n v_i \rangle
   \label{eq:reltn-ccf-confBlock}
\end{align}
of $n$ states $\{ v_{i=1,\cdots,n} \}$ in $\rho_{i=1,\cdots,n}$ inserted 
at $n$ points transform under $G_0$ as 
\begin{align}
  [f_{\otimes_i v_i}]_{\{ g_0^{-1} \cdot \beta\} } = [f_{\otimes_i v_i}]_{\{ \beta\}} \cdot 
   \left( \otimes_{i=1}^n \rho_i(g_0) \right).
     \label{eq:gn-confBlock-relatn-gen}
\end{align}

Now let us assume that the group $G_0$ is abelian, and choose an abelian 
subgroup $G$ of ${\rm Aut}(\Z[A_-])$ that contains $G_0$. The representations
$\rho_{i=1,\cdots, n}$ of $G_0$ are all 1-dimensional then. For a chosen 
$\{ v_{i=1,\cdots, n}\}$, consider the representation of the group $G$ on the 
vector space of chiral correlation functions 
${\rm Span}_\C \{ [f_{\otimes_i v_i}]_{\{ \beta\}} \; | \; 
  \{ \beta \} \in (A_-)^{3g-3+n} \}$,'' where on which $g \in G$ acts 
via (\ref{eq:autZG-acton-confblcks}). Some of those chiral correlation 
functions are trivially zero, 
\begin{align}
  [f_{ \otimes_i v_i}]_{\{\beta\}} = 0, \qquad \{ \beta\} \in 
       (A_-)^{3g-3+n} \; \backslash \; (A_-)^{3g-3+n}_*;
\end{align}
those with $\{ \beta \} \in (A_-)^{3g-3+n}_*$ are subject to the relation 
(\ref{eq:gn-confBlock-relatn-gen}).
Thus the vector space ${\rm Span}_\C \{ [f_{\otimes_i v_i}]_{\{ \beta\}} \; |
 \; \{ \beta \} \in (A_-)^{3g-3+n} \}$ is obtained from the quotient vector 
space 
\begin{align}
 \C_{(A_--)^{3g-3+n}}  \rightarrow  \C_{(A_-)^{3g-3+n}_*/G_0},
  \label{eq:vct-spce-qnt-by-G0}
\end{align}
where $\C_{(set)}:= {\rm Span}_\C \{ e_x \; | \; x \in (set)\}$ is a formal 
vector space over $\C$ generated by formal independent basis elements $e_x$ 
that are in one-to-one with the elements of the $(set)$. 
Restriction of the representation of $G$ on this quotient vector space 
to $G_0$ reproduces the representation $\otimes_{i=1}^n \rho_i$, 
with the multiplicity $|(A_-)^{3g-3+n}_*/G_0|$.    $\bullet$
\end{anythng}

Now, let us go back to the cases of our interest in this article. 
We consider models of rational CFT for $([E_z], f_\rho)$ in bosonic 
string theory (a $T^2$-target model) and the corresponding version 
in Type II string theory.
\begin{anythng}
In the model of rational CFT and in the rational model of ${\cal N}=(2,2)$ 
SCFT for a set of data $([E_z]_\C, f_\rho)$, the vertex operator algebra 
$(V_-,Y_-)$ has a symmetry group $G_0 = {\rm Aut}([E_z]_\C) \cong 
{\cal O}_{f_z}^\times$. Each one of the CM group 
$[{\cal O}_{f_z}/\underline{\mathfrak{c}}]^\times$ elements induces an 
automorphism of the fusion algebra $\Z[iReps]$, and hence the CM group 
is contained\footnote{
The fusion algebra here is the group algebra $\Z[iReps]$ of $iReps$ over $\Z$
where $iReps \cong \Lambda^\vee/\Lambda$ is an abelian group.  We know
(cf. \cite[p.198, Prop 36.1]{Sehgal}) that  the
automorphism group of the fusion algebra is isomorphic to the
automorphism group of the abelian group $iReps$.   Thus, the CM group
can be regarded as a subgroup of the automorphism group of the fusion
algebra. It is in general merely a proper subgroup; not all
automorphisms of the fusion algebra can be regarded as elements of the
CM group.
}
 in ${\rm Aut}(\Z[iReps])$ as a subgroup.\footnote{
Reference \cite{Fuchs-Galois-AutoFusionAlg}
identifies a subgroup of ${\rm Aut}(\Z[A_-])$ of a model of rational CFT
in the following way. The Galois action on the monodromy representation 
matrices (as reviewed in \ref{statmnt:GT-quotient}) 
 induces permutation on 
$A_-$ where $0 \in A_-$ (the vacuum repr) is mapped to itself. 
Such a Galois action is a symmetry of the fusion algebra and the charge 
conjugation combined \cite{Fuchs-Galois-AutoFusionAlg}. 
The authors of the present article are not ready to state the relation 
between this subgroup of ${\rm Aut}(\Z[A_-])$ and the CM group in the 
case of $T^2$-target models. 
} 
So, the CM group plays the role of the group $G$ in the discussion above.

Consider the vector space generated by the ``$(g,n)=(1,2)$ chiral 
correlation functions.'' We set $\{\alpha \} = \{0,0\}$ as in 
\ref{statmnt:fusn-alg-actn}; this is because the two vertex operators 
$(\partial_u X^\C)$ and $J_-$ both belong to $0 \in iReps$. 
The ``$(1,2)$ chiral correlation functions'' are identical to zero, 
if $\{\beta_1, \beta_2\}$ does not satisfy $\beta_1 =\beta_2$. So, 
$(iReps)^{3g-3+n}_* = G_{\Lambda_-}^* \subset G_{\Lambda_-}$.

The representation $\rho$ of $G_0={\rm Aut}([E_z]_\C)$ is $(\rho^1)^{-1}$ 
and trivial for the states $(\partial_u X^\C)$ and $J_-$, respectively. 
The relation (\ref{eq:fII-AutE-relatn}) in the model for 
$([E_z]_\C, f_\rho)$ can be read as a special case of 
(\ref{eq:gn-confBlock-relatn-gen}).
As a special case of Statement \ref{statmnt:fusn-alg-actn}, the CM group 
$G = [{\cal O}_{f_z}/\underline{\mathfrak{c}}]^\times$ acts on the 
vector space of ``the $(g,n)=(1,2)$ chiral correlation functions''\footnote{
The corresponding discussion in our previous paper \cite{prev.paper} 
exploited too much a non-canonical embedding of 
$\Lambda_{\rm winding} \cong \Lambda_-$ into a sublattice of ${\cal O}_K$; 
in fact, $\Lambda_{\rm winding}$ and $\Lambda_{\rm Cardy} \cong \Lambda_-^\vee$ 
are the modules of the order ${\cal O}_K$, which may well be rescaled to 
be within ${\cal O}_K$, but it is not mandatory or essential to do so. 
In the presentation here, it is clear that the complex multiplication 
operations ${\cal O}_{f_z}$ [resp. 
$[{\cal O}_{f_z}/\underline{\mathfrak{c}}]^\times$] forms a ring [resp. a group];
lattices $\Lambda_-$ and $\Lambda_-^\vee$ [resp. $iReps \cong 
\Lambda_-^\vee/\Lambda_-$] are modules/lattices [resp. is a set] on which 
the ring [resp. the group] acts on. We do not think of multiplication 
between elements of $\Lambda_-^\vee$ or $iReps$. We were, in fact, a bit 
confused about this in our previous paper.
} 
$\C_{G_{\Lambda_-}^*/G_0}$; when the representation 
is restricted from the group $G$ to $G_0 = {\rm Aut}([E_z]_\C)$, 
we have just $|G_{\Lambda_-}^*/G_0|$ copy of the representation $(\rho^1)^{-1}$.

We will see in Thm. \ref{thm:inj-surj-CMform} that the homomorphism 
$\C_{G_{\Lambda_-}^*/{\rm Aut}([E_z]_\C)} \rightarrow F([z],f_\rho)$ is not 
necessarily injective; ``the $(g,n)=(1,2)$ chiral correlation functions'' 
may still be subject to linear relations not captured in the 
quotient (\ref{eq:vct-spce-qnt-by-G0}). 
The representation space $\C_{G_{\Lambda_-}^*/G_0}$ may have multiple copies 
of the same representation of the CM group, and this redundancy is 
just gone on the representation space $F([z],f_\rho)$.
Because the CM group 
$[{\cal O}_{f_z}/\underline{\mathfrak{c}}]^\times$ is abelian, the representation 
space $F([z],f_\rho)$ can be decomposed into subspaces 
\begin{align}
  F([z],f_\rho) \cong
  \oplus_{\chi_f^{-1} \in {\rm Char}([{\cal O}_{f_z}/\underline{\mathfrak{c}}]^\times)}
      F([z],f_\rho)^{\chi_f^{-1}}
\label{eq:decomp-Fzfrho-FzfrhoChiF}
\end{align}
where each subspace is the representation space (allowing 
multiplicity larger than 1) of a character $\chi_f^{-1}$ of 
the group $[{\cal O}_{f_z}/\underline{\mathfrak{c}}]^\times$.   $\bullet$
\end{anythng}
%

\subsubsection{Interpretation}

\begin{anythng}
Before moving on to the sections, we discuss in the following, what the 
complex multiplication operation in 
$[{\cal O}_{f_z}/\underline{\mathfrak{c}}]^\times$ 
does on the conformal blocks (and the states in the models of rational CFT) 
to the extent that we can at this moment.
\end{anythng}

\begin{rmk}
The most solid interpretation on the complex multiplication operation 
\begin{align}
[{\cal O}_{f_z}/\underline{\mathfrak{c}}]^\times \ni [\alpha] & \; : 
   iReps 
    \rightarrow iReps    
    \\
 [\alpha] & \; : \C_{G_\Lambda^* /{\rm Aut}([E_z]_\C)}
   \rightarrow \C_{G_\Lambda^* / {\rm Aut}([E_z]_\C)}, 
\end{align}
we have at the moment is that it is a map from $iReps$ to $iReps$ in the 
open string language. Here, we see $iReps = \Lambda_-^\vee/\Lambda_- \cong 
\Lambda_{\rm Cardy}/\Lambda_{\rm winding}$ as the D0-brane Cardy states in a model 
of diagonal rational CFT \cite{GV}. 
The complex multiplication operation $[\alpha]$ induces a
1-to-1 exchange among the Cardy states. The Cardy states are labeled by 
the finite group $\Lambda_{\rm Cardy} / \Lambda_{\rm winding} \subset 
\C/\Lambda_{\rm winding} = [E_z]_\C$ of torsion points of the CM elliptic curve 
$[E_z]_\C$. So, mathematicians are familiar with this action of the complex 
multiplication operation as that on a finite quotient of 
$\otimes_{\ell \in {\rm Spec}(\Z)} {\rm Tate}_\ell(E)$, where 
${\rm Tate}_\ell(E) \cong \Z_\ell \oplus \Z_\ell$ is the Tate module for 
a rational prime integer $\ell$ associated with an elliptic curve 
$[E_z]_\C = \C/\Lambda_{\rm winding}$. When a data $([E_z]_\C, f_\rho)$ is given 
and fixed, then the finite group of torsion points 
$\Lambda_{\rm Cardy}/\Lambda_{\rm winding}$ is fixed; when we scan $f_\rho$ 
while $[E_z]_\C$ fixed, then multiple different finite quotients of 
$\otimes_\ell {\rm Tate}_\ell(E)$ are obtained as $iReps$' for such 
a class of models of rational CFT. 
\end{rmk}

The interpretation of the complex multiplication operation as permutation 
of the Cardy states is well-defined and solid, as above;  here, we will 
now try to see if we can interpret the complex multiplication operation 
as a homomorphism on the Hilbert space of the corresponding model(s) 
of rational CFT. There will be multiple versions of defining such a 
homomorphism, but one version is presented in the following. 

Let $[E_z]_\C$ be an elliptic curve with complex multiplication, and 
$A \in {\rm End}([E_z]_\C) = {\cal O}_{f_z}$. It defines a linear map 
$A: H_1([E_z]_\C; \Z) \rightarrow H_1([E_z]_\C; \Z)$. Its dual isogeny 
$A^\vee \in {\rm End}([E_z]_\C)$ (so $A^\vee \cdot A = A \cdot A^\vee = 
{\rm deg}(A) \cdot 1$ in ${\rm End}([E_z]_\C)$) also defines a linear map 
$(A^\vee)^*: H^1([E_z]_\C; \Z) \rightarrow H^1([E_z]_\C; \Z)$. 

Now, think of the model of rational CFT for a set of data $([E_z]_\C, f_\rho)$ 
for some $f_\rho  \in \N$. 
The linear map
\begin{align}
 A \oplus (A^\vee)^*: 
   H_1([E_z]_\C; \Z) \oplus H^1([E_z]_\C;\Z)
    \longrightarrow H_1([E_z]_\C; \Z) \oplus H^1([E_z]_\C;\Z)
\end{align}
can be extended to the entire Hilbert space of the model of rational CFT by 
having $(A \oplus (A^\vee)^*)$ act only on the U(1) charges while keeping
stringy oscillator excitations intact. 

The endomorphism $(A \oplus (A^\vee)^*)$ on the Hilbert space of the 
model of rational CFT is not one-to-one; $(A \oplus (A^\vee)^*) \in 
{\rm End}({\rm II}_{2,2})$ maps ${\rm II}_{2,2}$ to its index ${\rm deg}(A)^2$
subgroup. Moreover, 
it is a homomorphism of the abelian group ${\rm II}_{2,2} = 
H_1([E_z]_\C;\Z) \oplus H^1([E_z]_\C; \Z)$, but violates the intersection 
form of the lattice ${\rm II}_{2,2}$ by a scalar factor ${\rm deg}(A)$.
The homomorphism $(A\oplus (A^\vee)^*): {\rm II}_{2,2} \rightarrow 
{\rm II}_{2,2}$ maps $\Lambda_{\mp}$ to themselves, and $\Lambda_{\mp}^\vee$ 
also to themselves, where the image is an index ${\rm deg}(A)$ subgroup. 
The complex multiplication operation
$(A + \underline{\mathfrak{c}}) \in {\cal O}_{f_z}/\underline{\mathfrak{c}}$
on $iReps \cong \Lambda^\vee_-/\Lambda_-$ in section \ref{sssec:CM-group}
agrees with the map $(A \oplus (A^\vee)^*): \Lambda^\vee_-/\Lambda_- 
\rightarrow \Lambda_-^\vee/\Lambda_-$. 

\subsection{CM Group Character and Nebentypus}

We have seen in \ref{statmnt:Fzfrho-is-in-S2GammaN} that the vector space 
$F([z],f_\rho)$ of the rational model of ${\cal N}=(2,2)$ SCFT for a set 
of data $([z],f_\rho)$ is a subspace of $S_2(\Gamma(N_{D\Lambda}))$. On the other hand, the vector space 
$S_2(\Gamma(N))$ has a decomposition (\ref{eq:SGamma--SGamma00-decomp}).  
In the following, we will see that the 
decomposition (\ref{eq:decomp-Fzfrho-FzfrhoChiF}) of the vector space 
$F([z],f_\rho)$ under the action of the CM group is compatible with the 
decomposition (\ref{eq:SGamma--SGamma00-decomp}) of $S_2(\Gamma(N_{D\Lambda}))$ 
with respect to the nebentypus 
$\chi_{N_{D\Lambda}} \in {\rm Char}([\Z/(N_{D\Lambda})]^\times)$.  

In the rest of this article, the linear maps $\Omega'_{\pm}$ and 
$\Omega_{\pm}$, and the lattices $\Lambda_{\pm}$ are denoted by 
$\Omega'$, $\Omega$, and $\Lambda$, respectively. 

\begin{props}
\label{props:nebentypus-formula-CCF}
The vector space $F([z],f_\rho)^{\chi_f^{-1}}$ for $\chi_f^{-1} \in 
{\rm Char}([{\cal O}_{f_z}/\underline{\mathfrak{c}}]^\times)$ is in the 
component $S_2(\Gamma_0^0(N_{D\Lambda}),\chi_{N_{D\Lambda}}[\chi_f])$ where the 
nebentypus $\chi_{N_{D\Lambda}} \in {\rm Char}([\Z/(N_{D\Lambda})]^\times)$ 
for $\chi_f^{-1}$---denoted by $\chi_{N_{D\Lambda}}[\chi_f]$---is given by 
\begin{align}
 \chi_{N_{D\Lambda}}[\chi_f] = \epsilon_{D\Lambda}^{-1} \cdot \underline{\chi}_f^{-1} 
   = \left( \frac{D_K}{-} \right) \underline{\chi}_f^{-1}.
\end{align}
Both $\epsilon_{D\Lambda}^{-1}$ and $\underline{\chi}_f$ are characters 
of the multiplicative group $[\Z/(N_{D\Lambda})]^\times$; 
$\epsilon_{D\Lambda}^{-1}$ has been introduced in 
Lemma \ref{lemma:Weil-repr-kernel}, while 
$\underline{\chi}_f^{-1} := \chi_f^{-1} \cdot i$, where\footnote{
Note that $\underline{\mathfrak{c}} \; | \; (N_{D\Lambda})_{{\cal O}_{f_z}}$ 
(and hence this homomorphism $i: [\Z/(N_{D\Lambda})]^\times \rightarrow 
[{\cal O}_{f_z}/\underline{\mathfrak{c}}]^\times$ is well-defined). 
This is because the former is $\underline{\mathfrak{c}} = 
(f_\rho\sqrt{D_z}i)_{{\cal O}_{f_z}}$, and the latter is $(N_{D\Lambda})_{{\cal O}_{f_z}} 
= (f_\rho D_z)_{{\cal O}_{f_z}}$; finally, 
$\sqrt{D_z}i = f_z \sqrt{D_K} \in {\cal O}_{f_z}$.}  
$i: [\Z/(N_{D\Lambda})]^\times \rightarrow [{\cal O}_{f_z}/\underline{\mathfrak{c}}]^\times$
is to regard multiplication of an integer (modulo $N_{D\Lambda}$) as 
a complex multiplication. 
\end{props}

\begin{proof}: The vector space $F([z],f_\rho)^{\chi_f^{-1}}$ is generated by 
the linear combinations of the ``$(g,n)=(1,2)$ chiral correlation functions'' 
of the form 
\begin{align}
  f_{1\Omega'}^{\rm II}(\tau_{ws}; [\beta_0]; \chi_f)
   := \frac{1}{|{\rm Aut}([E_z]_\C)|
    |[{\cal O}_{f_z}/\underline{\mathfrak{c}}]^\times_{\beta_0}|}
    \sum_{g' \in [{\cal O}_{f_z}/\underline{\mathfrak{c}}]^\times} 
     \chi_f(g') \; f_{1\Omega'}^{\rm II}(\tau_{ws}; (g')^{-1} \cdot \beta_0)
  \label{eq:lin-comb-Fzfrho-chiF}
\end{align}
for $\beta_0$'s in $G_\Lambda^*$. Here, $[{\cal O}_{f_z}/
\underline{\mathfrak{c}}]^\times_{\beta_0}$ is the isotropy subgroup of 
the CM group at $\beta_0 \in iReps$.
To see this, one just has to note that 
\begin{align}
  g \cdot f_{1\Omega'}^{\rm II}([\beta_0]; \chi_f) & \; = 
  \frac{1}{ |{\rm Aut}([E_z]_\C)| 
            |[{\cal O}_{f_z}/\underline{\mathfrak{c}}]^\times_{\beta_0}|} 
    \sum_{g' \in [{\cal O}_{f_z}/\underline{\mathfrak{c}}]^\times}
      \chi_f(g') \; f_{1\Omega'}^{\rm II}(g^{-1}\cdot (g')^{-1}\cdot \beta_0)
     \nonumber \\
& \; = \frac{\chi_f^{-1}(g)}{ |{\rm Aut}([E_z]_\C)| 
               |[{\cal O}_{f_z}/\underline{\mathfrak{c}}]^\times_{\beta_0}| } 
    \sum_{g' \in [{\cal O}_{f_z}/\underline{\mathfrak{c}}]^\times}
      \chi_f(g'g) \; f_{1\Omega'}^{\rm II}(g^{-1}\cdot (g')^{-1}\cdot \beta_0)
     \nonumber \\
& \; = \chi_f^{-1}(g) \; f_{1\Omega'}^{\rm II}([\beta_0]; \chi_f).
\end{align}
The action of $g \in [{\cal O}_{f_z}/\underline{\mathfrak{c}}]^\times$ 
on $f_{1\Omega'}^{\rm II}$'s is that of (\ref{eq:autZG-acton-confblcks}, 
\ref{eq:reltn-ccf-confBlock}).

Now, it is straightforward to see, 
for $A = [a,b; c,d] \in \Gamma_0^0(N_{D\Lambda})$, that 
\begin{align}
 f_{1\Omega'}^{\rm II}([\beta_0]; \chi_f)|[A]_2 & \; =
   \frac{1}{ |{\rm Aut}([E_z]_\C)| 
             |[{\cal O}_{f_z}/\underline{\mathfrak{c}}]^\times_{\beta_0}| }
    \sum_{g' \in [{\cal O}_{f_z}/\underline{\mathfrak{c}}]^\times}
        \chi_f(g') \; f_{1\Omega'}^{\rm II}((g')^{-1}\cdot \beta_0)|[A]_2
      \nonumber \\
  & \; = \frac{\epsilon_{D\Lambda[-1]}^{-1}(a)}{|{\rm Aut}([E_z]_\C)|
            |[{\cal O}_{f_z}/\underline{\mathfrak{c}}]^\times_{\beta_0}| }
    \sum_{g' \in [{\cal O}_{f_z}/\underline{\mathfrak{c}}]^\times}
        \chi_f(g') \; f_{1\Omega'}^{\rm II}(a \cdot (g')^{-1}\cdot \beta_0)
      \nonumber \\
  & \; = \epsilon_{D\Lambda}^{-1}(d) \underline{\chi}_f^{-1}(d) \; 
      f_{1\Omega'}^{\rm II}([\beta_0]; \chi_f); 
\end{align}
we have used Prop. \ref{props:chi-corrl-II-lift}, 
eq. (\ref{eq:chi-corrl-fcn-4-L1}), and Lemma \ref{lemma:Weil-repr-kernel},  
in the second equality above. 
\end{proof}

\subsubsection{The Example $S_2(\Gamma(8))$ Once Again}
\label{sssec:exmpl-S2-Gamma8-followup}

Let us revisit the examples we worked on in section \ref{sssec:easiest-ex}
to see that linear combinations (\ref{eq:lin-comb-Fzfrho-chiF}) 
diagonalizing the CM group action correspond indeed to Hecke newforms and 
their images in the vector space of oldforms. The 5-dimensional vector space
$S_2(\Gamma(8)) \cong S_2(\Gamma_0^0(8),\chi_8(0,0)) \oplus 
S_2(\Gamma_0^0(8),\chi_8(0,1))$ is generated by ``the chiral correlation 
functions'' of the model with the data $(j,f_\rho)=(1728,2)$ and the model 
with the data $(j,f_\rho)=(20^3,1)$; the two models have their own 
CM groups, which act on the 3-dimensional $F([i],2) \cong 
S_2(\Gamma_0^0(8),\chi_8(0,0))$ and the 2-dimensional $F([\sqrt{2}i],1) \cong 
S_2(\Gamma_0^0(8),\chi_8(0,1))$ separately. So, we study the action of the 
CM group on $F([i],2)$ and $F([\sqrt{2}i],1)$ separately. 

\begin{exmpl}
\label{exmpl:Dk4-fz1-N64-followup}
In the model of rational CFT associated with the data $(j(z),f_\rho)=(1728,2)$, 
$K=\Q(\sqrt{-1})$, and 
$\underline{\mathfrak{c}}([E_{z=i}]_\C,2)=(4)_{{\cal O}_K}$.
$[{\cal O}_K/\underline{\mathfrak{c}}]^\times \cong \Z/(2) \times \Z/(4)$, 
where the first factor is generated by an order-2 element $(3+2i)$ mod 
$\underline{\mathfrak{c}}$ and the second factor by an order-4 element $(i)$ 
mod $\underline{\mathfrak{c}}$. The ${\rm Aut}([E_z]_\C) \cong \Z/(4)$ 
subgroup of $[{\cal O}_K/\underline{\mathfrak{c}}]^\times$ is precisely 
the $\Z/(4)$ factor generated by $[(i)\times]$, the multiplication of $i$. 

The CM group $[{\cal O}_K/\underline{\mathfrak{c}}]^\times$ acts on the 
vector space $F([z],f_\rho)$ in the way the restriction 
of the characters $\chi_f^{-1}$ on ${\rm Aut}([E_z]_\C) \subset 
[{\cal O}_K/\underline{\mathfrak{c}}]^\times$ is $(\rho^1)^{-1}$. 
Let us parametrize the characters $\chi_f$ of the CM group by 
$\chi_f(a,b)$ with $a \in \Z/(2)$ and $b \in \Z/(4)$; 
$\chi_f: [(3+2i)\times] \mapsto (-1)^a$ and $\chi_f: [(i)\times] \mapsto i^b$. 
The condition $(\chi_f^{-1})|_{{\rm Aut}([E_z]_\C)} = (\rho^1)^{-1}$
leaves $b = 1 \in \Z/(4)$, and ${}^\forall a \in \Z/(2)$. 
The conductor is $\mathfrak{c}_f = (4)_{{\cal O}_K}$ if $a=1 \in \Z/(2)$, 
and $\mathfrak{c}_f = (2+2i)_{{\cal O}_K}$ if $a = 0 \in \Z/(2)$. 

Let us work out the nebentypus of the lift of chiral correlation functions
by using Prop \ref{props:nebentypus-formula-CCF}. First, 
$\epsilon_{D\Lambda}^{-1}$ on $[\Z/(N_{D\Lambda})]^\times$ is 
$\epsilon_{D\Lambda}^{-1} = (D_K/-) = \chi_8(1,0)$ (recall that $N_{D\Lambda}=8$ in 
this example) in the convention of Example \ref{exmpl:Dk4-fz1-N64}. 
Secondly, the character $\underline{\chi_f} = \chi_8(1,0)$
when $\chi_f=\chi_f(a,1)$, regardless of $a \in \Z/(2)$; this is because the 
homomorphism $i: [\Z/(8)]^\times = [\Z/(N_{D\Lambda})]^\times \rightarrow 
[{\cal O}_{f_z}/\underline{\mathfrak{c}}]^\times = \Z/(2) \times \Z/(4)$ 
is given by $i: [7] \mapsto [(i)\times]^2$ and $i:[5] \mapsto [(1)\times]$.
So, the nebentypus is $\chi_{N_{D\Lambda}}[\chi_f(a,1)] = 
\chi_8(1,0)\cdot \chi_8(1,0)^{-1} = \chi_f(0,0)$ regardless of $a \in \Z/(2)$.
Therefore, we see by applying Prop. \ref{props:nebentypus-formula-CCF} that 
all of $F([z=i],f_\rho=2)^{\chi_f(a,1)^{-1}}$ with ${}^\forall a \in \Z/(2)$ 
of this model should contribute to the cuspforms of nebentypus 
$\chi_8(0,0)$ (as we knew in Example \ref{exmpl:Dk4-fz1-N64}).

In the following, we see explicitly that the newforms (\ref{eq:nf-64-00-byCCF})
and the oldform (\ref{eq:of-64-00-byCCF}) both of nebentypus $\chi_8(0,0)$
are not just random linear combinations of the ``$(g,n)=(1,2)$ chiral 
correlation functions,'' but are those (\ref{eq:lin-comb-Fzfrho-chiF}) 
that diagonalize the action of the CM group. 
As a preparation, let us have a look at how $iReps$ of this SCFT model 
decomposes into the orbits of the CM group action. In the image 
$\Omega(iReps) \cong {\cal O}_K/(4)$, 
\begin{align}
 {\rm orb}_{\mathfrak{q}={\cal O}_K} & \; = \{ 1, i, 3,3i, 3+2i, 2+3i, 1+2i, 2+i \}, \\
 {\rm orb}_{\mathfrak{q}=(1+i)} & \; =\{ 1+i, 3+i, 3+3i, 1+3i\},  
\end{align}
and $\{2,2i\} \amalg \{2+2i\} \amalg \{0\}$; notations for the orbits 
to be introduced systematically in Lemma \ref{lemma:CM-grp-orbit-dcmp} are 
already used here.  

Let us start off with working on $F([z],f_\rho)^{\chi_f(1,1)^{-1}}$. 
The CM-group-diagonalizing combination (\ref{eq:lin-comb-Fzfrho-chiF})
for the character $\chi_f(1,1)$ and the orbit ${\rm orb}_{{\cal O}_K}$
\begin{align}
   f_{1\Omega'}^{\rm II}(\tau_{ws}; [1]; \chi_f(1,1)) & \; = 
  \frac{1}{4} \left( 
     f_{1\Omega'}^{\rm II}(1) + i f_{1\Omega'}^{\rm II}(3i)
   - f_{1\Omega'}^{\rm II}(3) - i f_{1\Omega'}^{\rm II}(i)   \right) \nonumber \\
&  -\frac{1}{4} \left( 
     f_{1\Omega'}^{\rm II}(3+2i) + i f_{1\Omega'}^{\rm II}(2+i) 
   - f_{1\Omega'}^{\rm II}(1+2i) - i f_{1\Omega'}^{\rm II}(2+3i) \right) \nonumber \\
  & \; = f_{1\Omega'}^{\rm II}(1) - f_{1\Omega'}^{\rm II}(3+2i), 
\end{align}
is already the newform in (\ref{eq:nf-64-00-byCCF}); the level 
$M_f = 64$ is also reproduced by $|D_K|{\rm Nm}(\mathfrak{c}_f)=4 \cdot 16$
because $\mathfrak{c}_f=(4)$ for this $\chi_f$. 
 
Let us now work on $F([z],f_\rho)^{\chi_f(0,1)^{-1}}$. 
The CM-group-diagonalizing combination (\ref{eq:lin-comb-Fzfrho-chiF})
for this $\chi_f$ and the orbits ${\rm orb}_{{\cal O}_K}$ and ${\rm orb}_{(1+i)}$
are 
\begin{align}
 f_{1\Omega'}^{\rm II}(\tau_{ws}; [1]; \chi_f(0,1)) & \; = 
  \frac{1}{4} \left( 
     f_{1\Omega'}^{\rm II}(1) + i f_{1\Omega'}^{\rm II}(3i)
   - f_{1\Omega'}^{\rm II}(3) - i f_{1\Omega'}^{\rm II}(i)  \right) \nonumber \\
&  +\frac{1}{4} \left( 
     f_{1\Omega'}^{\rm II}(3+2i) + i f_{1\Omega'}^{\rm II}(2+i)
   - f_{1\Omega'}^{\rm II}(1+2i) - i f_{1\Omega'}^{\rm II}(2+3i) \right) \nonumber \\
  & \; = f_{1\Omega'}^{\rm II}(1) + f_{1\Omega'}^{\rm II}(3+2i), \\
  f_{1\Omega'}^{\rm II}(\tau_{ws}; [1+i]; \chi_f(0,1)) & \; = 
  \frac{1}{4} \left( 
     f_{1\Omega'}^{\rm II}(1+i)  + i f_{1\Omega'}^{\rm II}(1+3i)
   - f_{1\Omega'}^{\rm II}(3+3i) - i f_{1\Omega'}^{\rm II}(3+i)  \right) \nonumber \\
  & \; =  f_{1\Omega'}^{\rm II}(1+i). 
    \label{eq:of-64-00-byCCF-2ndApp}
\end{align}
Those two combinations agree with the level $M_f=32$ newform and oldform 
in (\ref{eq:nf-64-00-byCCF}, \ref{eq:of-64-00-byCCF}).

The example above might have given an impression that there is one-to-one 
correspondence between ``a choice of a character and a CM-group orbit'' 
and ``a newform or an oldform of certain type''. This is a little 
too naive, however. Theorems \ref{thm:newform-contrib-orbits} and 
\ref{thm:oldform-byCCF} describe which orbits of the CM group 
contribute to which newform/oldform.\footnote{In the language of 
Thm. \ref{thm:oldform-byCCF}, the newform/oldform of level $M_f=32$ above
is for $\mathfrak{c}_f=(2+2i)$, where $\underline{\mathfrak{c}}/\mathfrak{c}_f
=(1+i)$, $\mathfrak{q}_{p.\mathfrak{c}_f}={\cal O}_K$. For the $r=2$ oldform, 
we can choose $\mathfrak{q}_r = (1+i) = \mathfrak{q}_{r.f}$; 
$\mathfrak{q}_{r.p.\mathfrak{c}_f}={\cal O}_K$. Contributing to the 
$r=2$ oldform is the orbit ${\rm orb}_{\mathfrak{q}_0}$ with 
$\mathfrak{q}_0=\mathfrak{q}_{r.f}$.} 

Presentation on the example for $([z],f_\rho)=([i],2)$ so far 
illustrates how discussions in section \ref{ssec:CM-grp-def} and 
Prop. \ref{props:nebentypus-formula-CCF} work in practice. The following 
materials on this example, on the other hand, correspond to the 
general theory reviewed in section \ref{ssec:map2shimuraEC}. 
Notations and reasonings to be introduced there are used also here 
without an extra explanation. 

In the example above with $\chi_f(0,1)$ and $M_f=32$, the Hecke character 
$\varphi^{(10)}_K$ is such that (cf Lemma \ref{lemma:HeckeC} and 
Def. \ref{def:Hecke-theta-4K})
\begin{align}
  \varphi^{(10)}_K:&\;
    (3)_{{\cal O}_K} \mapsto -3, \qquad (2+i)_{{\cal O}_K} \mapsto -1+2i, \qquad
    (2-i)_{{\cal O}_K} \mapsto -1-2i, \\
  &\; 
    (11)_{{\cal O}_K} \mapsto -11, \qquad (4+3i)_{{\cal O}_K} \mapsto -3+4i, \qquad
    (4-3i)_{{\cal O}_K} \mapsto -3-4i.
\end{align}
So, for this newform, $T_{\varphi^{(10)}_{K}} = K(\varphi^{(10)}_K(\mathfrak{a}))$ 
is $\Q(i)$, which agrees with the CM field $K$. So, $[T_{\varphi^{(10)}_K}:K] = 1$. 
This Hecke character satisfies $\varphi^{(01)}_K = cc \cdot \varphi^{(10)}_K 
= \varphi^{(10)}_K \cdot cc$, and $K_f$ for the Hecke theta function 
$f=f_{\varphi^{(10)}_K}$ is $\Q$; see Lemma \ref{lemma:SchappPoH-2-1} and 
Example \ref{exmpl:Hecke-i-real}.
This $M_f=32$ newform alone forms a $\Q$-simple CM-type abelian variety 
$A_f/\Q$, which is 1-dimensional.  

To the Hecke character $\varphi^{(10)}_K : 
K^\times \backslash \mathbb{A}^\times_K \rightarrow \C$ of type [-1/2;1,0], 
there is one corresponding arithmetic model $E/K$ of $[E_{z=i}]_\C$ of 
Shimura type (see section \ref{sssec:shimuraEC} for a definition of 
elliptic curves of Shimura type; see section 4.2 of \cite{prev.paper} for 
the correspondence); the field of definition $k$ can be $K$ in this example.
So, $B/K = E/K$ is also 1-dimensional (Lemma \ref{lemma:def-B-as-WeilR}),
which is consistent with $[T_{\varphi^{(1,0)}_K}:K]=1$ in this example.  
There is an isogeny $A_f \times_\Q K \sim B/K = E/K$ 
(Prop. \ref{props:BK-as-quotient-of-Jac1}).

The arithmetic model $E/K$ in this example descends to a model $E^+/\Q$
(see section 4.2 of \cite{prev.paper}), so we also have an abelian variety 
$B^+/\Q$ (see Statement \ref{statmnt:wortmann}). $B^+/\Q=E^+/\Q$, and 
$[T^+:\Q]=1$. The isogenies $A_f \times_\Q K \sim B/K$ descends to 
isogenies $A_f/\Q \sim B^+/\Q$ as in \cite[Thm. 3.1]{Wortmann}. 

In the example above with $M_f=64$, precisely the same story can be repeated, 
although the Hecke character $\varphi^{(10)}_K$ is different from the one 
for $M_f=32$. 

The two Hecke theta functions (\ref{eq:nf-64-00})---therefore the ``chiral 
correlation functions'' (\ref{eq:nf-64-00-byCCF})---are related to 
the $L$-functions for (appropriate arithmetic models of) 
the corresponding elliptic curve $[E_{z}]_\C$ with complex multiplication. 
The choice $\chi_f(0,1)$ (the conductor $(2+2i)_{{\cal O}_K}$ and $M_f=32$) is 
known to be for the arithmetic model $E_{z=i}^+/\Q$ given by $y^2=x^3-x$ and 
its $\Q$-isogenous class.\footnote{
According to Table 1 of \cite{Cremona}, the $\Q$-isogenous class 
of the model $y^2=x^3-x$ over $\Q$, where 
$N_{E/\Q}=(32) \in {\rm Div}({\rm Spec}(\Z))$, consists of  
four models over $\Q$ not mutually isomorphic over $\Q$;
two have $j=1728$, and the other two have $j(2i) = 66^3$. 
All of them have the same ($M_f=32$) newform.
The newform for $M_f=64$ is for four $\Q$-isomorphism classes 
including $y^2=x^3-4x$; 
two are for $j=1728$ and the remaining two for $j(2i)=66^3$, 
$N_{E/\Q} = (64)$; the four forms one $\Q$-isogenous class, as before.
}  See \cite[Chap.II]{Koblitz}. $N_{E/\Q} =(32) \in {\rm Div}({\rm Spec}(\Z))$.
The choice $\chi_f(1,1)$ (the conductor $(4)_{{\cal O}_K}$ and level $M_f=64$)
is known to be for the arithmetic model $E^+_{z=i}/\Q$ given by 
$y^2 = x^3 - 4x$ and its $\Q$-isogenous class.
$N_{E/\Q}=(64) \in {\rm Div}({\rm Spec}(\Z))$. 
$\bullet$
\end{exmpl}

\begin{exmpl}
\label{exmpl:Dk8-fz1-N64-followup}
For $(j(z), f_\rho)=(j(\sqrt{2}i),1)=(20^3,1)$, 
$K=\Q(\sqrt{-2})$, and $\underline{\mathfrak{c}}([E_{\sqrt{2}i}]_\C, 1) = 
(2\sqrt{2}i)_{{\cal O}_K}$, which is an example of the general formula 
(Props. \ref{props:cUnderline-formula}) for $\underline{\mathfrak{c}}$ 
for the cases with complex multiplication by a maximal order ${\cal O}_K$, 
$f_z = 1$. 

The group of CM operations is 
$[{\cal O}_K/\underline{\mathfrak{c}}]^\times \cong 
\Z/(4)$ generated by $[(1+\sqrt{2}i)\times]$, which contains 
${\rm Aut}([E_z]_\C) = \{ [(\pm 1)\times] \}$ as the $\Z/(2) \subset \Z/(4)$ 
subgroup. Characters of the CM group are labeled by $a \in \Z/(4)$, where 
$\chi_f(a): [(1+\sqrt{2}i)\times] \mapsto i^a$. For the restriction of 
the character $\chi_f^{-1}$ of $[{\cal O}_K/\underline{\mathfrak{c}}]^\times$ 
to the subgroup ${\rm Aut}([E_z]_\C) = \{ \pm 1\}$ to be $(\rho^1)^{-1}$, 
$a \in \Z/(4)$ needs to be either one of $\{1,3\} \subset \Z/(4)$. 
For both of the two characters $\chi_f(a)$ with $a \in \{1,3\}$, 
the conductor $\mathfrak{c}_f$ is equal to $\underline{\mathfrak{c}}$. 

Prop. \ref{props:nebentypus-formula-CCF} determines the nebentypus 
of the cuspforms in $F([z],f_\rho)^{\chi_f(a)^{-1}}$; that should reproduce 
the result $\chi_8(0,1)$, which we found in Example \ref{exmpl:Dk8-fz1-N64} 
by comparing Fourier coefficients with database.
First, $\epsilon_{D\Lambda}^{-1} = (D_K/-) = \chi_8(1,1)$ in this case. 
Secondly, $\underline{\chi}_f = \chi_f \cdot i$ is 
$\chi_8(1,0)$ for any one of $\chi_f = \chi_f(a'')$ with $a'' \in \{1,3\}$.
So, the nebentypus $\chi_{N_{D\Lambda}}[\chi_f(a'')]$ of the linear 
combination of the ``chiral correlation functions'' is determined by 
Prop. \ref{props:nebentypus-formula-CCF} to be 
$\epsilon_{D\Lambda}^{-1} \cdot \underline{\chi}_f^{-1} = 
\chi_8(1,1) \cdot \chi_8(1,0)^{-1} = \chi_8(0,1)$ indeed. 

The set $G_\Lambda^* = G_\Lambda \backslash G_\Lambda[2]$ consists of single 
orbit under the CM group action, ${\rm orb}_{\mathfrak{q}={\cal O}_K}= 
\{ 1, 1+\sqrt{2}i, 3, 3+\sqrt{2}i \}$ in the image $\Omega(iReps)$. 
The ``chiral correlation functions'' summed over this orbit as 
in (\ref{eq:lin-comb-Fzfrho-chiF}), 
\begin{align}
 f_{1\Omega'}^{\rm II}(\tau_{ws};[1];\chi_f(a'')) & \; = 
  \frac{1}{2}\left[
         f_{1\Omega'}^{\rm II}(1)
 + i^{a''}f_{1\Omega'}^{\rm II}(3+\sqrt{2}i)
 - f_{1\Omega'}^{\rm II}([3]) 
 -i^{a''}f_{1\Omega'}^{\rm II}(1+\sqrt{2}i) \right), \nonumber \\
  & \; = \left(  f_{1\Omega'}^{\rm II}(1)
 + i^{a''}f_{1\Omega'}^{\rm II}(3+\sqrt{2}i) \right),
\end{align}
are precisely the linear combinations identified with the newforms 
in (\ref{eq:nf-64-01-byCCF}).
$F([\sqrt{2}i], 1) \cong S_2(\Gamma_0^0(8), \chi_8(0,1))$.

For the two characters $\chi_f(a'')$, $a'' \in \{1,3\} = \{ \pm 1\}$, we have 
two Hecke characters of $K^\times \backslash \mathbb{A}_K^\times$ of type 
[-1/2;1,0], which we denote by $\varphi^{(10)}_{K,\pm}$: 
\begin{align}
  \varphi^{(10)}_{K,\pm}:
 &\;  (1+\sqrt{2}i)_{{\cal O}_K} \mapsto \pm(-\sqrt{2}+i), \qquad
      (1-\sqrt{2}i)_{{\cal O}_K} \mapsto \pm(+\sqrt{2}+i). \\
  \varphi^{(01)}_{K,\pm}: 
 &\;   (1+\sqrt{2}i)_{{\cal O}_K} \mapsto \pm(-\sqrt{2}-i), \qquad
      (1-\sqrt{2}i)_{{\cal O}_K} \mapsto \pm(+\sqrt{2}-i). 
\end{align}
Two Hecke characters of type [-1/2;-1,0] are also presented here. 
$T_{\varphi^{(1,0)}_K} = K(\sqrt{2})  = K(i)$, and $[T_{\varphi^{(1,0)}_K}:K]=2$, 
in particular. 
The Hecke theta functions satisfy $f_{\varphi^{(10)}_{K,\pm}} = f_{\varphi^{(01)}_{K,\mp}}$, 
so there are two Hecke theta functions, not four, which are the 
newforms with level $M_f=64$ and nebentypus $\chi_8(0,1)$ we have seen
already. $K_f = \Q(i)$ for both. So, $[T_{\varphi^{(1,0)}_K}:K_f]=2$ 
(see Lemma \ref{lemma:SchappPoH-2-1} and Example \ref{exmpl:Hecke-i-CM}). 
This pair of newforms have one common $\Q$-simple abelian variety 
$A_f/\Q$ of CM type. 

To the two Hecke characters $\varphi^{(10)}_{K,\pm}$, there is just one 
isogenous class of elliptic curves of Shimura type, $E/k$. 
This statement follows from the fact that the character $\chi_{f}(\pm 1)$ 
takes value in ${\cal O}_K^\times = \{ \pm 1\}$ 
only in the $\Z/2\Z$ subgroup $\{ 1, 3\} \subset 
[{\cal O}_K/\mathfrak{c}_f]^\times \cong \Z/4\Z$; 
see section 4.2 of \cite{prev.paper} for why.  
$B/K$ is 2-dimensional, and is $K$-simple, because 
the pair $\varphi^{(10)}_{K,\pm}$ for this model $E/k$ forms 
just one Galois$_K$-orbit (see Lemma \ref{lemma:SchappPoH-2-1}).
There is an isogeny $A_f \times_\Q K \sim B/K$ defined over $K$  
(Props. \ref{props:BK-as-quotient-of-Jac1}).

The arithmetic model $E/k$ also descends to $E^+/F$, where $F$ must be 
a degree-2 subfield of $k$, and must be a degree-2 extension of $\Q$. 
We can see this by applying the condition of Thm. 4.2.25 of \cite{prev.paper}. 
Now, $\dim (B^+/\Q) = 2$. $B^+$ consists of just one $\Q$-simple factor. 
There is an isogeny $A_f \rightarrow B^+$ defined over $\Q$, 
as in \cite[Thm. 3.1]{Wortmann}
 $\bullet$
\end{exmpl}

\subsection{Hecke Theta Functions and Type II String 
Chiral Correlation Functions}
\label{ssec:HTheta-TypeII-CCF}

In Theorems \ref{thm:newform-byCCF}, \ref{thm:oldform-byCCF}, and 
\ref{thm:inj-surj-CMform}
in this section \ref{ssec:HTheta-TypeII-CCF}, we find that Hecke theta 
functions of an imaginary quadratic field $K$ 
(reviewed below by just writing down mathematical facts) are realized by 
the ``chiral correlation functions'' of $T^2$-target models of rational CFT 
associated 
with elliptic curves with complex multiplication by ${\cal O}_K$. 
A similar version of them has been presented in \cite{prev.paper}, but 
in the following, we present a version much more refined from mathematical 
perspectives. 
The relevance of those Hecke theta functions in arithmetic geometry 
of CM elliptic curves is reviewed in section \ref{ssec:map2shimuraEC}. 

First, we refer to some facts about Hecke characters of an imaginary 
quadratic field $K$, and then define Hecke theta functions associated 
with a Hecke character. 
\begin{lemma}
\label{lemma:HeckeC}
Let $K$ be an imaginary quadratic field, and $\mathfrak{c}_f$ an integral ideal 
of ${\cal O}_K$ that satisfies the following compatibility condition. 
For any Hecke character $\varphi: K^\times \backslash \mathbb{A}_K^\times 
\rightarrow \C^\times$ of the idele class group 
$K^\times \backslash \mathbb{A}_K^\times$ of type [-1/2; 1,0] with the 
conductor $\mathfrak{c}_f$ (where we assume that the ideal $\mathfrak{c}_f$ 
is compatible with the type [1,0]), one character 
$\chi'_f: [{\cal O}_K/\mathfrak{c}_f]^\times \rightarrow S^1$ is assigned. 
Conversely, suppose that $\mathfrak{c}_f$ is an integral ideal of ${\cal O}_K$
compatible with the type [1,0], and $\chi'_f: [{\cal O}_K/\mathfrak{c}_f]^\times 
\rightarrow S^1$ is a character whose restriction on the subgroup 
${\cal O}_K^\times \subset [{\cal O}_K/\mathfrak{c}_f]^\times$ is 
$(\rho^1)^{-1}$. Then there are $h({\cal O}_K)$ distinct 
Hecke characters $\{ \varphi_{\hat{a}} \; |
 \; \hat{a} = 1,\cdots, h({\cal O}_K) \}$ of 
the idele class group of $K$. The relation between $\chi'_f$ and $\varphi$ 
is characterized by the condition that 
\begin{align}
\varphi((\alpha)_{{\cal O}_K}) = \rho^1(\alpha) \chi'_f(\alpha), \qquad \quad 
 {}^\forall \alpha \in {\cal O}_K \; {\rm s.t.~}(\alpha) + \mathfrak{c}_f={\cal O}_K.
\end{align}
See section 4.2 of \cite{prev.paper} for a review. 
\end{lemma}

\begin{defn}
\label{def:Hecke-theta-4K}
Let $K$ be an imaginary quadratic field, $\varphi$ a Hecke character 
of the idele class group $K^\times \backslash \mathbb{A}_K^\times$, and 
$\mathfrak{c}_f$ its conductor. Then its Hecke theta function 
$\vartheta(\tau;\varphi)$ over $\tau \in {\cal H}$ is defined by 
\begin{align}
  \vartheta(\tau; \varphi) := \sum_{I {\rm integral}} \varphi(I) q^{{\rm Nm}(I)}, 
\end{align}
where the sum runs over all the integral ideals $I$ of ${\cal O}_K$, 
and $\varphi(I)$ for $I$ not mutually prime to $\mathfrak{c}_f$ is defined 
to be zero. 
\end{defn}

Here is a property of Hecke theta functions that we need. 
\begin{lemma} 
(\cite{Hecke-Werke}; Lemma 3 of \cite{Shimura-ellCM})
\label{lemma:Shimura-ellCM-Lemma3}
Let $\varphi: \mathbb{A}_K^\times/K^\times \rightarrow \C^1$ be a Hecke 
character of the idele class group of an imaginary quadratic field $K$ 
of type $[-1/2;1,0]$, and $\mathfrak{c}_f$ its conductor. Then its Hecke 
theta function $\vartheta(\tau; \varphi)$ is a Hecke newform in 
the vector space $S_2(\Gamma_0(N_\varphi), \chi_N)$ with 
\begin{align}
 N_\varphi & \; = |D_K| \; {\rm Nm}_{K/\Q}(\mathfrak{c}_f), \\
 \chi_{N_\varphi}(d) & \; = 
    \left(\frac{D_K}{d}\right) \chi'_f(d) 
\end{align}
for $d \in [\Z/(N_\varphi)]^\times$.
\end{lemma}

\subsubsection{Newforms Realized by Chiral Correlation Functions}

For a given imaginary quadratic field $K$ and its Hecke character $\varphi$
of type [-1/2; 1,0], 
its Hecke theta function $\vartheta(\tau; \varphi)$ is a weight-2 
cuspform with the level and nebentypus given above. It is in the 
space of newforms $[S_2(\Gamma_0(N_\varphi);\chi_{N_\varphi})]^{\rm new}_{N_\varphi}$, 
in particular. In the following, we will see that all of those Hecke 
theta functions are realized by the ``chiral correlation functions'' of 
rational $T^2$-target models of Type II string for $([z],f_\rho)$
with $[z] \in {\cal E}ll({\cal O}_K)$ and an appropriately chosen K\"{a}hler 
parameter $f_\rho \in \N$.

\begin{thm}
\label{thm:newform-byCCF}
Fix ${\cal O}_{f_z}$ with $f_z = 1$ (so ${\cal O}_{f_z}={\cal O}_K$) 
and $f_\rho$. Think of rational $T^2$-target models of ${\cal N}=(2,2)$ 
SCFT associated with data $([z],f_\rho)$, $[z] \in {\cal E}ll({\cal O}_K)$.
Then the CM group $[{\cal O}_K / \underline{\mathfrak{c}}]^\times$
is common to all of those $h({\cal O}_K)$ models 
(Prop. \ref{props:cUnderline-formula}). 
For a character $\chi_f$ of this common CM group whose restriction 
to the subgroup ${\rm Aut}([E_{z_a}]_\C) = {\cal O}_K^\times$ is $\rho^1$, 
the vector space 
\begin{align}
   F({\cal E}ll({\cal O}_K),f_\rho)^{\chi_f^{-1}} := 
       \oplus_{a=1}^{h({\cal O}_K)} F([z_a], f_\rho)^{\chi_f^{-1}}.
  \label{eq:Fzfrho-chif-sumoverClK}
\end{align}
of ``$(g,n)=(1,2)$ chiral correlation functions'' of those models 
is of our interest. 

Let $\mathfrak{c}_f$ be an integral ideal of ${\cal O}_K$ compatible 
with the type [1,0], and suppose that $\mathfrak{c}_f | 
\underline{\mathfrak{c}}$. Let $\varphi$ be a Hecke character of 
the idele class group $K^\times \backslash \mathbb{A}_K^\times$ of type 
[-1/2; 1, 0] with the conductor $\mathfrak{c}_f$.  Then {\it its Hecke theta 
function $\vartheta(\tau; \varphi)$ is found within the vector 
space (\ref{eq:Fzfrho-chif-sumoverClK}), if the character $\chi'_f$ 
associated with $\varphi$} (cf Lemma \ref{lemma:HeckeC}) 
{\it is equal to $\chi_f^{-1}$}. Here, 
the argument $\tau$ of the Hecke theta functions is related to 
the complex structure moduli parameter $\tau_{ws}$ of $g=1$ worldsheet 
through (\ref{eq:the-argument-rescaling}).  
Moreover, the Hecke theta functions $\vartheta(\tau; \varphi)$ are 
linear combinations of 
$f^{\rm II}_{1\Omega'}(\tau_{ws}; \beta)_{([z_a],f_\rho)}$'s in the vector 
space (\ref{eq:Fzfrho-chif-sumoverClK}) with all the coefficients 
having the same absolute value (modulus) unless the coefficient is zero 
(Thm. \ref{thm:newform-contrib-orbits} describes for which 
$\beta \in iReps_a$ the coefficient vanishes).
\end{thm}

\begin{rmk}
\label{rmk:sanity-check-nebentypus}
Before getting into a proof, here we run a consistency check. 
The level $N_{D\Lambda_a}$ and the character $\epsilon_{D\Lambda_a}^{-1}$ 
are common to all the $h({\cal O}_K)$ models associated with $([z_a],f_\rho)$, 
$a = 1,\cdots, h({\cal O}_K)$ (see and Lemmas \ref{lemma:NDLambda} 
and \ref{lemma:epsiln=DK}). The vector spaces $F([z_a],f_\rho)^{\chi_f^{-1}}$
for $a=1,\cdots, h({\cal O}_K)$ with a common character $\chi_f$
are therefore identified with a part of a common vector space 
$S_2(\Gamma_0^0(N_{D\Lambda}), \chi_{N_{D\Lambda}}[\chi_f])$ because of 
Prop. \ref{props:nebentypus-formula-CCF}. 

The rescaling of the argument (\ref{eq:the-argument-rescaling}) 
brings the chiral correlation functions in the vector space 
$S_2(\Gamma_0^0(N_{D\Lambda}),\chi_{N_{D\Lambda}}[\chi_f])$ into 
$S_2(\Gamma_0(N_{D\Lambda}^2), \chi_{N_{D\Lambda}}[\chi_f])$.

The vector space $S_2(\Gamma_0(N),\chi_N)$ contains 
${\rm NewForms}(M_f,\chi_N)|[\diag(1,1)]_2$ with $M_f | N$, and 
the case 
$N = N_{D\Lambda}^2 = |D_K|{\rm Nm}_{K/\Q}(\underline{\mathfrak{c}})$
is of our interest (Rmk. \ref{rmk:OK-CMgrp-ideal}). 
So, if a Hecke theta function $\vartheta(\tau; \varphi)$ 
is to be found in the vector space $S_2(\Gamma_0(N_{D\Lambda}^2),\chi_N[\chi_f])$, 
then $N_\varphi | N_{D\Lambda}^2$ should be satisfied, or equivalently 
${\rm Nm}_{K/\Q}(\mathfrak{c}_f) | {\rm Nm}_{K/\Q}(\underline{\mathfrak{c}})$ 
(see Lemma \ref{lemma:Shimura-ellCM-Lemma3}).  $\bullet$
\end{rmk}
 
In proving Thm \ref{thm:newform-byCCF}, 
the following two well-known results are used. 
\begin{lemma}
\label{lemma:Hecke-theta-each-K}
Notation is the same as in Def. \ref{def:Hecke-theta-4K}. The Hecke 
theta function $\vartheta(\tau; \varphi)$ has the following expression:
\begin{align}
  \vartheta(\tau; \varphi) & \; = \sum_{\mathfrak{K} \in {\rm Cl}_K}
     \vartheta(\tau; \varphi; \mathfrak{K}), 
     \label{eq:Hecke-theta-div-into-idealClasses} \\
  \vartheta(\tau; \varphi; \mathfrak{K}) & \; = \frac{1}{|{\cal O}_K^\times|}
    \frac{1}{\varphi(\mathfrak{a}(\mathfrak{K}))} 
     \sum_{\alpha \in \mathfrak{a}(\mathfrak{K})} \chi'_f(\alpha) \rho^1(\alpha)
      q^{ |\rho^1(\alpha)|_\C^2 / {\rm Nm}(\mathfrak{a}(\mathfrak{K}))},
     \label{eq:Hecke-theta-4-idealClass}
\end{align}
where $\mathfrak{a}(\mathfrak{K})$ is any integral\footnote{
Any fractional ${\cal O}_K$ ideal $\mathfrak{a}$ that is 
prime to $\mathfrak{c}_f$ and $[\mathfrak{a}] = \mathfrak{K}^{-1}$
can play the role of $\mathfrak{a}(\mathfrak{K})$, by extending 
the definition of $\chi'_f$ properly. But we stick 
to an integral $\mathfrak{a}(\mathfrak{K})$ in the main text to 
keep the presentation simple. } ideal of ${\cal O}_K$
that belongs to the ideal class $\mathfrak{K}^{-1}$ prime to $\mathfrak{c}_f$.  

The value $\varphi(\mathfrak{a}(\mathfrak{K})) \in \C$ should, 
when the ideal class $\mathfrak{a}(\mathfrak{K})$ is an order $r$ 
element in ${\rm Cl}_K$, satisfy $(\varphi(\mathfrak{a}(\mathfrak{K})))^r = 
\varphi((\alpha)_{{\cal O}_K}) = \chi'_f(\alpha) \rho^1(\alpha)$, where 
$\alpha \in K^\times$ is a generator of the principal ideal 
$(\mathfrak{a}(\mathfrak{K}))^r$. This is 
why the character $\chi'_f$ of $[{\cal O}_K/\mathfrak{c}_f]^\times$ 
almost determines the Hecke character $\varphi$ but not quite (finite 
choices $\{ \varphi_{\hat{a}} \; | \; \hat{a} =1,\cdots, h({\cal O}_K) \}$ 
remain for a given $\chi_f$, as stated in Lemma \ref{lemma:HeckeC}).  $\bullet$
\end{lemma}

\begin{lemma}
For $[E_{z_a}]_\C \in 
{\cal E}ll({\cal O}_K)$, the rank-2 lattices $\mathfrak{b}_{z_a}$ 
for $[E_{z_a}]_\C$ are all fractional ideals of ${\cal O}_K$, and 
$a_{z_a} \mathfrak{b}_{z_a}$ are integral ideals of ${\cal O}_K$. 
${\rm Nm}_{K/\Q} \mathfrak{b}_{z_a} = 1/a_{z_a}$, and 
${\rm Nm}_{K/\Q}(a_{z_a}\mathfrak{b}_{z_a}) = a_{z_a}$. $\bullet$
\end{lemma}

\begin{proof} of Thm. \ref{thm:newform-byCCF}: 
For each $[z_a] \in {\cal E}ll({\cal O}_K)$, let 
$[\mathfrak{b}_{z_a}] = \mathfrak{K}_a^{-1} \in {\rm Cl}_K$, so 
$\mathfrak{K}_a = [\mathfrak{b}_{z_a}^{-1}]$.
Let us first prove for each $\mathfrak{K}_a \in {\rm Cl}_K$ that 
$\vartheta(\tau; \varphi, \mathfrak{K}_a)$ is in the vector space 
$F([z_a], f_\rho)$.

The integral ideal $a_{z_a}\mathfrak{b}_{z_a}$ can be used as 
$\mathfrak{a}(\mathfrak{K_a})$ in Lemma \ref{lemma:Hecke-theta-each-K}, 
if $(a_{z_a}\mathfrak{b}_{z_a})$ is prime to $\mathfrak{c}_f$. 
If it is not, then choose $c'_a \in K$ so that 
$c'_a a_{z_a} \mathfrak{b}_{z_a}$ is integral and also prime to $\mathfrak{c}_f$. 
Such $\mathfrak{b}'_{z_a} := c'_a a_{z_a} \mathfrak{b}_{z_a}$ is used for 
$\mathfrak{a}(\mathfrak{K})$. 
The sum in (\ref{eq:Hecke-theta-4-idealClass}) is over the lattice 
$\mathfrak{b}'_{z_a}$ then. 
${\rm Nm}_{K/\Q}(\mathfrak{b}'_{z_a}) = |c'_a|_\C^2 a_{z_a}$. 
When the sum over $\alpha \in \mathfrak{b}'_{z_a}$ is grouped into  
$\amalg_{[\alpha_0]\in \mathfrak{b}'_{z_a}/\mathfrak{b}'_{z_a}\underline{\mathfrak{c}}} (\alpha_0 + \mathfrak{b}'_{z_a}\underline{\mathfrak{c}})$, 
the coefficient $\chi'_f(\alpha)$ in the sum remains the same in individual 
groups; $\chi'_f(\alpha_1) = \chi'_f(\alpha_2)$ if $\alpha_1, \alpha_2 \in \alpha_0 + \underline{c} \mathfrak{b}'_{z_a}$ because the conductor 
$\mathfrak{c}_f$ divides $\underline{\mathfrak{c}}$, and hence it also divides 
$\underline{\mathfrak{c}} \mathfrak{b}'_{z_a}$.

This observation helps in finding a relation between the Hecke theta 
functions and the ``chiral correlation functions''. The linear map $\Omega$ 
introduced in (\ref{eq:def-Omega}) maps the lattices $\Lambda_a$ and 
$\Lambda_a^\vee$ of the rational $T^2$-target model of ${\cal N}=(2,2)$ SCFT 
to 
\begin{align}
 & \Omega(\Lambda_a) = \underline{\mathfrak{c}} \mathfrak{b}_{z_a}, \qquad 
  \Omega(\Lambda_a^\vee) = \mathfrak{b}_{z_a},  \label{eq:image-LLdual-by-Omega} 
\end{align}
so the groups of the sum labeled by $\mathfrak{b}_{z_a}/
\underline{\mathfrak{c}}\mathfrak{b}_{z_a}$ are in one-to-one with 
the irreducible representations of the vertex operator algebra: 
\begin{align}
iReps_a \leftrightarrow \Omega(iReps_a) = 
\Omega(\Lambda_a^\vee)/\Omega(\Lambda_a) \leftrightarrow 
  \mathfrak{b}'_{z_a}/\underline{\mathfrak{c}}\mathfrak{b}'_{z_a}. 
\end{align}

For each group, 
\begin{align}
 i \frac{\sqrt{|G_{\Lambda_a}|}}{\sqrt{2a_{z_a}f_\rho}} 
   f_{1\Omega'}^{\rm II}(\tau_{ws}; \beta_0)_{([z_a],f_\rho)} & \; = 
   \sum_{w \in}^{w_0 + \Omega(\Lambda_a)} \; \rho^1(w) \; q^{a_{z_a} |w|_\C^2}
  = \sum_{\alpha \in}^{\alpha_0 + \mathfrak{b}'_{z_a}\underline{c}} 
     \frac{\rho^1(\alpha)}{c'_a a_{z_a}}  \; q^{\frac{|\alpha|_\C^2}{|c'_a|_\C^2a_{z_a}}}
   \label{eq:CCF-as-theta-in-Thm1}
\end{align}
indeed, where $\Omega(\beta_0) = [w_0] \in \mathfrak{b}_{z_a}/
\underline{\mathfrak{c}}\mathfrak{b}_{z_a}$, and 
the label $w$ in the sum in the middle is related to $\alpha$ in the 
sum (\ref{eq:Hecke-theta-4-idealClass}) are in the relation 
$c'_a a_{z_a} w = \alpha$. So, we can just take a sum of them to arrive at 
\begin{align}
 \vartheta(\tau; \varphi, \mathfrak{K}_a) & \; = 
    \frac{1}{|{\cal O}_K^\times|}
    \frac{i\sqrt{|G_\Lambda|}}{\sqrt{2f_\rho}}
    \frac{c'_a \sqrt{a_{z_a}}}{\varphi(\mathfrak{b}'_{z_a})}
  \sum_{\beta \in iReps_a} \chi'_f(c'_a a_{z_a}\Omega(\beta))       
      f_{1\Omega'}^{\rm II}(\tau_{ws}; \beta)_{([z_a],f_\rho)}.
     \label{eq:CCF-as-theta-in-Thm1-b}
\end{align}

Obviously, the value of $\chi'_f$ is either zero or a root of unity. 
The factor $c'_a \sqrt{a_{z_a}}/\varphi(\mathfrak{b}'_{z_a})$ is also 
a root of unity. 
So, all the $f^{\rm II}_{1\Omega'}(\tau_{ws};\beta)_{([z_a],f_\rho)}$'s
in the expression above have either the vanishing coefficient 
(when $\chi'_f=0$) or a $\C$-valued coefficient that is 
$|{\cal O}_K^\times|^{-1} \times \sqrt{f_\rho |D_K|/2}$ times a root of unity. 

Next, let us confirm that $\vartheta(\tau;\mathfrak{K}_a)$ is in the 
vector space $F([z_a],f_\rho)^{\chi_f^{-1}}$. The set $iReps_a$ can be grouped 
into orbits under the action of the CM group 
$[{\cal O}_K/\underline{\mathfrak{c}}]^\times$, and the sum over 
$\beta \in iReps_a$ in (\ref{eq:CCF-as-theta-in-Thm1-b}) can also be 
grouped accordingly. The sum in each group is 
(for some $\beta_* \in iReps_a$ representing an orbit)
\begin{align}
 \sum_{g \in [{\cal O}_K/\underline{\mathfrak{c}}]^\times} 
    \chi'_f(c'_a a_{z_a} \Omega(g \cdot \beta_*)) &
    f^{\rm II}_{1\Omega'}(g \cdot \beta_*)_{([z_a],f_\rho)}   \nonumber \\
& =  \chi'_f(c'_a a_{z_a} \Omega(\beta_*)) \times 
   \sum_{g \in [{\cal O}_K/\underline{\mathfrak{c}}]^\times}  \chi'_f(g)
    f^{\rm II}_{1\Omega'}(g \cdot \beta_*)_{([z_a],f_\rho)}, 
  \nonumber 
\end{align}
precisely the same form of sum in (\ref{eq:lin-comb-Fzfrho-chiF}), 
if and only if $\chi'_f = \chi_f^{-1}$.

Finally, because $\vartheta(\tau;\varphi, \mathfrak{K}_a)$'s are 
in the vector space $F([z_a],f_\rho)^{\chi_f^{-1}}$ for all 
$[z_a] \in {\cal E}ll({\cal O}_K)$ 
(if $\chi_f^{-1}=\chi'_f$ and $\mathfrak{c}_f | \underline{\mathfrak{c}}$), 
the Hecke theta functions $\vartheta(\tau;\varphi)$ 
in (\ref{eq:Hecke-theta-div-into-idealClasses}) must be in the 
vector space $F({\cal E}ll({\cal O}_K),f_\rho)^{\chi_f^{-1}}$ 
in (\ref{eq:Fzfrho-chif-sumoverClK}). 
For a given character $\chi_f^{-1} = \chi'_f \in 
{\rm Char}([{\cal O}_K/\mathfrak{c}_f]^\times)$, there are 
$h({\cal O}_K)$ distinct Hecke characters associated with the same 
$\chi'_f$ (Lemma \ref{lemma:HeckeC}), all of their Hecke theta functions 
are in the same vector space 
$F({\cal E}ll({\cal O}_K),f_\rho)^{\chi_f^{-1}}$. 
\end{proof}

Theorem \ref{thm:newform-byCCF} has not specified for 
which $\beta \in iReps_a$, $a = 1,\cdots, h({\cal O}_K)$ the 
linear combination coefficient of 
$f^{\rm II}_{1\Omega'}(\tau_{ws}; \beta)_{([z_a],f_\rho)}$ is non-zero. 
As a preparation for stating this information, we write down 
a decomposition of $iReps_a$ that is compatible with the CM group action. 
This decomposition (for any rational $T^2$-target model with $f_z = 1$)
is a generalization of what we already worked out in 
Examples \ref{exmpl:Dk4-fz1-N64-followup} and \ref{exmpl:Dk8-fz1-N64-followup}. 
\begin{lemma}
\label{lemma:CM-grp-orbit-dcmp}
Let $K$ be an imaginary quadratic field, 
$\underline{\mathfrak{c}}$ an integral ideal of ${\cal O}_K$, and 
$\mathfrak{b}$ a fractional ideal. Let 
$\underline{\mathfrak{c}} = \prod_i \mathfrak{p}_i^{n_i}$ be the prime 
ideal decomposition in ${\cal O}_K$. The group 
$[{\cal O}_K/\underline{\mathfrak{c}}]^\times$ within the ring 
${\cal O}_K/\underline{\mathfrak{c}}$ acts on the 
${\cal O}_K/\underline{\mathfrak{c}}$-module 
$\mathfrak{b}/\mathfrak{b}\underline{\mathfrak{c}}$, and the set 
$\mathfrak{b}/\mathfrak{b}\underline{\mathfrak{c}}$ is decomposed into 
orbits under the action of the group 
$[{\cal O}_K/\underline{\mathfrak{c}}]^\times$.

The orbit decomposition is given by 
\begin{align}
  \mathfrak{b}/\mathfrak{b}\underline{\mathfrak{c}} &\; = 
   \amalg_{\mathfrak{q} | \underline{\mathfrak{c}}} \; {\rm orb}_{\mathfrak{q}}, 
         \label{eq:iReps-MECE-dcmp} \\
 {\rm orb}_{\mathfrak{q}} & \; := \left( \mathfrak{b} \mathfrak{q} \; 
     \backslash \; \cup_{\mathfrak{q} \; | \; \mathfrak{q}' \; | \; 
     \underline{\mathfrak{c}} } \mathfrak{b} \mathfrak{q}' \right)
    / \mathfrak{b} \underline{\mathfrak{c}}, 
\end{align}
and the action of the group $[{\cal O}_K/\underline{\mathfrak{c}}]^\times$ 
on the orbit ${\rm orb}_{\mathfrak{q}_0}$ factors through 
$[{\cal O}_K/\mathfrak{c}_0]^\times$ where $\mathfrak{c}_0 := 
\underline{\mathfrak{c}}/\mathfrak{q}_0$. 
An element $x \in \mathfrak{b}/\mathfrak{b}\underline{\mathfrak{c}}$ 
is in ${\rm orb}_{\mathfrak{q}_0}$ iff $\mathfrak{c}_0 := {\rm Ann}(x)$. $\bullet$
\end{lemma}

The elementary Lemma above in algebraic number theory is translated 
to the following statement in rational $T^2$-target models of CFT. 
\begin{props}
For a fixed K\"{a}hler parameter $f_\rho \in \N_{>0}$, think 
of $h({\cal O}_K)$ models of $T^2$-target CFT associated 
with the data $([z_a],f_\rho)$, where $[z_a] \in {\cal E}ll({\cal O}_K)$, 
$a=1, \cdots, h({\cal O}_K)$. The set of irreducible representations $iReps_a$
is decomposed into the orbits of the CM group 
$[{\cal O}_K/\underline{\mathfrak{c}}]^\times$ where the orbits are 
labeled by integral ideals $\mathfrak{q}_0$ that divides 
$\underline{\mathfrak{c}}$. The orbit in $iReps_a$ that is mapped by 
$\Omega$ to ${\rm orb}_{\mathfrak{q}_0} \subset \mathfrak{b}_{z_a}/\mathfrak{b}_{z_a}\underline{\mathfrak{c}}$ is denoted by ${\rm orb}^a_{\mathfrak{q}_0}$.
The number of irreducible representations in ${\rm orb}^a_{\mathfrak{q}_0}$
depends only on $\mathfrak{q}_0$, but not on 
$a \in \{1,\cdots, h({\cal O}_K)\}$.

In Statement \ref{statmnt:Fzfrho-is-in-S2GammaN}, we have seen that 
$f^{\rm II}_{1\Omega'}(\tau_{ws};\beta)=0$ for $\beta$'s in 
$G_{\Lambda}=\Lambda^\vee/\Lambda$ where the subgroup $G_0 = {\rm Aut}([E_z]_\C)$
of the CM group does not act freely. 
One of such $\beta$'s are the 2-torsion 
points $G_\Lambda[2] \subset G_\Lambda$, which is non-empty iff 
$2|f_\rho^2 D_K$ ($f_z=1$ now). The CM group action preserves $G_\Lambda[2]$, 
and the orbits within $G_\Lambda[2]$ are ${\rm orb}_{\mathfrak{q}_0}$'s
with $\mathfrak{q}_0$ that divisible by $(f_\rho \sqrt{D_K}/2)_{{\cal O}_K}$. 
The other group of $\beta$'s are the fixed points of order-3 automorphisms 
available when $D_K = -3$. They correspond to ${\rm orb}_{\mathfrak{q}_0}$ 
with $\mathfrak{q}_0 = \underline{\mathfrak{c}}/(\sqrt{-3})_{{\cal O}_K}$. 
$\bullet$
\end{props}

\begin{thm}
\label{thm:newform-contrib-orbits}
When a Hecke theta function $\vartheta(\tau; \varphi)$ is written 
as a linear combination of $f^{\rm II}_{1\Omega'}(\tau_{ws}; \beta)_{([z_a],f_\rho)}$'s, 
the coefficients are non-zero only for $\beta$'s in the orbits 
\begin{align}
  \amalg_{a=1}^{h({\cal O}_K)} \; 
    \amalg_{\mathfrak{q} | \underline{\mathfrak{c}}, {\rm ~and~}
             \mathfrak{q} + \mathfrak{c}_f = {\cal O}_K}
      {\rm orb}_{\mathfrak{q}}^a. 
\end{align}
The condition on $\mathfrak{q}$ here is equivalent to 
$\mathfrak{q} | \mathfrak{q}_{p.\mathfrak{c}_f}$, where 
$\mathfrak{q}_{p.\mathfrak{c}_f} := 
  \prod_{i \nin {\rm Supp}(\mathfrak{c}_f)} \mathfrak{p}_i^{n_i}$ when 
$\underline{\mathfrak{c}} = \prod_i \mathfrak{p}_i^{n_i}$ is the prime ideal 
decomposition. 
\end{thm}
\begin{proof}
Suppose that 
$\alpha \in c'_a a_{z_a} \mathfrak{b}_{z_a} = \mathfrak{b}'_{z_a}$ 
in the proof of Thm. \ref{thm:newform-byCCF} is in ${\rm orb}^a_{\mathfrak{q}_0}$.
This means that $\mathfrak{q}_0$ is the largest ideal dividing 
$(\alpha)_{{\cal O}_K} (\mathfrak{b}'_{z_a})^{-1} + \underline{\mathfrak{c}}$. 
The coefficient $\chi'_f(\alpha)$ is non-zero if and only if 
$(\alpha)_{{\cal O}_K}$ is prime to $\mathfrak{c}_f$; this condition 
is the same as $(\alpha)_{{\cal O}_K} (\mathfrak{b}'_{z_a})^{-1}$ being 
prime to $\mathfrak{c}_f$, because we chose $\mathfrak{b}'_{z_a}$ to be 
prime to $\mathfrak{c}_f$ in the proof. So, $\chi'_f(\alpha)$ is 
non-zero if and only if 
$\mathfrak{q}_0 | \mathfrak{q}_{p.\mathfrak{c}_f}$.
\end{proof}

\begin{rmk}
\label{statmnt:CM-form} 
Reference \cite[\S3]{Ribet-Neben} introduces a notion of CM modular form. 
It is a Hecke newform $f \in S_k(\Gamma_0(N),\chi)$, where $\chi$ is the 
nebentypus, with certain properties described in \cite[\S3]{Ribet-Neben}.
Ref. \cite[\S3 and Thm 4.5]{Ribet-Neben} states further that those CM modular 
forms are in one-to-one with Hecke theta functions of some imaginary 
quadratic field $K$ and some Hecke character of 
$\varphi: K^\times \backslash \mathbb{A}_K^\times \rightarrow \C$ of 
type [$-(k-1)/2$; $(k-1)$, 0]. 

Therefore, Thm. \ref{thm:newform-byCCF} can be rephrased as follows. 
{\it All the CM modular forms in $S_2(\Gamma_1(N_{D\Lambda}^2))$, including 
the $|[\diag(1,1)]_2$ image of those of $S_2(\Gamma_1(M_f))$ with the 
level dividing $N_{D\Lambda}^2$, are found within the subspace 
$\oplus_{\chi_f} F({\cal E}ll({\cal O}_K),f_\rho)^{\chi_f^{-1}}$ of 
$S_2(\Gamma_1(N_{D\Lambda}^2))$}.  $\bullet$
\end{rmk}

\subsubsection{Oldforms Realized by Chiral Correlation Functions}
\label{sssec:oldform}

The vector space of ``$(g,n)=(1,2)$ chiral correlation functions'' 
often allows more independent 
linear combinations than those exploited in Theorem \ref{thm:newform-byCCF}; 
(\ref{eq:of-64-00-byCCF}) and (\ref{eq:of-64-00-byCCF-2ndApp}) are  
examples. We also know that the vector space $S_2(\Gamma_0(N),\chi_N[\chi_f])$ 
with 
$N = N_{D\Lambda}^2$ not only contains Hecke theta functions for Hecke 
characters with a conductor $\mathfrak{c}_f$ dividing 
$\underline{\mathfrak{c}}$ (as in Theorem \ref{thm:newform-byCCF}), but also 
their images under the linear operators $|[\diag(r,1)]_2$ for 
$r | {\rm Nm}_{K/\Q}(\underline{\mathfrak{c}}) / {\rm Nm}_{K/\Q}(\mathfrak{c}_f)$.
In the following, we study which of those oldforms are realized by 
chiral correlation functions, and for which of them multiple 
realizations exist (i.e., the homomorphism $\C_{G_{\Lambda}^*/{\rm Aut}([E_z]_\C)}
 \rightarrow F([z],f_\rho)$ is not an isomorphism but a projection). 

\begin{thm}
\label{thm:oldform-byCCF}
Notations being the same as in Theorem \ref{thm:newform-byCCF}. 
Suppose that $\mathfrak{c}_f | \underline{\mathfrak{c}}$, and 
$\chi'_f = \chi_f^{-1}$.
Then {\it the oldforms $\vartheta(\tau; \varphi)|[\diag(r,1)]_2$ associated 
with the Hecke theta function $\vartheta(\tau; \varphi)$ are found within 
the vector space (\ref{eq:Fzfrho-chif-sumoverClK}), 
if there exists an integral ideal $\mathfrak{q}_r$ satisfying} 
\begin{align}
 r = {\rm Nm}_{K/\Q}(\mathfrak{q}_r), \qquad 
    \mathfrak{q}_r \; | \; \underline{\mathfrak{c}}/\mathfrak{c}_f .
\end{align}
\end{thm}

\begin{proof}
The task now is to find a relation between the functions 
$\vartheta(r \tau; \varphi; \mathfrak{K}_a)$ with 
$a = 1,\cdots, h({\cal O}_K)$ and 
the ``chiral correlation functions,'' because
\begin{align}
  \vartheta(\tau;\varphi)|[\diag(r,1)]_2 = 
     r \sum_{\mathfrak{K} \in {\rm Cl}_K} \vartheta(r \tau; \varphi, \mathfrak{K}).
\end{align}

Let $a' \in \{ 1,\cdots, h({\cal O}_K)\}$ be such that 
$[\mathfrak{b}_{z_{a'}}] = [\mathfrak{b}_{z_a} \mathfrak{q}_r^{-1}]$ 
in the ideal class group ${\rm Cl}_K$. Then there exists  
$\xi_{a,\mathfrak{q}_r} \in K$ such that 
\begin{align}
\mathfrak{b}'_{z_a} \mathfrak{q}_r^{-1} = (\xi_{a,\mathfrak{q}_r})_{{\cal O}_K} 
    \mathfrak{b}'_{z_{a'}}. 
\end{align}
Using this, a Fourier series of $q^{r \times {\rm Nm}(I)}$'s with 
$[I] = \mathfrak{K}_a$ can be regarded as a Fourier series of 
$q^{{\rm Nm}(I')}$'s with $[I'] = \mathfrak{K}_{a'}$:
\begin{align}
  \vartheta(r \tau; \varphi; \mathfrak{K}_a) & \; = 
     \frac{1}{|{\cal O}_K^\times|} \frac{1}{\varphi(\mathfrak{b}'_{z_a})}
      \sum_{\alpha \in \mathfrak{b}'_{z_a}} \varphi((\alpha)) 
         q^{r \frac{|\alpha|_\C^2}{ {\rm Nm}(\mathfrak{b}'_{z_a})  }}, \\
  & \; =      \frac{1}{|{\cal O}_K^\times|} \frac{1}{\varphi(\mathfrak{b}'_{z_a})}
      \sum_{\alpha \in \mathfrak{b}'_{z_a}} \varphi((\alpha)) 
         q^{\frac{|\alpha|_\C^2}{ {\rm Nm}(\mathfrak{b}'_{z_a} \mathfrak{q}_r^{-1})  }}, \\
  & \; =      \frac{1}{|{\cal O}_K^\times|} \frac{1}{\varphi(\mathfrak{b}'_{z_a})}
      \sum_{\alpha \in \mathfrak{b}'_{z_a}}
        \chi'_f(\alpha) \rho^1(\alpha)
         q^{\frac{|\alpha/\xi_{a,\mathfrak{q}_r}|_\C^2}{ {\rm Nm}(\mathfrak{b}'_{z_{a'}} )  }}.
\end{align}
Note that the sum over $\alpha \in \mathfrak{b}'_{z_a}$ is equivalent to 
the sum over 
$\alpha/\xi_{a,\mathfrak{q}_r} \in \mathfrak{b}'_{z_{a'}} \mathfrak{q}_r$, and 
is also equivalent to the sum over integral ideals of the form 
$(\alpha/\xi_{a,\mathfrak{q}_r}) (\mathfrak{b}'_{z_{a'}})^{-1} \mathfrak{q}_r^{-1}$.

The character $\chi'_f$ has periodicity, $\chi'_f(\alpha_1) = \chi'_f(\alpha_2)$
when $\alpha_1-\alpha_2 \in \mathfrak{b}'_{z_a} \underline{\mathfrak{c}} 
\mathfrak{q}_r^{-1}$, because we assume in this Theorem that 
$\mathfrak{c}_f | \underline{\mathfrak{c}} \mathfrak{q}_r^{-1}$. This 
periodicity translates to $(\alpha_1-\alpha_2)/\xi_{a,\mathfrak{q}_r} \in 
\mathfrak{b}'_{z_{a'}} \underline{\mathfrak{c}}$. So, 
\begin{align}
  \vartheta(r \tau; \varphi, \mathfrak{K}_a) = \frac{1}{|{\cal O}_K^\times|}
    \frac{\xi_{a,\mathfrak{q}_r}}{\varphi(\mathfrak{b}'_{z_a})}
    \sum_{y \in \mathfrak{q}_r \mathfrak{b}'_{z_{a'}} / \mathfrak{b}'_{z_{a'}} \underline{\mathfrak{c}}}
      \chi'_f(y) \sum_{\gamma \in y} \rho^1(\gamma)
        q^{\frac{|\gamma|_\C^2}{{\rm Nm}(\mathfrak{b}'_{z_{a'}})}},
\end{align}
where $\gamma = \alpha/\xi_{a,\mathfrak{q}_r}$, and 
$\chi'_f(y) := \chi'_f((\xi_{a,\mathfrak{q}_r} \gamma)_{{\cal O}_K})$ for any 
$\gamma$ in the mod $\mathfrak{b}'_{z_{a'}} \underline{\mathfrak{c}}$ class, 
$y$. It is then easy to see that 
\begin{align}
 \vartheta(r \tau; \varphi, \mathfrak{K}_a) = 
    \frac{1}{|{\cal O}_K^\times|}
    \frac{i \sqrt{|G_\Lambda|}}{\sqrt{2f_\rho}}   
    \frac{c'_{a'} \sqrt{a_{z_{a'}}}\xi_{a,\mathfrak{q}_r}}{\varphi(\mathfrak{b}'_{z_a})}
    \sum_{y \in \mathfrak{q}_r \mathfrak{b}'_{z_{a'}} / \mathfrak{b}'_{z_{a'}} \underline{\mathfrak{c}}}
      \chi'_f(y)
      f^{\rm II}_{1\Omega'}(\tau_{ws}; \beta)_{([z_{a'}],f_\rho)}; 
  \label{eq:theta-K-rTau-byCCF}
\end{align}
$c'_{a'}a_{z_{a'}} \Omega(\beta) = y$.
The function $\vartheta(r \tau; \varphi; \mathfrak{K}_a)$ 
is 
in $F([z_{a'}],f_\rho)^{\chi_f^{-1}}$ with $\chi_f^{-1} = \chi'_f$, and 
the oldform $\vartheta(\tau,\varphi)|[\diag(r,1)]_2$ is in 
$F({\cal E}ll({\cal O}_K), f_\rho)^{\chi_f^{-1}}$. 

The overall factor in front of the sum in $y$ has the absolute value 
independent of $a \in 1,\cdots, h({\cal O}_K)$. 
The sum in (\ref{eq:theta-K-rTau-byCCF}) has non-zero contributions 
only from $c'_{a'} a_{z_{a'}} \Omega(\beta) = y$'s in the orbits 
${\rm orb}^{a'}_{\mathfrak{q}_0}$ in $iReps_{a'}$ with $\mathfrak{q}_0$ 
of the form 
\begin{align}
  \mathfrak{q}_0 = \mathfrak{q}_{r.f} \mathfrak{q}_{0.p.\mathfrak{c}_f}, 
   \qquad {}^\forall \mathfrak{q}_{0.p.\mathfrak{c}_f} \quad {\rm s.t.~}
     \mathfrak{q}_{r.p.\mathfrak{c}_f} | \mathfrak{q}_{0.p.\mathfrak{c}_f}
    \quad {\rm and} \quad 
     \mathfrak{q}_{0.p.\mathfrak{c}_f} | \mathfrak{q}_{p.\mathfrak{c}_f}, 
     \label{eq:necs-CM-orbits}
\end{align}
where $\mathfrak{q}_r =: \mathfrak{q}_{r.f} 
\mathfrak{q}_{r.p.\mathfrak{c}_f}$ is the decomposition into 
${\rm Supp}(\mathfrak{c}_f)$ and 
${\rm Supp}(\underline{\mathfrak{c}}) \backslash {\rm Supp}(\mathfrak{c}_f)$.
One can easily see that this condition on $\mathfrak{q}_0$ for the 
$r=1$ case (newform)---the $\mathfrak{q}_r = {\cal O}_K$ case---becomes 
the condition $\mathfrak{q}_0 | \mathfrak{q}_{p.\mathfrak{c}_f}$ in 
Thm. \ref{thm:newform-contrib-orbits}.  
\end{proof}

\begin{thm}
\label{thm:inj-surj-CMform}
The image of the homomorphism 
\begin{align}
\oplus_{a=1}^{h({\cal O}_K)} \C_{G_{\Lambda_a}^*/{\rm Aut}([E_{z_a}]_\C)}
  \longrightarrow S_2(\Gamma(N_{D\Lambda})) \rightarrow S_2(\Gamma_1(N_{D\Lambda}^2)),
\end{align}
denoted by 
\begin{align}
  F({\cal E}ll({\cal O}_K),f_\rho) := \oplus_{\chi_f} 
    F({\cal E}ll({\cal O}_K),f_\rho)^{\chi_f^{-1}} = \oplus_{a=1}^{h({\cal O}_K)}
     F([z_a],f_\rho),
   \label{eq:notatn-FEllOK-frho}
\end{align}
is within the vector subspace spanned by Hecke theta functions of 
the imaginary quadratic fields $K$ such that $\mathfrak{c}_f | 
\underline{\mathfrak{c}}$, and their oldforms.  

The homomorphism above is not necessarily surjective to this 
subspace. This is because there is not always an ideal $\mathfrak{q}_r$ 
for $r = {\rm Nm}_{K/\Q}(\mathfrak{q}_r)$; at a rational prime $(p)$ 
that is inert in the extension $K/\Q$, ${\rm Nm}_{K/\Q}(\mathfrak{q})$ 
cannot have an odd power of $p$. 

The homomorphism above is not necessarily injective. This is because 
there are multiple choices of $\mathfrak{q}_r$ for a given $r$ if it is 
divisible by a rational prime $p$ above which ${\cal O}_K$ is unramified, 
and which splits completely in ${\cal O}_K$.
\end{thm}

\begin{proof}
Consider the homomorphism 
\begin{align}
  \C_{{\rm orb}^a_{\mathfrak{q}_0}/{\rm Aut}([E_{z_a}]_\C)} \longrightarrow 
    F([z_a],f_\rho) = \oplus_{\chi_f} F([z_a],f_\rho)^{\chi_f^{-1}}
    \subset S_2(\Gamma(N_{D\Lambda}))
\end{align}
first. This homomorphism is injective. This is because 
the dimension of the vector space 
$\C_{{\rm orb}^a_{\mathfrak{q}_0}/{\rm Aut}([E_{z_a}]_\C)}$
is the same as the number of characters of $[{\cal O}_K/\mathfrak{c}_0]^\times$
(where $\mathfrak{c}_0 = \underline{\mathfrak{c}}\mathfrak{q}_0^{-1}$) 
whose restriction on $G_0 = {\rm Aut}([E_{z_a}]_\C)$ (use $G_0 \subset 
[{\cal O}_K/\underline{\mathfrak{c}}]^\times \rightarrow 
[{\cal O}_K/\mathfrak{c}_0]^\times$) is the same as $\rho^1$. This means that 
each one of those characters $\chi_f$ gives rise to a linear combination 
of the ``chiral correlation functions'' in 
$\C_{{\rm orb}^a_{\mathfrak{q}_0}}/{\rm Aut}([E_{z_a}]_\C)$, which is in the subspace 
$F([z_a],f_\rho)^{\chi_f^{-1}}$. So, all the images are mutually linearly 
independent, and hence the homomorphism above is injective. We have 
also learned that the linear combinations using the characters of the CM 
group---(\ref{eq:lin-comb-Fzfrho-chiF})---already realize all the possible 
linearly independent combinations of the ``chiral correlation functions'' 
associated with the orbit ${\rm orb}^a_{\mathfrak{q}_0}$. 

Let us now focus on the homomorphism projected onto the component 
$F([z_a],f_\rho)^{\chi_f^{-1}}$, and consider 
\begin{align}
  \oplus_{\mathfrak{q}_0}^{(\mathfrak{c}_f | \underline{\mathfrak{c}} \mathfrak{q}_0^{-1})}
     \C_{{\rm orb}^a_{\mathfrak{q}_0}/{\rm Aut}([E_{z_a}]_\C)}
         \longrightarrow F([z_a],f_\rho)^{\chi_f^{-1}}.
  \label{eq:maybe-inj-orbts-2-CCF}
\end{align}
This is not necessarily injective for the reason stated already 
in the statement of this Theorem; the proof of Thm. \ref{thm:oldform-byCCF}
(the argument around (\ref{eq:theta-K-rTau-byCCF})) indicates that 
two different choices of $\mathfrak{q}_r$ give rise to two different 
linear combinations of the ``chiral correlation functions'' that 
are identical. 

The arguments so far have shown that all that we obtain as 
$\oplus_a F([z_a],f_\rho)^{\chi_f^{-1}}$ are the 
$\vartheta(r \tau, \varphi,\mathfrak{K}_{a'})$'s with $r$ for which 
$\mathfrak{q}_r$ exists. Therefore, the image 
$F({\cal E}ll({\cal O}_K),f_\rho)$ remains within the Hecke theta functions 
of the imaginary quadratic field $K$ with the conductor $\mathfrak{c}_f$ 
dividing $\underline{\mathfrak{c}}$, and their oldforms. Moreover, 
the set of all possible values of $r$ is also determined by the 
availability of an ideal $\mathfrak{q}_r$. So, the oldform  
$\vartheta(\tau,\varphi)|[\diag(r,1)]_2$ is not obtained within 
the vector space $F({\cal E}ll({\cal O}_K),f_\rho)$ for $r$'s 
stated already in this Theorem. 
\end{proof}

\begin{rmk}
\label{rmk:discussion}
Thm. \ref{thm:newform-byCCF} states that Hecke theta functions, newforms 
of level $N_\varphi$, are realized within the vector space 
$F({\cal E}ll({\cal O}_K),f_\rho) \subset S_2(\Gamma_1(N_{D\Lambda}^2))$
only when we include the ``chiral correlation functions'' from 
all the components, $F([z_a],f_\rho)$, with $a=1,\cdots, h({\cal O}_K)$. 
The vector spaces $F([z_a],f_\rho)$ with different $[z_a] \in 
{\cal E}ll({\cal O}_K)$ are ``chiral correlation functions'' of SCFT's 
with different target-spaces.  When Hecke operators of 
$\Gamma_0(N_{D\Lambda}^2)$ act on ``chiral correlation functions'' 
of a model for $([z_a],f_\rho)$, the results are no longer 
within $(\partial_u X^\C)$--$J_-$ ``$(g,n)=(1,2)$ chiral correlation 
functions'' of the same model (cf \cite{Harvey}), 
but are linear combinations 
of those of models for $([z_b],f_\rho)$ with $[z_b] \in 
{\cal E}ll({\cal O}_K)$, including\footnote{There are 
infinitely many imaginary quadratic fields $K$, but 
the set ${\cal E}ll({\cal O}_K)$ consists of just one element 
(i.e., $h({\cal O}_K)=1$) for only nine $K$'s.} $[z_b] \neq [z_a]$. 

When we bring the ``$(g,n)=(1,2)$ chiral correlation functions'' 
$F^{(\Omega, \int J)}_{\{0,0\}}([z_a],f_\rho)$ together from the set 
of models with $([z_a], f_\rho)$ for all $[z_a] \in {\cal E}ll({\cal O}_K)$, 
the action of Hecke operators is closed within this vector space. 
Although it sounds a little odd to take a linear combination 
of correlation functions of multiple different field theory models, 
it turns out to be quite reasonable thing to do from the perspective 
of arithmetic geometry. See Statement \ref{statmnt:fromStrPerspectives}. 
$\bullet$
\end{rmk}

\subsection{An Example: $S_2(\Gamma(20))$}
\label{ssec:exmpl-S2-Gamma20}

We have seen that the vector space 
\begin{align}
  S_2(\Gamma(8)) = S_2(\Gamma_0^0(8), \chi_8(0,0)) \oplus 
  S_2(\Gamma_0^0(8), \chi_8(0,1))
\end{align}
is obtained entirely by the vector spaces of ``$(g,n)=(1,2)$ chiral 
correlation functions'' $F([i],2) \oplus F([\sqrt{2}i], 1)$ of the two 
rational models of ${\cal N}=(2,2)$ SCFT, those for $([z],f_\rho) = ([i],2)$ 
and $([\sqrt{2}i],1)$.  
Let us present here a less trivial example of the discussions 
in section \ref{ssec:HTheta-TypeII-CCF}. 

As such an example, we choose $S_2(\Gamma(20))$. For $N_{D\Lambda} = f_\rho D_z
=f_\rho f_z^2 (-D_K)$ to divide 20, the only possibilities are 
$(D_z, f_\rho)=(4,5)$ and $(20,1)$. Since $f_z=1$ for both of those two 
models, we can use all the results obtained in 
section \ref{ssec:HTheta-TypeII-CCF}.

\subsubsection{The Vector Space $S_2(\Gamma(20))$}
\label{sssec:G20-list-newforms}

First, we have a brief look at the number of newforms that belong to 
$S_2(\Gamma(20))$. All the concrete results in this 
section \ref{sssec:G20-list-newforms} are obtained by using SAGE. 
The vector space $S_2(\Gamma(20))$ of weight-2 cuspforms for $\Gamma(20)$
can be decomposed into $\oplus_{\chi_{20}} S_2(\Gamma_0^0(20), \chi_{20})$, where 
the nebentypus\footnote{
Here, a character of the group $[\Z/(20)]^\times \cong \Z/(2) \times \Z/(4)$ 
is denoted by $\chi_{20}(a,b)$, with $a \in \Z/(2)$ and $b \in \Z/(4)$. 
The factor $\Z/(2)$ of $[\Z/(20)]^\times$ is generated by $11$ mod 20, and 
$\Z/(4)$ of $[\Z/(20)]^\times$ by $17$ mod 20. $\chi_{20}(a,b): [11] \mapsto (-1)^a$ 
and $\chi_{20}(a,b): [17]\mapsto i^b$.}
 $\chi_{20}: [\Z/(20)]^\times \rightarrow S^1$ can be either one 
of $\chi_{20}(0,0)$, $\chi_{20}(1,1)$, $\chi_{20}(0,2)$, and $\chi_{20}(1,3)$; 
there are four other characters of the group $[\Z/(20)]^\times$, but there is 
no weight-2 cuspform of such a nebentypus. The number of newforms of level 
$M_f|20^2$ in $S_2(\Gamma_0^0(20),\chi_{20})$ for each one of the four $\chi_{20}$'s
are shown below. First, for the nebentypus $\chi_{20}(0,0)$, 
\begin{align}
 |{\rm NewForm}(400, \chi_{20}(0,0))|= 8, & \quad |{\rm NewForm}(80, \chi_{20}(0,0))| = 2, \\
|{\rm NewForm}(200, \chi_{20}(0,0))| = 5, & \quad |{\rm NewForm}(40, \chi_{20}(0,0))| = 1, \\
|{\rm NewForm}(100, \chi_{20}(0,0))| = 1, & \quad |{\rm NewForm}(20, \chi_{20}(0,0))| = 1, \\
|{\rm NewForm}(50, \chi_{20}(0,0))| = 2, & 
\end{align}
next for the nebentypus $\chi_{20}(0,2)$, 
\begin{align}
 |{\rm NewForm}(400, \chi_{20}(0,2))|= 8, & \quad |{\rm NewForm}(80, \chi_{20}(0,0))| = 2, \\
|{\rm NewForm}(200, \chi_{20}(0,2))| = 4, & \quad |{\rm NewForm}(40, \chi_{20}(0,0))| = 2, \\
|{\rm NewForm}(100, \chi_{20}(0,2))| = 2, & \quad |{\rm NewForm}(20, \chi_{20}(0,0))| = 0, \\
|{\rm NewForm}(50, \chi_{20}(0,2))| = 2, & 
\end{align}
and finally for $\chi_{20}(1,1)$ and $\chi_{20}(1,3)$, 
\begin{align}
 |{\rm NewForm}(400, \chi_{20}(1,1))|= 9, & \quad |{\rm NewForm}(80, \chi_{20}(0,0))| = 3, \\
|{\rm NewForm}(200, \chi_{20}(1,1))| = 0, & \quad |{\rm NewForm}(40, \chi_{20}(0,0))| = 0, \\
|{\rm NewForm}(100, \chi_{20}(1,1))| = 7, & \quad |{\rm NewForm}(20, \chi_{20}(0,0))| = 1.
\end{align}
The dimension of the entire vector space $S_2(\Gamma_0^0(20), \chi_{20}(1,1))$, 
for example, is 
\begin{align}
 1 \times 9 + 2 \times 0 + 3 \times 7 + 2 \times 3 + 4 \times 0 + 6 \times 1;
\end{align}
the factors multiplied to the number of newforms of level $M_f$ take account of 
the number of oldforms associated with the newforms, and are the number of 
divisors of $400/M_f$.

Now, we have a detailed look at the cuspforms for $\Gamma_0^0(20)$ with 
the nebentypus $\chi_{20}(1,1)$. At level $M_f = 400$, there are nine newforms,
\begin{align}
& \; q + \frac{10}{7}iq^{29} + \frac{16(i-1)}{7}q^{33} - \frac{16(1+i)}{7}q^{37} - \frac{3}{7}q^{41} 
 \nonumber \\
& + a_3 \left( q^3 - 6q^{23} + 7i q^{27} 
  \right) \nonumber \\
& + a_7 \left( q^7 + 3i q^{23} + 3q^{27} - iq^{43} 
  \right) \nonumber  \\
& + a_9 \left( q^9 - \frac{4}{7}q^{29} + \frac{2(1+i)}{7} q^{33} + \frac{2(i-1)}{7} q^{37} 
   + \frac{3}{7}i q^{41} 
   \right) \nonumber \\
& + a_{11} \left( q^{11} + (i - 1)q^{23} + (i + 1)q^{27} + q^{31} - iq^{39} 
 \right) \nonumber \\
& + a_{13} \left( q^{13} + \frac{3(1-i)}{7}q^{29} - \frac{3}{7}q^{33} - \frac{3}{7}iq^{37} 
   + \frac{3(1+i)}{7}q^{41} 
 \right) \nonumber \\
& + a_{17} \left( q^{17} - \frac{2(1+i)}{7}q^{29} - \frac{5}{7}iq^{33} - \frac{2}{7}q^{37} 
   + \frac{2(1-i)}{7} q^{41} 
  \right) \nonumber \\
& + a_{19} \left( q^{19} + 3(1+i) q^{23} + 3(1-i)q^{27} - iq^{31} - q^{39} 
  \right) \nonumber \\
& + a_{21} \left( q^{21} + \frac{6}{7}iq^{29} + \frac{3(1-i)}{7}q^{33} + \frac{3(1+i)}{7}q^{37} 
  - \frac{6}{7}q^{41} 
  \right),
\end{align}
with the nine choices of the Hecke eigenvalues 
$\{ a_3, a_7, a_9, a_{11}, a_{13}, a_{17}, a_{19}, a_{21}\}$:
\begin{align}
 \left[ \begin{array}{cccccccc} 
    0 & 0 & (-3i) & 0 & 5(1-i) & 5(1+i) & 0 & 0 \\
    \pm \sqrt{(10i)} & \pm \sqrt{(-10i)} & (7i) & 0 & 0 & 0 & 0 & 10 \\
    \pm \sqrt{(6i)} & \mp \sqrt{(-6i)} & (3i) & \pm (2\sqrt{3})i & (-1+i) & -(1+i) & \pm 4\sqrt{3} & -6 \\
   \sqrt{i} & 4\sqrt{(-i)} & (-2i) & \pm (3\sqrt{3})i & \pm 2\sqrt{(-3i)} & \mp 3\sqrt{(3i)} & \pm \sqrt{3} & 4 \\
   -\sqrt{i} & -4\sqrt{(-i)} & (-2i) & \pm (3\sqrt{3})i & \mp 2\sqrt{(-3i)} & \pm 3\sqrt{(3i)} & \pm \sqrt{3} & 4 
 \end{array} \right].
  \label{eq:nf-coeff-400-11}
\end{align}
Here, $\sqrt{i}$ and $\sqrt{(-i)}$ are short-hand notations of $e^{\pi i/4}$ and $e^{-\pi i/4}$, 
respectively.

There is no newform at level $M_f = 200$. There are seven newforms at level $M_f = 100$,
which are 
\begin{align}
&  q + i q^9 - 2q^{11} + 2i q^{14} + 2q^{16} - 2iq^{19} 
  + \cdots \nonumber \\
& + a_2 \left( q^2 -\frac{3(1+i)}{2}q^9 -\frac{3}{2}(1-i)q^{11} - 2q^{12} + iq^{13} + 3(i-1)q^{16} - q^{17} - iq^{18} + \frac{3}{2}(i+1)q^{19} 
   \right) \nonumber \\
& + a_3 \left( q^3 + (i- 1)q^9 -(1+i)q^{11} + 2(1+i)q^{16} + iq^{17} - 2q^{18} + (1-i)q^{19} 
 + \cdots \right) \nonumber \\
& + a_4\left( q^4 + \frac{1}{2}q^9 + \frac{3}{2}iq^{11} - q^{14} - 2iq^{16} - \frac{3}{2}q^{19} 
 + \cdots \right) \nonumber \\
& + a_6 \left( q^6 - \frac{1}{2}iq^9 - \frac{1}{2}q^{11} + iq^{14} - q^{16} - \frac{1}{2}iq^{19} 
  + \cdots \right)   \\
& + a_7 \left( q^7 -\frac{(1+i)}{2} q^9 + \frac{(1-i)}{2}q^{11} - 2q^{12} + iq^{13} + (i - 1)q^{16} + \frac{(i+1)}{2}q^{19} 
 + \cdots \right) \nonumber \\
& + a_8 \left(q^8 + \frac{(i-1)}{2}q^9 -\frac{(1+i)}{2} q^{11} -iq^{12} - q^{13} + (i + 1)q^{16} + iq^{17} - q^{18} + \frac{(1-i)}{2}q^{19}  + \cdots \right) \nonumber 
\end{align}
with the coefficients given by 
\begin{align}
 \left\{ a_2, a_3, a_4, a_6, a_7, a_8 \right\} \! = \!\left[ \begin{array}{cccccc}
   (1+i) & 0 & 2i & 0 & 0 & (-2+2i) \\
   \pm (1-i) & \pm (1+i) & -2i & 2 & \mp3(1-i) & \mp (2+2i) \\
   \pm \!\frac{\sqrt{3i}+\sqrt{-5i}}{2} & \mp \sqrt{5i} &
        \frac{-i + \sqrt{15}}{2} & \frac{-5-i\sqrt{15}}{2} & 0 & 
       \pm \!\! \sqrt{\frac{-11i+3\sqrt{5}}{2}} \\
   \pm \!\frac{\sqrt{3i}-\sqrt{-5i}}{2} & \pm \sqrt{5i} &
        \frac{-i - \sqrt{15}}{2} & \frac{-5+i\sqrt{15}}{2} & 0 &
       \pm \!\!\sqrt{\frac{-11i-3\sqrt{5}}{2}} 
 \end{array} \right]. 
    \label{eq:nf-coeff-100-11}
\end{align}

At level $M_f = 80$, there are three newforms, 
\begin{align}
& q + q^9 + (i - 3)q^{13} + (-3i - 1)q^{17} + (i - 3)q^{21} + 5iq^{25} - 4iq^{29} 
  + \cdots  \nonumber \\
& + a_3 \left( q^3 + iq^7 + (-i - 1)q^{11} + (-2i - 1)q^{15} + (2i - 2)q^{19} + q^{23} 
   \right) 
  \nonumber \\
& + a_5 \left( q^5 + (-i - 1)q^9 + 2iq^{13} - 2q^{17} + (-i + 1)q^{21} + (-i + 1)q^{25} 
   \right),
\end{align}
with the coefficients given by 
\begin{align}
 \{ a_3,a_5\} = \left[ \begin{array}{cc} 0 & (2+i) \\
   \pm \sqrt{(6i)} & -(1+2i) \end{array} \right].
  \label{eq:nf-coeff-80-11}
\end{align}

There is no newform at level $M_f = 40$, and there is just one newform 
at level $M_f = 20$, which is 
\begin{align}
q + (-i - 1)q^2 + 2iq^4 + (i - 2)q^5 + (-2i + 2)q^8 - 3iq^9 + (i + 3)q^{10}
   + (i - 1)q^{13} - 4q^{16} + \cdots .
 \label{eq:nf-Mf20-chi20[11]}
\end{align}

We will see that the ``$(g,n(=(1,2)$ chiral correlation functions'' of 
rational models of ${\cal N}=(2,2)$ SCFT end up having nebentypus 
$\chi_{20}(1,1)$ or $\chi_{20}(1,3)$, but not $\chi_{20}(0,0)$ or $\chi_{20}(0,2)$.
Although we have had detailed look at the newforms of the 
nebentypus $\chi_{20}(0,0)$ and those of the nebentypus $\chi_{20}(0,2)$, 
we are not including them here.
In light of Rmk. \ref{statmnt:CM-form}, we see that there is no CM modular 
form in $S_2(\Gamma_0(400), \chi)$ with $\chi = \chi_{20}(0,0)$ 
or $\chi_{20}(0,2)$.

\subsubsection{The Models for $([E_{z_a}]_\C, f_\rho)$ with $f_\rho = 1$, 
and $[E_{z_a}]_\C \in {\cal E}ll({\cal O}_K)$, $K = \Q(\sqrt{-5})$} 
\label{sssec:exmpl-Dk20-fz1-fr1}

For an imaginary quadratic field $K=\Q(\sqrt{-5})$, the ideal class 
group is ${\rm Cl}_K=\Z/(2)$, and hence ${\cal E}ll({\cal O}_K) = \Z/(2)$. 
There is a pair of $\C$-isomorphism classes of elliptic curves with 
complex multiplication by ${\cal O}_K$. 
One is $\C/(\Z+w_K\Z) = [E_{w_K}]_\C$ and the other 
$\C/(\Z+w'_K\Z) = [E_{w'_K}]_\C$, where 
$w_K = \sqrt{-5}$ and $w'_K = (1+w_K)/2$ in this 
section \ref{sssec:exmpl-Dk20-fz1-fr1}; 
they are labeled by $a \in \{ 0,1 \} \cong \Z/(2)$, respectively.
The two lattices $\mathfrak{b}_{z_0} := (\Z + w_K\Z) = {\cal O}_K$
and $\mathfrak{b}_{z_1} := (\Z + w'_K \Z) =: 2^{-1}\mathfrak{p}_2$
are both fractional ${\cal O}_K$-ideals. Their inverse, ${\cal O}_K$
and $\mathfrak{p}_2 := (1+w_K, 1-w_K)_{{\cal O}_K}$ respectively, 
represent the two ideal classes of ${\rm Cl}_K = \Z/(2)$.

Think of the pair of rational models of ${\cal N}=(2,2)$ SCFT corresponding 
to the pair of data $([z],f_\rho)=([z_a], 1)$ with $z_0 = w_K$ and $z_1 = w'_K$.
For both models, $N_{D\Lambda}=D_zf_\rho = 20 \cdot 1=20$. 
So, the lift of $(g,n)=(1,2)$ ``chiral correlation functions'' should be 
found within $S_2(\Gamma(20))$. Let us see how they fit into the decomposition 
of $S_2(\Gamma(20))$ in terms of the nebentypuses and level $M_f$. 
The CM group is common to both models, and is 
\begin{align}
  [{\cal O}_K/\underline{\mathfrak{c}}]^\times \cong \Z/(2) \times \Z/(4)
\end{align}
where the first and the second factor are generated by the CM operation 
$[(4+w_K)\times]$ and $[(3)\times]$, respectively. Their characters $\chi_f$ 
are labeled by $\chi_f(a,b)$ with $a \in \Z/(2)$ and $b \in \Z/(4)$, if 
$\chi_f: [(4+w_K)\times] \mapsto (-1)^a$ and $\chi_f: [(3)\times] \mapsto i^b$.
We will be interested in the characters with ${}^\forall a \in \Z/(2)$ and 
$b \in \{ 1,3\} \subset \Z/(4)$, so that the restriction of the representation 
$\chi_f^{-1}$ of the CM group is restricted to ${\rm Aut}([E_z]_\C) \cong 
\{ (\pm 1) \} \hookrightarrow [{\cal O}_K/\underline{\mathfrak{c}}]^\times$
to become $(\rho^1)^{-1}$. The conductor of those characters are 
$\mathfrak{c}_f = \underline{\mathfrak{c}}=(2w_K)_{{\cal O}_K}$ 
for $\chi_f(0,b'')$ with ${}^\forall b'' \in \{ 1,3\}$, and 
$\mathfrak{c}_f = (w_K)_{{\cal O}_K}$ for $\chi_f(1,b'')$ with 
${}^\forall b'' \in \{1,3\}$. 
Note also that $\chi_f(a,b'') \cdot cc \neq cc \cdot \chi_f(a,b'')$ for 
any choice of $a \in \Z/(2)$ and $b'' \in \{ 1,3 \} \subset \Z/(4)$.

The character $\epsilon_{DL}^{-1} = (D_K/-)$ on $[\Z/(20)]^\times$ is $
\chi_{20}(1,2)$. 
The character $\underline{\chi}_f$ of $[\Z/(20)]^\times$ is determined 
by using the fact that $i: [\Z/(N_{D\Lambda})]^\times \rightarrow 
[{\cal O}_K/\underline{\mathfrak{c}}]^\times$ is given by 
$[11] \mapsto [(1)\times]$ and $[17] \mapsto [(3)\times]^3$.
For $\chi_f = \chi_f(a,b'')$ with $a \in \Z/(2)$ and $b'' \in \{1,3\} 
\subset /\Z/(4)$, $\underline{\chi}_f = \chi_{20}(0,-b'')$.
So, the nebentypus $\chi_{N_{D\Lambda}}[\chi_f]$ is 
$\chi_{20}(1,2) \cdot \chi_{20}(0,b'') = \chi_{20}(1,-b'')$.
\begin{align}
   F({\cal E}ll({\cal O}_K),1)^{\chi_f(\underline{a},b'')^{-1}} 
 \hookrightarrow 
    \left[ S_2(\Gamma_0^0(20), \chi_{20}(1,b'')) \right]_{M_f(\chi_f(\underline{a},b''))}
  \label{eq:exmpl-Dk20-fz1-fr1-CCF-cusp-injmap}
\end{align}
for $\underline{a} \in \Z/(2)$ and $b'' \in \{ 1,3\} \subset \Z/(4)$. 
Remember that $M_f(\chi_f(\underline{a},b'')) = |D_K|{\rm Nm}_{K/\Q}(\mathfrak{c}_f)$.

In the following, we write down the linear combinations of the 
``$(g,n)=(1,2)$ chiral correlation functions'' explicitly, and confirm 
that they reproduce the newforms and oldforms in the way described 
in section \ref{ssec:HTheta-TypeII-CCF}. As a preparation, let us have 
a look at how 
$iReps$ of the two models decomposes under the action of the CM group. 
In the model for $([z_{a=0}],f_\rho)=([w_K], 1)$, the set $iReps_{a=0}$ is 
$G_{\Lambda_{a=0}} = {\cal O}_K/(2\sqrt{5}i)_{{\cal O}_K}$, and is decomposed into 
\begin{align}
  {\rm orb}^{a=0}_{\mathfrak{q}={\cal O}_K} & \; = 
      \{ 1,3,9,7,4+w_K, 2+w_K, 6+w_K, 8+w_K\} \subset iReps_{a=0}, \\
  {\rm orb}^{a=0}_{\mathfrak{q}=\mathfrak{p}_2} & \; = 
      \{ 1+w_K, 3+w_K, 9+w_K, 7+w_K \} , \\
  {\rm orb}^{a=0}_{\mathfrak{q}=(2)_{{\cal O}_K}} & \; = 
     \{ 2, 6, 8, 4 \} ,
\end{align}
and the orbits $G_{\Lambda_{a=0}}[2] = 
(\sqrt{5}i)_{{\cal O}_K}/(2\sqrt{5}i)_{{\cal O}_K}$. In the other model, 
where $([z_{a=1}], f_\rho) = ([w'_K], 1)$, the set  
$iReps_{a=1} = G_{\Lambda_{a=1}} = 2^{-1}\mathfrak{p}_2/\sqrt{5}i \mathfrak{p}_2$ 
is decomposed into 
\begin{align}
  {\rm orb}^{a=1}_{{\cal O}_K} & \; = \left\{ 
     \frac{7+w_K}{2}, \frac{1-w_K}{2}, \frac{3+w_K}{2}, \frac{9-w_K}{2}, \right.   \nonumber \\
  & \; \qquad \qquad \left. \frac{3-w_K}{2}, \frac{9+w_K}{2}, 
         \frac{7-w_K}{2}, \frac{1+w_K}{2} \right\} \subset iReps_{a=1}, \\
 {\rm orb}^{a=1}_{\mathfrak{q}=\mathfrak{p}_2} & \; = \{ 1,3,9,7 \}, \\
 {\rm orb}^{a=1}_{\mathfrak{q}=(2)_{{\cal O}_K}} & \; = \{ 2,6,8,4 \},
\end{align}
and the orbits $G_{\Lambda_{a=1}}[2] = 2^{-1}\sqrt{5}i\mathfrak{p}_2 /
\sqrt{5}i\mathfrak{p}_2$.

Let us first work out (\ref{eq:exmpl-Dk20-fz1-fr1-CCF-cusp-injmap}) 
explicitly for $(a,b'') = (0,3)$; the conductor is 
$\mathfrak{c}_f = (2w_K)_{{\cal O}_K} = \underline{\mathfrak{c}}$,
so the level $M_f$ is expected to be $20 \cdot {\rm Nm}_{K/\Q}((2w_K)_{{\cal O}_K}) 
= 20 \cdot 20 = 400$; we also expect that ``chiral correlation functions'' 
of the orbits ${\rm orb}^a_{\mathfrak{q}_0={\cal O}_K}$ should be enough 
in reproducing the Hecke theta functions associated with $\chi_f=\chi_f(0,3)$, 
because $\underline{\mathfrak{c}}/\mathfrak{c}_f = {\cal O}_K$ 
in (\ref{eq:necs-CM-orbits}). 
Indeed, straightforward calculation 
using (\ref{eq:CCF-as-theta-in-Thm1}) is enough to see that 
\begin{align}
 i \frac{\sqrt{|G_{\Lambda}|}}{\sqrt{2f_\rho}} 
    f_{1\Omega'}^{\rm II}([1]; \chi_f(0,3))_{([z_0],1)} & \; = 
  \left(q+3iq^9+2q^{21} + 6i q^{29} + \cdots \right) \nonumber \\
 & \qquad   + \left(4iq^9+8q^{21} -12q^{41} + \cdots \right).  \\
 i \frac{\sqrt{|G_{\Lambda}|}}{\sqrt{2f_\rho}} 
    f_{1\Omega'}^{\rm II}([(7+w_K)/2]; \chi_f(0,3))_{([z_1],1)} & \; = 
  \sqrt{10}(q^3 - iq^7 -3q^{23}+iq^{27} + \cdots). 
\end{align}
The two linear combinations of those ``chiral correlation functions'' 
from the two rational models of ${\cal N}=(2,2)$ SCFT, 
\begin{align}
  i \frac{\sqrt{|G_{\Lambda}|}}{\sqrt{2f_\rho}} \left( 
    f_{1\Omega'}^{\rm II}([1]; \chi_f(0,3))_{([z_0],1)} 
  \pm \sqrt{i} f_{1\Omega'}^{\rm II}([(7+w_K)/2]; \chi_f(0,3))_{([z_1],1)} \right), 
\end{align}
agree with two---those in the second row of (\ref{eq:nf-coeff-400-11})---of 
the nine newforms of the group $\Gamma_0^0(20)$, weight 2, level $M_f=400$, 
and nebentypus $\chi_{20}(1,1)$, as expected.  

Next, we work out (\ref{eq:exmpl-Dk20-fz1-fr1-CCF-cusp-injmap}) for the 
other character $\chi_f$---$\chi_f=\chi_f(1,3)$---of the CM group 
$[{\cal O}_K/\underline{\mathfrak{c}}]^\times$, which we expect to contribute 
to the cuspforms of nebentypus $\chi_{20}(1,1)$. The conductor is 
$\mathfrak{c}_f = (w_K)_{{\cal O}_K}$, and 
$\underline{\mathfrak{c}}/\mathfrak{c}_f = (2)_{{\cal O}_K}$.
So we expect that the level is $M_f = 20 \cdot {\rm Nm}_{K/\Q}((w_K)_{{\cal O}_K}) 
= 20 \cdot 5=100$. The CM group orbits of the form ${\rm orb}^a_{\mathfrak{q}_0}$
with $\mathfrak{q}_0 = {\cal O}_K$, $\mathfrak{p}_2$, and $(2)_{{\cal O}_K}$, 
of both of the models ($a=0,1$) are expected to contribute to the 
Hecke theta functions and their oldforms of this level $M_f = 100$. 
To verify this expectation, we do the straightforward computation 
using (\ref{eq:CCF-as-theta-in-Thm1}) once again.
\begin{align}
 i \frac{\sqrt{|G_\Lambda|}}{\sqrt{2f_\rho}} 
     f_{1\Omega'}^{\rm II}([1]; \chi_f(1,3))_{([z_0],1)} & \; = 
  \left(q+3iq^9+2q^{21} + 6i q^{29} + \cdots \right)  \nonumber \\ 
  &\; \qquad \qquad    - \left(4iq^9+8q^{21} -12q^{41} + \cdots \right). \\
 i \frac{\sqrt{|G_\Lambda|}}{\sqrt{2f_\rho}} 
     f_{1\Omega'}^{\rm II}([1+w_K]; \chi_f(1,3))_{([z_0],1)} & \; = 
    2q^6 + 6iq^{14} + 2q^{46} -8i q^{54} + \cdots, \\
 i \frac{\sqrt{|G_\Lambda|}}{\sqrt{2f_\rho}} 
     f_{1\Omega'}^{\rm II}([2]; \chi_f(1,3))_{([z_0],1)} & \; = 
    2q^4 -4i q^{16}+ 4q^{24} -2i q^{36} -8q^{64} + \cdots, 
\end{align}
are ``chiral correlation functions'' of the rational model of 
${\cal N}=(2,2)$ SCFT for $([E_z]_\C, f_\rho) = (E_{w_K},1)$, 
and 
\begin{align}
   i \frac{\sqrt{|G_\Lambda|}}{\sqrt{2f_\rho}} 
     f_{1\Omega'}^{\rm II}([(7+w_K)/2]; \chi_f(1,3))_{([z_1],1)} & \; = 
    \sqrt{2} \left( iq^3 - 3q^7 +iq^{23} + 7q^{27} + \cdots \right),   \\
   i \frac{\sqrt{|G_\Lambda|}}{\sqrt{2f_\rho}} 
     f_{1\Omega'}^{\rm II}([1]; \chi_f(1,3))_{([z_1],1)} & \; =
     \sqrt{2} \left( q^2 -iq^{18} -8q^{42} + \cdots \right),    \\
   i \frac{\sqrt{|G_\Lambda|}}{\sqrt{2f_\rho}} 
     f_{1\Omega'}^{\rm II}([2]; \chi_f(1,3))_{([z_1],1)} & \; = 
     \sqrt{2} \left( 2q^8 + 2iq^{12}-6q^{28} + 6i q^{72} + \cdots \right), 
\end{align}
are those for $([E_z]_\C, f_\rho)=([E_{w'_K}]_\C, 1)$.

Two out of the seven newforms of the group $\Gamma_0^0(20)$, weight 2, 
level $M_f=100$, and nebentypus $\chi_{20}(1,1)$ are reproduced by the 
following two linear combinations of the six ``chiral correlation functions''  
above (over ${\rm orb}^a_{\mathfrak{q}_0}$'s with $\mathfrak{q}_0$ 
that divides $\mathfrak{q}_{p.\mathfrak{c}_f} = \mathfrak{p}_2^2$):
\begin{align}
  & f_{1\Omega'}^{\rm II}([1];\chi_f)_{([z_0],1)}
  + f_{1\Omega'}^{\rm II}([1+w_K];\chi_f)_{([z_0],1)}
  -i f_{1\Omega'}^{\rm II}([2];\chi_f)_{([z_0],1)}   \nonumber \\
  & \pm \sqrt{(-i)} \left( 
     f_{1\Omega'}^{\rm II}\left([(7+w_K)/2];\chi_f\right)_{([z_1],1)}
   + f_{1\Omega'}^{\rm II}([1];\chi_f)_{([z_1],1)}
   -i f_{1\Omega'}^{\rm II}([2];\chi_f)_{([z_1],1)} \right) \nonumber 
\end{align}
[all $\chi_f$'s here are $\chi_f(1,3)$, just to save space], multiplied
by $i \sqrt{|G_\Lambda|/(2f_\rho)}$, is equal to 
\begin{align}
&  (q -2iq^4 + 2q^6 -iq^9+6iq^{14}-4q^{16} + \cdots ) \nonumber \\
& \qquad \quad   \pm (1-i)(q^2 + iq^3 -3q^7-2iq^8 -2q^{12} -iq^{18} + \cdots), 
\end{align}
which correspond to the second row of (\ref{eq:nf-coeff-100-11}) indeed. 

For each one of these two Hecke newforms of level $M_f=100$, there are 
two oldforms in $[S_2(\Gamma_0^0(20), \chi_{20}(1,1))]_{M_f=100}$; they are 
obtained from the newforms by $|[\diag(r,1)]_2$ with $r=2, 4$. We expect 
that all those oldforms are also obtained within 
$\oplus_{[E_z] \in {\cal E}ll({\cal O}_K)} F([z_a],1)^{\chi_f(1,3)^{-1}}$; this is 
because $\underline{\mathfrak{c}}/\mathfrak{c}_f = (2)_{{\cal O}_K} 
= \mathfrak{p}_2^2$, and hence 
we can choose $\mathfrak{q}_r = \mathfrak{p}_2$ for $r=2$, and 
$\mathfrak{q}_r=(2)_{{\cal O}_K}$ for $r=4$.  Indeed, the linear combinations 
(over ${\rm orb}^a_{\mathfrak{q}_0}$'s with $\mathfrak{q}_0=\mathfrak{q}_{0.p.\mathfrak{c}_f} = \mathfrak{p}_2$ and $\mathfrak{p}_2^2$)
\begin{align}
    f_{1\Omega'}^{\rm II}([1];\chi_f)_{([z_1],1)}
 -i f_{1\Omega'}^{\rm II}([2];\chi_f)_{([z_1],1)}
 \pm \sqrt{i} (f_{1\Omega'}^{\rm II}([1+w_K];\chi_f)_{([z_0],1)} 
      -i f_{1\Omega'}^{\rm II}([2];\chi_f)_{([z_0],1)}  )  \nonumber 
\end{align}
multiplied by $i \sqrt{|G_\Lambda|/(2f_\rho)}$ 
is equal to 
\begin{align}
&  \sqrt{2} (q^2 -2i q^8-2q^{12}-iq^{18}+6iq^{28} + \cdots)  \nonumber \\
&   \qquad \quad \pm \sqrt{2}(1-i) (q^4 + iq^6 -3q^{14}  -2iq^{16} + \cdots),
\end{align}
and the linear combination (over ${\rm orb}^a_{\mathfrak{q}_0}$'s with 
$\mathfrak{q}_0=\mathfrak{q}_{0.p.\mathfrak{c}_f} = \mathfrak{p}_2^2$)
\begin{align}
 i\frac{\sqrt{|G_\Lambda|}}{\sqrt{2f_\rho}} \left(
       f_{1\Omega'}^{\rm II}([2];\chi_f)_{([z_0],1)} 
  \pm \sqrt{(-i)} f_{1\Omega'}^{\rm II}([2];\chi_f)_{([z_1],1)} \right)
\end{align}
is equal to 
\begin{align}
   2(q^4 -2i q^{16} + 2q^{12}-iq^{18} + \cdots ) \pm 2(1-i)(q^8+iq^{12}-3q^{28}
    +\cdots )
\end{align}
indeed. The ``chiral correlation functions'' are summed in order to reproduce 
the oldforms only over the CM group orbits of the irreducible representations 
of the chiral algebras specified by (\ref{eq:necs-CM-orbits}).

To summarize, 
\begin{exmpl}
The ``$(g,n)=(1,2)$ chiral correlation functions'' (\ref{eq:chi-corrl-fcn-4-L1}, \ref{eq:f1Omega-as-12confBl}) 
of the rational model of ${\cal N}=(2,2)$ SCFT for two sets of 
data, $([E_{z_a}]_\C, f_\rho=1)$ with $[E_{z_a}]_\C \in {\cal E}ll({\cal O}_K)$, 
become single-valued when lifted to the modular curve $X(20)$, the 
compactification of ${\cal H}/\Gamma(20)$. They form a vector space 
$\oplus_{[E_z] \in {\cal E}ll({\cal O}_K)} F([z_a],1)$ over $\C$ of dimension 16, 
and occupies an 8-dimensional subspace of the vector space of cuspforms 
$S_2(\Gamma_0^0(20), \chi_{20}(1,1))$ and an 8-dimensional subspace of 
$S_2(\Gamma_0^0(20),\chi_{20}(1,3))$. 

The subspace of $S_2(\Gamma_0^0(20), \chi_{20}(1,1))$ obtained in this way 
is generated by two newforms of level $M_f=400$ (among the nine newforms 
of level $M_f=400$), two newforms of level $M_f=100$ (among the seven newforms 
of level $M_f=100$), and all possible oldforms of the two level $M_f=100$ 
newforms. Schematically, we might describe this situation in this way:
\begin{align}
 2_{400} + 2_{100}^{(1,2,4)} \subset 
 9_{400} + 7_{100}^{(1,2,4)} + 3_{80}^{(1,5)} + 1_{20}^{(1,2,4,5,10,20)}. 
    \qquad \qquad \bullet
\end{align}
\end{exmpl}

\subsubsection{The Model for $([E_z], f_\rho)$ with $f_\rho = 5$ and 
$[z] = i$}
\label{sssec:exmpl-Dk4-fz1-fr5}

For an imaginary quadratic field $K=\Q(\sqrt{-1})$, the ideal class 
group is trivial, and there is just one $\C$-isomorphism class 
of elliptic curves with complex multiplication by ${\cal O}_K$; $z = i$.
Now, we consider the rational model of ${\cal N}=(2,2)$ SCFT for 
the data $([z],f_\rho)=([i],5)$. Then $N_{D\Lambda}=D_zf_\rho=4 \cdot 5 = 20$, 
so the lift of the ``$(g,n)=(1,2)$ chiral correlation functions'' 
of this model should also contribute to $S_2(\Gamma(20))$.
$\mathfrak{b}_z = {\cal O}_K$ and 
$\mathfrak{b}_z\underline{\mathfrak{c}} = (10)_{{\cal O}_K}$ in this example.

In this model, the CM group is 
\begin{align}
  [{\cal O}_K/\underline{\mathfrak{c}}]^\times \cong \Z/(2) \times \Z/(4) 
     \times \Z/(4),
\end{align}
where we choose the generators of those three factors to be the CM operation 
$[(3i)\times]$, $[(1+4i) \times]$, and $[(i)\times]$. The characters of 
this group are denoted by $\chi_f(a,b,c)$ with $a \in \Z/(2)$, $b \in \Z/(4)$, 
and $c \in \Z/(4)$, when $\chi_f: [(3i)\times] \mapsto (-1)^a$, 
$\chi_f:[(1+4i)\times] \mapsto i^b$, and $\chi_f: [(i)\times] \mapsto i^c$.
Due to the condition that the restriction of the representation $(\chi_f)^{-1}$ 
to the subgroup ${\rm Aut}([E_z]_\C) \cong \{ \pm 1, \pm i\} 
\hookrightarrow [{\cal O}_K /\underline{\mathfrak{c}}]^\times$ should be 
$(\rho^1)^{-1}$, only the characters $\chi_f(a,b,c)$ with $c = 1 \in \Z/(4)$
are relevant. 
\begin{table}[tbp]
\begin{center}
\begin{tabular}{c|cccc|cccc}
$\chi_f(a,b,1)$ & (0,1) & (0,0) & (0,3) & (0,2) 
    & (1,0) & (1,3) & (1,2) & (1,1) \\
\hline
$\mathfrak{c}_f$ & $(10)$ & $(5)$ & $(2-4i)$ & $(1+2i)$ 
    & (10) & (5) & $(2+4i)$ & $(1-2i)$  \\
$M_f$ & 400 & 100 & 80 & 20 & 400 & 100 & 80 & 20 \\
\hline
$\chi_{20}[\chi_f]$ & $\chi_{20}(1,3)$ & & & & $\chi_{20}(1,1)$ & & &
\end{tabular}
\caption{\label{tab:exmpl-Dk4-fz1-fr5}This table summarizes properties 
of the eight characters of the CM group labeled by $(a,b) \in \Z/(2) \times \Z/(4)$ shown in the first row. Conductors $\mathfrak{c}_f$, the level $M_f$, and 
nebentypus $\chi_{20}[\chi_f]$ of the corresponding Hecke theta function 
are shown for each column. }
\end{center}
\end{table}
The conductor of those eight character $\chi_f(a,b,1)$, $a \in \Z/(2)$, 
$b \in \Z/(4)$ are summarized in Table~\ref{tab:exmpl-Dk4-fz1-fr5}.
Their corresponding nebentypus $\chi_{N_{D\Lambda}}[\chi_f]$ is also 
computed\footnote{
First, $\epsilon_{DL}^{-1} = (D_K/-) = (-4/-) = \chi_{20}(1,0)$.
The other character $\underline{\chi}_f$ of $[\Z/(20)]^\times$ 
is $\chi_{20}(0,2a+1)$ for $\chi_f(a,b,1)$, because $i: [\Z/(20)]^\times 
\rightarrow [{\cal O}_K/(10)]^\times$ is given by $i: [11] \mapsto 
[(1)\times]$ and $i:[17] \mapsto [(3i)\times]^3 [(i)\times]$. So, 
$\chi_{N_{D\Lambda}}[\chi_f] = \chi_{20}(1,0) \cdot \chi_{20}(0,1+2a)^{-1} 
= \chi_{20}(1,2a-1)$.
}
by using Prop. \ref{props:nebentypus-formula-CCF}, and are also written 
in the Table.

The CM group $[{\cal O}_K/\underline{\mathfrak{c}}]^\times$ acts 
on the set of irreducible representations of the chiral algebra 
$G_\Lambda$ and decomposes into the orbits. The 96 elements of $G_\Lambda$ 
are grouped into 
\begin{align}
  {\rm orb}_{{\cal O}_K} = \{ 1, \cdots \}_{32}, & \quad 
  {\rm orb}_{(1\pm 2i)} = \{ (1\pm 2i), (-2\pm i), \cdots, 
      (4\pm 3i), (-3\pm 4i),\cdots \}_8, \nonumber \\
  {\rm orb}_{(1+i)} = \{ 1+i, \cdots, \}_{16}, & \quad 
  {\rm orb}_{(1\mp 3i)} = \{ (1\mp 3i), (3\pm i), (-3\mp i), (-1\pm 3i)\}_{4}
      \nonumber  \\
  {\rm orb}_{(2)} = \{ 2, \cdots, \}_{16}, & \quad 
  {\rm orb}_{(2\pm 4i)} = \{ (2\pm 4i), (4\mp 2i), (8\mp 4i), (6\pm 2i) \}_4, 
    \nonumber
\end{align}
and $G_\Lambda[2] = (5)/(10)$; the number of $\beta$'s in the orbits are 
indicated by the subscripts $\{ \cdots \}_n$. 

Let us first work out $F([i], 5)^{\chi_f(1,0,1)^{-1}}$. We expect 
that the ``chiral correlation functions'' only of the orbit 
${\rm orb}_{{\cal O}_K}$ are necessary in reproducing one Hecke newform 
of level $M_f=400$ and nebentypus $\chi_{20}(1,1)$. By using 
(\ref{eq:CCF-as-theta-in-Thm1}), we find that 
\begin{align}
& i \frac{\sqrt{|G_\Lambda|}}{\sqrt{2f_\rho}} f_{1\Omega'}^{\rm II}([1]; \chi_f(1,0,1))
  = q - 3i q^9 + (5-5i)q^{13} +(5+5i)q^{17} -4iq^{29} + \cdots, \\
&  - (5 + 5i) q^{37} + 8 q^{41} + 7i q^{49} 
  - (5 - 5i) q^{53} - 12 q^{61} 
  - (5 - 5i)q^{73} - 9q^{81} + 16i q^{89} + \cdots \nonumber 
\end{align}
which agrees with one of the nine newforms of weight 2, level $M_f=400$, 
and nebentypus $\chi_{20}(1,1)$; the one corresponding to the first row of 
(\ref{eq:nf-coeff-400-11}).

Next, we confirm that $F([i],5)^{\chi_f(1,3,1)^{-1}}$ is indeed 
contained in $[S_2(\Gamma_0^0(20),\chi_{20}(1,1))]_{M_f=100}$. 
As a preparation, we work out $f_{1\Omega'}^{\rm II}([\beta]; \chi_f(1,3,1))$ 
of this model for the orbits ${\rm orb}_{\mathfrak{q}_0}$ with 
$\mathfrak{q}_0 = {\cal O}_K$, $(1+i)_{{\cal O}_K}$, and $(2)_{{\cal O}_K}$; 
this is because $\mathfrak{c}_f = (5)_{{\cal O}_K}$, and hence 
$\underline{\mathfrak{c}}/\mathfrak{c}_f = \mathfrak{q}_{p.\mathfrak{c}_f} 
= (2)_{{\cal O}_K}$. Brute force computation results in 
\begin{align}
 i \frac{\sqrt{|G_\Lambda|}}{\sqrt{2f_\rho}}f_{1\Omega'}^{\rm II}([1];\chi_f)
 & \; = q - 3iq^9 + (1-i)q^{13} - (3+3i)q^{17} + 4i q^{29} + \cdots, 
  \\
  i \frac{\sqrt{|G_\Lambda|}}{\sqrt{2f_\rho}}f_{1\Omega'}^{\rm II}([1+i];\chi_f)
   & \; = (1+i)q^2+(3-3i)q^{18} +2q^{26} + \cdots, \\
  i \frac{\sqrt{|G_\Lambda|}}{\sqrt{2f_\rho}}f_{1\Omega'}^{\rm II}([2];\chi_f)
   & \; = 2q^4 +(2+2i)q^8 + 4i q^{16} + \cdots.
\end{align}
Their linear combination 
\begin{align}
&   i\frac{\sqrt{|G_\Lambda|}}{\sqrt{2f_\rho}} \left( 
    f_{1\Omega'}^{\rm II}([1]; \chi_f)
 + f_{1\Omega'}^{\rm II}([1+i]; \chi_f)
 +i f_{1\Omega'}^{\rm II}([2]; \chi_f) \right) \nonumber \\
  = & \; q + (1+i)q^2 +2iq^4 +2(-1+i)q^8-3iq^9+(1-i)q^{13} -4q^{16} -3(1+i)q^{17} 
   + \cdots 
\end{align}
agrees with one of the seven newforms of level $M_f=100$ and 
nebentypus $\chi_{20}(1,1)$, the one in the first row 
of (\ref{eq:nf-coeff-100-11}). Two other linear combinations 
\begin{align}
   i\frac{\sqrt{|G_\Lambda|}}{\sqrt{2f_\rho}} & \left( 
    f_{1\Omega'}^{\rm II}([1+i]; \chi_f)
 +i f_{1\Omega'}^{\rm II}([2]; \chi_f) \right) \nonumber \\
& \; = (1+i) \left( q^2 + (1+i)q^4 + 2i q^8 +2(-1+i)q^{16} + \cdots \right), \\
  i\frac{\sqrt{|G_\Lambda|}}{\sqrt{2f_\rho}} & f_{1\Omega'}^{\rm II}([2]; \chi_f)
 = 2(q^4+(1+i)q^8+2iq^{16} + \cdots )
\end{align}
are its oldforms obtained by $|[\diag(2,1)]_2$ and $|[\diag(4,1)]_2$, 
respectively. The CM group orbits contributing to the two oldforms 
above follow the pattern of (\ref{eq:necs-CM-orbits}) 
with $\underline{\mathfrak{c}}/\mathfrak{c}_f = \mathfrak{q}_{p.\mathfrak{c}_f}
=(2)$, and $\mathfrak{q}_r= \mathfrak{q}_{r.p.\mathfrak{c}_f} = (1+i)$ for 
$r=2$ and $=(2)$ for $r=4$, respectively. 

The vector space $F([i],5)^{\chi_f(1,2,1)^{-1}}$ is contained in 
$[S_2(\Gamma_0^0(20), \chi_{20}(1,1))]_{M_f=80}$; to see this, 
we use 
\begin{align}
i\frac{\sqrt{|G_\Lambda|}}{\sqrt{2f_\rho}} f_{1\Omega'}^{\rm II}([1]; \chi_f) 
   & \; = q - 3i q^9 - (5 - 5i) q^{13} - (5 + 5i) q^{17} - 4i q^{29}
    + (5 + 5i) q^{37} + 8 q^{41} + , \nonumber \\
i \frac{\sqrt{|G_\Lambda|}}{\sqrt{2f_\rho}} f_{1\Omega'}^{\rm II}([2+i]; \chi_f)
  & \; = (2+i)q^5 + (3+4i)q^{25} + (3-6i)q^{45} + \cdots, 
\end{align}
the ``chiral correlation functions'' summed over $\beta$'s in 
${\rm orb}_{{\cal O}_K}$ and ${\rm orb}_{(1-2i)}$, respectively; 
because $\mathfrak{c}_f=(2+4i)$ implies $\underline{\mathfrak{c}}/\mathfrak{c}_f = \mathfrak{q}_{p.\mathfrak{c}_f} = (1-2i)$, they 
are all that we need. Their linear combinations 
\begin{align}
 i \frac{\sqrt{|G_\Lambda|}}{\sqrt{2f_\rho}} & \left( 
      f_{1\Omega'}^{\rm II}([1]; \chi_f) + f_{1\Omega'}^{\rm II}([2+i]; \chi_f) \right)
    \nonumber \\
  & \; = 
       q + (2+i)q^5 - 3iq^9 -5(1-i)q^{13}-5(1+i)q^{17} + (3+4i)q^{25} + \cdots, \\
 i \frac{\sqrt{|G_\Lambda|}}{\sqrt{2f_\rho}} & 
      f_{1\Omega'}^{\rm II}([2+i]; \chi_f) = 
    (2+i)\left( q^5 + (2+i)q^{25}-3i q^{45} + \cdots \right)
\end{align}
are one of two newforms at the level $M_f=80$ and nebentypus 
$\chi_{20}(1,1)$---the one corresponding to the first row 
of (\ref{eq:nf-coeff-80-11})---and its oldform obtained by 
$|[\diag(r,1)]_2$ with $r=5$; $\mathfrak{q}_r=\mathfrak{q}_{r.p.\mathfrak{c}_f} 
= (1-2i)$ can be used.

There is one more character of the CM group, $\chi_f(1,1,1)$, whose 
$F([i],5)^{\chi_f^{-1}}$ contribute to the cuspforms of nebentypus 
$\chi_{20}(1,1)$. The conductor $\mathfrak{c}_f=(1-2i)$ implies 
$\underline{\mathfrak{c}}/\mathfrak{c}_f = \mathfrak{q}_{p.\mathfrak{c}_f} = 
(1+i)^2_{{\cal O}_K}(1+2i)_{{\cal O}_K}$, so there are 
six orbits ${\rm orb}_{\mathfrak{q}_0}$ that contribute to 
$[S_2(\Gamma_0^0(20), \chi_{20}(1,1))]_{M_f=20}$; they are for 
$\mathfrak{q}_0 = (1+i)_{{\cal O}_K}^p (1+2i)_{{\cal O}_K}^q$ with 
$0 \leq p \leq 2$ and $0 \leq q \leq 1$. Because there is just one 
Hecke newform of level $M_f=20$ and nebentypus $\chi_{20}(1,1)$, 
as we have seen in Thm. \ref{thm:newform-contrib-orbits}, a sum over 
all the six orbits 
will yield the newform (\ref{eq:nf-Mf20-chi20[11]}). We expect that 
all the possible oldforms of this newform are also realized by 
the ``chiral correlation functions'' of this model; we can choose 
$\mathfrak{q}_r = (1+i)^p (1+2i)^q$ with $0 \leq p \leq 2$ and 
$0 \leq q \leq 1$. Because the choice of $\mathfrak{q}_r$ is 
unique for any one of $r \in \{ 2,4,5,10,20\}$, the homomorphism 
(\ref{eq:maybe-inj-orbts-2-CCF}) is expected to be injective in 
this example. 
We have not verified those expectations by explicit Fourier series 
expansion for this case $\chi_f=\chi_f(1,1,1)$, however. 

To summarize, 
\begin{exmpl}
The ``$(g,n)=(1,2)$ chiral correlation functions'' (\ref{eq:chi-corrl-fcn-4-L1}, \ref{eq:f1Omega-as-12confBl}) 
of the rational model of ${\cal N}=(2,2)$ SCFT for the set of data 
$([E_{z}]_\C, f_\rho)$ with $[z] = [i]$ and $f_\rho = 5$ become single-valued 
when lifted to the modular curve $X(20)$. They form a vector space 
$F([i],5)$ over $\C$ of dimension 24, and occupies a 12-dimensional 
subspace of the vector space of cuspforms 
$S_2(\Gamma_0^0(20), \chi_{20}(1,1))$ and a 12-dimensional subspace of 
$S_2(\Gamma_0^0(20),\chi_{20}(1,3))$. 

The subspace of $S_2(\Gamma_0^0(20), \chi_{20}(1,1))$ obtained in this way 
is generated by one newform of level $M_f=400$ (among the nine newforms 
of level $M_f=400$), one newform of level $M_f=100$ (among the seven newforms 
of level $M_f=100$), one newform of level $M_f=80$ (among the two newforms), 
one (unique) newform of level $M_f=20$, and all possible oldforms of those newforms. Schematically, we might describe this situation in this way:
\begin{align}
 1_{400} + 1_{100}^{(1,2,4)} + 1_{80}^{(1,5)} + 1_{20}^{(1,2,4,5,10,20)} \subset 
 9_{400} + 7_{100}^{(1,2,4)} + 3_{80}^{(1,5)} + 1_{20}^{(1,2,4,5,10,20)}.
\end{align}

There is no overlap between the Hecke newforms obtained from the 
models of $[E_z]_\C \in {\cal E}ll({\cal O}_K)$ of $K=\Q(\sqrt{-5})$,
and the newforms obtained from the model of $[E_z]_\C$ with $z=i$.
There are $(6+4+2)$ newforms of level $M_f=400, 100$ and $80$ in the 
vector space of cuspforms $S_2(\Gamma_0^0(20), \chi_{20}(1,1))$ that are 
NOT realized by the ``$(g,n)=(1,2)$ chiral correlation functions'' 
(\ref{eq:chi-corrl-fcn-4-L1}, \ref{eq:f1Omega-as-12confBl}) of 
rational models of ${\cal N}=(2,2)$ SCFT. 
None of the weight-2 cuspforms for $\Gamma_0^0(20)$ of nebentypus 
$\chi_{20}(0,0)$ or $\chi_{20}(0,2)$ is realized by the ``chiral correlation 
functions'' either.
\end{exmpl}

\subsection{Chiral Correlation Functions of Models $([E_z]_\C, f_\rho)$ 
with Different $f_\rho$'s}
\label{ssec:diff-frho}

In section \ref{ssec:HTheta-TypeII-CCF}, we studied how the vector space 
$F^{(\Omega, \int J)}_{\{0,0\}}([z],f_\rho)$ of 
``$(g,n)=(1,2)$ chiral correlation functions'' of the rational 
model of ${\cal N}=(2,2)$ SCFT for a set of data $([E_z]_\C, f_\rho)$
fills a part of the vector space $H^{1,0}(X(N_{D\Lambda});\C)$, or 
a part of the vector space $H^{1,0}(X_1(N_{D\Lambda}^2);\C)$. 
Here, the level $N$ of the principal congruence 
subgroup $\Gamma(N) \subset {\rm SL}(2;\Z)_{ws}$ acting on 
$\tau_{ws} \in {\cal H}$ is chosen to be $N_{D\Lambda}=D_zf_\rho$, 
because $\Gamma(D_zf_\rho)$ is contained 
in the kernel of the monodromy representation of ${\rm SL}(2;\Z)_{ws} \cong 
\Gamma_{1,1}$ of this model of SCFT. It then makes sense also 
to pose a question which subspace of $H^{1,0}(X(N);\C)$ (and also of 
$H^{1,0}(X_1(N^2);\C)$ is filled by $F^{(\Omega, \int J)}_{\{0,0\}}([z],f_\rho)$ for 
any one of $N$'s that are divisible by $N_{D\Lambda}$, because such 
$\Gamma(N)$'s are also contained in the kernel of the monodromy representation. 
 
Take $N=N_{D\Lambda} d$, where $d \in \N_{>1}$. Then there are embeddings 
\begin{align}
  \imath^{(K,f_\rho; d^2,d'')}: 
   & \; F^{(\Omega, \int J)}_{\{0,0\}}({\cal E}ll({\cal O}_K),f_\rho) 
       \hookrightarrow H^{1,0}(X(N_{D\Lambda} d);\C),
      \label{eq:ccf-embed2-H10ofXNs} \\
  \imath^{(K,f_\rho; d^2,d'')}_1: 
   & \; F^{(\Omega, \int J)}_{\{0,0\}}({\cal E}ll({\cal O}_K),f_\rho) 
    \hookrightarrow H^{1,0}(X_1(N_{D\Lambda}^2 d^2);\C),  
      \label{eq:ccf-embed2-H10ofX1N2s}
\end{align}
labeled by integers $d''$ that divide $d^2$. Here, the embedding
$\imath^{(K,f_\rho;d^2,d'')}_1$ is the simple lift of the 
``chiral correlation functions'' to $X(N_{D\Lambda})$ followed by 
$|[\diag(N_{D\Lambda},1)]_2$ to $X_1(N_{D\Lambda}^2)$ 
(as in sections \ref{sec:lift} and and \ref{ssec:HTheta-TypeII-CCF} so far), 
and further by $|[\diag(d'',1)]_2$ to $X_1(N_{D\Lambda}^2 d^2)$; 
the embedding $\imath^{(K,f_\rho;d^2,d'')}$ in (\ref{eq:ccf-embed2-H10ofXNs}) 
brings the image of $\imath^{(K,f_\rho;d^2,d'')}_1$ to $X(N_{D\Lambda}d)$ by 
$|[\diag(1/(dN_{D\Lambda}),1)]_2$. 
In the discussion that follows, 
we are primarily concerned about the embeddings $\imath^{(K,f_\rho;d^2,d'')}_1$'s
(than $\imath^{(K,f_\rho;d^2,d'')}$'s).

The images of $F^{(\Omega, \int J)}_{\{0,0\}}({\cal E}ll({\cal O}_K),f_\rho)$
under $\imath_1^{(K,f_\rho;d^2,d'')}$ with multiple different $d''$'s mutually 
overlap within $H^{1,0}(X_1(N_{D\Lambda}^2d^2);\C)$, in general. Let  
$f$ be the Hecke theta function associated with a Hecke character of $K$ 
whose infinity type is [-1/2;1,0] and conductor $\mathfrak{c}_f$. 
This newform is in the image of $\imath_1^{(K,f_\rho;1,1)}$ in 
$H^{1,0}(X_1(N_{D\Lambda}^2);\C)$, 
so long as $\mathfrak{c}_f | \underline{\mathfrak{c}}(K,f_z=1,f_\rho)$
(Thm. \ref{thm:newform-byCCF}).
When the image of $\imath_1^{(K,f_\rho;1,1)}$ contains the $|[\diag(d',1)]_2$ 
image of the level $M_f = |D_K|{\rm Nm}_{K/\Q}(\mathfrak{c}_f)$ newform 
(cf Thm. \ref{thm:oldform-byCCF}) 
the image of $\imath_1^{(K,f_\rho;d^2,d'')}$ 
contains the $|[\diag(d'd'',1)]_2$ image of the level $M_f$ newform $f$. 
When one can find another pair of $d'$ and $d''$ whose product $d'd''$ 
remains the same as before, then the image of $\imath_1^{(K,f_\rho;d^2,d'')}$
overlaps with that of $\imath_1^{(K,f_\rho;d^2,d'')}$ with $d''$ that of the 
other $d'$--$d''$ pair.\footnote{
Such phenomena can be understood also explicitly in terms of theta functions.
For example, for $\delta \in \N_{>0}$, 
\begin{align}
  \vartheta_\Lambda (\tau_{ws}; \delta \beta) = 
    \sum_{\beta' \in [\beta]} \vartheta_\Lambda(\delta^2 \tau_{ws}; \beta'),
\end{align}
where $[\beta]$ is the subset of $iReps \cong \mathfrak{b}
/(f_\rho D_z i)\mathfrak{b}$ that all fall into the same image 
as $\beta \in iReps$ does in the projection 
$\mathfrak{b}/(f_\rho D_zi)\mathfrak{b} \rightarrow 
\mathfrak{b}/(\delta^{-1} f_\rho D_z i)\mathfrak{b}$. Similar relations for 
the theta functions of infinity type [-1/2;1,0] describe the overlap
of the images for $(d',d'')=(d'_* \delta^2 ,d''_*)$ and 
$(d',d'') = (d'_* , d''_* \delta^2)$. The relation above may hold also 
for $\delta \in {\cal O}_{f_z=1} = {\cal O}_K$, not just for $\delta \in \N$;
the expression $\delta^2$ should be generalized to ${\rm Nm}_{K/\Q}(\delta)$ 
then.} 

The largest value of $d'$ is ${\rm Nm}_{K/\Q}(\underline{\mathfrak{c}}
/\mathfrak{c}_f)$, and the largest 
possible value of $d''$ is $d^2$. The combination of both yields 
the $|[d'd'',1]_2$ image of the newform in $H^{1,0}(X_1(N_{D\Lambda}^2d^2);\C)$
with the largest possible 
\begin{align}
   d'd'' \leq {\rm Nm}_{K/\Q}(f_\rho \sqrt{D_K})/\mathfrak{c}_f) d^2
   = {\rm Nm}_{K/\Q}((df_\rho \sqrt{D_K})_{{\cal O}_K}/\mathfrak{c}_f).
\end{align}
To summarize the discussion so far, 
\begin{thm}
Let $K$ be an imaginary quadratic field, and $f_\rho \in \N_{>0}$.
Let $N$ be $(D_Kf_\rho)^2 d^2$ for some $d \in \N_{\geq 1}$. Then the 
images of the vector space of ``chiral correlation functions''  
$F^{(\Omega,\int J)}_{\{0,0\}}({\cal E}ll({\cal O}_K),f_\rho)$ by 
$\imath_1^{(K,f_\rho; d^2,1)}$ in $H^{1,0}(X_1(N);\C)$'s with all the 
$d''$'s (s.t. $d'' | d^2$) combined contain all the Hecke newforms 
associated with Hecke characters of $K$ of type [-1/2;1,0] and conductor 
$\mathfrak{c}_f$'s that divide $\underline{\mathfrak{c}} = 
(f_\rho \sqrt{D_K})_{{\cal O}_K}$, and their $|[\diag[d'd'',1]_2$ images
with $d''|d^2$ and $d' = {\rm Nm}_{K/Q}(\mathfrak{q_r})$ with 
${}^\exists \mathfrak{q}_r$ characterized in Thm. \ref{thm:oldform-byCCF}. 
\end{thm}

Let us now turn the question around. For $N=(D_Kf)^2$, where $K$ is 
an imaginary quadratic field and $f \in \N_{>0}$, think of the set of 
embeddings 
\begin{align}
 \{\imath_1^{(K,f_\rho; (f/f_\rho)^2,d'')} \; | \; d''={\rm Nm}_{K/\Q}(\mathfrak{q}''), \; \mathfrak{q}''|(f/f_\rho)_{{\cal O}_K}\}
   \label{eq:embeddings-2-larger-frho}
\end{align}
from $F^{(\Omega, \int J)}_{\{0,0\}}({\cal E}ll({\cal O}_K),f_\rho)$ to 
$H^{1,0}(X_1(N);\C)$, for any $f_\rho \in \N_{>0}$ with $f_\rho|f$.
The images of $F^{(\Omega,\int J)}_{\{0,0\}}({\cal E}ll({\cal O}_K),f_\rho)$ by 
$\{ \imath_1^{(K,f_\rho;(f/f_\rho)^2,{\rm Nm}(\mathfrak{q}''))} \}$
with one $f_\rho | f$ overlaps a lot 
with those with another $f_\rho|f$. When $f_\rho | f'_\rho | f$, 
in particular, the images of 
$F^{(\Omega, \int J)}_{\{0,0\}}({\cal E}ll({\cal O}_K),f'_\rho)^{\chi_f^{-1}}$ 
by (\ref{eq:embeddings-2-larger-frho}) with $f'_\rho$ instead of $f_\rho$
completely overlap with the images of 
$F^{(\Omega, \int J)}_{\{0,0\}}({\cal E}ll({\cal O}_K),f_\rho)^{\chi_f^{-1}}$ by 
(\ref{eq:embeddings-2-larger-frho}), if and only if 
the conductor $\mathfrak{c}_f$ of $\chi_f$ divides 
$\underline{\mathfrak{c}}(K,f_\rho)$, not just 
$\underline{\mathfrak{c}}(K,f'_\rho)$.
So, some of the linear combinations (those that fall within 
$F([z_a],f_\rho)^{\chi_f^{-1}}$ with such a conductor) of ``chiral correlation 
functions'' of the model for $([E_{z_a}],f'_\rho)$ are not {\it new} 
but {\it old}, in the sense that they can be realized as ``chiral correlation 
functions'' of some of the models for $(E_{z_a}]_\C, f_\rho)$ with 
$f_\rho | f'_\rho$. This observation is implemented by using a direct limit 
in section \ref{ssec:KahlerVsShimuraC}. 

\newpage

\section{Modular Parametrizations and RCFT Chiral Correlation Functions 
of an Elliptic Curve of Shimura-type}
\label{sec:parametrizatn}

We have introduced the idea of assigning a modular curve for 
a $T^2$-target rational model of ${\cal N}=(2,2)$ SCFT in 
section \ref{sec:lift}, 
and worked out in detail where the vector space of ``$(g,n)=(1,2)$ 
chiral correlation functions'' of those models fit into in the vector 
space of holomorphic 1-forms on the modular curves in 
section \ref{sec:CCF-eigenform}. Having looked at the upper-left edge 
of the triangle in (\ref{eq:schematic-triangle}), we will now describe 
the upper-right edge 
in the triangle (\ref{eq:schematic-triangle})---modular parametrization
for elliptic curves of Shimura type---in section \ref{ssec:map2shimuraEC}.
Most of the task in section \ref{ssec:map2shimuraEC} is to bring 
known math together. This sets the stage for discussion in 
section \ref{ssec:KahlerVsShimuraC}, where we see that the computation 
of chiral correlation functions in string theory plays a role quite 
analogous to the modular parametrizations (non-constant maps from 
modular curves). 

\subsection{Modular Parametrization of Elliptic Curves of Shimura-type}
\label{ssec:map2shimuraEC}

\subsubsection{The $K$-simple Abelian Varieties of a Hecke
Theta Function}
\label{sssec:Ksmpl-abelV-CMcupsF}

Let us start off by recalling well-known facts about a Hecke newform $f$
of weight 2, level $M_f$, and nebentypus $\epsilon$.  
\begin{props}
\label{props:GL2-abel}
Let $f=\sum_{n \geq 1} a_n q^n$ be its Fourier series expansion, and 
$K_f = \Q(a_n| n \geq 1)$ be the number field generated by the eigenvalues 
of the Hecke operators acting on $f$. Then all of 
$\{ f^\sigma  = \sum_{n \geq 1} a_n^\sigma q^n \; | \; \sigma \in 
{\rm Gal}(\overline{\Q}/\Q)\}$ are also weight-2 and level-$M_f$ Hecke 
newforms with nebentypus $\epsilon^\sigma$, not necessarily the 
same as $\epsilon$. $[K_f:\Q]$ distinct newforms are obtained in this way. 
We will call this set of newforms {\it the Fourier-coefficient-Galois orbit}
(${\rm Gal}^{\rm FC}$ orbit) {\it of a newform $f$}. 

Those $[K_f:\Q]$ weight-2 modular forms $f^\sigma$ can be regarded as 
holomorphic 1-forms on the modular curve $X_1(M_f)$; 
$f^\sigma \mapsto f^\sigma(\tau)d\tau$. They span a vector subspace of 
$H^{1,0}(X_1(M_f);\C)$ of dimension $[K_f:\Q]$. This subspace specifies 
an abelian 
variety $A_f$ that is a quotient of ${\mathrm{Jac}}(X_1(M_f))$. 
The abelian variety $A_f$ is of $[K_f:\Q]$ dimensions. 
It is further known that we can choose the field of definition of $A_f$ 
to be $\Q$, and it is assumed implicitly that we choose so unless 
otherwise stated. It is also known that $A_f$ is $\Q$-simple and 
is of ${\rm GL}_2$-type for any Hecke newform $f$ 
\cite[\S 14 and Prop. 15.1.5]{Ribet-Stein}.

It is also known that the number field $K_f$ is either totally real 
or a CM field, and moreover, $K_f$ is CM only when $\epsilon \neq 1$
\cite[ Props. 3.2 and 3.3]{Ribet-Neben}. 

For a positive integer $N$ divisible by $M_f$, the newform $f$
has its images $f|[\diag(r,1)]_2 \in S_2(\Gamma_1(N))$ with $r|(N/M_f)$. 
One can define an abelian variety $A_{f,N}$ associated with $f$ itself ($r=1$)
as a quotient of ${\rm Jac}(X_1(N))$, where now $\dim A_{f,N}$ is $[K_f:\Q]$ 
times the number of divisors of $(N/M_f)$ (i.e, the number of images 
of the form $f|[\diag(r,1)]_2$). One can also take the field of definition 
of this $A_{f,N}$ to be $\Q$. 

For a positive integer $N$, there is an isogeny over $\Q$ between 
${\rm Jac}(X_1(N))$ and $\prod_f A_{f,N}$ where the product runs over 
all the newforms $f$ of level $M_f$ that divides $N$ 
\cite[\S 14.2]{Ribet-Stein}. $\bullet$
\end{props}

From now on, we focus on the ${\rm GL}_2$-type abelian varieties $A_f$ 
associated with a special class of newforms. Let $K$ be an imaginary 
quadratic field. For a Hecke character 
$\varphi^{(10)}_K: K^\times \backslash \mathbb{A}_K^\times \rightarrow \C$ 
of the idele class group of $K$ of infinity type [-1/2; 1,0], let 
$f_{\varphi^{(10)}_K}$ be its Hecke theta function, 
which is a Hecke newform of weight 2; see 
Lemma \ref{lemma:Shimura-ellCM-Lemma3} for the level $M_f$ and 
nebentypus $\epsilon$. In the rest of section \ref{sssec:Ksmpl-abelV-CMcupsF}, 
we quote known results on whether and how the $\Q$-simple abelian variety 
$A_f$ for such a newform $f=f_{\varphi^{(10)}_K}$ is decomposed over $K$ 
into $K$-simple 
varieties.\footnote{Even if an abelian variety $A/F$ is $F$-simple, 
$A \times_F E$ may not be $E$-simple, where $E/F$ is a field extension. 
If $A \times_F E$ is $E$-simple, however, $A/F$ is $F$-simple. 
Similarly, even if $A/F$ is $F$-simple, $A \times_F \overline{\Q}$ may not 
be simple over $\overline{\Q}$.}

\begin{lemma}
\label{lemma:SchappPoH-2-1}
For a Hecke character $\varphi^{(10)}_K$ of an imaginary quadratic field, 
let $T_{\varphi^{(10)}_K}$ be the number field generated over $\Q$ 
\cite[\S1.5 p.26]{SchPoH} by the values 
$\varphi^{(10)}_K(\mathfrak{a}) \in \overline{\Q}$ of 
integral ideals $\mathfrak{a}$ of ${\cal O}_K$. Then 
$T_{\varphi^{(10)}_K}$ is always a CM field \cite[p.4]{SchPoH}, and 
$T_{\varphi^{(10)}_K} = K K_{f_{\varphi^{(10)}_K}}$, so 
$[T_{\varphi^{(10)}_K}: K_{f_{\varphi^{(10)}_K}}]$ is either 
2 or 1 \cite[Ch.5 Lemma 2.1.1 (p.141)]{SchPoH}. 
Now we see that one and only one of the following three must 
hold true \cite[Ch. 5 \S2.1.2--2.1.4]{SchPoH}.

{\bf Case (i-real)}: $[T_{\varphi^{(10)}_K}: K_{f_{\varphi^{(10)}_K}}] = 2$, and 
$K_{f_{\varphi^{(10)}_K}}$ is totally real. In this case, 
the Galois${}_K$ orbit\footnote{
For a Hecke character $\varphi^{(10)}: K^\times \backslash \mathbb{A}_K^\times 
\rightarrow \C$ of infinity type [-1/2; 1, 0] where the conductor is 
an integral ${\cal O}_K$ ideal $\mathfrak{c}_f$, let 
$T_{\varphi^{(10)}} \subset \overline{\Q} \subset \C$ be the number field 
generated by the images of $I_K(\mathfrak{c}_f) \subset \mathbb{A}_K^\times$
under $\varphi^{(10)}$; $I_K(\mathfrak{c}_f)$ is the group 
of ideals of ${\cal O}_K$ prime to $\mathfrak{c}_f$.
This $T_{\varphi^{(10)}}$ contains (an embedded image of)
$K$ in $\overline{\Q} \subset \C$. 
Then there is a unique algebraic Hecke character $\chi: I_K(\mathfrak{c}_f) 
\rightarrow T'$ of infinity type $n=1$ (weight 1) 
where $T'$ an abstract algebraic extension field of $K$ isomorphic to $T$
(cf \cite[\S 0--1]{SchPoH}). 
The Hecke characters in the Galois${}_K$-orbit
are all of those that give rise to a common algebraic Hecke character $\chi$
over $K$ with value in $T'$; the $[T:K]$ distinct Hecke characters 
in the Galois${}_K$-orbit correspond to the $[T':K]$ distinct 
embeddings ${\rm Hom}_K(T', \overline{\Q})$ with a fixed embedding  
$\rho^1: K \hookrightarrow \overline{\Q}$. This algebraic Hecke character 
$\chi$ also 
determines an object called the motive of $\chi$ over $K$ with value in $T'$, 
denoted by $M(\chi)$. So, the Galois${}_K$ orbit of Hecke characters of 
the idele class group of $K$ in the main text can be regarded as a synonym of 
the motive over $K$. We decided, however, not to use the word ``motive'' 
in the main text because it is so a powerful jargon that it would fend 
off large fraction of readers in string-theory community 
(including one of the authors).} 
$\{ \sigma \cdot \varphi^{(10)}_K \; | \; \sigma \in 
{\rm Gal}(\overline{\Q}/K)\}$
of the Hecke character $\varphi^{(10)}_K$ 
consists of $[T_{\varphi^{(10)}_K}:K]$ distinct 
Hecke characters, and each one of them, $\sigma \cdot \varphi^{(10)}_K$, 
satisfies the relation $(\sigma \cdot \varphi^{(10)}_K)^{\wedge}  = 
(\sigma \cdot \varphi^{(10)}_K)$; here, the $^{\wedge}$-operation converts 
one Hecke character $\varphi: K^\times \backslash \mathbb{A}_K^\times 
\rightarrow \C$ of infinity type [-1/2; 1,0] into another 
Hecke character $\varphi^{\wedge} := cc_\C \cdot \varphi \cdot cc_K: 
K^\times \backslash \mathbb{A}_K^\times \rightarrow \C$ of the same infinity 
type. Because $[K_{f_{\varphi^{(10)}_K}}:\Q] = [T_{\varphi^{(10)}_K}:K]$ in this case, 
the ${\rm Gal}^{\rm FC}$ orbit of the newform $f_{\varphi^{(10)}_K}$ 
agrees with the set of Hecke theta functions corresponding to the 
Galois${}_K$-orbit of the Hecke character $\varphi^{(10)}_K$. 

{\bf Case (i-CM):} $[T_{\varphi^{(10)}_K}: K_{f_{\varphi^{(10)}_K}}]=2$, and
$K_{f_{\varphi^{(10)}_K}}$ is also a CM field.  The Galois${}_K$ orbit 
of $\varphi^{(10)}_K$ as a whole remains invariant under 
the ${}^{\wedge}$-operation, 
but the individual Hecke characters in the Galois${}_K$ orbit are not;
they are exchanged within the orbit.
The ${\rm Gal}^{\rm FC}$ orbit of $f_{\varphi^{(10)}_K}$
agrees with the set of Hecke theta functions corresponding to the 
Hecke characters in the Galois${}_K$ orbit of $\varphi^{(10)}_K$ 
also in this case. 

{\bf Case (ii):} $K_{f_{\varphi^{(10)}_K}} = T_{\varphi^{(10)}_K}$. 
The ${\rm Gal}^{\rm FC}$ orbit of $f_{\varphi^{(10)}_K}$ consists 
of the Hecke theta functions corresponding to two Galois${}_K$ orbits 
of Hecke characters of the same infinity type. One is the Galois${}_K$ orbit 
of $\varphi^{(10)}_K$ itself, and the other that of $(\varphi^{(10)}_K)^{\wedge}$.
$\bullet$
\end{lemma}

\begin{props}
(\cite[Ch.5 Prop. 2.2 (p.142)]{SchPoH}; \cite[p.41]{Ribet-Neben})
\label{props:SchappPoH-2-2}
Let $K$ be an imaginary quadratic field, $\varphi^{(10)}_K$ a Hecke character 
of infinity type [-1/2;1,0], and $f_{\varphi^{(10)}_K}$ its Hecke theta function. 
Then the base change $A_f \times_\Q K$ of the $\Q$-simple abelian variety $A_f$
remains $K$-simple in both of case (i-real) and case (i-CM). 
In the case (ii), 
$A_f \times_\Q K$ is isogenous over $K$ to a product of the pair of $K$-simple
abelian varieties defined over $K$, 
$B_{\varphi^{(10)}_K} \times B_{(\varphi^{(10)}_K)^{\wedge}}$,
where $B_{(\varphi^{(10)}_K)^{\wedge}}$ is a $K$-simple abelian variety
whose associated\footnote{
$A_f \times_\Q K \cong M(\chi)$ as motives over $K$ in both of the 
cases (i-real) and (i-CM), where $\chi$ is the algebraic Hecke character 
corresponding to the Galois${}_K$-orbit of $\varphi^{(10)}$. 
$A_f \times_\Q K \cong M(\chi) \oplus M(\chi^{\wedge})$ in the case (ii), where 
$\chi$ and $\chi^{\wedge}$ are the algebraic Hecke characters 
corresponding to the Galois${}_K$-orbit of $\varphi^{(10)}$ and 
$(\varphi^{(10)})^{\wedge}$, respectively.
 } 
Hecke character is $(\varphi^{(10)}_K)^{\wedge}$. 
\end{props}

We have already encountered examples of Hecke characters of imaginary quadratic 
fields (and their Hecke theta functions) of all the three cases above. 

\begin{exmpl}
\label{exmpl:Hecke-i-real}
{\bf Case (i-real)}: 
For the imaginary quadratic field $K=\Q(\sqrt{-1})$, there is one Hecke 
character of infinity type [-1/2;1,0] with the conductor 
$\mathfrak{c}_f=(2+2i)$, and there is another with the conductor 
$\mathfrak{c}_f=(4)$. The Hecke theta functions of those two 
Hecke characters appeared in Example \ref{exmpl:Dk4-fz1-N64};  
a follow-up discussion is in Example \ref{exmpl:Dk4-fz1-N64-followup}.
$M_f=32$ and $64$ there. For both of those Hecke characters, $T=K$ and 
$K_f=\Q$; $\dim_\Q A_f = 1$, and $[T:K]=1=[K_f:\Q]$, the case (i-real).  
\end{exmpl}

\begin{exmpl}
\label{exmpl:Hecke-i-CM}
{\bf Case (i-CM)}
For the imaginary quadratic field $K=\Q(\sqrt{-2})$, there is a pair 
of Hecke characters with the infinity type [-1/2;1,0] and the conductor 
$\mathfrak{c}_f = (2\sqrt{2}i)$; we have already encountered their Hecke 
theta functions in Example \ref{exmpl:Dk8-fz1-N64}, where $M_f=64$ 
(additional discussion is in Example \ref{exmpl:Dk8-fz1-N64-followup}). 
For this pair, $T = K(i)=\Q(\sqrt{2}i,i)$ and $K_f=\Q(i)$; there is just 
one $\Q$-simple and CM-type $A_f/\Q$, with $\dim_\Q A_f=2$. 
$[T:K]=2=[K_f:\Q]$. This is a case (i-CM).

For the imaginary quadratic field $K=\Q(\sqrt{-5})$, there are two pairs 
of Hecke characters all of infinity type [-1/2;1,0] and conductor 
$\mathfrak{c}_f=(\sqrt{5}i)$. Their Hecke theta functions appeared 
in section \ref{sssec:exmpl-Dk20-fz1-fr1}; $M_f=100$ for all the four;  
two of them have nebentypus $\chi_{20}(1,1)$, while the other two 
$\chi_{20}(1,3)$. The corresponding Hecke characters are denoted by 
$\varphi^{(10)}_{\epsilon;2\pm}$ and $\varphi^{(10)}_{\epsilon';2\pm}$, respectively. 
In this example, $T=H_K$ (the Hilbert class field of $K=\Q(\sqrt{-5})$) 
and $K_f=\Q(i)$, so $[T:K]=2=[K_f:\Q]$. This is another case (i-CM).
The pair $\{ \varphi^{(10)}_{\epsilon;2+},\varphi^{(10)}_{\epsilon';2+}\}$ 
forms one Galois${}_K$ orbit, and their corresponding 
newforms also forms one ${\rm Gal}^{FC}$ orbit, so there is one 
$\Q$-simple CM-type $A_f/\Q$ of $\dim_\Q A_f=2$. The pair 
$\{ \varphi^{(10)}_{\epsilon;2-},\varphi^{(10)}_{\epsilon';2-}\}$ also has 
one $\Q$-simple CM-type abelian variety of dimension 2. 
\end{exmpl}

\begin{exmpl}
\label{exmpl:Hecke-ii}
{\bf Case (ii)}: 
For the imaginary quadratic field $K=\Q(\sqrt{-1})$, there is one 
pair of Hecke characters with the conductor $\mathfrak{c}_f=(10)$, 
another pair with $\mathfrak{c}_f=(5)$, another pair with 
$\mathfrak{c}_f=(2\pm 4i)$, and another pair with $\mathfrak{c}_f=(1\mp 2i)$, 
all with infinity type [-1/2;1,0]. Their Hecke theta functions appeared 
in section \ref{sssec:exmpl-Dk4-fz1-fr5}, where $M_f=400$, $100$, $80$, 
and $20$, respectively. For all of $M_f$'s, $T=K_f=K$ for the pair of 
Hecke characters, and they remain to form one-element Galois${}_K$ orbits
separately that are mutually ${}^\wedge$-operation pair.
The corresponding pair of newforms forms one ${\rm Gal}^{\rm FC}$ orbit, 
and hence one $\Q$-simple CM-type abelian variety $A_f/\Q$ of 
$\dim_\Q A_f = 2$.

For the imaginary quadratic field $K=\Q(\sqrt{-5})$, there is 
a quartet of Hecke characters with the infinity type[-1/2;1,0] and 
conductor $\mathfrak{c}_f=(2\sqrt{5}i)$. Their Hecke theta functions 
appeared in section \ref{sssec:exmpl-Dk20-fz1-fr1}, where $M_f=400$; 
two of them---denoted by $\varphi^{(10)}_{\epsilon,2\pm}$---have the nebentypus 
$\chi_{20}(1,1)$, and the nebentypus is $\chi_{20}(1,3)$ for the other two, 
which are denoted by $\varphi^{(10)}_{\epsilon';2\pm}$. In this example,   
$T=K_f=\Q(\sqrt{(10i)})=H_K$, and the Hecke theta functions of those four 
Hecke characters of type [-1/2;1,0] form one ${\rm Gal}^{FC}$ orbit. 
There is just one $\Q$-simple CM-type $A_f/\Q$.  
The Galois${}_K$ orbit decomposition of the four Hecke characters is 
$\{\varphi^{(10)}_{\epsilon;2+}, \varphi^{(10)}_{\epsilon';2-}\} \amalg 
 \{ \varphi^{(10)}_{\epsilon';2+}, \varphi^{(10)}_{\epsilon;2-}\}$. 
The ${}^\wedge$-operation exchanges the two orbits.  $\bullet$
\end{exmpl}

\subsubsection{Elliptic Curves of Shimura-type}
\label{sssec:shimuraEC}
We collect some more math, concerning elliptic curves of Shimura type.
This subsubsection is independent of the previous section 
\ref{sssec:Ksmpl-abelV-CMcupsF},
and will be used in the following subsubsections. 
\begin{defn}
\label{defn:Shimura type}
Let $E/k$ be an elliptic curve with complex multiplication by an order 
${\cal O}_{f_z}$ of an imaginary quadratic field $K$, and let 
$k$ be the field of definition of the elliptic curve $E$. We assume 
that $k$ is a finite abelian extension of $K$.
This elliptic curve $E/k$ is said to be of {\it Shimura type}, 
if $k(E_{\rm tor})/K$ is abelian.
\end{defn}

\begin{lemma}
\label{lemma:def-B-as-WeilR}
Let $E/k$ be an elliptic curve with complex multiplication by some 
order ${\cal O}_{f}$ of an imaginary quadratic field $K$, and suppose 
that the field of definition $k$ is an abelian extension of $K$ 
(and contains the ring class field $L_f$).
Then there is an isomorphism
\begin{align}
  B := {\rm Res}_{k/K}(E) \cong_{/k} \prod_{\sigma \in {\rm Gal}(k/K)} ({}^\sigma E)
\end{align}
over $k$, where $\mathrm{Res}_{k/K} E$ is 
an abelian variety over $K$ of dimension $[k:K]$, called the Weil 
restriction of $E$.
This statement (existence of the isomorphism over $k$) is found at the 
beginning 
of \S 4 of \cite{GSchappacher} (and also in \cite[\S 15]{Gross-LNM}).\footnote{ 
For this statement to be true, an elliptic curve $E/k$ does not have 
to have complex multiplications, or the extension $k/K$ Galois; 
when $k/K$ is not Galois, the product on the right hand side is 
over $\sigma \in {\rm Hom}_K(k, \bar{k})$ 
\cite[just before Lemma 13.3]{Milne-pedestrian}. 
}
\end{lemma}

\begin{lemma}
\label{lemma:Lfcn-abelV-restrictScalar}
Let $E$ be an elliptic curve defined over a number field $k$, and 
$k$ be a finite extension over a number field $K$; $k/K$ is not 
necessarily Galois, nor $E/k$ is of CM type. Then 
\begin{align}
  L(E/k,s) = L(B/K,s).
\end{align}
See \cite[Lemma 13.3]{Milne-pedestrian} for a sketch of proof. 
\end{lemma}
\begin{props}\cite[Thm 4.1 and Lem. 4.8(i)]{GSchappacher}
\label{props:GS}
Let $E/k$ be an elliptic curve with complex multiplication by an order 
of $K$, and suppose that the field of definition $k$ is an abelian extension 
of $K$. Then the following four conditions are equivalent. 
\begin{itemize}
\item (i) $E/k$ is of Shimura type. 
\item (v) For the Hecke character $\psi^{(1,0)}_{E/k}: 
k^\times \backslash \mathbb{A}_k^\times \rightarrow \C^\times$ 
[resp. Serre--Tate character $\psi_{E/k}: k^\times \backslash \mathbb{A}_k^\times 
\rightarrow K^\times$]
associated with $E/k$, there exists a Hecke character 
$\varphi^{(1,0)}: K^\times \backslash \mathbb{A}_K^\times \rightarrow \C^\times$ 
such that $\psi^{(1,0)}_{E/k} = \varphi^{(1,0)} \cdot {\rm Nm}_{k/K}$ 
[resp. an algebraic Hecke character $\chi$ over $K$ such that 
$\psi_{E/k} = \chi \circ {\rm Nm}_{k/K}$].
\item (iv) Let $T := {\rm End}_K(B) \otimes_\Z \Q$ be the ring of 
endomorphisms tensored by $\Q$. There exists a finite set $I$ such that 
$T \cong \oplus_{e \in I} T_e$ with each $T_e$ a CM field (commutative in 
particular); the set $I$ consists of primitive idempotents of the algebra 
$T$ ($e \in T$). Each one of those CM fields ($T_e$'s) contains $K$, and 
$[k:K] = \dim (B) = \sum_{e\in I} [T_e:K]$. 
\item (ii)--(iii) The representation 
${\rm Ind}^{{\rm Gal}(\overline{\Q}/K)}_{{\rm Gal}(\overline{\Q}/k)} (\rho_\ell (E/k))
 = \rho_\ell(B/K)$ is abelian. 
\end{itemize}
Here, $\rho_\ell(B/K)$ and $\rho_\ell(E/K)$ denote the Galois representations
associated with the dual of the respective $\ell$-adic Tate modules tensored 
with $\Q_\ell$.  $\bullet$
\end{props}

Some remarks are in order here so we can digest all that are crammed 
into the Proposition above. We begin with additional information 
related to the condition (v).
\begin{rmk}\cite[Thm. 7.42 and 7.44]{Shimura-AA} 
\label{rmk:shimura-7-44}
When those conditions are satisfied, (v) in particular, there are 
$|{\rm Gal}(k/K)| = [k:K]$ distinct Hecke characters $\varphi^{(1,0)}$ 
satisfying the condition $\psi^{(1,0)}_{E/k} = \varphi^{(1,0)} \cdot {\rm Nm}_{k/K}$. 
They are denoted by $\varphi^{(10)}_{K,i}$ labeled by $i=1,\cdots, [k:K]$.
Similarly, $\varphi^{(01)}_{K,i}$ with $i=1,\cdots, [k:K]$ are the Hecke 
characters $K^\times \backslash \mathbb{A}_K^\times \rightarrow \C^\times$ for 
$\psi^{(01)}_{E/k}$.
\begin{align}
  L(E/k,s) = L(\psi^{(1,0)}_{E/k},s) L(\psi^{(0,1)}_{E/k},s) = 
    \prod_{i=1}^{[k:K]} L(\varphi_{K,i}^{(1,0)}, s) \; 
    \prod_{i=1}^{[k:K]} L(\varphi_{K,i}^{(0,1)}, s). 
  \label{eq:Lfcn-GalE-Autok}
\end{align}
\end{rmk}

The conditions (ii)--(iii) in the Proposition and the Hecke characters 
that appeared in the Remark above are related as follows. 
\begin{rmk}
\label{rmk:Shimura-type-GalChartzn}
Suppose $k \supset K$.
On the two dimensional 
$\Q_\ell$-vector space ${\rm Tate}_\ell(E) \otimes_{\Z_\ell} \Q_\ell$, 
we have a representation of the Galois group 
${\rm Gal}(\overline{\Q}/k)$ dual to $\rho_\ell(E/k)$; the image 
of the group ${\rm Gal}(\overline{\Q}/k)$ under this representation 
$\rho_\ell$ is in an abelian 
subgroup $M \subset {\rm GL}_2(\Q_\ell)$, because $E$ has complex 
multiplication (defined over $k$). 
The two-dimensional representation 
$\rho_\ell(E/k)$ on the dual of
${\rm Tate}_\ell(E)\otimes_{\Z_\ell} \Q_\ell$ splits into a pair 
of one-dimensional representation over the dual vector space of 
${\rm Tate}_\ell(E) \otimes_{\Z_\ell} \overline{\Q}_\ell$,
so the two representations are denoted by 
$\rho^{(10)}_\ell(E/k)$ and $\rho^{(01)}_\ell(E/k)$. 
The Hecke characters 
$\psi^{(10)}_{E/k}$ and $\psi^{(01)}_{E/k}$
are those that correspond to the Galois representations
$\rho^{(10)}_\ell(E/k)$ and $\rho^{(01)}_\ell(E/k)$
via class field theory (or the Langlands correspondence).

Consider the induced representation of $\rho_\ell(E/k)$, and 
also of $\rho^{(10)}_\ell(E/k) \oplus \rho^{(01)}_\ell(E/k)$, for the 
inclusion of the groups ${\rm Gal}(\overline{\Q}/k) \subset 
{\rm Gal}(\overline{\Q}/K)$, namely, 
\begin{align}
  {\rm Ind}^{{\rm Gal}(\overline{\Q}/K)}_{{\rm Gal}(\overline{\Q}/k)} ( \rho_\ell(E/k)), 
   \qquad 
  {\rm Ind}^{{\rm Gal}(\overline{\Q}/K)}_{{\rm Gal}(\overline{\Q}/k)} \left( 
     \rho^{(10)}_\ell(E/k) \oplus \rho^{(01)}_\ell(E/k) \right).
\end{align}
The subgroup $[G_K, G_K]$ of $G_K := {\rm Gal}(\overline{\Q}/K)$ may not 
be in the kernel. 
The conditions (ii)--(iii) on an arithmetic model $E/k$ of an elliptic 
curve $[E]_\C$ with complex multiplication is that $[G_K,G_K]$ is in the 
kernel, so the representation of $G_K$ on the dual space 
of ${\rm Tate}_\ell(B) \otimes_{\Z_\ell} \overline{\Q}_\ell$ splits into 
$[k:K] + [k:K]$ separate one-dimensional representations. 
The Hecke characters corresponding via class field theory 
to one-dimensional representations of $G_K = {\rm Gal}(\overline{\Q}/K)$ 
are the Hecke characters $\varphi^{(10)}_{K,i=1,\cdots, [k:K]}$ and 
$\varphi^{(01)}_{K,i=1,\cdots,[k:K]}$.   $\bullet$
\end{rmk}

Now, the following remark elaborates more on the condition (iv) 
in the Proposition.
\begin{rmk}
\label{rmk:gross-LNM-s15}
The direct sum structure of the algebra $T \cong \oplus_{e \in I} T_e$
corresponds to $K$-simple decomposition of the abelian variety $B/K$. 
To be more explicit, let 
\begin{align}
 B \sim_{/K} \prod_{e \in I} B_e 
  \label{eq:BfromE-split-by-K}
\end{align}
be an isogeny defined over $K$, where each factor $B_e$ is a CM-type  
abelian variety defined over $K$ and is $K$-simple; we have already 
used the set $I_{{\rm for}T}$ of primitive idempotents of $T$ also in labeling 
the $K$-simple factors of $B$, $I_{{\rm for}B}$. 
To see that $I_{{\rm for}B}=I_{{\rm for}T}$, 
note that the endomorphism algebra ${\rm End}_{K}(B_e/K) \otimes_\Z \Q$ of 
a $K$-simple CM-type abelian variety is a CM field $T_e$, $e \in I_{{\rm for}B}$;
if there were an isogeny over $K$ between a pair of the $K$-simple factors 
$B_e$ and $B_{e'}$ with $e,e' \in I_{{\rm for}B}$, the total endomorphism 
algebra $T = {\rm End}_K(B/K) \otimes_\Z \Q$ would have a structure 
of $\oplus_{{\rm isogeny~classes}} {\rm Matrix}({\rm CM~field})$. The structure 
$T \cong \oplus_{e \in I} T_e$ with all the $T_e$'s commutative implies that 
there must be no isogeny over $K$ between a pair of $K$-simple factors 
in (\ref{eq:BfromE-split-by-K}), and that $I_{{\rm for}B}$ is equal 
to $I_{{\rm for}T}$ (as announced already). 

Due to \cite[Thm. 5]{Shimura-zetaA}, the absence of an isogeny between 
different $K$-simple factors $B_e$ and $B_{e'}$ 
in (\ref{eq:BfromE-split-by-K}) indicates that 
their Serre--Tate characters $\psi_{e/K}: K^\times \backslash \mathbb{A}_K^\times 
\rightarrow T_e^\times$ 
and $\psi_{e'/K}: K^\times \backslash \mathbb{A}_K^\times 
\rightarrow T_{e'}^\times$ should be different.\footnote{That does not rule 
out a possibility that their CM fields $T_e$ and $T_{e'}$ happen to be 
isomorphic, however.} Now, 
for each factor $B_e/K$, and its Serre--Tate character $\psi_{e/K}$, 
there are $[T_e:K] + [T_e:K]$ Hecke characters $\chi^{(10)}_{e;K,\alpha}$
and $\chi^{(01)}_{e;K,\alpha}$, where $\alpha = 1,\cdots, [T_e:K]$, 
corresponding to the same number of embeddings of $T_e$ to $\overline{\Q}$; 
the superscripts ${}^{(10)}$ and ${}^{(01)}$ refer to how the subfield 
$K \subset T_e$ is embedded into $\overline{\Q}$. 
Applying general results on CM-type abelian varieties 
\cite[18.1.6]{Gross-LNM}, 
\begin{align}
  L(B/K,s) =  \prod_{e \in I} \left( \prod_{\alpha =1}^{[T_e:K]} 
  L(\chi_{e; K,\alpha}^{(10)},s) \; L(\chi^{(01)}_{e;K,\alpha}, s) \right). 
  \label{eq:Lfcn-GalB-AutoK}
\end{align}

The decomposition (\ref{eq:BfromE-split-by-K}) into $K$-simple factors 
is determined modulo isogenies defined over $K$. Neither we require 
that $B_e$'s are subvarieties of $B$ nor are obtained as a quotient 
of $B$. That is fine, because the $L$-function of an abelian variety 
defined over $K$ remains the same for all that belong to the same class 
of isogeny over $K$. $\bullet$
\end{rmk}

\begin{rmk}
\label{rmk:compare-ShimuraType-varphi-chi}
The set of Hecke characters $\{\varphi^{(10)}_{K,i} \; | \; i=1,\cdots, [k:K]\}$
in Remark \ref{rmk:shimura-7-44} agree with 
$\cup_{e\in I} \; \{ \chi^{(10)}_{e;K,\alpha} \; | \; \alpha = 1,\cdots, [T_e:K]\}$ 
in Remark \ref{rmk:gross-LNM-s15}, and so do the Hecke characters of 
the infinity type [-1/2;-1,0]. 
cf. \cite[Lemma 4.8 (i)]{GSchappacher}. 
\end{rmk}

There are 13 elliptic curves $[E]_\C$ with complex multiplication 
that have arithmetic models $E^+/\Q$; 
even for such an $[E]_\C$, not all its models $E/K$ defined over the 
imaginary quadratic field $K$ can be regarded as the base change of a 
model over $\Q$, $E = E^+ \times_\Q K$. A model $E^+/\Q$ is found for some 
models $E/K$, but there is none for other models $E/K$. 
A similar story holds for all the CM elliptic curves $[E_z]_\C$ not 
necessarily with $h({\cal O}_{f_z}) = 1$. 

\begin{anythng}\cite{Wortmann}
\label{statmnt:wortmann}
Let $E/k$ be an elliptic curve of Shimura type. It is said to be 
an {\it elliptic curve of Shimura type with the (S2) condition}, if it 
satisfies the following:
\begin{itemize}
\item[(S2)] there exist a subfield $F$ of $k$ which does not contain $K$, 
$k = FK$, and an elliptic curve $E^+$ defined over $F$, such that the 
base change reproduces $E/k$. That is, $E^+ \times_F k = E$. 
\end{itemize}
Let $B^+ := {\rm Res}_{F/\Q}(E^+)$ be the Weil restriction, an abelian variety 
over $\Q$ of dimension $[F:\Q] = [k:K]$. Then $B = B^+ \times_\Q K$.
The algebra of endomorphisms $T^+ := {\rm End}_\Q (B^+) \otimes_\Z \Q$ 
is an index 2 subalgebra of $T = {\rm End}_K(B) \otimes_\Z \Q$ fixed 
by a non-trivial involution on $T$. Furthermore, the condition (S2) 
is equivalent to $j([E]_\C)^{cc} = j([E]_\C)$ and 
$\{ (\varphi^{(10)}_{e,K,\alpha})^\wedge \} = \{ (\varphi^{(10)}_{e,K,\alpha}) \}$.
So, it follows that $e^\wedge = e$ for $e \in I$ of case (i-real) and (i-CM),
while $e \in I$ of case (ii) is accompanied by $e^\wedge \in I$ different 
from $e$ itself. Let $I^+ := \{ [e] \; | e\in I \}_{e: {\rm case~(i)}} \amalg 
\{ [e] \; | \; (e,e^\wedge) \subset I \}_{e,e^\wedge: {\rm case~(ii)}}$.
Then $T^+ \cong \oplus_{d \in I^+} T^+_d$, where $[T_e:T^+_{d=[e]}]=2$
is the invariant part of the CM field $T_e$ under an involution on $T$ 
for $e \in I$ of case (i), and 
$T^+_{d=[e]} \cong T_e$ is the invariant part of $(T_e \oplus T_{e^\wedge})$
for $(e,e^\wedge) \subset I$ of case (ii). 
The $\Q$-simple decomposition of $B^+$ is $B^+ \cong_{/\Q} \prod_{d \in I^+}B^+_d$, 
and $T^+_d = {\rm End}_\Q(B^+_d) \otimes_\Z \Q$. The base change of those 
$\Q$-simple abelian varieties $B^+_d$ are 
$B_e = B^+_{d=[e]} \times_\Q K$ for $e \in I$ of case (i), and 
$B_e \times B_{e^\wedge} \cong B^+_d \times_\Q K$ for $(e,e^\wedge) \subset I$ 
of case (ii).

The representation $\rho_\ell(B^+) = 
{\rm Ind}^{{\rm Gal}(\overline{\Q}/\Q)}_{{\rm Gal}(\overline{\Q}/F)}(\rho_\ell(E^+))$ 
of ${\rm Gal}(\overline{\Q}/\Q)$ on the dual vector space of 
${\rm Tate}_\ell(B^+) \otimes_{\Z_\ell} \overline{\Q}_\ell$ splits into 
$[F:\Q] = [k:K]$ distinct two-dimensional representations. 
This is because the restriction of $\rho_\ell(B^+)$ on the index two subgroup
${\rm Gal}(\overline{\Q}/K)$ should be the same as $\rho_\ell(B) = 
{\rm Ind}^{{\rm Gal}(\overline{\Q}/K)}_{{\rm Gal}(\overline{\Q}/k)}(\rho_\ell(E))$, 
which is the sum of $2[k:K]$ distinct one-dimensional representations.
Their induced representations can only be two-dimensional. $\bullet$
\end{anythng}

\subsubsection{A Non-constant Map from the Jacobian 
of Modular Curves to a Elliptic Curve of Shimura-type}
\label{sssec:surjective-morph-2BK}

The aim of this section is to construct a morphism over $K$
from the Jacobian of modular curves to the Weil restriction 
$B={\rm Res}_{k/K}(E)$ of an elliptic curve $E/k$ of Shimura type, 
so that we have a good description of the pullback
map induced on $H^{1,0}$ of those abelian varieties. 

\begin{props}
\label{props:BK-as-quotient-of-Jac1}
Let $E/k$ be an elliptic curve of Shimura type. The notation being 
the same as in sections \ref{sssec:Ksmpl-abelV-CMcupsF} and 
\ref{sssec:shimuraEC}. Then there exists a surjective morphism 
\begin{align}
  \nu_N: {\rm Jac}(X_1(N)) \times_\Q K \rightarrow B
   \label{eq:surj-fromJac-toBofE}
\end{align}
defined over $K$, when a positive integer $N$ is divisible by 
a certain positive integer $N_E$ that is specified (eq. (\ref{eq:def-NE})) 
in the proof below.\footnote{Because $N=N_{D\Lambda}^2 = |D_K| {\rm Nm}_{K/\Q}(\mathfrak{c}_f) = 
D_K^2 f_\rho^2 \geq 9$ cannot be very small,
$X_1(N)$ is an algebraic variety, not an orbifold (defined over $\Q$). } 
\end{props}

\begin{proof}
For each $K$-simple factor $B_e/K$ of $\prod_e B_e$, the set of 
$[T_e:K]$ Hecke characters $\{ \chi^{(10)}_{e;K,\alpha}\}$ share the 
same conductor $\mathfrak{c}_{f_e}$, and the set of $[T_e:K]$ Hecke 
theta functions $\{ f_{\chi^{(10)}_{e;K,\alpha}} \}$ have the same level 
$M_{f_e} = |D_K| {\rm Nm}_{K/\Q}(\mathfrak{c}_{f_e})$. Now, to those 
Hecke theta functions, a $\Q$-simple CM-type abelian variety is assigned 
by the argument in Prop. \ref{props:GL2-abel} 
(also \cite[Thm.7.14 p.183]{Shimura-AA}), which we denote $A_{f[e]}/\Q$.
One such $A_{f[e]}$ is for two $B_e$'s, if and only if $A_{f[e]}$ in question 
is the case (ii) in Lemma \ref{lemma:SchappPoH-2-1} and Prop. \ref{props:SchappPoH-2-2}, and 
both $B_{\varphi}$ and $B_{\varphi^\wedge}$ appear in the product $\prod_e B_e$.
Otherwise, one $A_{f[e]}/\Q$ is for just one factor $B_e$. 
Let $I' := \{ [e] \; | \; e \in I\}$ be the label of 
such abelian varieties $A_{f[e]}$'s without duplication
(this set $I'$ is in one-to-one with the ${\rm Gal}^{\rm FC}$ orbits 
of such Hecke theta functions); $[e] = [e'] \in I'$ when 
$A_{f[e]} = A_{f[e']}$ for $e, e' \in I$. 

Because there is a surjective morphism $A_{f[e]}\times_\Q K \rightarrow B_e$ 
defined over $K$ (Prop. \ref{props:SchappPoH-2-2}), 
there is a morphism 
\begin{align}
  \prod_{[e] \in I'} \left( A_{f[e]} \times_\Q K \right) 
     \rightarrow \prod_{e \in I} B_e \cong B
  \label{eq:fromAf-toB}
\end{align}
defined over $K$. The morphism above is surjective. 

Now, we set\footnote{At this moment, we do not know whether $M_{f[e]}$ can 
be different for different $[e]$'s, or whether $N_E \in \N$ is somehow 
related to the conductor describing the degeneration of $E/k$ over 
the curve ${\rm Spec}({\cal O}_k)$.}
\begin{align}
  N_{E} := {\rm LCM} \left\{ M_{f[e]} \; | \; [e] \in I' \right\}. 
    \label{eq:def-NE}
\end{align}
Then for any positive integer $N$ divisible by $N_E$, there is a 
natural surjective morphism $\pi_{[e]}: X_1(N) \rightarrow X_1(M_{f[e]})$
defined over $\Q$ for each $[e]$ appearing in the product $\prod_{[e]}A_{f[e]}$, 
and there is also a surjective morphism $(\pi_{[e]})_*: {\rm Jac}(X_1(N)) 
\rightarrow {\rm Jac}(X_1(M_{f[e]}))$. By combining this set of morphisms 
$(\pi_{[e]})_*$ over $\Q$ with the surjective morphisms over $\Q$ 
from ${\rm Jac}(X_1(M_{f[e]}))$ to $A_{f[e]}/\Q$, a morphism over $\Q$ from 
${\rm Jac}(X_1(N))$ to $\prod_{[e]}A_{f[e]}$ is obtained; this morphism is 
surjective, because all the Hecke characters $\{ \varphi^{(10)}_{K,i}\}$ 
of infinity type [-1/2;1,0] are distinct from one another, and no two 
$A_{f[e]}$'s are isogenous with each other over $\Q$ 
(cf Prop. \ref{props:GL2-abel}).

Combining the two surjective morphisms, ${\rm Jac}(X_1(N)) \rightarrow 
\prod_{[e]} A_{f[e]}$ over $\Q$, and the morphism (\ref{eq:fromAf-toB}) over 
$K$, we obtain a surjective 
morphism (\ref{eq:surj-fromJac-toBofE}) defined over $K$.
\end{proof}

Composing $\nu$ above with the morphism 
$\mu: X_1(N) \to \mathrm{Jac}(X_1(N))$,
we obtain $X_1(N)\times_\Q K \to B$.
We know that $H^{1,0}(X_1(N); \C)$
is isomorphic to the space of cusp forms of level $N$.
We also know that 
$H^{1,0}(B;\C)$ is spanned by eigenstates 
of the action of the algebra $T \cong \oplus_{e\in I} T_e$ 
that are in one to one with the $[k:K]$ Hecke characters
$\{\chi_{e;K,\alpha}\}_{e,\alpha}$. We have the following corollary.
\begin{cor}
The image of the pullback map $H^{1,0}(B; \C) \to H^{1,0}(X_1(N);\C)$
is the span of the Hecke theta function $f_{\chi_{e;K,\alpha}}(\tau)$ of 
$\chi_{e;K,\alpha}$.
\end{cor}

\begin{anythng}
For a given positive integer $M$, the set of integers 
$\N[M] := \{ N \in \N \; | M|N \}$ forms a directed partially ordered set
where the partial order is set by the divisibility. 
The set of modular curves $\{ X_1(N) \; | \; N \in \N[M]\}$ labeled by this 
directed partially ordered set forms an inverse system, where we use 
the ordinary projection $\pi_{NN'}: X_1(N') \rightarrow X_1(N)$, that is, 
the projection ${\cal H}/\Gamma_1(N') \rightarrow {\cal H}/\Gamma_1(N)$ 
extended to 
$X_1(N')$ and $X_1(N)$, as the morphism of the inverse system for any 
pair $N, N' \in \N[M]$ with $N|N'$. Therefore, for any positive integer $M$, 
we can think of the projective limit of the curves, 
\begin{align}
   \mathop{\varprojlim}\limits_{N \in \N[M] } X_1(N). 
  \label{eq:ShimuraC-M}
\end{align}
We note that all the morphisms $\pi_{NN'}$ are defined over $\Q$ 
(see \cite[\S7]{DS} and \cite[\S11--12]{Ribet-Stein}). 

For an elliptic curve $E/k$ of Shimura type, there is a 
surjective morphism $\nu_N$ in (\ref{eq:surj-fromJac-toBofE})
defined over the field $K$ for any $N \in \N[N_E]$;  
those morphisms satisfy 
\begin{align}
{\rm End}_K(B) \circ \nu_N = {\rm End}_K(B) \circ \nu_{N'} \circ (\pi_{NN'})_*
 \label{eq:comm-as-awhole}
\end{align}
(Props. \ref{props:BK-as-quotient-of-Jac1}). 
Thus, there is a left-$({\rm End}_K(B) \circ)$-coset of surjective morphisms 
\begin{align}
  \nu:  \mathop{\varprojlim}\limits_{N \in \N[N_E]} \left( {\rm Jac}(X_1(N))
  \times_\Q K \right)
      \longrightarrow B  
   \label{eq:map-fromShimuraC-toBKandEk}
\end{align}
defined over $K$.  $\bullet$
\end{anythng}

\subsubsection{
Galois Representations Associated 
with Elliptic Curves of Shimura Type, and the Corresponding Cuspforms}
\label{sssec:Langlands}

In Remark \ref{rmk:triangle-shcematic}, 
it has been stated that there exists 
a morphism $\nu: {\rm Jac}(X_0(N)) \rightarrow E/\Q$ defined 
over $\Q$ for any elliptic curve $E$ defined over $\Q$; 
in the Langlands correspondence, the Galois representation 
$\rho_\ell(E/\Q)$ of ${\rm Gal}(\overline{\Q}/\Q)$ corresponds 
to the automorphic representation $\pi_f$ of ${\rm GL}_2(\mathbb{A}_\Q)$
associated with $f(\tau)d\tau= \mu^* \circ \nu^*(\omega_E)$, where 
$f(\tau)$ can also be regarded as a weight-2 cuspform for $\Gamma_0(N)$.
This correspondence is expressed in the form of 
\begin{align}
  \pi(\rho_\ell(E/\Q)) = \pi_f, \qquad 
  \rho(\pi_f) = \rho_\ell(E/\Q)
\end{align}
using $\pi()$ and $\rho()$.
The $L$-functions of the Galois representation $\rho_\ell(E/\Q)$ and 
the automorphic representation $\pi_{f=(\nu \circ \mu)^*(\omega_E)}$ agree, 
\begin{align}
  L(\rho_\ell(E/\Q),s) := 
  L(H^1_{et}(E\times_\Q \overline{\Q}; \overline{\Q}_\ell),s)  
   = L(\mu^* \circ \nu^*(\omega_E),s) 
  =: L(\pi_f , s) .
\end{align}

There have been attempts of generalizing the statements above in 
multiple different directions.
Section \ref{sssec:Langlands} leaves 
a statement available for elliptic curves of Shimura type $E$ defined 
over a number field $k$ that is an abelian extension of the imaginary 
quadratic field $K$ of complex multiplication. Such $E/k$'s are more 
general than $E/\Q$'s above in that the field of definition $k$ is 
allowed not to be $\Q$, while we have paid the price by restricting 
ourselves to $E$'s that have complex multiplication, and furthermore 
are of Shimura type.  

\begin{anythng}
For an elliptic curve $E/k$ of Shimura type, the representation 
${\rm Ind}^{{\rm Gal}(\overline{\Q}/K)}_{{\rm Gal}(\overline{\Q}/k)}(\rho_\ell(E/k))
$ of the group ${\rm Gal}(\overline{\Q}/K)$ splits into 
the direct sum of $2[k:K]$ distinct 1-dimensional representations 
(Prop.\ \ref{props:GS}). 
They are denoted by $\rho^{(10)}_{e;\alpha=1,\cdots, [T_e:K]}$
and $\rho^{(01)}_{e;\alpha=1,\cdots, [T_e:K]}$. 
If we are to use the notation above, then 
$\pi_{\varphi^{(10)}_{e;\alpha}} = \pi(\rho^{(10)}_{e;\alpha})$ and 
$\pi_{\varphi^{(01)}_{e;\alpha}} = \pi(\rho^{(01)}_{e;\alpha})$, where 
the Galois representations of ${\rm Gal}(\overline{\Q}/K)$ are 
${\rm GL}_1(\overline{\Q}_\ell)$-valued, and the automorphic representations 
$\pi_{\varphi}$'s are those of the group ${\rm GL}_1(\mathbb{A}_K)$.
For any one of $(e;\alpha)$, one can further think of the representation 
of ${\rm Gal}(\overline{\Q}/\Q)$, 
\begin{align}
  \rho_{e;\alpha} :=  
 {\rm Ind}^{{\rm Gal}(\overline{\Q}/\Q)}_{{\rm Gal}(\overline{\Q}/K)} 
        (\rho^{(10)}_{e;\alpha}) \cong  
  {\rm Ind}^{{\rm Gal}(\overline{\Q}/\Q)}_{{\rm Gal}(\overline{\Q}/K)} 
        (\rho^{(01)}_{e;\alpha})  .
  \label{eq:indcd-repr-GalQbarQ-fromGeometric}
\end{align}
This representation of ${\rm Gal}(\overline{\Q}/\Q)$ is 
${\rm GL}_2(\overline{\Q}_\ell)$ valued, and is irreducible 
\cite[p.29 Thm 2.3]{Ribet-Neben}.
The automorphic representation 
$\pi(\rho_{e;\alpha})$ of 
${\rm GL}_2(\mathbb{A}_\Q)$ that corresponds to this Galois representation 
must be associated with some weight-2 cuspform $f$ in general, and 
in this special case, the cuspform in question 
is the Hecke theta function $f_{\varphi^{(10)}_{e;\alpha}}$ associated with 
the Hecke character $\varphi^{(10)}_{e;\alpha}$. 
\begin{align}
   \pi(\rho_{e;\alpha}) = \pi_{f_{\varphi^{(10)}_{e;\alpha}}}  , \qquad 
      {}^\forall e, \alpha = 1,\cdots, [T_e:K].
\end{align}

Just like the cuspforms $f$ with $\pi_f = \pi(\rho_\ell(E/\Q))$ 
are obtained through $\mu^* \circ \nu^*(\omega_E)$ for elliptic curves 
$E$ defined over $\Q$, the Hecke theta functions 
$\{ f_{\varphi^{(10)}_{e;\alpha}} \; | \; (e,\alpha) \}$ of an elliptic 
curve $E/k$ of Shimura type are characterized in terms of the surjective 
morphism $\nu: {\rm Jac}(X_1(N))\times_\Q K \longrightarrow B$ and a (1,0)-form 
$\omega_E$ of $E/k$, as follows. 

First, the representation 
${\rm Ind}^{{\rm Gal}(\overline{\Q}/K)}_{{\rm Gal}(\overline{\Q}/k)}
(\rho_\ell(E/k))$ of the group ${\rm Gal}(\overline{\Q}/K)$ is on 
the dual of the vector space 
${\rm Tate}_\ell(B/K) \otimes_{\Z_\ell} \overline{\Q}_\ell$; 
the representation space is 
\begin{align}
  H^1_{et}(B \times_K \overline{\Q}; \overline{\Q}_\ell) \cong 
  {\rm Hom}_{\Z_\ell} ({\rm Tate}_\ell(B/K), \overline{\Q}_\ell) .
  \label{eq:H1et-as-dual-Tate}
\end{align}
For this representation space 
\begin{align}
  H^1_{et}(B\times_K \overline{\Q}; \overline{\Q}_\ell) \cong H^1(B(\C); \C) 
\end{align}
over the field $\overline{\Q}_\ell \cong \C$, we may first choose a set 
of representatives 
\begin{align}
 \left\{ g_\sigma \in {\rm Gal}(\overline{\Q}/K)  \; | \; 
  [g_\sigma] = \sigma \in {\rm Gal}(k/K) \right\}
\end{align}
of one element $g_\sigma$ for each element $\sigma \in {\rm Gal}(k/K)$, 
and then set 
\begin{align}
  \left\{ \omega_{E^{g_\sigma}} \; | \; \sigma \in {\rm Gal}(k/K) \right\}
\end{align}
as a basis.  Here, we fix one (1,0)-form $\omega_E$ for the original 
elliptic curve $E/k$ of Shimura type, and then determine 
$\omega_{E^{g_\sigma}}$ on $E^\sigma/k$ as the `Galois conjugate'\footnote{
For $g \in {\rm Gal}(\overline{\Q}/K)$, there is an isogeny  
$\phi_g: E \rightarrow E^g$ so that the Galois action of $g$ 
on the torsion points of $E$ is reproduced by $\phi_g$. 
The 1-form $\omega_{E^g}$ in the text is characterized by 
$\phi_g^*(\omega_{E^g}) = \omega_E$.
} 
by $g_\sigma$. 
The basis elements 
$\omega_{E^{g_\sigma}}$'s in (\ref{eq:pullback-omegaEs-bynu}) generate
$H^{1,0}(E^\sigma(\C); \C)$ of all $\sigma \in {\rm Gal}(k/K)$, 
and hence the entire $H^{1,0}(B(\C); \C)$. 

Secondly, due to Prop.~\ref{props:GS} and 
Remark~\ref{rmk:Shimura-type-GalChartzn}, 
the representation 
of ${\rm Gal}(\overline{\Q}/K)$ on the vector 
space (\ref{eq:H1et-as-dual-Tate}) splits into 1-dimensional representations,
$[k:K]$ of which are generated by the linear combinations 
\begin{align}
  \sum_{\sigma \in {\rm Gal}(k/K) } \omega_{E^{g_\sigma}} 
    \left( \rho^{(10)}_{e;\alpha}(g_\sigma) \right)^{-1} \in 
  H^1_{et} (B \times_K \overline{\Q}; \overline{\Q}_\ell) 
  \label{eq:basis-omegaE-diagonalize-Gal}
\end{align}
labeled by $\{ (e,\alpha) \; | \; \alpha = 1,\cdots, [T_e:K] \}$.

Finally, the surjective morphism in Props.~\ref{props:BK-as-quotient-of-Jac1}
identifies 
the vector space (\ref{eq:H1et-as-dual-Tate}) as a subspace of 
$H^1_{et}({\rm Jac}(X_1(N)) \times_\Q \overline{\Q}; \overline{\Q}_\ell)$
for any $N$ divisible by $N_E$ in (\ref{eq:def-NE}).\footnote{
Although the choice of a surjective 
morphism $\nu$ is not unique in that it may well be followed by any 
endomorphism of $B/K$ defined over $K$, the subspace of 
$H^1_{et}({\rm Jac}(X_1(N))\times_\Q \overline{\Q}; \overline{\Q}_\ell)$
identified in this way does not depend on the choice of the morphism $\nu$.
}
The $[k:K]$ 1-forms in (\ref{eq:basis-omegaE-diagonalize-Gal}) are mapped to 
\begin{align}
  \left\{  \left. 
     \nu^* \left(  \sum_{\sigma \in {\rm Gal}(k/K) } \omega_{E^{g_\sigma}} 
      \left( \rho^{(10)}_{e;\alpha}(g_\sigma) \right)^{-1} \right)  \; 
    \right| \; (e, \alpha) \; \alpha = 1,\cdots, [T_e:K] \right\}.
\end{align}
Those $[k:K]$ weight-2 cuspforms of $\Gamma_1(N)$ are the newforms 
$f_{\varphi^{(10)}_{e;\alpha}}$ (except that they may not be properly normalized):
\begin{align}
   f_{\varphi^{(10)}_{e;\alpha}} \propto  
     \mu^* \circ \nu^* \left(  \sum_{\sigma \in {\rm Gal}(k/K) } \omega_{E^{g_\sigma}} 
      \left( \rho ^{(10)}_{e;\alpha}(g_\sigma) \right)^{-1} \right), 
  \qquad {}^\forall e, \alpha = 1,\cdots, [T_e:K] .    
         \label{eq:pullback-omegaEs-bynu}
\end{align}
Those representation spaces can also be regarded as a subspace of 
the direct limit, 
\begin{align}
  \nu^*: H^1_{et}(B\times_K \overline{\Q};\overline{\Q}_\ell) \hookrightarrow
    \mathop{\varinjlim}\limits_{N \in \N[N_E]} \left(   
 H^1_{et}({\rm Jac}(X_1(N))\times_\Q \overline{\Q}; \overline{\Q}_\ell) \right).
  \label{eq:H1-ResE-in-directLim-Jac}
\end{align}
The arrow on the upper right edge of (\ref{eq:schematic-triangle}) meant this. 
$\bullet$
\end{anythng}

\vspace{5mm}

In the argument above, we have defined representations $\rho_{e;\alpha}$ 
of the group ${\rm Gal}(\overline{\Q}/\Q)$ 
through (\ref{eq:indcd-repr-GalQbarQ-fromGeometric}).
While the representations $\rho^{(10)}_{e;\alpha}$ and $\rho^{(01)}_{e;\alpha}$
of ${\rm Gal}(\overline{\Q}/K)$ have construction associated with the 
arithmetic geometry of $B$ defined over $K$, the representations 
$\rho_{e;\alpha}$ of the group ${\rm Gal}(\overline{\Q}/\Q)$ are not 
directly related to a geometric object defined over $\Q$. 
The latter representations are relevant to geometry only through 
an abstract procedure in group representation theory, induced representation.

\begin{anythng}
\label{statmnt:wortmann2}
The representations $\rho_{e;\alpha}$ of ${\rm Gal}(\overline{\Q}/\Q)$, however, 
can be constructed directly from geometry when the elliptic curve of Shimura 
type, $E/k$, satisfies the (S2) condition in \ref{statmnt:wortmann}
(see \cite{Wortmann}). The $[F:\Q]$ distinct 2-dimensional representations 
of ${\rm Gal}(\overline{\Q}/\Q)$ in \ref{statmnt:wortmann} are $\rho_{e;\alpha}$'s.
\end{anythng}

\subsection{Chiral Correlation Functions on 
the Tower of Modular Curves}
\label{ssec:KahlerVsShimuraC}

In \cite{prev.paper}, we have observed that the 
$L$-functions of elliptic curves $E/k$ of Shimura type can be 
reproduced by using as building blocks the ``$(g,n)=(1,2)$ chiral 
correlation functions'' $f_{1\Omega'}^{\rm II}(\tau_{ws};\beta)$'s of a 
rational model of ${\cal N}=(2,2)$ SCFT for $([E]_\C, f_\rho)$, 
when $f_\rho \in \N_{>0}$ is chosen appropriately.  
Having done preparation in sections \ref{sec:lift} 
and \ref{sec:CCF-eigenform}, we are now ready to present a polished-up 
version of that statement in \cite{prev.paper} and discuss implications 
of the observation. 

\begin{anythng}
Let $K$ be an imaginary quadratic field, and $f_\rho \in \N_{>0}$. 
Then we have seen that the vector space of ``$(g,n)=(1,2)$ chiral correlation 
functions'' $F^{(\Omega, \int J)}_{\{0,0\}} ({\cal E}ll({\cal O}_K), f_\rho)$
can be identified with the subspace of $H^{1,0}(X_1(N);\C)$ so long as 
\begin{align}
   (D_K^2 f_\rho^2) = N_{D\Lambda}^2 \; | \; N.
\end{align}
Moreover, the embeddings 
\begin{align}
  \imath_1^{(K,f_\rho;d^2,1)} : 
   F^{(\Omega, \int J)}_{\{0,0\}}({\cal E}ll({\cal O}_K),f_\rho) 
   \hookrightarrow
   H^{1,0}(X_1(N_{D\Lambda}^2 d^2);\C), \qquad \quad d \in \N_{>0}
  \label{eq:embed-ccf4OK-fixKahler-limHX1N-1}
\end{align}
are compatible with the natural 
pull-backs $H^1(X_1(N');\C) \rightarrow H^1(X_1(N' d^2);\C)$, 
so there is a well-defined embedding
\begin{align}
  \imath_1^{(K,f_\rho;*,1)} : 
   F^{(\Omega, \int J)}_{\{0,0\}}({\cal E}ll({\cal O}_K),f_\rho) 
   \hookrightarrow
   \mathop{\varinjlim}\limits_{N \in \N^2[(D_Kf_\rho)^2]} 
         \left( H^{1,0}(X_1(N);\C) \right);
  \label{eq:embed-ccf4OK-fixKahler-limHX1N-2}
\end{align}
here, 
\begin{align}
\N^2[M] := \left\{ N \in \N_{>0} \; | M | N, \; 
   {}^\exists N_0 \in \N_{>0} {\rm ~s.t.~} N_0^2 = N \right\} 
\end{align}
for a positive integer $M$ is a directed partially ordered set (by 
the divisibility). 

As a side remark, we just note one thing. 
Although we have introduced the projective limits (\ref{eq:ShimuraC-M})
labeled by a positive integer $M$ (and the ${\rm Jac}$ version 
in (\ref{eq:map-fromShimuraC-toBKandEk})), and also 
the direct limits (\ref{eq:embed-ccf4OK-fixKahler-limHX1N-2}) 
labeled by $(D_Kf_\rho)^2$ and (\ref{eq:H1-ResE-in-directLim-Jac}) by $N_E$, 
those limits are independent of those $M$'s, $(D_K f_\rho)^2$'s, and $N_E$'s.
For example, take the limits over the directed partially ordered set $\N[M]$. 
Then the limit over $\N[M]$ and the limit over $\N_{}$ are isomorphic; 
for any $N \in \N_{}$, one can find $N' \in \N[M]$ so that $N|N'$, and that 
is enough to prove the isomorphism. 

Now, the set of positive integers $\{ f_\rho \in \N_{>0}\}$ also forms 
a directed partially ordered set (by the divisibility). The vector spaces 
$F^{(\Omega, \int J)}_{\{0,0\}}({\cal E}ll({\cal O}_K),f_\rho)$ forms a direct system 
along with the homomorphism $\imath_1^{(K,f_\rho;d^2,1)}: F^{(\Omega, \int J)}_{\{0,0\}}
     ({\cal E}ll({\cal O}_K),f_\rho) \rightarrow 
  F^{(\Omega, \int J)}_{\{0,0\}}({\cal E}ll({\cal O}_K),f_\rho d)$ over 
$f_\rho \in \N_{}$. So, using the vector space 
\begin{align}
  F^{(\Omega, \int J)}_{\{0,0\}}({\cal E}ll({\cal O}_K)) :=   
 \mathop{\varinjlim}\limits_{f_\rho \in \N} F^{(\Omega,\int J)}_{\{0,0\}}
      ({\cal E}ll({\cal O}_K),f_\rho), 
     \label{eq:def-all-CCF-K}
\end{align}
we can describe the embeddings (\ref{eq:embed-ccf4OK-fixKahler-limHX1N-2}) for 
various $f_\rho$'s as one injective map
\begin{align}
   F^{(\Omega, \int J)}_{\{0,0\}}({\cal E}ll({\cal O}_K)) \hookrightarrow 
       \mathop{\varinjlim}\limits_{N \in \N} \left( 
         H^{1,0}(X_1(N);\C) \right) .
    \label{eq:embed-ccf4OK-allKahler-limHX1N}
\end{align}
The vector subspace introduced here depends on the choice of an imaginary 
quadratic field $K$, and $f_z=1$ implicitly, but not on $f_\rho$. 
Although there is no choice but to specify a K\"{a}hler parameter $f_\rho$ 
in formulating a geometry-target model of rational CFT, such notion as 
``the vector space of chiral correlation functions with $[E_{z}]_\C$ 
as a target space'' (without specifying the metric on $[E_{z}]_\C$) 
can be formulated by taking the direct limit\footnote{
To think of such things as a vector space of (chiral) correlation functions 
for a target space with a given complex structure but without a specific 
choice of K\"{a}hler parameter, we need to introduce some rule of how to 
identify correlation functions of models with different K\"{a}hler parameters. 
The direct limit is one way to do the identification. Remember that, 
for a directed partially ordered set $A$, 
a direct system $\{ \{F_a\}_{a\in A}, \{\phi_{ba} | a \leq b\}\}$
consists of objects $F_a$ for $a\in A$ and one morphism 
$\phi_{ba}: F_a \hookrightarrow F_b$ for $a,b \in A$, $a \leq b$. 
The direct limit $\mathop{\varinjlim}\limits_{a \in A}F_a$ is constructed 
by taking a quotient of $\oplus_{a\in A}F_a$ by relations determined 
by $\{\phi_{ba}\}$. \\
Here is also an alternative idea for identification. 
Think of a system 
$\{\{F_a\}_{a\in A}, \{ M_{ba} | a \leq b\} \}$ instead where 
$M_{ba} \subset {\rm Hom}(F_a,F_b)$; this is a generalization of a direct 
system, 
in that $M_{ba}$ may consist of more than one element of ${\rm Hom}(F_a,F_b)$.
One might require $M_{cb} \circ M_{ba} = M_{ca}$ for $a \leq b \leq c$, for example. 
A limit may be defined by taking a quotient of $\oplus_{a \in A} F_a$ by 
all the relations determined by $\phi_{ba} \in M_{ba}$.  \\
When this idea is applied to the system of 
$\{F_a\} = \{ H^{1,0}(X_1(N);\C) \; | \; N \in \N\}$ and 
$M_{N'N} = \{ |[\diag(r,1)]_2 \; | \; r | (N'/N) \}$, for example, 
the limit will retain only the newforms, not oldforms. 
Similarly, one can think of a limit for the system 
of $\{ F_a \} = \{ F({\cal E}ll({\cal O}_K, f_\rho) \}$ and 
$M_{f'f} = \{ \imath^{(K,f;(f'/f)^2, d'')} \; | \; d''|(f'/f)^2 \}$. 
} 
 with respect 
to the choice of the K\"{a}hler parameter. $\bullet$
\end{anythng}

\begin{thm}
\label{thm:reproduce-all-HeckeTheta-by-CCF}
Let $K$ be an imaginary quadratic field. 
The vector space $F^{(\Omega, \int J)}_{\{0,0\}}({\cal E}ll({\cal O}_K))$
contains all the Hecke theta functions of all the Hecke characters of 
$K^\times \backslash \mathbb{A}^\times_K$ of 
infinity type [-1/2;1,0]. At least some of the $|[\diag(r,1)]_2$ images 
for $r \in \N_{>1}$ may also be contained in this vector space. 
\end{thm}
\begin{proof}
For a given Hecke character $\varphi^{(10)}_K$ of 
$K^\times \backslash \mathbb{A}_K^\times$ of infinity type [-1/2;1,0], 
its Hecke theta function $f_{\varphi^{(10)}_{K}}$ is contained in 
$F^{(\Omega, \int J)}_{\{0,0\}}({\cal E}ll({\cal O}_K),f_\rho)$ if and 
only if $|D_K| {\rm Nm}_{K/\Q}(\mathfrak{c}_f)$ divides $(D_Kf_\rho)^2$;
here $\mathfrak{c}_f$ is the conductor of $\varphi^{(10)}_{K}$. 
Because there are such $f_\rho$'s in the directed partially ordered set 
$\N = \{ f_\rho \in \N \}$, the Hecke theta function $f_{\varphi^{(10)}_K}$ 
is contained in the vector space (\ref{eq:def-all-CCF-K}).
\end{proof}

\begin{thm}
\label{thm:reproduce-LLdual-by-CCF}
Let $E/k$ be an elliptic curve of Shimura type, and $K$ be its 
imaginary quadratic field of complex multiplication. 
The modular forms $f_{{\varphi^{(1,0)}_{e;\alpha}}}$ for 
$\pi \left( \rho_{e;\alpha} \right)$'s
given by (\ref{eq:pullback-omegaEs-bynu}) are always found within 
the vector space of ``the $(g,n)=(1,2)$ chiral correlation functions 
of ${\cal E}ll({\cal O}_K)$''. Put differently, 
\begin{align}
   F^{(\Omega, \int J)}_{\{0,0\}}({\cal E}ll({\cal O}_K)) \supset
    \nu^* \left( H^{1,0}(B;\C) \right),
\end{align}
where both are regarded as subspaces of $\mathop{\varinjlim}\limits_{N \in \N}
 H^{1,0}(X_1(N);\C)$ through (\ref{eq:embed-ccf4OK-allKahler-limHX1N}) 
and (\ref{eq:pullback-omegaEs-bynu}, \ref{eq:H1-ResE-in-directLim-Jac}), 
respectively.   
\end{thm}

The Theorem above is one of the main results in this article that fit 
into the form of mathematical statements (such as non-trivial relations 
among well-defined objects). Observations of importance in physics/string 
theory cannot always be stated in that way, however:
\begin{anythng}
\label{statmnt:fromStrPerspectives}
To find the modular form $f$ that corresponds (via the Langlands correspondence) 
to the Galois representation $\rho_\ell(E/k)$ of an elliptic curve of Shimura 
type $E/k$, the standard story in math has been to find a surjective morphism 
$\nu$ from the Jacobian of a modular curve and work out 
$d\tau f = \mu^* \circ \nu^*(\omega_E)$ 
(see review in sections \ref{sssec:shimuraEC} and 
\ref{sssec:surjective-morph-2BK}).
Theorem \ref{thm:reproduce-LLdual-by-CCF} states that string theory computation 
provides an alternative. 

Instead of surjective morphisms $\nu$ from ${\rm Jac}(X_1(N))$ to $E/k$ 
which pull back the 1-form $\omega_E$, string theory pulls back $\omega_E$ 
in the form of the vertex operator $\Omega=du (\partial_u X^\C)$ by the map 
$\phi: \Sigma_{ws} \rightarrow [E]_\C$. 
A choice of arithmetic model $E/k$ has been thrown away because string theory
deals with the target space only as $E \times_k \C = [E]_\C$, and further 
the map $\phi$ is not required to be holomorphic in string theory. 
The totality of the ``chiral correlation functions'' in string theory that 
result 
from the insertion of the operator $\Omega$---$F^{(\Omega,\int J)}_{\{0,0\}}({\cal E}ll({\cal O}_K))$---yields all the possible modular forms that 
correspond to the Galois representations associated with elliptic curves 
of Shimura type, with complex multiplication by ${\cal O}_K$. 

When it comes to an elliptic curve of Shimura type $E/k$, all of its 
Galois conjugates $E^\sigma$ with $\sigma \in {\rm Gal}(k/K)$ share the 
same $L$-function.
Put differently, the $2[k:K]$ Galois representations that are reviewed in 
Remark \ref{rmk:Shimura-type-GalChartzn} are not just for $E/k$ but 
for all of $\{ E^\sigma/k \; | \; 
\sigma \in {\rm Gal}(k/K)\}$. So, it makes sense now
to find the modular forms that correspond to those Galois representations 
within the vector space of the ``chiral correlation functions'' with any one of 
$\{ E^\sigma \times_k \C = [E^\sigma]_\C \;| \; \sigma \in {\rm Gal}(k/K) \}$
as the target space. If $[E/k]_\C$ belongs to ${\cal E}ll({\cal O}_K)$, 
then\footnote{Theorem \ref{thm:reproduce-LLdual-by-CCF} also states that 
the modular forms that correspond to the Galois representations 
of elliptic curves of Shimura type $E/k$ are found within the same 
vector space $F^{(\Omega,\int J)}_{\{0,0\}}({\cal E}ll({\cal O}_K))$, 
even when $E/k$ (and $[E]_\C$) has complex multiplication by a non-maximal 
order ${\cal O}_{f_z}$. This is not totally an odd thing; the Galois 
representation of an elliptic curve $E/k$ is the same as that of another 
elliptic curve $E'/k$ iff there is an isogeny (not isomorphism) 
$E \sim_{/k} E'$ defined over $k$. So, the vector space 
of chiral correlation functions with ${\cal E}ll({\cal O}_K)$ as the target 
spaces, $F^{(\Omega, \int J)}_{\{0,0\}}({\cal E}ll({\cal O}_K))$, may contain 
the modular forms that correspond to the Galois representation of $E/k$ 
with complex multiplication by ${\cal O}_{f_z}$, {\it if} $E/k$ has a 
$k$-isogenous $E'/k$ whose base change $[E']_\C$ belongs to ${\cal E}ll({\cal O}_K)$. 
}
the base change operation $E^\sigma \longmapsto E^\sigma \times_k \C$
is a projection ${\rm Gal}(k/K) \rightarrow {\rm Gal}(H_K/K) \cong 
{\cal E}ll({\cal O}_K)$, where $H_K$ is the Hilbert class field of $K$.
The vector space $F^{(\Omega, \int J)}_{\{0,0\}}({\cal E}ll({\cal O}_K))$ 
is a collection of the ``chiral correlation functions'' with any one of 
the elliptic curves in ${\cal E}ll({\cal O}_K)$ as the target space 
indeed. $\bullet$
\end{anythng}

Results in this article have already improved the results and theoretical 
understanding in our previous work \cite{prev.paper} 
in various respects, but Ref. \cite{prev.paper} is not entirely 
a proper subset of this article. Here are two remarks. 
\begin{rmk}
Instead of specifying an elliptic curve of Shimura-type $E/k$ first, and then 
asking whether the modular forms 
$\pi({\rm Ind}^{{\rm Gal}(\overline{\Q}/K)}_{{\rm Gal}(\overline{\Q}/k)}(\rho_\ell(E/k)))$
are found within the vector space of ``chiral correlation functions of 
${\cal E}ll({\cal O}_K)$'', we may ask a converse question. Namely, 
we may specify a set of chiral correlation functions of 
a set of CM elliptic curves ${\cal E}ll({\cal O}_K)$ first, and then ask 
if there is any elliptic curve of Shimura-type $E/k$ whose Galois 
representation is $\rho(f)$ of the modular form in this vector space of 
the ``chiral correlation functions''; we may ask what the field of 
definition $k$ 
can be at the same time. This question was surveyed (in the form of a 
pure math subject) in section 4.2 of \cite{prev.paper}. So, we 
do not explore this question further in this article.
\end{rmk}

\begin{rmk}
In our previous paper \cite{prev.paper}, we considered using only the 
chiral correlation functions in $F^{(\Omega, \int J)}_{\{0,0\}}([z],f_\rho)$ 
along with their images under $|[\diag(r,1)]_2$ and $|[\diag(1,s)]_2$ 
for integers $r,s$, in reproducing the modular 
forms (\ref{eq:pullback-omegaEs-bynu}) dual to the Galois representation
${\rm Ind}^{{\rm Gal}(\overline{\Q}/K)}_{{\rm Gal}(\overline{\Q}/k)}(\rho_\ell(E_z/k))$ 
of an elliptic curve of Shimura type. In this article,  
Theorem \ref{thm:reproduce-LLdual-by-CCF} in particular, we do not 
allow to include the images of chiral correlation functions by 
$|[\diag(r,1)]_2$ or $|[\diag(1,s)]_2$ in recovering those modular forms, 
but we allow ourselves to 
use the ``chiral correlation functions'' of rational models of ${\cal N}=(2,2)$ 
SCFT with other elliptic curves $[E_{z_a}]_\C$ in ${\cal E}ll({\cal O}_K)$; 
we take a direct sum over $a = 1,\cdots, h({\cal O}_K)$ 
in (\ref{eq:Fzfrho-chif-sumoverClK}, \ref{eq:notatn-FEllOK-frho}). 
So, this article poses a question a little different from the one 
in \cite{prev.paper}.

From conventional perspectives in physics, it is not a natural idea 
to take a linear combination of correlation functions of multiple 
different models of quantum field theory. The authors---we---were still 
within this physics tradition in our previous paper \cite{prev.paper}.
Although it is a little too artificial and less motivated
rule of game from math perspectives (if not ill-defined) 
to allow including arbitrary 
$|[\diag(r,1)]_2$ and $|[\diag(1,s)]_2$ operations, 
we decided not to be concerned about this artificial aspect too much 
in \cite{prev.paper}. In this article, however, we chose 
not to allow those operations that bring the ``chiral correlation functions'' 
out of $F^{(\Omega, \int J)}_{\{0,0\}}({\cal E}ll({\cal O}_K))$; instead 
we allow to include the ``chiral correlation functions'' of all of $[E_z]_\C 
\in {\cal E}ll({\cal O}_K)$ at the same time; the linear 
combination (\ref{eq:Hecke-theta-div-into-idealClasses}, 
\ref{eq:Hecke-theta-4-idealClass}) with 
$\mathfrak{a}(\mathfrak{K}_a) = \mathfrak{b}_{z_a}$ just does that
(Rmk. \ref{rmk:discussion}).
Math motivation in allowing this has been stated 
in \ref{statmnt:fromStrPerspectives}.
\end{rmk}

As one more comment, 
\begin{rmk}
We have not studied what the vector space $\oplus_{a=1}^{h({\cal O}_{f_z})} 
F^{(\Omega, \int J)}_{\{0,0\}}([z_a],f_\rho)$ is like for the set of models 
for a fixed value of $f_\rho$ and a set of elliptic curves with 
complex multiplication by an order ${\cal O}_{f_z}$ that is not 
maximal ($f_z >1$).  Neither this article nor \cite{prev.paper} has 
tried to characterize the vector space of the ``chiral correlation functions'' 
of {\it non-diagonal} models of rational CFT with an elliptic curve $[E]_\C$ 
with complex multiplication as the target space; by non-diagonal, it is meant 
that the complexified 
K\"{a}hler parameter $\rho \in {\cal O}_{f_z}$ is not of the form 
$\rho = f_\rho a_z z$ labeled by $f_\rho \in \N_{>0}$. 
The task will be to repeat analysis like the one in 
section \ref{sec:CCF-eigenform}, 
but we leave this task beyond the scope of this article. 
(This is partly because the ``chiral correlation functions of 
${\cal E}ll({\cal O}_K)$,'' (\ref{eq:def-all-CCF-K}), is already enough 
for Theorems \ref{thm:reproduce-all-HeckeTheta-by-CCF} 
and \ref{thm:reproduce-LLdual-by-CCF}.)
\end{rmk}

\section{Two Families of Galois Group Actions Associated with a 
CM Elliptic Curve}
\label{sec:two-Gal-actn}

In this section, we wish to share an observation (\ref{statmnt:GTtheory-final})
with readers. 
The observation is presumably not particularly striking, 
or not of immediate consequences perhaps (at least to the minds of the authors).
Parts of the observation written down in the following have been 
written or discussed in various literatures, but not in an entire form 
anywhere (to the knowledge of the authors); the observation emerges 
naturally after we understand the material in the preceding sections. 
For this reason, the authors consider that it is worthwhile to have 
this section as a part of this article. Materials in 
section \ref{ssec:Grothendieck} comes mostly from Ref. \cite{Schneps-survey} 
and references therein, while those in \ref{statmnt:GT-quotient} mostly 
from \cite{CG, Degiovanni-french}.

\subsection{Galois Group and Etale Coverings of ${\cal M}_{g,n}$}
\label{ssec:Grothendieck}

There are infinitely many number fields, and they form an intricate 
network under the relation of one number field being an extension of 
another. The absolute Galois group $G_\Q := 
{\rm Gal}(\overline{\Q}/\Q)$ contains the whole information of this 
intricate network; a better understanding of this intricate network 
is one of big dreams in number theory. A quite common strategy 
in understanding the structure of a complicated group is to study 
representations of the group.   

\begin{anythng}
An idea of Galois--Teichm\"{u}ller theory \cite{Ihara-86, Esquisse}
is to use a fact\footnote{When $X$ is a non-singular variety, 
then $\pi_1^{\rm alg}(X \times_L \overline{\Q};\bar{x}) \cong 
\hat{\pi}_1^{\rm top}(X(\C);\bar{x})$, where $\pi_1^{\rm top}(X(\C))$ is the 
topological fundamental group \cite{Milne-etale}. } that 
\begin{align}
  1 \rightarrow \pi_1^{\rm alg}(X \times_L \overline{\Q} ;\bar{x})
    \rightarrow \pi_1^{\rm alg}(X, \bar{x}) 
   \rightarrow {\rm Gal}(\overline{\Q}/L) \rightarrow 1
\end{align}
is exact; here, $X$ is a variety or stack defined over a number field $L$, 
and $\bar{x}$ a geometric point of $X$. $\pi_1^{\rm alg}$ stands for 
the etale fundamental group (e.g., \cite{Milne-etale}), and $\hat{G}$ 
is the profinite completion of a group $G$ (see \cite{Oort} for def). 
From the fact that this short sequence is exact, it follows that 
there is a homomorphism 
\begin{align}
  \phi_X : {\rm Gal}(\overline{\Q}/L) \longrightarrow {\rm Out}
   \left( \pi_1^{\rm alg}(X \times_L \overline{\Q} ; \bar{x} ) \right);
  \label{eq:hom-GalQbarOvL-2-OutPi1Top}
\end{align}
for $\sigma \in {\rm Gal}(\overline{\Q}/L)$, $\phi_X(\sigma) \in 
{\rm Out}(\pi_1^{\rm alg}(X \times_L \overline{\Q},\bar{x}))$ 
is realized by 
\begin{align}
  \phi_X(\sigma): \pi_1^{\rm alg}(X \times_L \overline{\Q}) \ni \gamma
  \longmapsto 
  \gamma^\sigma := 
  \tilde{\sigma} \cdot \gamma \cdot (\tilde{\sigma})^{-1} \in 
  \pi_1^{\rm alg}(X \times_L \overline{\Q})
     \subset \pi_1^{\rm alg}(X;\bar{x}),
\end{align}
where $\tilde{\sigma} \in \pi_1^{\rm alg}(X;\bar{x})$ is a lift 
of $\sigma \in {\rm Gal}(\overline{\Q}/L)$, as usual. 

When the geometric point $\bar{x}$ of $X$ lies over a point $x \in X$ rational over some Galois extension field $k$ over $L$, there exists 
(see \cite{Matsumoto}) a section (lift homomorphism)
${\rm Gal}(\overline{\Q}/k) \rightarrow \pi_1^{\rm alg}(X,\bar{x})$, 
from which one can determine a homomorphism 
\begin{align}
  \phi_{X,\bar{x}}: {\rm Gal}(\overline{\Q}/k) \longrightarrow 
  {\rm Aut}(\pi_1^{\rm alg}(X \times_L \overline{\Q}; \bar{x}));
   \label{eq:hom-GalQbarOvk-2-AutPi1Top}
\end{align}
further quotient ${\rm Aut}(\pi_1^{\rm alg}(X \times_L \overline{\Q})) 
\rightarrow {\rm Out}(\pi_1^{\rm alg}(X \times_L \overline{\Q}))$ by 
the subgroup ${\rm Int}(\pi_1^{\rm alg}(X\times_L \overline{\Q}))$ 
brings this back to the homomorphism (\ref{eq:hom-GalQbarOvL-2-OutPi1Top}).  
 $\bullet$
\end{anythng}

\begin{exmpl}
The simplest example of this construction is to take $X = \mathbb{G}_m 
= \P^1\backslash \{ 0, \infty\}$ with the field of definition $L=\Q$. 
Then $X(\C) = \C^\times$, $\pi_1^{\rm top}(X(\C)) \cong \Z$, and 
$\hat{\pi}_1^{\rm top}(X (\C);\bar{x}) \cong \hat{\Z}$. 
The representation $\phi$ in this example is the cyclotomic character, 
$\phi = \chi: {\rm Gal}(\overline{\Q}/\Q) \rightarrow {\rm Gal}(\Q^{ab}/\Q) 
\cong \hat{\Z}^\times$, which acts on $\hat{\pi}_1^{\rm top}(X(\C)) \cong 
\hat{\Z}$ by multiplication (eg \cite{Oort}). 

The moduli space of $(g,n)=(0,4)$ pointed Riemann surfaces, ${\cal M}_{0,4}$, 
is isomorphic to $\P^1 \backslash \{0,1,\infty\}$, and we can 
take its field of definition to be $L=\Q$. Another example is 
${\cal M}_{1,1}$ which is a stack defined over $L=\Q$. 
More generally, we can think of $X = {\cal M}_{g,n}$ defined over $L = \Q$
for all $(g,n)$ (e.g., \cite[p.200]{Schneps-survey}), and 
a geometric point $\bar{x}$ over a tangential base point---a maximal 
degeneration 
limit of a $(g,n)$ curve---can be used as the base point, where $k=\Q$ as well
\cite{tng-base-pt}. 
So, there exists an exterior representation $\phi_{X,\bar{x}}$ of 
${\rm Gal}(\overline{\Q}/\Q)$---(\ref{eq:hom-GalQbarOvk-2-AutPi1Top})---on 
the group $\hat{\pi}_1^{\rm alg}({\cal M}_{g,n} \times_\Q \overline{\Q})$ for 
$X={\cal M}_{g,n}$. 
The group $\pi_1^{\rm alg}({\cal M}_{g,n} \times_\Q \overline{\Q})$ is the 
profinite completion of the mapping class group $\Gamma_{g,n}$ of a 
genus $g$ and $n$-pointed Riemann surface.  $\bullet$
\end{exmpl}

\begin{anythng}
Both of the groups ${\rm Gal}(\overline{\Q}/L)$ and 
$\pi_1^{\rm alg}(X_L; \bar{x})$ are infinitely large, but the 
exterior representation $\phi_{X,\bar{x}}$ for $(X,\bar{x})$ defined over $L$
can be studied in the form of infinitely many representations 
$\phi_\Gamma :{\rm Gal}(\overline{\Q}/L) \rightarrow 
{\rm Aut}(\pi_1^{\rm alg}(X\times_L \overline{\Q})/\Gamma)$ associated 
with normal subgroups $\Gamma$ of $\pi_1^{\rm alg}(X\times_L \overline{\Q})$ 
of finite index. Choosing a normal subgroup $\Gamma$ of finite index is 
equivalent to choosing a finite etale Galois cover 
$\pi: Y_{L_0} \rightarrow X_L$ with 
$\pi_1^{\rm alg}(Y \times_{L_0} \overline{\Q}) \cong \Gamma$; the field of 
definition $L_0$ of $Y$ and $\pi$ is an extension of $L$ here. 
The quotient group 
$\pi_1^{\rm alg}(X \times_L \overline{\Q})/
\pi_1^{\rm alg}(Y \times_{L_0} \overline{\Q})$ can be regarded as the 
group of covering transformation of $\pi: Y\times_{L_0} \overline{\Q} 
\rightarrow X\times_L \overline{\Q}$, denoted by 
${\rm Aut}_{X \times_L \overline{\Q}}(Y\times_{L_0} \overline{\Q})$, and is also 
in one-to-one with the finite set of points of the fiber of 
$\pi: Y(\C) \rightarrow X(\C)$ over $\bar{x} \in X(\C)$.
Now, we have an exact sequence of finite groups, 
\begin{align}
1   \rightarrow
\pi_1^{\rm alg}(X\times_L \overline{\Q}) / 
   \pi_1^{\rm alg}(Y\times_{L_0}\overline{\Q})   \rightarrow
\pi_1^{\rm alg}(X_L)/\pi_1^{\rm alg}(Y \times_{L_0} L')   \rightarrow 
{\rm Gal}(L'/L) \rightarrow 1, 
 \label{eq:short-ex-seq-Gal-fin-etaleCov}
\end{align}
where $L'$ is any finite Galois extension over $L_0 \supset L$ that contains 
the number field $L_{Y/X}$ that is the normal closure of the field of moduli 
of all the covering transformations in ${\rm Aut}_{X}(Y_{L_0})$.
Applying the same logic as before, we have a homomorphism 
\begin{align}
  \phi_{(X,\bar{x})/(Y,\bar{y})}: {\rm Gal}(\overline{\Q}/L) \rightarrow 
   {\rm Gal}(L'/L) \rightarrow 
   {\rm Aut}\left( \pi_1^{\rm alg}(X \times_L \overline{\Q}) / 
                   \pi_1^{\rm alg}(Y \times_{L_0} \overline{\Q})
            \right);  
  \label{eq:ext-repr-fin-etaleCov}
\end{align}
this exterior representation factors through ${\rm Gal}(L_{Y/X}/L)$; 
a choice of the number field $L'$ above is irrelevant in the end. 
See \cite{Oort} or \cite{Milne-etale}, for example, for more information. 
$\bullet$
\end{anythng}

\begin{exmpl}
\label{exmpl:Gal-actn-on-XN}
In the case $X = {\cal M}_{1,1}$ and $\bar{x}$ is the tangential base 
point\footnote{
This base point is placed at the $\tau = i\infty$ limit of ${\cal H}$
approached from the positive imaginary axis; a path comes out 
from the $\tau = i\infty$ base point to the direction of $\tau = i$. 
In the algebraic coordinate of ${\cal M}_{1,1}$, $1/j$, it is 
$\overrightarrow{01}$ pointing from $1/j=0$ to the direction of 
$1/j=1/1728$. In the coordinate $q$, it is also pointing from $q=0$ to 
the direction of $q \in \R_{>0}$.}
 $\tau = \overrightarrow{(i\infty)i}$, 
for example, we can choose 
$Y = {\cal H}/\Gamma(N)$ for a positive integer $N$. The number fields 
in the discussion above are $L=\Q$, and $L_0 = L_{Y/X} = \Q(\zeta_N)$ 
\cite[\S7]{DS}. The exact sequence (\ref{eq:short-ex-seq-Gal-fin-etaleCov}) 
is 
\begin{align}
 1 \rightarrow {\rm SL}(2;\Z)/\Gamma(N) \rightarrow 
   {\rm GL}(2;\Z/N\Z) \rightarrow {\rm Gal}(\Q(\zeta_N)/\Q) \rightarrow 1
\end{align}
when we set $L'$ to be the minimal possible one, and we have an exterior 
representation 
\begin{align}
 \phi_{{\cal M}_{1,1}/ {\cal H}/\Gamma(N)}: {\rm Gal}(\Q(\zeta_N)/\Q) \rightarrow 
    {\rm Aut}\left( {\rm SL}(2;\Z/N\Z) \right).
  \label{eq:hom-quotnt-GalCyclo-2-SL2ZNZ}
\end{align}
\end{exmpl} 

\begin{anythng}
\label{statmnt:how2workout-quot-repr}
There are two different (and equivalent) ways to see how Galois 
transformation acts on the group $\pi_1^{\rm alg}(X \times_L \overline{\Q}) / 
\pi_1^{\rm alg}(Y \times_{L_0} \overline{\Q})$ through the exterior 
representation $\phi_{(X,\bar{x})/(Y,\bar{y})}$ in (\ref{eq:ext-repr-fin-etaleCov}). 
Suppose that the base point $\bar{x}$ is lying over a point $x \in X_L$ 
rational over $L$. Then the representation $\phi_{(X,\bar{x})/(Y,\bar{y})}(\sigma): 
\gamma \mapsto  \tilde{\sigma}^{-1} \cdot \gamma \cdot \tilde{\sigma}^{-1}$ 
(where $\gamma \in \pi_1^{\rm alg}(X\times_L \overline{\Q}) / 
\pi_1^{\rm alg}(Y \times_{L_0} \overline{\Q})$) can be understood 
in a highly intuitive (geometric) way, a combination of Galois transformation 
on geometric points in the fiber of $\bar{x}$ and tracking of those points 
along a path in $X(\C)$. 
See \cite[``well-known facts'' in \S 2.8]{Schneps-PS}.

The alternative (and equivalent) way to understand the action of the Galois 
group is to keep track of functions on $Y$ instead of points on $Y$. 
Rational functions on $Y$ are expanded into power series of a local coordinate 
$t_X$ in $X$ around $\bar{x}$ with coefficients in $L_{Y/X}$ 
(called Puiseux series);\footnote{Fractional powers of 
$t_X$ need to be allowed when the closure of $Y$ is ramified over the closure 
of $X$ at $\bar{x}$.} the lift of $\sigma \in {\rm Gal}(L'/L)$, 
$\tilde{\sigma}$, simply acts as Galois transformation on all the coefficients 
of the Puiseux series expansion, and $\gamma$ acts on those functions by 
their analytic continuation along the path $\gamma$ in $X(\C)$.
For a test function $f \in L'(Y)$, 
\begin{align}
   f \longmapsto  (\gamma \circ (f)^{\tilde{\sigma}^{-1}})^{\tilde{\sigma}} =: 
    \gamma^\sigma \circ f,
  \label{eq:algorithm-EmsLck}
\end{align}
and a path $\gamma^\sigma$ should be found from 
$\pi_1^{\rm alg}(X \times_L \overline{\Q}) /
 \pi_1^{\rm alg}(Y \times_{L_0} \overline{\Q})$ so that the relation above 
should hold for any test function. 
This algorithm determines the element $\gamma^\sigma$ in 
$\phi_{(X,\bar{x})/(Y,\bar{y})}: \gamma \longmapsto \gamma^\sigma$
\cite{I-200GT, Ems-Lck}. 

The same material as in this section \ref{ssec:Grothendieck} have been 
reviewed in literatures addressed toward string theorists (at least 
as much as to mathematicians) already;  
see \cite{Degiovanni-french, Degiovanni-200} and \cite[\S6.3]{Gannon}.
$\bullet$
\end{anythng}

\subsection{Galois Group Action on Modular Curves and Monodromy Representations}
\label{ssec:CG-D} 

\begin{anythng}
Just like a representation of a group $\rho: G \rightarrow {\rm Aut}(V)$ 
on a vector space $V$ has a sub-representation when there is a vector subspace 
$W \subset V$ where any one of $\rho(G)$ maps $W$ to itself, 
the exterior representation of a group $\phi: G \rightarrow {\rm Aut}(\Gamma)$
on a group $\Gamma$ has a sub-representation when there is a subgroup 
$H \subset \Gamma$ where any one of $\phi(G)$ maps $H$ to itself. 
In the case $G = {\rm Gal}(\overline{\Q}/\Q)$, 
$\Gamma = \pi_1^{\rm alg}({\cal M}_{1,1} \times_\Q \overline{\Q})$, and 
$\phi = \phi_{{\cal M}_{1,1}, \overrightarrow{(i\infty)i}}$, a subgroup 
$H = \pi_1^{\rm alg}( ({\cal H}/\Gamma_1(M))_\Q \times_\Q \overline{\Q})$ 
has that property for arbitrary positive integer $M$ because the modular 
curve $({\cal H}/\Gamma_1(M))_\Q$ and a natural projection 
$\pi: {\cal H}/\Gamma_1(M) \rightarrow {\cal M}_{1,1}$ are defined 
over $\Q$ (\cite[\S7]{DS} and \cite[\S11--12]{Ribet-Stein}).

The sub-representation $\phi_{X,\bar{x}}: {\rm Gal}(\overline{\Q}/\Q) \rightarrow 
{\rm Aut}(H)$ with $H = \pi_1^{\rm alg}(Y_1 \times_\Q \overline{\Q})$ and  
$Y_1 = ({\cal H}/\Gamma_1(M))_\Q$ also induces a representation 
on the module, 
\begin{align}
 {\rm Hom}(H, \Z/(\ell^m)) = {\rm Hom}(H/[H,H], \Z/(\ell^m)) 
  \cong H^1_{et}(Y_1 \times_\Q \overline{\Q}; \Z/(\ell^m)),
\end{align}
and also a representation on the module 
\begin{align}
  \mathop{\varprojlim}\limits_{m \in \N} 
      H^1_{et}(Y_1 \times_\Q \overline{\Q}, \Z/(\ell^m)) = 
  H^1_{et}(Y_1 \times_\Q \overline{\Q}; \Z_\ell ).
\end{align}
The compactified modular curve $X_1(M)$ contains 
$Y_1(M) = {\cal H}/\Gamma_1(M)$ as its Zariski open subvariety, and 
\begin{align}
 H^1_{et}(X_1(M)\times_\Q \overline{\Q}; \Z_\ell) \cong 
   H^1_{et}({\rm Jac}(X_1(M)) \times_\Q \overline{\Q}; \Z_\ell)
  \label{eq:et-coh-modCurv-X0}
\end{align}
(the dual of the Tate module of ${\rm Jac}(X_1(M))$) is regarded 
as a submodule of $H^1_{et}(Y_1(M)\times_\Q \overline{\Q}; \Z_\ell)$
that is preserved by any one of 
$\phi_{{\cal M}_{1,1},\overrightarrow{(i\infty)i}}({\rm Gal}(\overline{\Q}/\Q))$.

Therefore, the set of $[k:K]$ representations of 
${\rm Gal}(\overline{\Q}/\Q)$ associated with an elliptic curve 
$E/k$ of Shimura type (reviewed in section \ref{sssec:Langlands})
correspond to some appropriate submodules of 
$H^1_{et}(X_1(M)\times_\Q \overline{\Q}; \Z_\ell)$ above, and hence 
are associated with the sub-representation of 
$\phi_{{\cal M}_{1,1},\overrightarrow{(i\infty)i}}: {\rm Gal}(\overline{\Q}/\Q) 
\rightarrow {\rm Aut}(\pi_1^{\rm alg}({\cal M}_{1,1}\times_\Q \overline{\Q};
 \overrightarrow{(i\infty)i}))$ in (\ref{eq:hom-GalQbarOvk-2-AutPi1Top}) 
corresponding to the subgroup 
$\pi_1^{\rm alg}(({\cal H}/\Gamma_1(M))_\Q \times_\Q \overline{\Q})$ of 
$\pi_1^{\rm alg}({\cal M}_{1,1} \times_\Q \overline{\Q})$. In the perspective 
of string theory, one may well say that the set of the ``chiral correlation 
functions'' $\{ f^{\rm II}_{1\Omega'}(\tau_{ws}; \beta) \}_{\beta \in iReps}$ for 
the data $([E_{z_a}]_\C, f_\rho)$ with $[E_{z_a}]_\C \in {\cal E}ll({\cal O}_K)$
facilitate a part of the sub-representations 
${\rm Gal}(\overline{\Q}/\Q) \rightarrow {\rm Aut}(H)$.  $\bullet$
\end{anythng}

\begin{anythng}
\label{statmnt:GT-quotient}
Just like a representation of a group on a vector space $\rho: G \rightarrow 
{\rm Aut}(V)$ with a non-trivial sub-representation 
$\rho: G \rightarrow {\rm Aut}(W)$ with $W \subset V$ is accompanied with 
the quotient representation $\rho: G \rightarrow {\rm Aut}(V/W)$, 
we can think of a quotient exterior representation for 
$\phi_{X,\bar{x}}: {\rm Gal}(\overline{\Q}/L) \rightarrow 
{\rm Aut}(\pi_1^{\rm alg}(X \times_L \overline{\Q}; \bar{x}))$, 
using a finite etale Galois cover $\pi: Y \rightarrow X$. 
The homomorphism (\ref{eq:ext-repr-fin-etaleCov}) does that; note 
that the field of definition $L_0$ of the finite cover $Y$ over $X$ 
can be larger than the field of definition $L$ of $(X,\bar{x})$.

For $X = {\cal M}_{1,1}$ and $\bar{x} = \overrightarrow{(i\infty)i}$ defined 
over $L=\Q$ and its finite etale and Galois cover 
$Y:= {\cal H}/\Gamma(N) \rightarrow {\cal M}_{1,1}$, 
(\ref{eq:hom-quotnt-GalCyclo-2-SL2ZNZ}) is such a quotient representation.

Just like the ``chiral correlation functions'' of Type II string theory 
with the target spaces $\{ ([E_{z_a}]_\C, f_\rho) \; | \; 
[E_{z_a}]_\C \in {\cal E}ll({\cal O}_K) \}$ constitute a part of the 
module (\ref{eq:et-coh-modCurv-X0}) of the sub-action 
$\phi_{X,\bar{x}}: {\rm Gal}(\overline{\Q}/\Q) \rightarrow 
{\rm Aut}(\pi_1^{\rm alg}(Y_1(M)\times_\Q \overline{\Q}))$, 
chiral correlation functions with an elliptic curve of complex multiplication 
as a target space also provide some of test functions 
(see \ref{statmnt:how2workout-quot-repr}) in working out the quotient 
representation $\phi_{(X,\bar{x})/(Y,\bar{y})}: {\rm Gal}(\overline{\Q}/\Q)
 \rightarrow {\rm Gal}(\Q(\zeta_N)/\Q) \rightarrow 
{\rm Aut}({\rm SL}(2;\Z/N\Z))$. 

To be more specific, think of the chiral correlation functions 
$\{ f^{\rm bos}_0(\tau_{ws};\beta) \; | \; \beta \in iReps. \}$ obtained 
in bosonic string theory for the data $([E_z]_\C, f_\rho)$; they are 
functions of $\tau_{ws} \in {\cal H}$, 
and are also regarded as meromorphic functions on $Y(N_{\rm bos}) := 
{\cal H}/\Gamma(N_{\rm bos})$
where $N_{\rm bos} = {\rm LCM}(12, N_{D\Lambda})$ as we have seen in Prop. 
\ref{props:mndrmy-kernel-bos-str}. They are genuine functions rather than 
a section of such bundles as $(T^*Y(N_{\rm bos}))^{\otimes h}$ over the curve 
$Y(N_{\rm bos})$ for some $h \in \N_{>0}$, because 
$\{f^{\rm bos}_0(\tau_{ws};\beta)\}$ is a weight-zero modular form (when 
they may well be regarded as 
$(g,n)=(1,1)$ chiral correlation functions, but the inserted 
operator, {\bf 1}, is a Virasoro primary operator of weight $h=0$).
The Fourier series expansion (power series in $q_{ws} = e^{2\pi i \tau_{ws}}$) 
of those chiral correlation functions $f^{\rm bos}_0$ is much the same as the 
Puiseux series expansion of rational functions on the finite etale Galois 
cover $Y(N_{\rm bos}) \rightarrow {\cal M}_{1,1}$ at the tangential base point 
$\bar{x}= \{ \tau_{ws}= \overrightarrow{(i\infty)i} \}$; certainly it is 
$1/j(\tau_{ws})$ rather than $q_{ws}$ that should be used for the Puiseux 
series expansion (because that is the algebraic coordinate of ${\cal M}_{1,1}$), 
but the $q_{ws}$-series expansion $1/j(\tau_{ws}) \simeq q_{ws} + 
(\cdots) q_{ws}^2 + \cdots$ only involves coefficients in $\Q$, 
so it does not matter whether we use $q_{ws}$ or $1/j(\tau_{ws})$ for 
the Puiseux series expansion in working out $\gamma^\sigma$ for 
$\gamma \in \pi_1^{\rm alg}({\cal M}_{1,1}\times_\Q \overline{\Q}) / 
\pi_1^{\rm alg}(Y(N_{\rm bos}) \times_{\Q(\zeta_{N_{\rm bos}})}\overline{\Q}) = 
{\rm SL}(2;\Z/N\Z)$
through the algorithm (\ref{eq:algorithm-EmsLck}).

The action of ${\rm Gal}(\overline{\Q}/\Q)$ in (\ref{eq:algorithm-EmsLck}) 
is, in the end, the same as carrying out the Galois transformation 
of the individual matrix entries of the monodromy representation 
of $\gamma$ in the mapping class group $\Gamma_{(g,n)=(1,1)} = {\rm SL}(2;\Z)$
associated with the\footnote{
There can be multiple $(g,n)=(1,1)$ conformal blocks of a given rational 
model of CFT, ${\cal F}_{\{\alpha_1\}}$, depending on which irreducible 
representation is to be inserted at $n=1$ point in a $g=1$ Riemann surface. 
In the case of a $T^2$-target rational model of CFT, however, $(g,n)=(1,1)$
conformal blocks are non-empty only when $\alpha_1 = 0 \in iReps$, 
so there is just one conformal block ${\cal F}_{\{\alpha_1=0\}}$ of interest.
}\raisebox{4pt}{,}\footnote{
For $(g,n)=(1,1)$, conformal blocks give rise to monodromy 
representation of the group $\Z . \Gamma_{1,1}$ rather than $\Gamma_{1,1}$. 
In the case of a $(g,n)=(1,1)$ conformal block with $\alpha_1 = 0$, however, 
the monodromy representation reduces to that of $\Gamma_{1,1}$ because 
the conformal weights in the vacuum representation $\alpha_1=0$ are 
all integers.} 
$(g,n)=(1,1)$ conformal block of the rational model 
of CFT for $([E_z]_\C ,f_\rho)$. 
This is because the Puiseux series coefficients of $f$ get Galois 
transformation once by $\sigma^{-1}$ in (\ref{eq:algorithm-EmsLck}) and 
then once by $\sigma$; linear combination coefficients that emerge 
as a result of monodromy along the path $\gamma$, however, is subject 
to the Galois transformation by $\sigma$, and hence the 
Galois transformation acts only on the linear combination coefficients 
describing the monodromy along $\gamma$ \cite[\S5.2]{Degiovanni-french}. 

The individual entries of the monodromy representation matrices of 
${\rm SL}(2;\Z) \cong \Gamma_{1,1}$ should be within the number field\footnote{
The minimal number field that contains the entries of the monodromy 
representation of ${\rm SL}(2;\Z)$ have been studied for the 
Wess--Zumino--Witten model \cite{CG} and $S^1$-target rational models of 
CFT \cite{Degiovanni-S1, Coste-Degiovanni}.
}  
$\Q(\zeta_{N_{\rm bos}})$, because the quotient representation factors through 
${\rm Gal}(\Q(\zeta_{N_{\rm bos}})/\Q)$ (Example \ref{exmpl:Gal-actn-on-XN}). 
There is indeed an argument \cite{dBG, CG} that those matrix 
entries should be in $\Q^{ab}$ (the maximal cyclotomic extension), and moreover, 
it is known that they are indeed in $\Q(\zeta_{N_{\rm bos}})$ 
\cite[just above Lemma 5.1]{Bruinier}. $\bullet$
\end{anythng}

\begin{anythng}
\label{statmnt:GTtheory-final}
By bringing together all the pieces of discussions above (which are in 
the literature), and presenting in one common narrative, we wish i) 
to write down an observation that (a) the Galois group representations 
on $H^1_{et}$ of the modular curves
and (b) the Galois group action on the covering transformation groups 
of finite etale coverings of ${\cal M}_{1,1}$ can be regarded as partial 
information of one and the same thing, the action of 
${\rm Gal}(\overline{\Q}/\Q)$ on $\pi_1^{\rm alg}({\cal M}_{1,1} 
\times_\Q \overline{\Q})$ given in (\ref{eq:hom-GalQbarOvk-2-AutPi1Top}).

The other observation, ii), in the collection of discussions in this section
is that the group $\pi_1^{\rm alg}({\cal M}_{1,1} \times \overline{\Q})$ is 
split into the quotient $\pi_1^{\rm alg}({\cal M}_{1,1}\times \overline{\Q}) / 
\pi_1^{\rm alg}(Y(N_{\rm bos}) \times_{\Q(\zeta_{N_{\rm bos}})} \overline{\Q})$ and 
the subgroup whose abelianization is dual to the etale cohomology group 
of the modular curve,\footnote{There is a discrepancy between $f^{\rm bos}_0$ 
in bosonic string and $f^{\rm II}_{1\Omega'}$ in superstring for the choice 
of the level, of the modular curves in question.
} once a set of data $\{ ([E_{z_a}]_\C, f_\rho) \; | \; 
[E_{z_a}]_\C \in {\cal E}ll({\cal O}_K) \}$ is specified,
depending on the kernel of the monodromy representation of 
${\rm SL}(2;\Z) \cong \Gamma_{1,1}$.
For a fixed set ${\cal E}ll({\cal O}_K)$ of elliptic curves with complex 
multiplication by ${\cal O}_K$, the K\"{a}hler parameter 
$\{ f_\rho \in \N_{>0} \}$ plays the role of the index set in 
promoting\footnote{
We (the authors) have a vague memory 
that we might have seen in a literature an idea that the 
K\"{a}hler parameter plays the role of realizing Galois action on 
the \underline{profinite} group, not just on a finite quotient group. 
As we are posting this article to arXiv, however, we have not been 
able to dig out such a paper from our desktop.
 } 
the action of ${\rm Gal}(\overline{\Q}/\Q)$ 
on ${\rm SL}(2;\Z/N\Z) = \pi_1^{\rm alg}({\cal M}_{1,1} \times_\Q \overline{\Q})
 / \pi_1^{\rm alg}(Y(N_{\rm bos}) \times_{\Q(\zeta_{N_{\rm bos}})} \overline{\Q})$ into the 
action of the profinite group 
$\pi_1^{\rm alg}({\cal M}_{1,1} \times_\Q \overline{\Q})$.
The K\"{a}hler parameter also plays the role of the index set 
in section \ref{ssec:KahlerVsShimuraC} in packing the ``chiral correlation 
functions'' of the models with complex multiplication by ${\cal O}_K$ 
together to form the vector space 
$F^{(\Omega, \int J)}_{\{0,0\}}({\cal E}ll({\cal O}_K))$. 
Such an observation on the role played by the K\"{a}hler parameter of 
the string theory target space is also a reason of existence of this section, 
and is also an answer to a question posed at Introduction of our previous 
article \cite{prev.paper}.  $\bullet$
\end{anythng}

 \subsection*{Acknowledgments}

The authors are grateful to Scott Carnahan and Yuji Tachikawa for 
useful comments and discussions.
This work is support in part by the World Premier International Research 
Center Initiative (WPI) and    
Grant-in-Aid New Area no. 6003, MEXT, Japan.  


%


\begin{thebibliography}{99}
%
\bibitem[Rsector]{Rsectr-SVOA}
%
C. Dong, H. Li, and G. Mason, 
``Modular invariance of trace functions in orbifold theory and generalized Moonshine,'' 
Comm. Math. Phys. {\bf 214} 1--56 (2000); 
%
C. Dong and Z. Zhao, 
``Modularity in orbifold theory for vertex operator superalgebras,'' 
Comm. Math. Phys. {\bf 260} 227--256 (2005); 
%
J. van Ekeren 
``Modular invariance for twisted modules over a vertex operator superalgebra,'' 
Comm. Math. Phys. {\bf 322} (2013) 333--371, 
[arXiv:1111.0682 [math.RT]].
%
\bibitem[BK]{BK}
%
B. Bakalov and A. Kirillov, ``Lectures on tensor categories and modular functor,'' University Lecture Series, 21. AMS, 2000.
%
\bibitem[Br]{Bruinier}
%
J. H. Bruinier, ``Borcherds products on O(2,l) and Chern classes of Heegner divisors,'' Springer, 2004.
%
\bibitem[BCLDB]{Coste-Degiovanni}
%
E.~Buffenoir, A.~Coste, J.~Lascoux, P.~Degiovanni, and A.~Buhot, 
``Precise study of some number fields and Galois actions occurring in conformal field theory,''Annales de l'I.H.P., {\bf A63} (1995) pp. 41--79.
%

%
\bibitem[CG]{CG}
%
A.~Coste and T.~Gannon,
  ``Remarks on Galois symmetry in rational conformal field theories,''
  Phys.\ Lett.\ B {\bf 323} (1994) 316.
%
\bibitem[Crem]{Cremona}
%
J. E. Cremona, ``Algorithms for modular elliptic curves,'' Cambridge U. Press, 1992. A web interface is also available. 
%
\bibitem[dBG]{dBG}
%
J.~ de Boer and J.~Goeree, ``Markov traces and ${\rm II}_1$ factors in conformal field theory,'' Comm. Math. Phys. {\bf 139} (1991) pp.267--304.
%
%
\bibitem[Deg-S1]{Degiovanni-S1}
%
P. Degiovanni, 
``Z/NZ conformal field theories,'' 
 Comm. Math. Phys. {\bf 127} (1990) pp.71--99.
%
\bibitem[Deg-eMS]{Degiovanni-french}
%
P. Degiovanni, 
 ``Equations de Moore et Seiberg, Theories Topologiques et Theorie de Galois,''
  Helv.\ Phys.\ Acta {\bf 67} (1994) 799. 
%
\bibitem[Deg-200]{Degiovanni-200}
%
P. Degiovanni, ``Moore and Seiberg equations, topological field theories
and Galois theory,'' an article in \cite{200}.
%
\bibitem[tbp]{tng-base-pt}
%
P. Deligne, ``Le groupe fondamental de la droite projective moins trois points,'' in {\it Galois groups over Q}, MSRI Publ. {\bf 16} (89) p.79. \\
%
See also \cite{Ems-Lck} and \cite{Schneps-survey}.
%
\bibitem[Deu]{Deuring}
%
M. Deuring, ``Die Zetafunktion einer algebraischen Kurve vom Geschlecht Eins, I, II, III, IV,'' Nachr. Akad. Wiss. G\"{o}ttingen, (1953) 85--94, (1955) 13--42, (1956) 37--76, and (1957) 55--80. 
%
\bibitem[DS]{DS}
%
F. Diamond and J. Shurman, ``A First Course in Modular Forms,'' GTM 228, 
Springer, 2005.
%
\bibitem[EL]{Ems-Lck}
%
M. Emsalem and P. Lochak, ``The action of the absolute Galois group 
on the moduli space of spheres with four marked points,'' as an appendix 
to the second article in \cite{I-200GT}. 
%
\bibitem[FB-Z]{FBz}
%
E. Frenkel and D. Ben-Zvi ``Vertex Algebras and Algebraic Curves,'' 
Mathematical Surveys and Monographs, vol 88. AMS, 2004.
%
\bibitem[FS]{FS}
%
 D.~Friedan and S.~H.~Shenker,
  ``The Analytic Geometry of Two-Dimensional Conformal Field Theory,''
  Nucl.\ Phys.\ B {\bf 281} (1987) 509.
%
\bibitem[FG-RSS]{Fuchs-Galois-AutoFusionAlg}
%
J.~Fuchs, B.~Gato-Rivera, B.~Schellekens and C.~Schweigert,
  ``Modular invariants and fusion rule automorphisms from Galois theory,''
  Phys.\ Lett.\ B {\bf 334} (1994) 113
  [hep-th/9405153].
%

\bibitem[Gan]{Gannon}
%
T. Gannon, ``Moonshine beyond the Monster'' Cambridge monographs on 
mathematical physics. 2006.
%
\bibitem[GS]{GSchappacher}
%
C. Goldstein and N. Schappacher, 
``Series d'Eisenstein et fonctions L de courbes elliptiques a multiplication complexe.''
 (French) [Eisenstein series and L functions of elliptic curves with complex multiplication] 
J. Reine Angew. Math. 327 (1981), 184--218.


%
\bibitem[Gross]{Gross-LNM}
%
B. H. Gross, 
``Arithmetic on elliptic curves with complex multiplication.''
With an appendix by B. Mazur. Lecture Notes in Mathematics, 776. 
Springer, Berlin, 1980.  iii+95 pp. ISBN: 3-540-09743-0. 

%
\bibitem[Groth]{Esquisse}
%
A. Grothendieck, ``Esquisse d'un programme,'' recorded in \cite{242}.
%
\bibitem[GV]{GV}
%
S.~Gukov and C.~Vafa,
  ``Rational conformal field theories and complex multiplication,''
  Commun.\ Math.\ Phys.\  {\bf 246} (2004) 181
  [hep-th/0203213].
%
\bibitem[HW]{Harvey}
%
 J.~A.~Harvey and Y.~Wu,
  ``Hecke Relations in Rational Conformal Field Theory,''
  JHEP {\bf 1809} (2018) 032
  [arXiv:1804.06860 [hep-th]].
%
\bibitem[Hecke]{Hecke-Werke}
%
E. Hecke, 
``Zur Theorie der elliptischen Modulfunktionen'' 
(Math. Werke, the 23th article, pp.428--460), 
``Bestimmung der Perioden gewisser Integrale duruch 
die Theorie der Klassenk\"{o}per,'' 
(Math. Werke, the 27th article, pp.505--524).
%
%
\bibitem[I-1]{Ihara-86}
%
Y. Ihara, ``Profinite braid groups, Galois representations and complex multiplications,'' 
Ann. of Math. {\bf 123} (1986) pp. 43--106. 
%
\bibitem[I-2]{I-200GT}
%
Y. Ihara, ``Braids, Galois groups, and Some Arithmetic Functions,''
Proc. ICM Kyoto (1990), 99--120; \\
%
"On the embedding of Gal(Qbar/Q) into GT-hat,'' in \cite{200}.
%
\bibitem[Kob]{Koblitz}
%
N. Koblitz, ``Introduction to Elliptic Curves and Modular Forms,''
GTM 97, 1993, Springer.
%
\bibitem[SVOA]{SVOA}
%
V. Kac and M. Wakimoto, ``Unitarizable Highest Weight Representations of the Virasoro, Neveu--Schwarz and Ramond Algebras''in {\it Conformal groups and related symmetries: physical results and mathematical background}, Lect. Notes in Phys. {\bf 261} (1986) 345--371, \\
V. Kac and Weiqiang Wang, ``Vertex Operator Superalgebras and Their Representations,'' Contemp. Math. {\bf 175} (1994) 161--191 [hep-th/9312065].
%
\bibitem[Koh]{Kohno}
%
T. Kohno, ``Topological invariants for 3-manifolds using representations of mapping class groups I,'' Topology {\bf 31} (1992) 203--230.
%
\bibitem[KW]{prev.paper}
%
S.~Kondo and T.~Watari,
  ``String-theory Realization of Modular Forms for Elliptic Curves with Complex Multiplication,''
  Commun.\ Math.\ Phys.\  {\bf 367} (2019) 89
  [arXiv:1801.07464 [hep-th]].
%
\bibitem[Mat]{Matsumoto}
%
M. Matsumoto, ``Difference between Galois representations in automorphism and outerautomorphism groups of a fundamental group,'' Proc. Amer. Math. Soc. {\bf 139} (2011) 1215--1220. 
%
\bibitem[Mil-pd]{Milne-pedestrian}
%
J. S. Milne, 
``Abelian varieties with complex multiplication (for pedestrians)'' 
 arXiv:math/9806172 [math.NT].
%
\bibitem[Mil-et]{Milne-etale}
%
J. S. Milne, ``Etale Cohomology,'' (PMS-33) Princeton University Press (1980).
``Lectures on etale cohomology'' available online:
www.jmilne.org/math/CourseNotes/ 
%
\bibitem[Miy]{Miyake}
%
T.~Miyake, ``{\it Modular Forms},'' Springer, Berlin (1989).
%
\bibitem[Moo]{MooreArth}
%
G.~Moore, ``Arithmetic and attractors,'' hep-th/9807087.
%
\bibitem[MS]{MS}
%
G.~W.~Moore and N.~Seiberg,
  ``Classical and Quantum Conformal Field Theory,''
  Commun.\ Math.\ Phys.\  {\bf 123} (1989) 177.
%
\bibitem[Oort]{Oort}
%
F. Oort, ``The algebraic fundamental group,''  an article in \cite{242}.
%
\bibitem[Ribet]{Ribet-Neben}
%
K. A. Ribet, 
``Galois representations attached to eigenforms with Nebentypus. Modular functions of one variable, V'' 
(Proc. Second Internat. Conf., Univ. Bonn, Bonn, 1976), pp. 17--51. 
Lecture Notes in Math., Vol. 601, Springer, Berlin, 1977. 
%
%
\bibitem[RS]{Ribet-Stein}
%
K. A. Ribet and W. Stein 
``Lectures on Modular Forms and Hecke Operators'',
online book, available at https://wstein.org/
%
\bibitem[SAGE]{SAGE}
%
SageMath, a free open-source mathematics software system available online. 
%
\bibitem[Schap]{SchPoH}
%
N. Schappacher, 
``Periods of Hecke characters.''
Lecture Notes in Mathematics, 1301. Springer-Verlag, Berlin, 1988. xvi+160 pp. 
ISBN: 3-540-18915-7 
%
\bibitem[Schimm]{Schimmrigk-05}
%
R.~Schimmrigk,
  ``Arithmetic spacetime geometry from string theory,''
  Int.\ J.\ Mod.\ Phys.\ A {\bf 21} (2006) 6323
  [hep-th/0510091].
%
\bibitem[Schn-GT]{Schneps-survey}
%
L. Schneps, ``The Grothendieck--Teichmueller group GT-hat: a survey,'' 
in \cite{242}.
%
\bibitem[Schn-btc]{Schneps-PS}
%
L. Schneps, ``Fundamental groupoids of genus zero moduli spaces and 
braided tensor categories,'' as section 2 of 
``Moduli Spaces of Curves, Mapping Class Groups and Field Theory,''
by Buff, Fehrenbach, Lochak, Schneps, and Vogel. 
%
SMF/AMS Texts and Monographs vol. 9.
%
\bibitem[LNS200]{200}
%
L. Schneps, ``The Grothendieck Theory of Dessins d'Enfants,'' London Mathematical Society LNS vol. 200, Cambridge U Press, 1994.
%
\bibitem[LNS242]{242}
%
L. Schneps and P. Lochak, ``Geometric Galois Actions'' 
London Mathematical Society LNS vol. 242, Cambridge U Press, 1997.
%
\bibitem[WeilRep-1]{Weil-repr}
%
B. Schoeneberg, ``Das Verhalten von mehrfachen Thetareihen bein Modulsubstitutionen,'' Math. Ann. {\bf 116} (1939) 511. 
%
N. Scheithauer, ``The Weil representation of SL(2;Z) and some applications,''
Int. Math. Res. Not. {\bf 8} (2009) 1488.
%
\bibitem[Sehg]{Sehgal}
%
S. Sehgal "Units in integral group rings,'' Monographs and Surveys in Pure and Applied Mathematics, Chapman and Hall, 1993.
%
\bibitem[Shim-eCM]{Shimura-ellCM}
%
G. Shimura, 
``On elliptic curves with complex multiplication as factors of the Jacobians of modular function fields.''
Nagoya Math. J. 43 (1971), 199--208. 
%
\bibitem[Shim-zA]{Shimura-zetaA}
%
G. Shimura  ``On the zeta-function of an abelian variety with complex multiplication.'' Annals. Math {\bf 94} (1971) 504--533.  
%
%
\bibitem[Shim-AA]{Shimura-AA}
%
G. Shimura, ``Introduction to the Arithmetic Theory of Automorphic Functions,''
1971, Iwanami/Princeton U. Press.
%
%
%
%
\bibitem[WeilRep-2]{Stroemberg}
%
F. Str\"{o}mberg, 
``Weil representation associated with finite quadratic modules,'' 
Math. Zeitschrift (2013) {\bf 275} 509.
%
\bibitem[TK]{TK}
%
A. Tsuchiya and Y. Kanie, ``Vertex operators in the conformal field theory 
on P1 and monodromy representations of the braid group,'' Lett. Math. Phys. {\bf 13} (1987) 303--312.
%
\bibitem[TUY]{TUY}
%
A. Tsuchiya, K. Ueno, and Y. Yamada, 
 ``Conformal field theory on universal family of stable curves with gauge symmetries,'' in p. 459--566 of Adv. Stud. Pure Math. 
 ``Integrable Systems in Quantum Field Theory and Statistical Mechanics'', 
 Math Soc. Japan, 1989. 
 available from Project Euclid. 
%
\bibitem[BCDT]{Taylor-Wiles}
%
%
C. Breuil, B. Conrad, F. Diamond, and R. Taylor, ``On the modularity of elliptic curves over Q: wilde 3-adic exercises,'' J. Amer. Math. Soc. {\bf 14} (2001) 843--939.
%
\bibitem[Wort]{Wortmann}
%
S. Wortmann, 
``Generalized Q-curves and factors of {$J_1(N)$.}''
Abh. Math. Sem. Univ. Hamburg 70 (2000), 51--61.  
%

\end{thebibliography}
\end{document}